\documentclass[oneside,reqno,10pt,a4paper]{amsart}

\usepackage[margin=1.3in]{geometry}
\usepackage[utf8]{inputenc}

\usepackage[hypertexnames=false,hyperfootnotes=false,colorlinks=true,linkcolor=blue,%
citecolor=purple,filecolor=magenta,urlcolor=cyan,unicode,linktocpage=true]{hyperref}
\usepackage{float}
\usepackage{scrextend}
\usepackage{enumitem}
\usepackage{braket}
\usepackage{amsmath, amssymb}
\usepackage{latexsym}
\usepackage{IEEEtrantools}
\usepackage{amsthm}
\usepackage{stmaryrd}
\SetSymbolFont{stmry}{bold}{U}{stmry}{m}{n} 
\usepackage{mathtools}
\usepackage{tikz-cd}
\usepackage{ytableau}

\usepackage[T1]{fontenc}
\usepackage{anyfontsize}
\usepackage[libertine,libaltvw,liby]{newtxmath}
\usepackage[oldstyle]{libertine}%
\usepackage{courier}
\usepackage{csquotes}
\usepackage{microtype}
\usepackage[backend=biber,style=alphabetic,maxalphanames=5,giveninits,sorting=nyt,%
maxbibnames=99]{biblatex}
\renewbibmacro{in:}{%
  \ifboolexpr{%
     test {\ifentrytype{article}}%
     or
     test {\ifentrytype{inproceedings}}%
  }{}{\printtext{\bibstring{in}\intitlepunct}}%
}
\usepackage[noabbrev]{cleveref}
\expandafter\def\csname ver@etex.sty\endcsname{3000/12/31}

\crefname{part}{Part}{Parts}
\crefname{section}{Section}{Sections}
\crefname{theorem}{Theorem}{Theorems}
\crefname{lemma}{Lemma}{Lemmata}
\crefname{conjecture}{Conjecture}{Conjectures}
\crefname{question}{Question}{Questions}
\crefname{proposition}{Proposition}{Propositions}
\crefname{equation}{Equation}{Equations}
\crefname{definition}{Definition}{Definitions}
\crefname{remark}{Remark}{Remarks}
\newtheorem{theorem}{Theorem}[section]
\newtheorem{lemma}[theorem]{Lemma}
\newtheorem{proposition}[theorem]{Proposition}
\newtheorem{corollary}[theorem]{Corollary}
\newtheorem{conjecture}[theorem]{Conjecture}
\theoremstyle{definition}
\newtheorem{definition}[theorem]{Definition}
\newtheorem{example}[theorem]{Example}

\theoremstyle{remark}
\newtheorem{remark}[theorem]{Remark}

\renewcommand{\thepart}{\Roman{part}}
\usepackage[newparttoc]{titlesec}
\usepackage{appendix}
\usepackage{titletoc}
\contentsmargin{2.5em}
\titleformat{\part}
{\normalfont\LARGE\scshape}{Part \thepart \hspace{0.5em}--}{0.5em}{}
\titlecontents{part}
	[1.5em]{\vspace{5pt}}
	{\contentslabel{1.5em}\scshape} 
	{}
	{\hfill\contentspage}
\titleformat{\section}
{\normalfont\LARGE\scshape}{\thesection}{1em}{}
\titlecontents{section}
	[3.5em]{}
	{\contentslabel{1.5em}}
	{}
	{\hfill\contentspage}
\titleformat{\subsection}
{\normalfont\Large\scshape}{\thesubsection}{1em}{}
\titlecontents{subsection}
	[5.5em]{}
	{\contentslabel{1.5em}\itshape}
	{}
	{\hfill\contentspage}
\titleformat{\subsubsection}
{\normalfont\Large\itshape}{\thesubsubsection}{1em}{}
\titlecontents{subsubsection}
	[7.5em]{}
	{\contentslabel{1.5em}\itshape}
	{}
	{\hfill\contentspage}

\setlist[enumerate,1]{label=(\roman*)}
\newcommand{\define}[1]{\textit{#1}}

\newcommand{\mc}{\mathcal}
\newcommand{\mr}{\mathrm}
\newcommand{\mf}{\mathfrak}
\renewcommand{\C}{\mathbb{C}}
\newcommand{\N}{\mathbb{N}}
\newcommand{\Z}{\mathbb{Z}}
\newcommand{\Q}{\mathbb{Q}}
\renewcommand{\P}{\mathbb{P}}
\newcommand{\Rational}{\mathbb{Q}}
\newcommand{\Complex}{\mathbb{C}}
\renewcommand{\theta}{\vartheta}
\renewcommand{\phi}{\varphi}
\newcommand{\<}{\langle}
\renewcommand{\>}{\rangle}
\newcommand{\plh}{\mathord{\cdot}}

\DeclareMathOperator{\ad}{ad}
 
\DeclareMathOperator{\End}{End}

\DeclareMathOperator{\Spec}{Spec}
\DeclareMathOperator*{\Res}{Res}
\newcommand{\dd}{\mathrm{d}}
\newcommand{\id}{\mathord{\mathrm{id}}}
\newcommand{\WeylAlg}[2]{\mc{D}_{#1}^{#2}}
\newcommand{\WeylAlgComp}[2]{\widehat{\mc{D}}_{#1}^{#2}}
\newcommand{\degProj}[2]{\pi_{#1}\!\left(#2\right)}
\newcommand{\vsum}[2]{\boldsymbol{#1}_{#2}}
\newcommand{\rsum}[1]{\vsum{r}{#1}}
\newcommand{\ssum}[1]{\vsum{s}{#1}}

\newcommand{\vsumind}[2]{\vsum{#1}{[#2]}}
\newcommand{\rsumind}[1]{\vsumind{r}{#1}}
\newcommand{\ssumind}[1]{\vsumind{s}{#1}}

\newcommand{\higherorderterms}[1]{\mc{O}({#1})}

\newcommand{\descPart}{\vdash}

\newcommand{\norder}[1]{{\vcentcolon}\!\mathrel{#1}\!{\vcentcolon}}

\newcommand{\ii}{\mathbf{i}}
\newcommand{\degOneMat}{\mathcal{M}}

\newcommand{\kLowerBound}[3]{\mf{d}^{>}_{\mathbf{#2}, \mathbf{#3}}(#1)}

\newcommand*{\Scale}[2][4]{\scalebox{#1}{$#2$}}

\title{Higher Airy structures and topological recursion for singular spectral curves}

\author[G.~Borot]{Ga\"etan Borot}

\author[R.~Kramer]{Reinier Kramer}

\author[Y.~Sch\"uler]{Yannik Sch\"uler}

\address[G.~B. \& R.K.]{Max-Planck-Institut f\"ur Mathematik, Vivatsgasse 7, 53111 Bonn, Germany}
\address[G.~B.]{Institut f\"ur Mathematik und Institut f\"ur Physik, Humboldt-Universit\"at zu Berlin, Unter den Linden 6, 10099 Berlin, Germany}
\email[G.~B.]{gaetan.borot@hu-berlin.de}

\address[R.~K.]{Department of Mathematical \& Statistical Sciences, University of Alberta, 632 CAB, Edmonton, Alberta, Canada, T6G 2G1.}
\email[R.~K.]{reinier@ualberta.ca}

\address[Y.~S.]{School of Mathematics and Statistics, University of Sheffield, Hicks Building, Hounsfield Road, S3 7RH, Sheffield, United Kingdom.\newline
Bethe Center for Theoretical Physics, Universit\"at Bonn, Nu\ss{}allee 12, 53115 Bonn, Germany.}
\email[Y.~S.]{yschuler1@sheffield.ac.uk}
\thanks{We thank Nezhla Aghaei for participation in an early phase of the project, Ran Tessler for his lights on open $r$-spin theory, and Nitin K. Chidambaram for pointing out typos. We also thank the referee for their comments. G.B. and R.K benefited from the support of the Max-Planck-Gesellschaft. During revision, R.K. was supported by the Natural Sciences and Engineering Research Council of Canada, and the Pacific Institute for the Mathematical Sciences (PIMS). The research and findings may not reflect those of these institutions.\\
The University of Alberta respectfully acknowledges that we are situated on Treaty 6 territory, traditional lands of First Nations and Métis people.}

\begin{document}

\begin{abstract}
We give elements towards the classification of quantum Airy structures based on the $W(\mathfrak{gl}_r)$-algebras at self-dual level based on twisted modules of the Heisenberg VOA of $\mathfrak{gl}_r$ for twists by arbitrary elements of the Weyl group $\mathfrak{S}_{r}$. In particular, we construct a large class of such quantum Airy structures. We show that the system of linear ODEs forming the quantum Airy structure and determining uniquely its partition function is equivalent to a topological recursion \`a la Chekhov--Eynard--Orantin on singular spectral curves. In particular, our work extends the definition of the Bouchard--Eynard topological recursion (valid for smooth curves) to a large class of singular curves, and indicates impossibilities to extend naively the definition to other types of singularities. We also discuss relations to intersection theory on moduli spaces of curves, giving a general ELSV-type representation for the topological recursion amplitudes on smooth curves, and formulate precise conjectures for application in open $r$-spin intersection theory.
\end{abstract}

\maketitle

\newpage

\tableofcontents

\newpage

\section{Introduction}
Introduced by Kontsevich and Soibelman in \cite{KSTR}, Airy structures consist of a system of linear PDEs depending on a parameter $\hbar$ and satisfying a compatibility condition, so that they admit a simultaneous solution which is unique when properly normalised. This solution is called ``partition function'' and is encoded in Taylor coefficients $F_{g,n}$ indexed by integers $n \geq 1$ and (half-)integers $g \geq 0$, which are determined by a topological recursion, i.e. a recursion on $2g - 2 + n$. The interest in Airy structures comes from the numerous applications of the topological recursion in enumerative geometry, see e.g. \cite{EynardICM}.

The purpose of \cref{part:AiryStr} of this work is to construct new Airy structures from representations of the $\mc{W}(\mathfrak{gl}_r)$-algebra. The latter is realised as a sub-VOA of the Heisenberg VOA $\mathcal{F}_{\mathbb{C}^r}$. Twisting the free field representation of $\mathcal{F}_{\mathbb{C}^{r}}$ by an arbitrary element of the Weyl group $\sigma \in \mathfrak{S}_{r}$ gives rise to an (untwisted) representation of $\mathcal{W}(\mathfrak{gl}_r)$. Applying a ``dilaton shift'', \cite{BBCCN18} constructed all the Airy structures that can arise when $\sigma$ is an $r$-cycle or an $(r - 1)$-cycle. Our work explores the possibility of constructing Airy structures from arbitrary $\sigma \in \mathfrak{S}_{r}$: we will obtain in \cref{thm:W_gl_Airy_arbitrary_autom} conditions on $\sigma$ and the dilaton shifts that are sufficient for the success of this construction. This gives rise to many new Airy structures and thus partition functions, for which it would be desirable to find enumerative interpretations. If one contents oneself with the existence of a partition function only solving the PDEs to leading order in $\hbar$ (\textit{i.e.} existence of $F_{g=0,n}$) we give less restrictive conditions for this in \cref{thm:genus0_soln_admissible_Hik}. Our approach and the results are presented in \cref{S2} while \cref{sec:W_gl_airy_structs_proof} is devoted to the proof of the main \cref{thm:W_gl_Airy_arbitrary_autom}.

In \cref{SCpart}, we show that the topological recursion for all these Airy structures can be equivalently formulated via residue (hence, period) computations on possibly singular spectral curves. Roughly speaking, a spectral curve is a branched cover of complex curves $x\,:\,C \rightarrow C_0$ equipped with a meromorphic function $y$ and a bidifferential $\omega_{0,2}$ on $C^2$. The original formulation of the topological recursion, with a spectral curve as input, was developed by Chekhov, Eynard and Orantin \cite{EyOr07,EORev} in the case of smooth curves with simple ramifications. The output of this CEO recursion is a family of multidifferentials $\omega_{g,n}$ on $C^n$ that have poles at ramification points of $x$, are symmetric under permutation of the $n$ copies of $C$, and obtained recursively by residue computations on $C$. As observed already in \cite{EORev} and revisited by \cite{KSTR,ABCO}, the corresponding Airy structure is based on the Virasoro algebra and related to the $\mc{W}(\mathfrak{gl}_2)$-algebra; for each $g,n$, $F_{g,n}$ or $\omega_{g,n}$ contain the same information packaged in a different way.  The definition of the CEO topological recursion was extended in \cite{BHLMR14,BoEy13} to smooth curves with higher order ramification points, and its correspondence with $\mc{W}(\mathfrak{gl}_r)$-Airy structures when $\sigma$ is an $r$-cycle was established in \cite{BBCCN18}. An important application of this correspondence is a conceptual proof of symmetry of the $\omega_{g,n}$ based on representation theory of $\mc{W}$-algebras. This led the discovery of a non-trivial criterion on the order of $y$ at ramification points for the symmetry to hold. From \cite{BHLMR14}, it is easy to propose a definition of the CEO recursion also valid for singular spectral curves (see \cref{TRsum}), i.e. if $C$ has several irreducible components intersecting at ramification points of arbitrary order. It is however unclear (and in fact not always true) that this definition leads to symmetric $\omega_{g,n}$.

In \cref{HASequivALE,mainth2}, we show that this definition naturally arises from the $\mc{W}(\mathfrak{gl}_r)$-Airy structures with arbitrary permutation $\sigma$ encoded the ramification profile over a branchpoint in a normalisation of $C$, and dilaton shifts specifying the order of $y$ at the ramification points. Smooth spectral curves correspond to $\sigma$ having a single cycle. Besides, the basic properties of Airy structures guarantee that the corresponding $\omega_{g,n}$ are symmetric. The results of \cref{part:AiryStr} therefore give sufficient conditions on the ramification type and the order of $y$ at ramification points for the symmetry to hold, see \cref{dereg,deirreg,deexpreg,dall} and \labelcref{dequasi}. The central result of \cref{SCpart} is then \cref{mainth2} giving the correspondence between those Airy structures and our extension of the CEO topological recursion.

In terms of spectral curves, the CEO-like topological recursion provides the (unique when properly normalised) solution to the so-called ``abstract loop equations''. The latter express that certain polynomial combinations of the $\omega_{g,n}$ are holomorphic at the ramification points. This in fact provides tools that have been used to establish applications of the topological recursion e.g. in matrix models \cite{BEO15}, in Hurwitz theory \cite{Hurwitzspin,DBKPS19}, and to the reconstruction of WKB expansions \cite{dxdsys,Marchalwaki,BEMar}. The setup of abstract loop equations was developed in \cite{BEO15,BoSh17} for smooth curves with simple ramifications and extended in \cite[Section 7.6]{Kra19} to higher order ramifications. In Section~\ref{sec:absloop}, we define a notion abstract loop equations for arbitrary spectral curves (\cref{def:absloopgen}) and prove they admit at most one normalised solution (\cref{combinabsloop}), which must then be given by our extension of the CEO recursion (\cref{TRsum}). The question of existence of a solution is then reduced to proving this formula yields symmetric $\omega_{g,n}$. As we establish the equivalence of the abstract loop equations with the differential constraints built from $\mc{W}(\mathfrak{gl}_r)$-algebra representations (\cref{HASequivALE}), it is sufficient to check that the latter form an Airy structure to establish symmetry. In case the associated PDEs admit a partition function solving them to leading order in $\hbar$ this is sufficient to prove symmetry of the $\omega_{0,n}$ (\cref{mainth2gequal0}).

Our results give a definition of topological recursion for many spectral curves that could not be treated before, for instance
\begin{equation}
\begin{split}
& y(x - y^r) = 0 \qquad r \geq 1\,, \\
& (1 - xy^3) = 0\,, \\
& x^2 - y^2 = 0 \\
& (1 - xy^3)(x - y + 1) = 0\,, \\
& (1 - xy^3)(x - y + 1)(x - y + t) = 0\qquad t \in \mathbb{C}\setminus\{1\} \,,\\
& (1 - xy^3)(x - y + 1)(y - t) = 0\qquad t\in \mathbb{C}\,,\\
& (1 - xy^2)(x - y + 1)(x - ty) = 0\qquad t \in \mathbb{C}\setminus\{1\}\,. \\
\end{split}
\end{equation}
In general, nodal points such that $y$ tends to distinct non-zero values on both sides are admissible. if $C_i\,:\,P_i(x,y) = 0$ are admissible for $i \in \{1,2\}$ do not intersect in $\mathbb{C}^2$ at zeroes of $\dd x$ or singular points of $C_i$, then $C\,: P_1(x,y)P_2(x,y) = 0$ is admissible. The complete list of conditions defining admissible spectral curves according to our work can be found in \cref{SecTR}, and they only regulate the behavior of $x$ and $y$ at the points where the cardinality of the fibers of $x$ jump (zeroes of $\dd x$ and singular points). Examples of non-admissible curves are:
\begin{equation}
\begin{split}
& 1 - x^2y^5 = 0\,, \\
& (x - y^2)(1 - xy^2) = 0\,, \\
&  (x - y)(x - y^r) = 0\qquad r \in \mathbb{Z}\,, \\
& x^p - y^q = 0 \qquad p,q\,\,{\rm coprime}\,\,{\rm and}\,\,p > 1, q > 0\,, \\
& (x - y^2)^2 = 0\,, \\
&  (x - y^2)(x - y^3) = 0\,.
\end{split}
\end{equation}
In general, the following cases are not admissible for us: non-reduced curves, curves where $\dd y$ and $\dd x$ have a common zero, reducible curves where there is a point at which at least two irreducible components $C_1$ and $C_2$ meet and where $\dd x$ has a zero and $y$ is regular on $C_1$ and $y$ is not identically zero on $C_2$.
It would be desirable to understand if the admissibility conditions can be weakened even more with a suitable modification of the topological recursion residue formula.

We expect all the partition functions of the Airy structures present in this article to admit an enumerative interpretation, i.e. that $\omega_{g,n}$ or $F_{g,n}$ can be computed via intersection theory on a certain moduli space of curves. This was achieved by Eynard in \cite{Eyn11,Eyn14} in the case of simple ramification points on smooth curves and $y$ holomorphic, in a form that has a structure similar to the ELSV formula \cite{ELSV}, and found applications in Gromov--Witten theory \cite{EyOr15} and Hurwitz theory \cite{SSZ15,KLPS17}. The case of $y$ having a simple pole was later treated by Chekhov and Norbury \cite{Nor17,CheNor19}. \Cref{part:appl} explores the generalisations of this link to other spectral curves that can directly be reached by combining known results with the results of \cref{part:AiryStr,SCpart}. This stresses the role of Laplace-type integrals on the spectral curve. Although it is not essential in the theory, in the case of global spectral curves the Laplace transform of $\omega_{0,2}$ enjoy a factorisation property reminiscent to the use of $R$-matrices in the theory of Frobenius manifolds. This was known by \cite[Appendix B]{Eyn14} for smooth spectral curves with simple ramifications, and we show in \Cref{compactfactospin} that it extends to singular spectral curves in a slightly different form. 

In \cref{higheret}-\cref{ELSVwgnsec}, building on \cite{BBCCN18} and \cref{mainth2} we generalise Eynard's formula of \cite{Eyn11,Eyn14} to all smooth spectral curves with arbitrary ramification and $y$ of order $1$ at the ramification points --- this involves Witten spin classes. In the case of global smooth spectral curves (i.e. obeying the precise conditions \cref{presequde}) it leads in \Cref{ELSVwgnCor} to an intersection-theoretic formula for decomposition of the $\omega_{g,n}$ on the descendants of the primary differentials. We deduce in \Cref{TRESLVcor} that the coefficients of expansion of $\omega_{g,n}$ near simple poles of $\dd x$ are given by an ELSV-like formula and displaying a quasi-polynomiality property. This answers a question of Shadrin to the first-named author. In \cref{sec:open}, we apply our general results to the $\mc{W}(\mathfrak{gl}_3)$-constraints of Alexandrov \cite{Ale15} for the open intersection theory developed in \cite{PST15,Bur15,BT17,ST}. The open intersection numbers can be packaged in a generating series $\omega_{g,n}^{{\rm open}}$. As Alexandrov's $\mc{W}(\mathfrak{gl}_3)$-constraints have been identified with an Airy structure in \cite{BBCCN18} for $\sigma = (12)(3)$, we deduce from the results of \cref{SCpart} that $\omega_{g,n}^{{\rm open}}$ satisfies our extension of the CEO topological recursion applied to the curve $y(y^2 - 2x) = 0$ (\Cref{coroKP}). Using modified $\mc{W}$-constraints, Safnuk had derived in \cite{Saf16} a residue formula associated to this curve, but its structure is different and not easily generalisable. On the other hand, we can easily conjecture a residue formula for the open $r$-spin theory. This conjecture is equivalent to the $\mc{W}(\mathfrak{gl}_r)$ constraints first mentioned in \cite[Section 6.3]{BBCCN18}: we expose it in more detail, in particular specifying the normalisations necessary for the comparison, and give some support in its favor by comparison with \cite{BCT18}.

\begin{remark} At the time of writing, there are several foundational conjectures in open intersection theory. In \cref{OpenIntKP} and \cref{courserspin}, we formulate the ones that are directly relevant for us and explain the logical dependence of our statements on these conjectures. 
\end{remark}

In fact, one of the initial motivation of our work was to generalise the structure of Safnuk's residue formula \cite{Saf16} to higher $r$, and seek along this line for a definition of topological recursion for curves with many irreducible components. Our conclusion is that, although we do not know how to generalise the structure of Safnuk's recursion, there is a simpler and general definition of the recursion, which retrieves open intersection theory when applied to the reducible curve $y(y^2 - 2x) = 0$.

\medskip

The $\mc{W}(\mathfrak{gl}_r)$-representations that we consider have an explicit though lengthy expression. We extract from them a few concrete calculations:
\begin{itemize}
\item General formulas for $F_{0,3}$, $F_{\frac{1}{2},2}$ and $F_{1,1}$, and partial computations for $F_{0,4}$ in the undeformed case. Their symmetry gives necessary constraints for $\sigma$ and the dilaton shifts. Remarkably, the symmetry conditions governing $F_{0,3}$ and $F_{0,4}$ exactly match those found in \cref{thm:genus0_soln_admissible_Hik} which are sufficient for all $(F_{0,n})_{n\geq 0}$ to be symmetric. Taking the additional symmetry conditions of $F_{\frac{1}{2},2}$ for a generic choice of $F_{\frac{1}{2},1}$ into account we find that these are only slightly weaker than the sufficient conditions under which we have an Airy structure according to \cref{thm:W_gl_Airy_arbitrary_autom}. Hence, we are inclined to think that also this result is generically optimal. This could be implied by the analysis of the symmetry of $F_{\frac{1}{2},n}$ for higher $n \geq 3$.
\item If $\sigma = (1\cdots r -1)(r)$ or $(1\cdots r)$, we derive in \cref{SecHom} from the $\mathcal{W}$-constraints a homogeneity property, the dilaton equation, and the string equation when it applies.
\end{itemize}

\medskip

While this work was in the final writing stage, we learned that results similar to our \cref{thm:W_gl_Airy_arbitrary_autom} but restricted to the case of cycles of equal lengths (perhaps with a fixed point) are obtained in an independent work of Bouchard and Mastel~\cite{BM20}. The question of defining a Chekhov-Eynard-Orantin topological recursion on singular curves mentioned in their work is solved by \cref{SCpart} of our work.

\medskip

\subsection*{Notation}
$\N$ is the set of nonnegative integers and $\N^* = \N \setminus \{0\}$. For $n \in \N$, we denote $[n] = \{ 1, \ldots, n\}$, and in particular $ [0] = \emptyset$. More generally, if $a \leq b$ are integers, we denote $[a,b]$, $[a,b)$, $(a,b]$, $(a,b)$ the integer segments, where the bracket (resp. parenthesis) means we include (resp. exclude) the corresponding endpoint. If $a > b$ we set $[a,b] = \emptyset$.
Furthermore $ z_{[n]} = \{ z_1, \ldots, z_n\}$. \par
A partition of $r \in \N$, denoted $ \lambda \vdash r$, is $\lambda=(\lambda_1,\ldots,\lambda_\ell)$ such that $\lambda_1+\cdots+\lambda_\ell = r$. We will sometimes (but not always) require $\lambda_j \geq \lambda_{j+1}$, in which case we say that $\lambda$ is a \emph{descending partition}. Often it is convenient to express equal blocks $\lambda_j=\lambda_{j+1}=\cdots=\lambda_{j+n}$ as $\lambda_j^n$. Moreover, to any descending partition $\lambda$ one can associate a Young diagram $\mathbb{Y}_{\lambda}$ in a bijective way. For example all notations
\begin{equation*}
\lambda=(4,3,3,1) \quad \longleftrightarrow \quad \lambda=(4,3^2,1) \quad  \longleftrightarrow \quad  \mathbb{Y}_\lambda=\Scale[0.5]{\ytableausetup{centertableaux} \begin{ytableau}%
	\, & \, & \, & \, \\
	\, & \, & \, \\
	\, & \, & \, \\
	\, \\
	\end{ytableau}}
\end{equation*}
characterise the same descending partition $\lambda \descPart 11$. The \emph{size} of $ \lambda$ is $ |\lambda| = \sum_i \lambda_i $ and its \emph{length} is $\ell (\lambda )= \max \{ i \mid \lambda_i > 0\} $. Given partitions $\lambda^1,\ldots,\lambda^n$ of integers $r_1,\ldots,r_n$ we will write $(\lambda^1,\ldots,\lambda^n)\coloneqq (\lambda^1_1,\ldots,\lambda^1_{\ell(\lambda^1)},\lambda^2_1,\ldots,\lambda^n_{\ell(\lambda^n)})$ for their concatenation, which is a partition of $\sum_{j=1}^n r_j$.\par
If $A$ is a finite set, we write $\mathbf{L} \vdash A$ to say that $\mathbf{L}$ is a partition of $A$, that is an unordered tuple of pairwise disjoint non-empty subsets of $A$ whose union is $A$. We denote $|\!|\mathbf{L}|\!|$ the number of sets in the partition $\mathbf{L}$. \par
All of our vector spaces or algebraic spaces are over $\C$. We denote $\mathbb{C}\langle y \rangle$ a $1$-dimensional complex vector space equipped a non-zero linear form $y$.

\newpage

\part{Classification of \texorpdfstring{$\mc{W}(\mf{gl}_r)$}{W(glr)}-Airy structures}
\label{part:AiryStr}

\section{Constructing Airy structures}
\label{S1}
\label{S2}

\medskip

In this \namecref{S1}, we recall the definition of Airy structures and their partition function and its adaptation to the infinite-dimensional setting for which it will be used. We present the main results of \cref{part:AiryStr}, i.e. we exhibit the new Airy structures that can be constructed from $\mc{W}(\mf{gl}_r)$-algebra modules, while the proofs are carried out in  \cref{sec:W_gl_airy_structs_proof}.

\medskip

\subsection{Preliminaries on Airy structures}

\medskip

\subsubsection{Finite dimension}
\label{sec:AiryStructsFinDim}
\medskip

We first present the definition when $E$ is a finite-dimensional $\mathbb{C}$-vector space. For convenience, let us fix a basis $(e_a)_{a \in A}$ of $E$ and $(x_a)_{a \in A}$ be the dual basis of $E^*$.  We consider the graded algebra of differential operators $\WeylAlg{E}{\hbar}$, also called \textit{Weyl algebra}. It is the quotient of the free algebra generated by $\hbar^{\frac{1}{2}}$, $(x_a)_{a \in A}$ and $(\hbar \partial_{x_a})_{a \in A}$,  modulo the relations generated by 
\[
[\hbar \partial_{x_a},x_b] = \hbar \, \delta_{a,b}\,,\quad [x_a,x_b]=[\hbar\partial_{x_a},\hbar\partial_{x_b}] = 0 \qquad a,b\in I,\qquad \hbar\,\,{\rm central}.
\]
and equipped with the grading
\[
\deg x_a = \deg \hbar \partial_{x_a} = \deg \hbar^{\frac{1}{2}} = 1\,,
\]
We will write $P = Q + \higherorderterms{n}$ for two elements $P,Q\in \WeylAlg{E}{\hbar}$ if they agree up to at least degree $n-1$ and use the notation $P = Q + o(\hbar^n)$ if $P-Q\in \hbar^n \WeylAlg{E}{\hbar}$.

\begin{definition}
	\label{defn:higher_airy_struct}
	A family $(H_i)_{i \in I}$ of elements of $\WeylAlg{E}{\hbar}$ is an \emph{Airy structure on $E$ in normal form} with respect to the basis $(x_a)_{a \in A}$ if $I = A$ and it satisfies:
		\begin{enumerate}
		\item[$\bullet$] \textit{The degree $1$ condition}: for all $i \in I$, we have
		\begin{equation}
		\label{eq:deg_one_cond}
		H_i= \hbar \partial_{x_i} + \higherorderterms{2}\,.
		\end{equation}
		\item[$\bullet$] \textit{The subalgebra condition}: there exist $f^k_{i,j}\in \WeylAlg{E}{\hbar}$ such that for all $i,j\in I$
		\begin{equation}
		\label{eq:gr_lie_subalg_cond}
		[H_i,H_j] = \hbar \sum_{k \in A} f^k_{i,j} \, H_k\,.
		\end{equation}
	\end{enumerate}
A family $(H_i)_{i \in I}$ of elements of $\WeylAlg{E}{\hbar}$ is an \emph{Airy structure} if there exists two matrices $\mc{N} \in \mathbb{C}^{A \times I}$ and $\mc{M} \in \mathbb{C}^{I \times A}$ such that
\begin{equation}
\label{inversMS}
\begin{split}
\forall a,b \in A,& \qquad \sum_{i \in I} \mc{N}_{a,i}\mc{M}_{i,b} = \delta_{a,b} \\
\forall i,j \in I,& \qquad \sum_{a \in A} \mc{M}_{i,a}\mc{N}_{a,j} = \delta_{i,j}
\end{split}
\end{equation}
and the family $\tilde{H}_a = \sum_{i \in I} \mc{N}_{a,i}\,H_i$ indexed by $a \in A$ is an \emph{Airy structure in normal form}. 
\end{definition}

Being an Airy structure does not depend on a choice of basis, but being an Airy structure in normal form does. Airy structures in \cref{defn:higher_airy_struct} would be called in \cite{BBCCN18} ``crosscapped higher quantum Airy structure''. ``Crosscapped'' refers to the presence of half-integer powers of $\hbar$ and we comment it in \cref{nohalf}. ``Higher'' means that compared to the definition in \cite{KSTR,ABCO}, it can contain terms of degree higher than $2$. ``Quantum'' is used to distinguish it in \cite{KSTR,ABCO} from the classical Airy structures where the Weyl algebra is replaced by the Poisson algebra of polynomial functions on $T^*E$. We simplified the terminology as the restriction to maximum degree $2$ and the classical Airy structures will not play any role in this article and handling half-integer powers of $\hbar$ does not lead to any complication in the theory.

The essential property of an Airy structure is that it specifies uniquely a formal function on $E$.

\begin{theorem} \cite[Theorem 2.4.2]{KSTR}, \cite[Proposition 2.20]{BBCCN18}
	\label{thm:partition_fct_exists}
	If $(H_i)_{i\in I}$ is an Airy structure on $E$, then the system of linear differential equations
	\begin{equation}
	\label{HihIH}\forall a \in A,\qquad Z^{-1} H_a Z \cdot 1 = 0
	\end{equation}
	admits a unique solution $Z$ of the form $Z= \exp(F)$ with
	\begin{equation}
	\label{eq:free_energy_of_Airy_struct}
	F = \sum_{\substack{g \in \frac{1}{2}\mathbb{Z}_{\geq 0} \\ n \in \mathbb{Z}_{>0} \\ 2g+2-n>0}} \frac{\hbar^{g - 1}}{n!}\,F_{g,n}\,\qquad F_{g,n} \in {\rm Sym}^n(E^*)
	\end{equation}
\end{theorem}

$Z$ (resp. the $F_{g,n}$) is called the partition function (resp. free energies). Equation \eqref{HihIH} implies a recursive formula for $F_{g,n}$ on $2g - 2 + n> 0$. We will typically be interested in the coefficients of decomposition of the free energies on a given basis $(x_a)_{a \in A}$ of linear coordinates of $E$, for which we use the notation:
\begin{equation}
\label{Fgnanudgsgn}
F_{g,n} = \sum_{a_1,\ldots,a_n \in A} F_{g,n}[a_1,\ldots,a_n] x_{a_1}\cdots x_{a_n}
\end{equation}
When the Airy structure has normal form with respect to this basis, the explicit recursion to obtain the $F_{g,n}[a_1,\ldots,a_n]$ appears in \cite[Sections 2.2 and 2.3]{BBCCN18}. We reproduce it in \cref{theFGNsum} at the only place where it is used in the article. We will be led to work with Airy structures that are not given in normal form (cf. \cref{sec:deg_one_part}), but for which there is an equivalent formulation of the recursion in terms of spectral curves (\cref{SCpart}), that is often more efficient for calculations (cf. \cref{sec:free_energ_arb_autom}).

\medskip

\subsubsection{Infinite dimension}
\label{sec:infinit}

\medskip

In this article we need to handle Airy structures for certain infinite-dimensional vector spaces. This requires some amendments of the previous definitions which we now explain.

A filtered vector space is a vector space $E$ together with a collection of subspaces $ 0 \subseteq \mc{F}_1E \subseteq \mc{F}_2E \subseteq \dotsb \subseteq E$, called the filtration. Throughout this paper we will assume that for any filtered vector space $E$, the $\mc{F}_pE$ are finite-dimensional, and that $E = \bigcup_{p > 0} \mc{F}_pE$. Two filtrations $ \mc{F}$, $\mc{F}'$ on a vector space $E$ are equivalent if for any $p > 0$ there exists $p' > 0$ such that $\mc{F}_pE \subseteq \mc{F}'_{p'}E$ and for any $q' > 0$ there exists $q > 0$ such that $\mc{F}'_{q'}E \subseteq \mc{F}_{q}E$. In particular, all filtrations on a given vector space satisfying our extra assumptions are equivalent.\par
 A filtered set is a set $A$ together with a collection of subsets $ \emptyset \subseteq f_1A \subseteq f_2A \subseteq \dotsb \subseteq A$. Again, we assume all $f_pA$ are finite and $ A = \bigcup_{p>0} f_pA$. A filtered basis of a filtered vector space $(E,\mc{F})$ is the data of a filtered set $(A,f)$ and a family $(e_a)_{a \in A}$ of elements of $E$ such that $(e_a)_{a \in f_pA}$ is a basis of $\mc{F}_pE$.

Let $(E,\mc{F})$ be a filtered vector space, and for convenience choose a filtered basis $(e_a)_{a \in A}$ and the corresponding linear coordinates $(x_a)_{a \in A}$. We consider the completed Weyl algebra with respect to this filtration, $\WeylAlgComp{E}{\hbar}$. It consists of elements of the form
\begin{equation}
\label{decomConp}
\sum_{\substack{m,n \in \mathbb{Z}_{\geq 0} \\ j \in \frac{1}{2}\mathbb{Z}_{\geq 0}}} \sum_{a_1,\ldots,a_n \in A} \sum_{\overline{a}_1,\ldots,\overline{a}_n \in A} \frac{\hbar^j}{n!m!}\,C^{(j)}[\overline{a}_1,\ldots,\overline{a}_m ; a_1,\ldots,a_n]\,x_{\overline{a}_1} \cdots x_{\overline{a}_m} \hbar\partial_{x_{a_1}} \cdots \hbar \partial_{x_{a_n}}\,,
\end{equation}
where for any $p>0$, the coefficients $C^{(j)}[\overline{a}_1,\ldots,\overline{a}_m; a_1,\ldots,a_n]$ vanish for all but finitely many $ a_1, \dotsc, a_n \in f_pA$ and $j,\overline{a}_1,\ldots,\overline{a}_m$. One can check this does form a (graded) algebra.

If $(I,f')$ is a filtered set (which is unrelated to $(A,f)$), we say that a family $(D_i)_{i \in I}$ of elements of $\WeylAlgComp{E}{\hbar}$ is filtered if for any $ p > 0$, the coefficients $C^{(j)}_i[\overline{a}_1,\ldots,\overline{a}_m; a_1,\ldots,a_n]$ in the decomposition \eqref{decomConp} of $D_b$ vanish for all but finitely many $ a_1, \dotsc, a_n \in f_pA$ and $i,j,\overline{a}_1,\ldots,\overline{a}_n$.

\begin{definition}
\label{deffilter}A filtered family $(H_i)_{i \in I}$ of elements of $\WeylAlgComp{E}{\hbar}$ is an Airy structure on $E$ in normal form with respect to $(x_a)_{a \in A}$ if $(I,f') = (A,f)$,  and
\begin{itemize}
\item[$\bullet$] the degree one condition holds ;
\item[$\bullet$] the subalgebra condition holds for some filtered family $(f^k_{i,j})_{i,j,k \in I}$ of elements of $\WeylAlgComp{E}{\hbar}$.
\end{itemize}
A filtered family $(H_i)_{i \in I}$ of elements of $\WeylAlgComp{E}{\hbar}$ is an Airy structure on $E$ if there exist $\mc{N} \in \mathbb{C}^{A \times I}$ and $\mc{M} \in \mathbb{C}^{I \times A}$ such that:
\begin{itemize}
\item[$\bullet$] for each $i \in I$, $\mc{N}_{a,i}$ vanish for all but finitely many $a \in A$ ;
\item[$\bullet$] for each $a \in A$,  $\mc{M}_{i,a}$ vanish for all but finitely many $i \in I$ ;
\item[$\bullet$] $\mc{N}$ and $\mc{M}$ are inverse to each other in the sense of \eqref{inversMS} --- where the sums are finite due to the previous two points ;
\item[$\bullet$] the family $\tilde{H}_a = \sum_{i \in I} \mc{N}_{a,i}H_i$  indexed by $a \in A$ ---  which is a filtered family of elements of $\WeylAlgComp{E}{\hbar}$ by the first point --- is an Airy structure on $E$ in normal form with respect to $(x_a)_{a \in A}$.
\end{itemize}
\end{definition}

The notion of Airy structure does not depend on the choice of filtered basis, and only depends on the equivalence class of the filtration of $E$. We will sometimes omit to specify the filtration when it is evident. The existence and uniqueness of the partition function (\cref{thm:partition_fct_exists}) extends to this infinite-dimensional setting and for each $g$, $n$, and $p$, the summation of \eqref{Fgnanudgsgn} with $ a_1$ restricted to $ f_pA$ is finite, that is $F_{g,n} \in \widehat{{\rm Sym}^n(E^*)}$.

\medskip

\subsection{Preliminaries on the \texorpdfstring{$\mc{W}(\mf{gl}_r)$}{W(glr)}-algebra and its twisted modules}

\medskip

The Airy structures constructed in this paper are obtained by considering twisted modules of Heisenberg vertex operator algebras (VOAs), taking subalgebras of the associated algebras of modes, and using a dilaton shift to break homogeneity. The idea of this construction dates back to \cite{Mil16}. It was developed more systematically in \cite[Sections 3 and 4]{BBCCN18} and we refer to that paper for details. Here we summarise the main points of the construction.

\medskip

\subsubsection{The Heisenberg VOA and  \texorpdfstring{$\mc{W}(\mf{gl}_r)$}{W(glr)}-algebras}

\medskip

\begin{definition}
Let $ \mf{h} $ be a finite-dimensional vector space with non-degenerate inner product $ \< \plh, \plh \>$. The \emph{Heisenberg Lie algebra} associated to $\mf{h}$ is given by
\begin{equation*}
\hat{\mf{h}} = \big( \mf{h} [t^{\pm 1}] \oplus \C K\big)  \otimes \C_{\hbar} \,, \qquad [\xi_l+ aK, \eta_m + bK] = \hbar\,\< \xi, \eta \>l \delta_{l+m,0}K\,,
\end{equation*}
where we write $ \xi_l \coloneqq \xi \otimes t^l$, and $ \C_{\hbar} \coloneqq \C[\hbar^{\frac{1}{2}}]$. The associated \emph{Weyl algebra} is defined as a quotient of its universal enveloping algebra: $ \mc{H}_\mf{h} \coloneqq \mf{U}(\hat{\mf{h}})/(K-1)$. The \emph{Fock space} $\mc{F}_{\mf{h}}$ is the representation of $\mc{H}_{\mf{h}}$ generated by a vector $ |0\> $ and relations $ \xi_l |0\> = 0$ for $ \xi \in \mf{h}$, $ l \geq 0$. The \emph{Heisenberg VOA} is the vertex operator algebra with underlying vector space $\mc{F}_\mf{h}$ equipped with vacuum $ |0\> $, state-field correspondence $ Y \colon \mc{F}_\mf{h} \to \End \mc{F}_\mf{h} \llbracket \tilde{x}^{\pm 1} \rrbracket $ given by
\begin{equation}
\label{productformY} 
\begin{split}
Y\big(|0\>, \tilde{x}\big) &= \id_{\mc{F}_\mf{h}}\,, \\
Y\big(\xi_{-1}|0\>, \tilde{x}\big) &= \sum_{l \in \Z} \xi_l \tilde{x}^{-l-1}\,, \\
Y\big(\xi^1_{-k_1} \cdots \xi^n_{-k_n} |0\>, \tilde{x}\big) &= \norder{\frac{1}{(k_1-1)!} \frac{\dd^{k_1-1}}{\dd \tilde{x}^{k_1-1}} Y(\xi^1_{-1}|0\>, \tilde{x}) \cdots \frac{1}{(k_n-1)!} \frac{\dd^{k_n-1}}{\dd \tilde{x}^{k_n-1}} Y(\xi^n_{-1}|0\>,\tilde{x})}\,,
\end{split}
\end{equation}
and conformal vector $ \varpi = \frac{1}{2} \sum_{j} \chi^j_{-1} \chi^j_{-1} |0\> $ for an orthonormal basis $(\chi^j)_{j}$ of $ \mf{h}$.
\end{definition}

The $\mc{W}$-algebra associated to the general linear Lie algebra $ \mf{gl}_r$ at the self-dual level is denoted $\mc{W}(\mathfrak{gl}_r)$. It can be constructed as a sub-VOA of the Heisenberg VOA $\mc{F}_{\mathfrak{h}}$ attached to its Cartan subalgebra $ \mf{h} \cong \mathbb{C}^{r} \subset \mf{gl}_r$. We identify $ \mf{h}$ with its dual using the Killing form and take $ \chi^j \in \mf{h}$ to correspond to the roots under this identification. Note that they are orthonormal.

\begin{theorem}[\cite{Fateev:1988,AM17}]
$\mc{W}(\mf{gl}_r)$ is the sub-VOA of $\mc{F}_\mf{h}$ freely and strongly generated by the elements
\begin{equation}
\label{strongenw} w_i = {\rm e}_i \big(\chi^1_{-1},\ldots,\chi^{r}_{-1}\big) \ket{0},\quad i\in [r]\,,
\end{equation}
where ${\rm e}_i$ is the $i$th elementary symmetric polynomial in $r$ variables.
\end{theorem}
We introduce the modes $W_{i,k}$ and their ($i$-form valued) generating series $W_i(\tilde{x})$ with the formulas
\begin{equation}
\label{Wgenerating}
W_i(\tilde{x}) \coloneqq \sum_{k \in \Z} \frac{W_{i,k} (\dd \tilde{x})^i}{\tilde{x}^{i+k}} \coloneqq Y(w_i , \tilde{x})(\dd \tilde{x})^{i}\,.
\end{equation}

\begin{remark}
\label{notationWik}
Contrarily to \cite{BBCCN18}, we do not include a factor of $r^{i - 1}$ in the definition of $w_i$. Our convention that $W_{i,k}$ is the coefficient of $\tilde{x}^{-(k + i)}$. This coincides with the convention taken in \cite[Section 3.3.4]{BBCCN18} but differs from the convention $\tilde{x}^{-(k + 1)}$ used in the rest of \cite{BBCCN18}. We also find convenient to consider the generating series of modes of conformal weight $i$ to be $i$-differential forms. The variable $\tilde{x}$ is often denoted $z$ in the VOA literature. However, in \cref{SCpart}, we will see that this variable can be interpreted as the pullback under the normalisation morphism of a function usually denoted $x$ for a spectral curve. For consistency, we therefore chose to use the letter $\tilde{x}$, and use the letter $z$ for local coordinates on (the normalisation of) the spectral curve.
\end{remark}

\medskip

\subsubsection{The mode algebra and its subalgebras}

\medskip

Let $\mc{A}$ be the associative algebra of modes of $\mc{W}(\mf{gl}_r)$, see \cite[Section 4.3]{FB04}.  Furthermore, let $L(\mc{A})$ denote the set of possibly infinite sums of ordered monomials in $\mc{A}$ whose degree and conformal weight is bounded;  we equip it with the bracket $\hbar^{-1}[\cdot,\cdot]$, making it a Lie algebra.

\begin{definition}
We say that a subset $S \subset [r] \times \mathbb{Z}$ of modes generates a Lie subalgebra of $L(\mc{A})$ if the left $\mc{A}$-ideal generated by $\oplus_{(i,k) \in S} W_{i,k}$ is a Lie subalgebra of $L(\mc{A})$. An equivalent condition is the existence of $f^{(j,l)}_{(i,k),(i',k')} \in L(\mc{A})$ such that
	\begin{equation*}
	\forall (i,k),(i',k') \in S,\qquad [ W_{i,k}, W_{i',k'} ] = \hbar \sum_{(j,l)\in I_\lambda} f^{(j,l)}_{(i,k),(i',k')} ~ W_{j,l}\,.
	\end{equation*}
	More precisely, it is  first required that the right-hand side defines an element of $L(\mc{A})$, and then that it coincides with the left-hand side.
\end{definition}

Given a partition $\lambda \vdash r$, we set
\[
\lambda(i)\coloneqq \min \bigg\{m \geq 0 \quad \Big|\quad  \sum_{j=1}^m \lambda_j \geq i\bigg\}\,,
\]
and we define the index set
\begin{equation}
\label{eq:part_to_index_set}
I_{\lambda} \coloneqq \{ (i,k) \in [r] \times \mathbb{Z} \quad | \quad  \lambda(i) + k > 0\}\,.
\end{equation}
\begin{theorem}
	\label{prop:desc_part_Lie_subalg}
	\cite[Theorem 3.16]{BBCCN18} For any descending partition $\lambda \vdash r$, $I_{\lambda}$ generates a Lie subalgebra of $L(\mc{A})$.
\end{theorem}

\medskip

\subsubsection{Twisted modules}
\label{sec:twisted_modules}

\medskip

Let $\sigma\in \mathfrak{S}_{r}$ be an arbitrary element of the Weyl group of $\mf{gl}_r$. It is a permutation of the elements $\chi^i$ and can thus be decomposed into $d\geq1$ cycles $\sigma=\sigma_1\cdots \sigma_d$ with each cycle $\sigma_\mu$ of length $r_\mu \geq 1$ such that $r_1+\cdots+r_d = r$. If $d=1$ then $\sigma$ is a transitive element. After relabelling of the elements of the basis of the Cartan, we can assume that $\sigma_\mu$ acts as
\[
\sigma_\mu : \quad \chi^{1+\rsumind{\mu-1}} \longrightarrow \chi^{2+\rsumind{\mu-1}} \longrightarrow \cdots \longrightarrow \chi^{r_\mu+\rsumind{\mu-1}} \longrightarrow \chi^{1+\rsumind{\mu-1}}\,,
\]
while keeping all other $\chi^i$ fixed, where we introduced the notation $\rsumind{\mu}\coloneqq\sum_{\nu=1}^\mu r_\nu$. It is then easy to check that
\[
v^{\mu,a}\coloneqq \sum_{j=1}^{r_\mu} \theta_{r_\mu}^{-aj} \chi^{j+\rsumind{\mu-1}} \,, \qquad \mu \in [d],\qquad a \in [r_{\mu}], \qquad \theta_{r_\mu}\coloneqq e^{2\ii\pi/r_\mu}
\]
is an eigenvector of the action of $\sigma$ on $\mathfrak{h}$, with eigenvalue $\theta_{r_\mu}^{a}$. We define the \emph{currents} via the state field-correspondence
\begin{equation}
\label{DefCurrents}
J^{\mu,a}(\tilde{x}) := {}^{\sigma}Y\big(v^{\mu,a}_{(-1)} \ket{0},\tilde{x}\big)\dd \tilde{x} \coloneqq  \sum_{k \in \frac{a}{r_\mu}+\mathbb{Z}} J^\mu_{r_\mu k} \frac{\dd \tilde{x}}{\tilde{x}^{k+1}}
\end{equation}
with fractional mode expansion on these eigenvectors, after introducing the differential operators
\begin{equation}
\label{repJdiff} J^\mu_k = \left\{\begin{array}{lll} \hbar \partial_{x^\mu_k} & & {\rm if}\,\,k > 0 \\ \hbar^{\frac{1}{2}}Q_\mu & & {\rm if}\,\,k = 0 \\ -k\, x^\mu_{-k} & & {\rm if}\,\,k < 0 \end{array}\right.\,,
\end{equation}
with $Q_\mu \in \Complex$, acting on the space $\mc{T}\coloneqq \Complex_{\hbar^{\frac{1}{2}}}[(x^\mu_a)_{\mu \in [d],\,\, a>0}]$. Note that the formal variables $ \tilde{x}$ and $x^\mu_i$ are unrelated. The state-field correspondence ${}^{\sigma}Y$ can be extended to the whole space $\mc{F}_{\mathfrak{h}}$ by using formula \eqref{productformY}. This turns $(\mc{T},{}^{\sigma}Y)$ into a twisted representation of $\mc{F}_{\mathfrak{h}}$, whose restriction to $\mc{W}(\mf{gl}_r)$ becomes an (untwisted) representation of $\mc{W}(\mathfrak{gl}_r)$.  For details see \cite{Doyon,BBCCN18}. We define the \define{twist modes} $ {}^{\sigma}W_{i,k} $ by
\begin{equation*}
{}^{\sigma}W_i(\tilde{x}) \coloneqq {}^{\sigma}Y(w_i, \tilde{x})\, (\dd \tilde{x})^{i} = \sum_{k \in \Z} {}^{\sigma}W_{i,k} \frac{ (\dd \tilde{x})^i}{\tilde{x}^{k+i}}\,,
\end{equation*}
where $w_i$ are the strong generators of $\mc{W}(\mf{gl}_r)$ from \eqref{strongenw}. They are differential operators acting on $\mc{T}$.

\begin{lemma} \cite[Proposition 4.5 and Lemma 4.15]{BBCCN18}.
\label{TwistedWmodes}
	For an automorphism $\sigma$ with $d$ cycles of respective lengths $r_\mu$ the $\sigma$-twisted modes read
	\begin{equation}
	\label{eq:twist_mode_arbitr}
	{}^{\sigma}W_{i,k} = \sum_{M \subseteq [d]} \sum_{\substack{i_{\mu} \in [r_\mu],\,\mu \in M \\ \sum_{\mu} i_{\mu} = i}} \sum_{\substack{\mathbf{k} \in \mathbb{Z}^{M} \\ \sum_{\mu} k_{\mu} = k}} \prod_{\mu \in M} \frac{1}{r_\mu^{i_{\mu} - 1}}\,W^{\mu}_{i_\mu,k_\mu}\,,
	\end{equation}
	where the $W^{\mu}_{i_\mu,k_\mu}$ are defined as
	\begin{equation}
	\label{eq:twist_mode_coxeter}
	 W^{\mu}_{i_\mu,k_\mu} = \frac{1}{r_\mu} \sum_{j_\mu = 0}^{\lfloor i_\mu/2 \rfloor} \frac{i_\mu!\,\hbar^{j_\mu}}{2^{j_\mu}j_\mu!(i_\mu - 2j_\mu)!} \sum_{\substack{p^\mu_{2j_{\mu} + 1},\ldots,p^\mu_{i_\mu} \in \mathbb{Z} \\ \sum_{l} p^\mu_{l} = r_\mu k_\mu }} \Psi^{(j_\mu)}_{r_\mu}(p^\mu_{2j_\mu + 1},p^\mu_{2j_\mu + 2},\ldots,p^\mu_{i_\mu})\,\,\,\norder{\prod_{l = 2j_\mu + 1}^{i_\mu} J^\mu_{p^\mu_{l}}}\,,
	\end{equation}
	with coefficients $\Psi^{(j_\mu)}_{r_\mu}(\cdots) \in \Rational$ admitting a representation in terms of sums over $r$th roots of unity:
\begin{equation}
\label{eq:psidef}
\Psi^{(j)}_r (a_{2j+1},\ldots, a_i) \coloneqq  \frac{1}{i!} \sum_{\substack{m_1, \ldots , m_{i}=0 \\ m_{l} \neq m_{l'}}}^{r-1} \left( \prod_{l'=1}^{j}\frac{\theta^{m_{2l' - 1}+m_{2l'}}}{(\theta^{m_{2l'}} - \theta^{m_{2l' - 1}})^2}   \prod_{l=2j+1}^{i}\theta^{-m_{l} a_{l}}\right)\,,
\end{equation}
where $\theta=e^{2\ii\pi/r}$.
\end{lemma}
In terms of generating series, this lemma can be restated as
\begin{align}
{}^{\sigma}W_{i}(\tilde{x}) & = \sum_{M \subseteq [d]} \,\,\sum_{\substack{i_{\mu} \in [r_{\mu}]\,\,\mu \in M \\  \sum_{\mu} i_{\mu} = i}} ~~\prod_{\mu \in M} \frac{W^{\mu}_{i_{\mu}}(\tilde{x})}{r_\mu^{i_\mu - 1}}\,, \nonumber \\
 W^\mu_{i}(\tilde{x}) & \coloneqq \frac{1}{r_\mu} \sum_{j = 0}^{\lfloor i/2 \rfloor} \frac{i!\,\hbar^{j}}{2^{j}j!(i - 2j)!} \bigg(\frac{\dd \tilde{x}}{\tilde{x}}\bigg)^{2j}\bigg(\sum_{a_{2j + 1},\ldots,a_{i} = 0}^{r_{\mu} - 1} \Psi^{(j)}_{r_\mu}(a_{2j + 1},\ldots,a_{i})\,\, \norder{ \prod_{l = 2j + 1}^{i} J^{\mu,a_l}(\tilde{x})}\,\,\bigg)\,. \label{Wikz}
\end{align}
If we also take a generating series in $i$ by defining
\begin{equation*}
\mc{W}^\mu(\tilde{x},u) \coloneqq \frac{u^{r_\mu}}{r_\mu} + \sum_{i = 1}^{r_\mu} W^\mu_i (\tilde{x}) u^{r_\mu -i}\,, \qquad \mc{W}(\tilde{x},u) \coloneqq u^r + \sum_{i = 1}^r {}^{\sigma}W_i(\tilde{x}) u^{r -i}\,,
\end{equation*}
this can be written compactly as
\begin{equation*}
\mc{W}(\tilde{x},u) = \prod_{\mu = 1}^{d} r_\mu \mc{W}^\mu (\tilde{x},u)\,.
\end{equation*}

Introducing the filtered vector space
\begin{equation}
\label{Evectspace}
E = \bigoplus_{k > 0} \bigoplus_{\mu = 1}^d \mathbb{C}\langle x_k^\mu \rangle\,,\qquad E_p \coloneqq \bigoplus_{k \leq p} \bigoplus_{\mu = 1}^d \mathbb{C}\langle x_k^\mu \rangle\,,
\end{equation}
we see from the condition of summations in \eqref{eq:twist_mode_arbitr}-\eqref{eq:twist_mode_coxeter} that each $W_{i,k}$ belongs to the completed Weyl algebra $\WeylAlgComp{E}{\hbar}$ according to the definitions in \cref{sec:infinit}. Even more, any family $(W_{i,k})_{(i,k) \in I}$ for which $\min_I k > -\infty$ is a filtered family of elements of $\WeylAlgComp{E}{\hbar}$.  Our goal is to construct Airy structures from these operators. The degree one condition requires the operators to have the  form $\hbar \partial + \higherorderterms{2}$. Unfortunately, ${}^{\sigma}W_{i,k}$ does not have this form as it is homogenous of degree $\deg {}^{\sigma}W_{i,k} = i$ in view of \eqref{eq:twist_mode_arbitr}. Thus in order to construct an Airy structure we first have to break up this homogeneity.

\medskip

\subsection{Airy structures from twisted \texorpdfstring{$\mc{W}(\mf{gl}_r)$}{W(glr)} modules}

\medskip

We can now summarise the main findings of \Cref{part:AiryStr}. Before presenting the new basic Airy structures we have found in Section~\ref{sec:W_gl_Airy_structs_all_shift}, we review in Section~\ref{sec:ninrign} the ones known from \cite{BBCCN18}. Section~\ref{sec:AllDilatonPolarize} gives the general form of the Airy structures that can be obtained by deformation of those basic ones through conjugation in the Weyl algebra, which will be important for the correspondence with spectral curves in \Cref{SCpart}. This provides so far \emph{sufficient conditions} on the parameters of our construction so that it does produce Airy structures. It is natural to wonder if these conditions are also necessary (classification problem). In Section~\ref{sec:NecCond} we summarise the conditions we have found \emph{necessary} to get Airy structures.

\medskip

\subsubsection{From a twist with a transitive element}

\medskip

\label{sec:ninrign}

First let $\sigma$ be a transitive element, i.e. $d=1$. Then ${}^{\sigma}W_{i,k} =r^{1-i} W^1_{i,k}$, and to simplify the notation we can omit the superscript $1$. One can break up the homogeneity of the differential operators ${}^{\sigma}W_{i,k}$ by performing a \define{dilaton shift}
\begin{equation}
\label{eq:dilaton_shift_coxeter}
J_{-s} \longrightarrow J_{-s} - 1
\end{equation}
for a fixed $s>0$ while keeping all other $J_a$ for $a\neq -s$ unchanged. Formally one defines
\begin{equation*}
{}^{\sigma}H_{i,k} \coloneqq  \hat{T}\cdot {}^{\sigma}W_{i,k} \cdot \hat{T}^{-1}\,, \qquad \hat{T} \coloneqq  \exp\left(-\frac{J^\mu_{s}}{\hbar s}\right) \,.
\end{equation*}
It follows from the Baker--Campbell--Hausdorff formula that conjugating with $\hat{T}$ means shifting the $J$s as in \eqref{eq:dilaton_shift_coxeter}. The action on the completed Weyl algebra $\WeylAlgComp{E}{\hbar}$ is then well-defined. Then, certain subsets of the modes ${}^{\sigma} H_{i,k}$ yield Airy structures.

\begin{theorem}
	\label{thm:W_gl_Airy_Coxeter}
	\cite[Theorem 4.9]{BBCCN18}
	Let $\sigma \in \mf{S}_r $ be transitive and $s\in [r + 1]$ with $r = \pm 1 \,\,{\rm mod}\,\, s$ and $J_0=0$. Let us define $\lambda^{r,s} \descPart r$ to be the descending partition
		\begin{equation}
	\label{eq:corresp_part_coxeter_Airy}
	\lambda^{r,s} = \left\{\begin{array}{lll}
	(r) & & {\rm if}\,\,s=1 \\
	(1^r) & & {\rm if}\,\,s = r+1 \\
	\big((r'+1)^{r''}, (r')^{s-r''}\big) & & {\rm if}\,\,r=r's + r''\,\,{\rm with}\,\,r''\in\{1,s-1\}
\end{array}\right.\,.
	\end{equation}
	Then, the family of operators
	\begin{equation*}
	r^{i - 1}\, {}^{\sigma}H_{i,k} \,\qquad i\in [r], \quad  k\geq 1 - \lambda^{r,s}(i) + \delta_{i,1}
	\end{equation*}
	forms an Airy structure in normal form on $E = \bigoplus_{k > 0} \mathbb{C}\langle x_k \rangle$, seen as a filtered vector space when equipped with the filtration $E_p = \bigoplus_{k \leq p} \mathbb{C}\langle x_k \rangle$.
\end{theorem}

The partition $\lambda^{r,s}$ chosen in \eqref{eq:corresp_part_coxeter_Airy} to define the mode set determines the subalgebra associated to this Airy structure by using \cref{prop:desc_part_Lie_subalg}. The corresponding Young diagrams are depicted in \cref{tab:coxeter_partition_diagrams}.

\begin{table}[t]
	\centering
		{\def\arraystretch{2.2}
			\begin{tabular}{|c|@{}p{2pt}@{}|c|c|c|c|}
				\cline{1-1}\cline{3-6}
				$r,s$ && $s=1$  & $r=r's + 1$ & $r=r's + s - 1$ & $s=r+1$ \\
				\cline{1-1}\cline{3-6}
				$\lambda^{r,s}$ &&
				\begin{minipage}{0.14\textwidth}
					\centering
					\vspace{0.7em}
				$\overbrace{\Scale[0.6]{\begin{ytableau}%
					\, & \, & \, & \none[\dots] & \, & \,\\
					\none \\
					\none \\
					\none \\
					\none \\
					\none \\
					\end{ytableau}}}^{r\text{ boxes}}$
				\vspace{0.7em}
				\end{minipage} &
				\begin{minipage}{0.14\textwidth}
			\centering
			\vspace{0.7em}
			$\overbrace{\Scale[0.6]{\begin{ytableau}%
					\, & \, & \none[\dots] & \, & \, \\
					\, & \, & \none[\dots] & \, \\
					\, & \, & \none[\dots] & \, \\
					\none[\vdots] & \none[\vdots] & \none[\ddots] & \none[\vdots] \\
					\, & \, & \none[\dots] & \, \\
					\, & \, & \none[\dots] & \, \\
					\end{ytableau}}}^{(r' + 1)\text{ columns}}$
				\vspace{0.7em}
				\end{minipage} & 
				\begin{minipage}{0.14\textwidth}
					\centering
					\vspace{0.7em}
					$\overbrace{\Scale[0.6]{\begin{ytableau}%
					\, & \, & \none[\dots] & \, & \, \\
					\, & \, & \none[\dots] & \, & \, \\
					\, & \, & \none[\dots] & \, & \, \\
					\none[\vdots] & \none[\vdots] & \none[\ddots] & \none[\vdots] & \none[\vdots] \\
					\, & \, & \none[\dots] & \, & \, \\
					\, & \, & \none[\dots] & \, \\
						\end{ytableau}}}^{(r' + 1)\text{ columns}}$
					\vspace{0.7em}
				\end{minipage} &
				\begin{minipage}{0.14\textwidth}
					\centering
					\vspace{0.7em}
					$\Scale[0.6]{\begin{ytableau}%
					\none \\
					\none \\
					\, \\
					\, \\
					\, \\
					\none[\vdots] \\
					\, \\
					\, \\
					\end{ytableau}}$
					\vspace{0.7em}
				\end{minipage} \\
				\cline{1-1}\cline{3-6}
			\end{tabular}}
		\vspace{0.3cm}
		\caption{The partitions $\lambda^{r,s}$ associated to different values $s$.}
		\label{tab:coxeter_partition_diagrams}
\end{table}

\medskip

\subsubsection{A generalisation to arbitrary twists}
\medskip
\label{sec:W_gl_Airy_structs_all_shift}
Let $\sigma\in\mf{S}_r$ be a permutation with $d$ cycles of respective lengths $r_1,\ldots,r_d$, so that $r_1+\cdots+r_d=r$. Then the differential operators ${}^{\sigma}W_{i,k}$ act on the space
\begin{equation*}
\Complex_{\hbar^{\frac{1}{2}}}\big[(x^\mu_a)_{\mu\in [d],\,\,a>0}\big]\,.
\end{equation*}
Again, we will break up the homogeneity of ${}^{\sigma}W_{i,k}$ by performing a dilaton shift. There are $d$ independent families of variables $(x^\mu_a)_{a>0}$ labelled by $\mu \in [d]$ in which we can perform the shift. Two types of shifts will in fact lead to Airy structures:
\begin{itemize}
	\item simultaneous shifts in each of the $d$ sets of variables.
	
	\item simultaneous shifts in all but one set of variables and the label $\mu$ of the unshifted set of variables correspond to a fixed point $r_\mu=1$.
\end{itemize}
Let us thus choose $s_\mu \in \N^* \cup \{ \infty \} $ and, for each $\mu$ such that $ s_\mu \neq \infty $, $t_\mu\in\Complex^*$, for each $\mu\in [d]$ and define
\begin{equation}
\label{eq:dilaton_arbitr_autom}
{}^{\sigma}H_{i,k} \coloneqq  \hat{T} \cdot {}^{\sigma}W_{i,k} \cdot \hat{T}^{-1}\,, \qquad \hat{T} \coloneqq  \prod_{\mu \in [d]\,\,:\,\,s_\mu \neq \infty} \exp\left( -r_{\mu}t_\mu \frac{J^\mu_{s_{\mu}}}{\hbar s_{\mu}} \right)\,.
\end{equation}
We chose to include a normalisation factor $r_{\mu}$ to simplify later computations.  Remember that conjugating with $\hat{T}$ means nothing but shifting
\begin{equation*}
\forall \mu\in[d],\quad k\in\mathbb{Z},\qquad J^\mu_{-k} \longrightarrow J^\mu_{-k}  -r_\mu t _\mu \, \delta_{k, s_\mu}.
\end{equation*}
and its action on the completed Weyl algebra $\WeylAlgComp{E}{\hbar}$ is well-defined. It turns out that, with the right choice of parameters, certain subsets of these operators ${}^{\sigma} H_{i,k}$ indeed form an Airy structure.

\begin{theorem}
	\label{thm:W_gl_Airy_arbitrary_autom}
	Let $ d \geq 2$, $ r_1, \ldots, r_d \geq 1$ and $ s_1, \ldots, s_d \in \N^* \cup \{ \infty \} $ be such that
	\begin{equation}
	\label{eq:fracordering}
	\frac{r_1}{s_1} \geq \frac{r_2}{s_2} \geq \cdots \geq \frac{r_d}{s_d}.
	\end{equation}
Let $Q_1,\ldots,Q_{d} \in \mathbb{C}$, $t_1,\ldots,t_{d - 1} \in \mathbb{C}^*$, and if $s_{d} \neq \infty$ also $t_{d} \in \mathbb{C}^*$. Assume that
\[
\sum_{\mu = 1}^{d} Q_{\mu} = 0,\qquad t_{\mu}^{r_{\mu}} \neq t_{\nu}^{r_{\nu}}\,\,\,{\rm whenever}\,\, \mu \neq \nu\,\,{\rm and}\,\,(r_{\mu},s_{\mu}) = (r_{\nu},s_{\nu}).
\]
Define $ r = \sum_{\mu=1}^d r_\mu $ and let $ \sigma \in \mf{S}_r $ be a permutation made of disjoint cycles of respective lengths $ (r_1, \ldots, r_d)$. Assume that
\begin{itemize}
\item $r_1 = -1\,\,{\rm mod}\,\, s_1$.
\item $ s_\mu = 1$ for any $ \mu \notin \{1, d\}$;
\item $r_d = 1 \,\,{\rm mod}\,\, s_{d}$.
\end{itemize}
and define $\lambda \descPart r$ to be the descending partition 
	\begin{equation}
\label{thepartlam}	\lambda =\begin{cases}
	\big({(r_{1}'+1)}^{s_{1}} , r_2 , r_3 , \ldots , r_{d-1} , {r_{d}'}^{s_{d}}\big) & {\rm if}\,\,r_d\neq 1 \\
	\big({(r_{1}'+1)}^{s_{1}} , r_2 , r_3 , \ldots , r_{d-1}\big) & {\rm if}\,\,r_d=1
	\end{cases}\,,
	\end{equation}
where we set $r_{\mu}'\coloneqq \bigl\lfloor \frac{r_\mu}{s_\mu} \bigr\rfloor$. Then, the family
\[
{}^{\sigma}H_{i,k}\,,\qquad  i \in [r],\qquad k \geq 1 - \lambda(i) + \delta_{i,1}
\]
is an Airy structure (not necessarily in normal form) on the filtered vector space $E$ given in \eqref{Evectspace}.
\end{theorem}
	We call the case $ s_d = \infty $ the \emph{exceptional case} and the other case the \emph{standard case}. The proof of the \nameCref{thm:W_gl_Airy_arbitrary_autom} will be presented in \cref{sec:W_gl_airy_structs_proof}. The exceptional $d = 2$ case was already obtained in \cite[Theorem 4.16]{BBCCN18}.
	
\begin{remark}
Note that the conditions imply that $ s_1 \in [r_1+1]$ and if $ r_d >1$, $s_d \in [r_d-1]$.
\end{remark}

The partition $\lambda$ in \eqref{thepartlam} can be depicted as
\begin{equation}
\label{eq:arb_autom_corresp_diagram}
\includegraphics[scale=1]{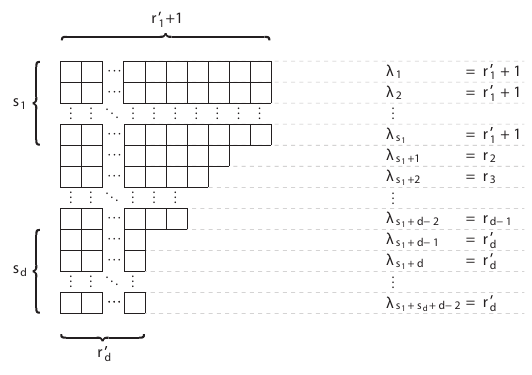}
\end{equation}
In case $r_d=1$ the last block ${r_d'}^{s_d}$ is simply absent. Going through all cases one thus finds that every descending partition $\lambda\descPart r$ is either of the form depicted in \cref{tab:coxeter_partition_diagrams} or of the form \eqref{eq:arb_autom_corresp_diagram}. This implies that all the subalgebras mentioned in \cref{prop:desc_part_Lie_subalg} support two Airy structures: one standard and one exceptional.

\medskip

\Cref{thm:W_gl_Airy_arbitrary_autom} guarantees us the existence of a partition function $Z$ solving the system of differential equations \eqref{HihIH} associated to the operators ${}^{\sigma}H_{i,k}$ but especially the conditions on the choice of the integers $(r_\mu,s_\mu)_{\mu=1}^d$ are rather restrictive. It turns out that if we content ourselves with the existence of a partition function that solves the differential equations \eqref{HihIH} only up to corrections in $\hbar^{\frac{1}{2}}$ we can relax the conditions on the input data drastically.
\begin{theorem}
	\label{thm:genus0_soln_admissible_Hik}
	Suppose the $(t_\mu)_{\mu = 1}^{d}$ satisfy the condition stated in \cref{thm:W_gl_Airy_arbitrary_autom} while $(r_\mu,s_\mu)_{\mu = 1}^{d}$ are subject to
	\begin{itemize}
		\item $s_\mu = \infty$ for at most one $\mu\in[d]$ and in this case $r_\mu=1$.
		\item  $r_\mu = \pm 1 \,\,{\rm mod}\,\, s_\mu$ for all $\mu\in[d]$.
		
		\item For all $\mu_1 \neq \mu_2$ with $s_{\mu_i}>2$ such that either 
		\begin{equation*}
			r_{\mu_1} = 1 \,\,{\rm mod}\,\, s_{\mu_1} \text{ and } r_{\mu_2} = 1 \,\,{\rm mod}\,\, s_{\mu_2}
		\end{equation*}
		or
		\begin{equation*}
			r_{\mu_1} = -1 \,\,{\rm mod}\,\, s_{\mu_1} \text{ and } r_{\mu_2} = -1 \,\,{\rm mod}\,\, s_{\mu_2}
		\end{equation*}
		one has $\big\lfloor \tfrac{r_{\mu_1}}{s_{\mu_1}} \big\rfloor \neq \big\lfloor\tfrac{ r_{\mu_2}}{s_{\mu_2}}\big\rfloor$.
		
		\item If there are pairwise distinct $\mu_1,\mu_2,\mu_3\in[d]$ with $\big\lfloor\frac{r_{\mu_1}}{s_{\mu_1}}\big\rfloor = \big\lfloor\frac{r_{\mu_2}}{s_{\mu_2}}\big\rfloor = \big\lfloor\frac{r_{\mu_3}}{s_{\mu_3}}\big\rfloor$, then there is an $m\in \{1,2,3\}$ for which $s_{\mu_m}=1$.
	\end{itemize}
	Then the family
	\[
	{}^{\sigma}H_{i,k}\,,\qquad  i \in [r],\qquad k \geq 1 - \lambda(i) + \delta_{i,1}
	\]
	satisfies the degree one condition of an Airy structure and there is a unique solution $Z=\exp \hbar^{-1}F_0$ to
	\begin{equation}
		\label{eq:Hik_diff_eq_const_ord}
		\forall i \in [r],\, k \geq 1 - \lambda(i) + \delta_{i,1} \qquad Z^{-1} \, {}^{\sigma}H_{i,k} \, Z \cdot 1 = o(\hbar^{\frac{1}{2}})  
	\end{equation}
	with $F_0=\sum_{n\geq 3} \frac{1}{n!}F_{0,n}$ where $F_{0,n} \in {\rm Sym}^n(E^*)$.
\end{theorem}
Here, the partition $\lambda$ specifying the mode set is obtained from the concatenation $(\lambda^{r_1,s_1},\ldots,\lambda^{r_d,s_d})$, which is build from the partitions $\lambda^{r_\mu,s_\mu}$ specified in \Cref{tab:coxeter_partition_diagrams}, by shifting a box from the first row of $\lambda^{r_{\mu+1},s_{\mu+1}}$ to the last row of $\lambda^{r_{\mu},s_{\mu}}$ for each $\mu\in[d-1]$. In this construction we are assuming that we chose a labelling satisfying \eqref{eq:fracordering}. See also equation \eqref{eq:lambda_decomposed} for an explicit description of this partition. The proof of \cref{thm:genus0_soln_admissible_Hik} can be found in \cref{sec:genuszerosolutionexists}.

\medskip

\subsubsection{Arbitrary dilaton shifts and changes of polarisation}
\label{sec:AllDilatonPolarize}
\medskip
In order to connect with the theory of the Chekhov--Eynard--Orantin topological recursion, we ought to be able to conjugate the ${}^{\sigma}W_{i,k} $ with more general operators, inducing dilaton shifts in several of the variables and also making a \emph{change in polarisation}. This section is completely parallel to \cite[Section 4.1.5]{BBCCN18}.

First, let us consider a general dilaton shift
\begin{equation*}
 \hat{T} \coloneqq \exp\left(\sum_{\substack{\mu \in [d] \\ k > 0}} \Big(\hbar^{-1} F_{0,1}\big[\begin{smallmatrix} \mu \\ -k \end{smallmatrix}\big] + \hbar^{-\frac{1}{2}} F_{\frac{1}{2},1}\big[\begin{smallmatrix} \mu \\ -k \end{smallmatrix}\big]\Big)\,\frac{J^{\mu}_{k}}{k}\right)\,,\qquad s_\mu \coloneqq \min \big\{ k  > 0 \quad \big| \quad  F_{0,1}\big[\begin{smallmatrix} \mu \\ -k \end{smallmatrix}\big] \neq 0 \big\}.
\end{equation*} 
with arbitrary scalars $F_{h,1}\big[\begin{smallmatrix} \mu \\ -k \end{smallmatrix}\big]$ for $h \in \{0,\tfrac{1}{2}\}$. The seemingly complicated way to denote these scalars will become natural in \Cref{SCpart}, see e.g. \cref{locun}. Effectively, this shifts
\[
J_{-k}^{\mu} \rightarrow J_{-k}^{\mu} + F_{0,1}\big[\begin{smallmatrix} \mu \\ -k \end{smallmatrix}\big] + \hbar^{\frac{1}{2}}\,F_{\frac{1}{2},1}\big[\begin{smallmatrix} \mu \\ -k \end{smallmatrix}\big]
\]
and by construction of the completed Weyl algebra, its action on $\WeylAlgComp{E}{\hbar}$ is well-defined. It should be interpreted as a deformation of the case where there is a single non-zero shift
\begin{equation}
\label{F01tum}
F_{0,1}\big[\begin{smallmatrix} \mu \\ -s_\mu \end{smallmatrix}\big] = - r_{\mu}t_{\mu} \,.
\end{equation}
The $F_{\frac{1}{2},1}$ give an extra possible deformation as we have allowed half-integer powers of $\hbar$.

 \par
Another deformation we would like to consider is the \emph{change of polarisation}, given by conjugation with the operator
\begin{equation*}
\hat{\Phi} = \exp\left( \frac{1}{2\hbar} \sum_{\substack{\mu,\nu \in [d] \\ k,l > 0}}  F_{0,2}\big[\begin{smallmatrix} \mu & \nu \\ -k & -l \end{smallmatrix}\big]\,\frac{J_k^\mu J_l^\nu}{kl} \right)\,.
\end{equation*} 
where $F_{0,2}\big[\begin{smallmatrix} \mu & \nu \\ -k & -l \end{smallmatrix}\big]$ are arbitrary scalars. Effectively, it shifts
\begin{equation}
\label{PolChange}
J^\mu_{-k} \to J^\mu_{-k} + \sum_{\substack{\nu \in [d]\\ l > 0}} F_{0,2}\big[\begin{smallmatrix} \mu & \nu \\ -k & -l \end{smallmatrix}\big]\,\frac{J^\nu_l}{l}\,.
\end{equation}
Once again, the action on $\WeylAlgComp{E}{\hbar}$ is well-defined. We want to consider the conjugated operators
\begin{equation}
\label{HAllDilatonPolarize}
{}^{\sigma}H_{i,k} \coloneqq \hat{\Phi} \hat{T}\cdot {}^{\sigma}W_{i,k} \cdot \hat{T}^{-1} \hat{\Phi}^{-1}\,.
\end{equation}

\begin{theorem}\label{prop:AiryAllDilatonPolarize}
Defining $t_{\mu}$ by \eqref{F01tum} and with the same conditions for $d$, $(r_\mu, s_\mu, t_\mu, Q_{\mu})_{\mu = 1}^{d}$ and the same range for $ (i,k)$ as in \cref{thm:W_gl_Airy_arbitrary_autom}, the operators ${}^{\sigma}H_{i,k}$ in \cref{HAllDilatonPolarize} form an Airy structure on $E$.
\end{theorem}

\begin{remark}
	\label{rem:AiryAllDilatonPolarizegenuszero}
	In the same vein, the result of \cref{thm:genus0_soln_admissible_Hik} with the exact same conditions on $(r_\mu, s_\mu, t_\mu)_{\mu = 1}^{d}$ holds for the operators \eqref{HAllDilatonPolarize} as well. 
\end{remark}

These results are proved in \cref{sec:W_gl_airy_structs_proof} and will be reformulated in terms of spectral curves in \cref{globalsp}.  

\par

\subsubsection{Necessary conditions}
\label{sec:NecCond}

\medskip

It turns out that for a family of differential operators $({}^\sigma H_{i,k})$ of the type considered in \Cref{sec:W_gl_Airy_structs_all_shift}, i.e.\ $F_{0,1}\big[\begin{smallmatrix} \mu \\ k \end{smallmatrix}\big] = - r_{\mu}t_{\mu}\delta_{k,-s_\mu}$, the implication in \Cref{thm:genus0_soln_admissible_Hik} can be reversed.
\begin{proposition}
	\label{prop:gen_zero_nec=suff}
	For a family $({}^\sigma H_{i,k})_{k>-\lambda(i)+\delta_{i,1}}$ of differential operators as defined in \eqref{eq:dilaton_arbitr_autom} there exists a solution $Z$ to equation \eqref{eq:Hik_diff_eq_const_ord} if and only if $(r_\mu,s_\mu,t_\mu)_{\mu=1}^d$ satisfy the conditions in \Cref{thm:genus0_soln_admissible_Hik}.
\end{proposition}
To prove this proposition we will exploit a reinterpretation of the results of this section in terms of spectral curves which will be presented in \Cref{globalsp}. The proof of this proposition can be found in \cref{secneccons}.\par

Motivated by \Cref{prop:gen_zero_nec=suff} it is natural to ask whether the sufficient conditions stated in \Cref{thm:W_gl_Airy_arbitrary_autom} for the family of differential operators to be an Airy structure are also necessary to impose in order to guarantee an all order solution to the associated system of differential equations. We will investigate the analogous question phrased in the setting of spectral curves in \cref{sec:free_energ_arb_autom} by analysing the symmetry of the multidifferentials $\omega_{0,3}$, $\omega_{0,4}$, and $\omega_{\frac{1}{2},2}$ which are the counterparts of the free energies $F_{0,3}$, $F_{0,4}$ and $F_{\frac{1}{2},2}$. Doing so we can prove that for generic values of $(t_\mu,Q_\mu)_{\mu=1}^d$ most of the conditions stated in \Cref{thm:W_gl_Airy_arbitrary_autom} are also necessary and we believe that a more thorough analysis would actually lead to the conclusion that they are all indeed necessary.

\begin{proposition}
\label{propneccons} Let $ d \geq 2$, $ r_1, \ldots, r_d \geq 1$ and $ s_1, \ldots, s_d \in \N^* \cup \{ \infty \} $ be such that $\frac{r_1}{s_1} \geq \cdots \geq \frac{r_d}{s_d}$ and choose $\lambda$ as in \eqref{thepartlam}. Assume that for all $Q_1,\ldots,Q_{d} \in \mathbb{C}$, $t_1,\ldots,t_{d - 1} \in \mathbb{C}^*$, and in case $s_{d} \neq \infty$ also $t_{d} \in \mathbb{C}^*$ such that
\[
\sum_{\mu = 1}^{d} Q_{\mu} = 0,\qquad {\rm and}\qquad t_{\mu}^{r_{\mu}} \neq t_{\nu}^{r_{\nu}}\,\,\,{\rm whenever}\,\, \mu \neq \nu\,\,{\rm and}\,\,(r_{\mu},s_{\mu}) = (r_{\nu},s_{\nu})\,,
\]
the operators $({}^\sigma H_{i,k})_{k>-\lambda(i)+\delta_{i,1}}$ defined in \eqref{eq:dilaton_arbitr_autom} form an Airy structure. Then necessarily
\begin{itemize}
	\item $r_1 = -1\,\,{\rm mod}\,\, s_1$ ;
	\item $ s_\mu \in\{1,2\}$ for all $ \mu \notin \{1, d\}$ ;
	\item $r_d = 1 \,\,{\rm mod}\,\, s_{d}$.
\end{itemize}
Moreover for $d>2$, if $(r_{\mu},s_{\mu})=(r_{\nu},s_{\nu})$ for $\mu\neq\nu$ then necessarily $s_{\mu}=s_{\nu}=1$.
\end{proposition}
Again we need to postpone the proof of this proposition to \cref{secneccons} since we first have to establish the required tools. Besides that, in the course of the proof of \cref{thm:W_gl_Airy_arbitrary_autom} in \cref{sec:W_gl_airy_structs_proof}, we will see that coprimality of $r_\mu$ and $s_\mu$ and non-resonance condition for the $t_\mu$ (\cref{rem:t_cond}), as well as non-vanishing of all but maybe one dilaton shift (\cref{rem:generality_number_shifts}) are obvious necessary conditions to obtain Airy structures with our method. The assumption $\sum_{\mu = 1}^d Q_\mu = 0$ may not always be necessary and we obtain finer information on this in \cref{prop:symm_cond_omega1half2}, but we adopted it here to simplify the statement of \cref{propneccons}.

\medskip

\subsubsection{Half-integer or integer powers of \texorpdfstring{$\hbar$}{hbar}?}

\medskip

\label{nohalf}

There are several reasons to allow half-integer powers of $\hbar$ in Airy structures instead of just integer power. Our construction admits natural extra degrees of freedom when $\sigma$ has at least two cycles, namely the parameters $Q$ in Theorem~\ref{thm:W_gl_Airy_arbitrary_autom}. This is relevant for applications to open intersection theory, where we have to allow indices $g$ to be both integer or \textit{half-integer} --- see \cref{OpenIntKP,ZKPAiry}.  As it does not lead to any complication, we write the whole article allowing $g \in \frac{1}{2}\mathbb{Z}_{\geq 0}$. The possibility of half-integer genus was already addressed in \cite{BBCCN18}, which called such Airy structures crosscapped. If all monomials in $H_i$ only feature integer powers of $\hbar$ then $F_{g,n}$ in \eqref{eq:free_energy_of_Airy_struct} vanishes for half-integers $g$. It is therefore straightforward to specialise our results to allow only integer $g$, as is more common in topological recursions. Let us however note that half-integer $g$ already made their appearance in certain other applications of the topological recursion, such as non-hermitian matrix models \cite{CE06}, enumeration of non-orientable discretised surfaces \cite{CEMq1}, Chern--Simons theory with gauge groups ${\rm SO}(N)$ or ${\rm Sp}(2N)$ \cite{BESeifert}, etc. 

\medskip
\subsection{String, dilaton and homogeneity equations}
\medskip
In this section, we give more explicit forms of some of the lowest order constraints of the Airy structures found in \cref{thm:W_gl_Airy_Coxeter} and \cref{thm:W_gl_Airy_arbitrary_autom}. In the particular case of a transitive twist or its exceptional analogue, we obtain the dilaton equation. With a transitive twist and $ s_1 = r_1+1$ or its exceptional analogue --- i.e. for the Airy structures already known from \cite{BBCCN18} --- we also get a string equation. The string equation does not occur in any of the other (new) Airy structures.

\label{SecHom}
\begin{lemma}
\label{lem:homogen0}
Consider one of the Airy structure of \cref{thm:W_gl_Airy_Coxeter} or \cref{thm:W_gl_Airy_arbitrary_autom}. For any $g,n \geq 0$ such that $2g - 2 + (n + 1) > 0$, any $\mu,\mu_1,\ldots,\mu_n \in [d]$ and $p_1,\ldots,p_n > 0$, we have
\begin{equation}
\label{thedesdil}\sum_{\mu = 1}^d t_\mu F_{g,n + 1}\big[\begin{smallmatrix} \mu & \mu_1 & \cdots & \mu_n \\  s_\mu & p_1& \cdots & p_n \end{smallmatrix}\big] = \sum_{m = 1}^n \frac{p_m}{r_{\mu_m}} F_{g,n}\big[\begin{smallmatrix}   \mu_1 & \cdots & \mu_n \\  p_1& \cdots & p_n \end{smallmatrix}\big] + \delta_{g,1}\delta_{n,0}\bigg(\frac{r_\mu^2 - 1}{24r_\mu} + \frac{Q_\mu^2}{2r_\mu}\bigg)\,.
\end{equation}
\end{lemma} 
\begin{proof}
We express the constraint ${}^{\sigma}H_{i = 2,k = 0} \cdot Z = 0$, as it is always part of the Airy structure. From \eqref{eq:twist_mode_arbitr}-\eqref{eq:twist_mode_coxeter} and taking into account $J_0^\mu = \hbar^{\frac{1}{2}}Q_\mu$ and the evaluations
\begin{equation}
\label{psiorun}\begin{split}
\Psi^{(0)}_{r}(q_1,q_2) & = \frac{1}{2}\big(r^2\delta_{r|q_1}\delta_{r|q_2} - r\delta_{r|q_1 + q_2}\big)\,, \\
\Psi^{(1)}_{r}(\emptyset) & = -\frac{r(r^2 - 1)}{24}\,,
\end{split}
\end{equation}
we obtain
\begin{equation}
\label{Wi2k9}
\begin{split}
{}^{\sigma} W_{i = 2,k = 0} & = \sum_{\mu = 1}^d \frac{1}{r_\mu} W^\mu_{2,0} + \sum_{\substack{\mu \neq \nu \\ k > 0}} W^\mu_{1,-k}W^\nu_{1,k}  + \sum_{\mu < \nu} W^\mu_{1,0}W^\nu_{1,0} \\
& = \sum_{\mu = 1}^d \bigg(\sum_{k > 0}  \frac{r_\mu \delta_{r_\mu|k} - 1}{r_\mu} J^\mu_{-k}J^\mu_k + \frac{(r_\mu - 1)Q_\mu^2 \hbar}{2r_\mu} - \frac{(r_\mu^2 - 1)\hbar}{24r_\mu}\bigg) + \sum_{\mu \neq \nu } \bigg( \sum_{k > 0} J^\mu_{-r_\mu k}J^\nu_{r_\nu k} + \frac{\hbar Q_\mu Q_\nu}{2}\bigg)\,.
\end{split}
\end{equation}
To get ${}^\sigma H_{i = 2,k = 0}$ we have to apply the dilaton shifts $J^\mu_{-s_\mu} \rightarrow J^\mu_{-s_\mu } - r_\mu t_\mu$, which simply results in adding a term $\sum_{\mu = 1}^d t_\mu J^\mu_{s_\mu}$ to \eqref{Wi2k9}. Expressing the constraint
\[
\forall k > 0\,,\qquad  {}^\sigma H_{i = 1,k}\cdot Z = \bigg(\sum_{\mu = 1}^d J_{r_\mu k}^\mu\bigg)\cdot Z\ = 0\,,
\]
and using the assumption $\sum_{\mu = 1}^d Q_\mu = 0$, we see that ${}^\sigma H_{i = 2,k = 0} \cdot Z = 0$ implies that $Z$ is annihilated by the operator
\[
\sum_{\mu  = 1}^d \bigg\{t_\mu J^\mu_{s_\mu} - \sum_{k > 0} \frac{1}{r_\mu} J^\mu_{-k}J^\mu_k - \hbar\bigg(\frac{r_\mu^2 - 1}{24r_\mu} + \frac{Q_\mu^2}{2r_\mu}\bigg)\bigg\}\,.
\]
By the representation \eqref{repJdiff} of the $J$s, this yields \eqref{thedesdil} for the coefficients  \eqref{eq:free_energy_of_Airy_struct} of the partition function.
\end{proof}

The partition functions of these Airy structures for $\sigma = (1\cdots r)$ or $(1\cdots r - 1)(r)$ (already constructed in \cite{BBCCN18}) enjoy an extra property of homogeneity, which turn \eqref{thedesdil} into an analog of the dilaton equation.

\begin{corollary}
\label{lem:homogen} Assume $d = 1$ and $(r_1,s_1) = (r,s)$ with $r = \pm 1\,\,{\rm mod}\,\,s$. Then, the coefficients \eqref{eq:free_energy_of_Airy_struct} of the partition function of the Airy structure of \cref{thm:W_gl_Airy_Coxeter} satisfy $F_{g,n}[p_1,\ldots,p_n] = 0$ whenever $\sum_{m = 1}^n p_m \neq s(2g - 2 + n)$, and the dilaton equation
\[
F_{g,n + 1}[s,p_1,\ldots,p_n] = s(2g - 2 + n)\,F_{g,n}[p_1,\ldots,p_n] + \frac{r^2_1 - 1}{24} \delta_{g,1}\delta_{n,0}\,.
\]
\end{corollary}
\begin{corollary}
\label{lem:homogen2} Assume $d = 2$, $r_1 = -1\,\,{\rm mod}\,\,s_1$, $(r_2,s_2) = (1,\infty)$ and $t_1 = \tfrac{1}{r_1}$. Then, the coefficients of the partition function of the Airy structure described in \cref{thm:W_gl_Airy_arbitrary_autom} (also appearing in \cite[Theorem 4.16]{BBCCN18}) satisfy $F_{g,n}\big[\begin{smallmatrix} 1 & \cdots & 1 \\ p_1 & \cdots & p_n \end{smallmatrix}\big] = 0$ whenever $\sum_{m = 1}^n p_m \neq s_1(2g - 2 + n)$, and the dilaton equation
\[
F_{g,n + 1}\big[\begin{smallmatrix} 1 & 1 & \cdots & 1 \\ s_1 & p_1 & \cdots & p_n \end{smallmatrix}\big] = s_1(2g - 2 + n)\,F_{g,n}\big[\begin{smallmatrix} 1 & \cdots & 1 \\ p_1 & \cdots & p_n \end{smallmatrix}\big] + \bigg(\frac{r^2_1 - 1}{24} + \frac{(r_1 + 1)Q_1^2}{2}\bigg)\delta_{g,1}\delta_{n,0}\,.
\]
\end{corollary}

The proof of these two corollaries goes by induction on $2g - 2 + n > 0$, analysing the structure of the Airy structure constraints.
 
\begin{proof}[Proof of \Cref{lem:homogen}]
In this case we only need to consider $g \in \mathbb{Z}_{\geq 0}$. The Airy structure is in normal form up to an overall normalisation, and we can decompose for $i \in (r]$ and $k \geq 1 - \lambda(i)$
\[
r^{i - 1}\,{}^{\sigma} H_{i,k} = J_{\Pi(i,k)} - \sum_{\substack{ \ell,j \in \mathbb{Z}_{\geq 0} \\ 2 \leq \ell + 2j \leq r}} \frac{\hbar^{j}}{\ell!} \sum_{\mathbf{q} \in (\mathbb{Z}^*)^{\ell}} C^{(j)}[\Pi(i,k)|q_1,\ldots,q_{\ell}]\,\norder{J_{q_1}\cdots J_{q_{\ell}}}\,,
\]
where $\Pi(i,k) \coloneqq rk + s(i - 1)$. Setting $F_{0,2}[p_1,p_2] \coloneqq |p_1|\delta_{p_1 + p_2,0}$, \cite[Corollary 2.16]{BBCCN18} gives the following formula for the coefficients of the partition function:
\begin{equation} 
\label{theFGNsum} F_{g,n}[p_1,\ldots,p_n] =  \sum_{\substack{ \ell,j \in \mathbb{Z}_{\geq 0} \\ 2 \leq \ell + 2j \leq r \\ \mathbf{q} \in (\mathbb{Z}^*)^{\ell}}} \frac{C^{(j)}[p_1|q_1,\ldots,q_{\ell}]}{\ell!} \sum_{\boldsymbol{\rho} \vdash [\ell]} \sum_{\substack{\boldsymbol{\mu} \vdash_{\boldsymbol{\rho}} (n] \\ h\,:\,\boldsymbol{\rho} \rightarrow \mathbb{Z}_{\geq 0} \\ \ell + j + \sum_{\rho \in \boldsymbol{\rho}} (h_{\rho} - 1) = g}}'' \prod_{\rho \in \boldsymbol{\rho}} F_{h_{\rho},|\rho| + |\mu_{\rho}|}[\mathbf{q}_{\rho},\mathbf{p}_{\mu_\rho}]\,.
\end{equation} 
Here, $\boldsymbol{\rho} \vdash [\ell]$ means that $\boldsymbol{\rho}$ is a set of non-empty subsets of $[\ell]$ which are pairwise disjoint and whose union is $\ell$, and for $\rho \in \boldsymbol{\rho}$ we denote $\mathbf{q}_{\rho} \coloneqq (q_{m})_{m \in \rho}$. Then, $\mu \vdash_{\boldsymbol{\rho}} (n]$ is a family of (possibly empty) pairwise disjoint subsets $\mu_{\rho}$ of $(n]$ indexed by $\rho \in \boldsymbol{\rho}$, whose union is $(n]$. The double prime over the summation means that the terms involving $F_{0,1}[q_m]$ or $F_{0,2}[q_m,q_{m'}]$ are excluded from the sum. We note that the summation condition is equivalent to
\[ 
2g - 2 + n + (1 - \ell - 2j) = \sum_{\rho \in \boldsymbol{\rho}} (2h_{\rho} - 2 + |\rho| + |\mu_{\rho}|)\,.
\]
Since $\ell + 2j \geq 2$, this is indeed a recursion on $2g - 2 + n \geq 0$ to compute $F_{g,n}$ starting from the value of $F_{0,2}$.

$F_{0,2}$ obviously satisfies the homogeneous property. Assume the $F_{g',n'}$ for $0 \leq 2g' - 2 + n' < 2g - 2 + n$ satisfy homogeneity. So, the summands that may contribute to \eqref{theFGNsum} are such that
\begin{equation} 
\label{sumsssnynbbh} s\big(2g - 2 + n + (1 - \ell - 2j)\big)= \sum_{\rho \in \boldsymbol{\rho}} s(2h_{\rho} - 2 + |\rho| + |\mu_{\rho}|) = \sum_{l = 1}^{\ell} q_{l} + \sum_{m = 2}^{n} p_m\,.
\end{equation}
Writing $p_1 = \Pi(i,k)$ and applying the dilaton shift $J_{-s} \rightarrow J_{-s} - 1$ to \eqref{eq:twist_mode_coxeter} we know that $C^{(j)}[p_1|\mathbf{q}]$ is a linear combination of terms inside which
\begin{equation*}
\sum_{l = 1}^{\ell} q_l  - s\ell' = rk\,,
\end{equation*}
where $\ell + \ell' = i - 2j$. Hence
\begin{equation}
\label{sumimimim} \sum_{l = 1}^{\ell} q_{l} = p_1 + s(1 - 2j - \ell)\,.
\end{equation}
Together with \eqref{sumsssnynbbh}, this proves homogeneity of $F_{g,n}$. By induction, homogeneity is established for all $g,n$.

We then apply \Cref{lem:homogen0}. In our case, there is a single $\mu$, $t_\mu = \tfrac{1}{r_\mu}$ and $Q_\mu = 0$. Using homogeneity to simplify the right-hand side of \eqref{thedesdil}, we obtain
\[
F_{g,n + 1}[s,p_1,\ldots,p_n] = s(2g - 2 + n)F_{g,n}[p_1,\ldots,p_n] + \delta_{g,0}\delta_{n,1}\,\frac{r^2 - 1}{24}\,,
\]
\end{proof}

\begin{proof}[Proof of \Cref{lem:homogen2}]
The argument is similar and we only point the minor differences that must be taken into consideration. Although half-integer $g$ and $j$ are now allowed, this does not spoil the sum constraints appearing in the recursive formula for $F_{g,n}$ and which were used in the argument. According to \eqref{eq:twist_mode_arbitr} and \eqref{eq:dilaton_arbitr_autom}, we have for $i \in (r_1 + 1]$
\begin{equation}
\label{HihigI}{}^{\sigma} H_{i,k} = r_1^{-(i - 1)}\bigg(W_{i,k}^{1} + \sum_{k' \in \mathbb{Z}} r_1W_{i - 1,k'}^1 J_{k - 1 - k'}^2\bigg)\bigg|_{J_{-s_1}^1 \rightarrow J_{-s_1}^1 - 1}\,,
\end{equation}
where by convention $W_{r_1 + 1,k}^1 = 0$. The analysis of the previous proof applies to the term $W_{i,k}^1$. Due to the equation ${}^{\sigma} H_{1,k}\cdot Z = (J^{1}_{r_1k} + J_k^2) = 0$, we can obtain a recursion in normal form involving only $F_{g,n}\big[\begin{smallmatrix} 1 & \cdots & 1 \\ * & \cdots & * \end{smallmatrix}\big]$ by substituting
\begin{equation}
\label{subsJ2}J_{k_2}^2 \longrightarrow \left\{\begin{array}{lll} 0 & & {\rm if}\,\,k_2 < 0 \\ -\hbar^{\frac{1}{2}}Q_1 & & {\rm if}\,\,k_2 = 0 \\ -J_{r_1k_2}^{1} & & {\rm if}\,\,k_2 > 0 \end{array}\right.
\end{equation}
in \eqref{HihigI}. This converts $W_{i - 1,k'}^{1}J_{k -k'}^{2}$ into $-W_{i - 1,k'}^1J^1_{r_1(k - k')}$ or $0$, and makes it contribute to a coefficient $C^{(j)}[p_1|q_1,\ldots,q_{\ell}]$ where now $q_{\ell} = r_1(k - k')$ and
\[
\bigg(\sum_{l = 1}^{\ell - 1} q_l\bigg) - s_1\ell' = r_1k'\,,
\]
with $(\ell - 1) + \ell' = i - 1 - 2j$. Hence
\[
\sum_{l = 1}^{\ell} q_l =  r_1k - s_1\ell' = p_1 + s_1(1- 2j - \ell)\,,
\]
which is the same as \eqref{sumimimim} and is all what we need to repeat the previous proof and establish homogeneity.

We then specialise \Cref{lem:homogen0} to our case, that is $t_{1} = \frac{1}{r_{1}}$ and $t_2 = 0$, while $Q_2 = -Q_1$. Setting $\mu_i = 1$ for all $i \in [n]$ in \eqref{thedesdil} and using homogeneity to simplify the right-hand side, we deduce
\[
F_{g,n + 1}\big[\begin{smallmatrix} 1 & 1 & \cdots & 1 \\ s_1 & p_1 & \cdots & p_n \end{smallmatrix}\big] = s_1(2g - 2 + n)F_{g,n}\big[\begin{smallmatrix} 1 & \cdots & 1 \\ p_1 & \cdots & p_n \end{smallmatrix}\big] + \bigg(\frac{r_1^2 - 1}{24} + \frac{(r_1 + 1)Q_1^2}{2}\bigg)\delta_{g,1}\delta_{n,1}\,.
\]
\end{proof}

For the two above cases, we also have a string equation when $s_1 = r_1 + 1$.

\begin{lemma}
Assume $d = 1$ and $(r_1,s_1) = (r,r + 1)$. Then, the coefficients of the partition function of the Airy structure described in \cref{thm:W_gl_Airy_Coxeter} satisfy
\[
F_{g,1 + n}\big[1,p_1,\ldots,p_n\big] = \sum_{m = 1}^n p_m F_{g,n}\big[p_1,\ldots,p_{m - 1},p_m - r,p_{m + 1},\ldots,p_n\big] + \delta_{g,0}\delta_{n,2}\delta_{p_1 + p_2,r}\,.
\]
\end{lemma}
\begin{lemma}
Assume $d = 2$, $s_1 = r_1 + 1$ and $(r_2,s_2) = (1,\infty)$. Then, the coefficients of the partition function of the Airy structure described in \cref{thm:W_gl_Airy_arbitrary_autom} satisfy
\[
F_{g,1 + n}\big[\begin{smallmatrix} 1 & 1 & \cdots & 1 \\ 1 & p_1 & \ldots & p_n\end{smallmatrix}\big] = \sum_{m = 1}^n p_m F_{g,n}\big[\begin{smallmatrix} 1 & \cdots & 1 & 1 & 1 & \cdots & 1 \\ p_1 & \cdots & p_{m- 1} & p_m - r & p_{m + 1} & \cdots & p_n\end{smallmatrix}\big] + \delta_{g,0}\delta_{n,2}\delta_{p_1 + p_2,r} + \delta_{g,\frac{1}{2}}\delta_{n,1}\delta_{p_1,r}\,Q_1 \,.
\]
\end{lemma}
\begin{proof} The string equation corresponds to the operator $H_{i = 2,k = -1}$, which is only present in the Airy structure for $d = 1$ when $s = r + 1$, and for $d = 2$ if $s_1 = r_1 + 1$. The proof of the above relations is then similar to the one in \Cref{lem:homogen0}.
\end{proof}

\medskip
\section{Proof of \texorpdfstring{\cref{thm:W_gl_Airy_arbitrary_autom,prop:AiryAllDilatonPolarize,thm:genus0_soln_admissible_Hik}}{Theorem}}
\label{sec:W_gl_airy_structs_proof}

\medskip

The first part of this section is devoted to the proof of \cref{thm:W_gl_Airy_arbitrary_autom,prop:AiryAllDilatonPolarize}. These theorems state that certain collections of operators $({}^{\sigma}H_{i,k})_{(i,k)\in I}$ defined in \cref{eq:dilaton_arbitr_autom,HAllDilatonPolarize} are Airy structures. In fact we only prove the second theorem, and note that the first is a special case. We first prove this in the standard case, proceeding as follows.
\begin{enumerate}[label=(\Roman*)]
	\item \label{item:step_deg_one_calc} The operators of an Airy structure must be of the form $J^\mu_a + \higherorderterms{2}$. It is thus necessary to first identify the degree zero and degree one term of ${}^{\sigma}H_{i,k}$ in order to check this condition.
	
	\item In general, one will find that
	\begin{equation}
	\label{eq:Hik_general_deg_one}
	{}^{\sigma}H_{i,k} = c_{i,k} + \sum_{\substack{\mu \in [d] \\ a\in\mathbb{Z}}} \degOneMat_{(i,k),(\mu,a)} \, J^\mu_a + \higherorderterms{2}\,.
	\end{equation}
	for some matrix $\degOneMat$ and constants $c_{i,k}$. This means that for generic $i,k$ we may expect terms proportional to $J^\mu_{-a}=a\,x^\mu_a$. We will thus construct an index set $I\subset [r] \times \mathbb{Z}$ such that for all $(i,k)\in I$ we have $c_{i,k} = 0$ and $\degOneMat_{(i,k),(\mu,a)}=0$ if $a\leq 0$.	
	\item \label{item:step_deg_one_cond} Nevertheless, even restricted to $I$ the degree one term \eqref{eq:Hik_general_deg_one} is in general a linear combination of many $J^\mu_a$s for $a>0$. In order to bring the operators into the normal form of an Airy structure $J^\mu_a + \higherorderterms{2}$ where each $(\mu,a)$ appears in a unique operator, we will show that the matrix $\degOneMat$ restricted to $I$ is invertible under certain constraints on the dilaton shifts.
	One can then obtain the operators
	\begin{equation*}
	{}^{\sigma}\tilde{H}^\mu_a \coloneqq  \sum_{(i,k)\in I} \left(\degOneMat^{-1}\right)_{(\mu, a), (i,k)} {}^{\sigma}H_{i,k} = \hbar \partial_{x^{\mu}_a} + \higherorderterms{2}\,,
	\end{equation*}
	which are of desired form.
	\item \label{item:step_Lie_subalg} The modes $(\tilde{H}^\mu_a)_{(\mu\in[d],\,\,a>0)}$ satisfy the subalgebra condition if and only if the $(^{\sigma}W_{i,k})_{(i,k)\in I}$ do. The latter is satisfied when $I$ is induced by a descending partition of $r$ as specified in \cref{prop:desc_part_Lie_subalg}. This criterion thus allows for an easy check whether the mode set constructed in \labelcref{item:step_deg_one_cond} satisfies the subalgebra condition \eqref{eq:gr_lie_subalg_cond} of a higher quantum Airy structure.
\end{enumerate}

Together the results of \labelcref{item:step_deg_one_cond} and \labelcref{item:step_Lie_subalg} will directly imply \cref{prop:AiryAllDilatonPolarize} and hence \cref{thm:W_gl_Airy_arbitrary_autom} in the standard case. The steps \labelcref{item:step_deg_one_calc} to \labelcref{item:step_deg_one_cond} will be carried out in \cref{sec:deg_one_part} and step \labelcref{item:step_Lie_subalg} is performed in \cref{sec:match_lie_subalg}. In \cref{sec:fixed_pt_no_shift} we treat in less details the exceptional case, as it resembles the standard case in many ways.\par
The purpose of \Cref{sec:genuszerosolutionexists} is the proof of \Cref{thm:genus0_soln_admissible_Hik} that addresses the question in which case the operators ${}^{\sigma}H_{i,k}$ admit a partition function solving the associated system of differential equations to leading order in $\hbar^{\frac{1}{2}}$. We will start in \cref{sec:genuszerosolutionexists2} by deriving sufficient conditions for this question to have a positive answer for general families of differential operators satisfying the degree one condition. Then we apply this result to our case at hand and finally prove \Cref{thm:genus0_soln_admissible_Hik} in \cref{sec:genus0_soln_admissible_Hik}.
\vspace{10pt}

First, let us recall some notation. Let $\sigma\in\mf{S}_r$ be a permutation with $d$ cycles of respective length $r_\mu$ such that $r=r_1+\cdots+r_d$. We can then define the dilaton shifted modes
\begin{equation*}
{}^{\sigma}H_{i,k} \coloneqq  \hat{T} \cdot {}^{\sigma}W_{i,k} \cdot \hat{T}^{-1} \,,\qquad  \hat{T} \coloneqq \exp\left(\sum_{\substack{\mu \in [d] \\ k > 0}} \hbar^{-1} F_{0,1}\big[\begin{smallmatrix} \mu \\ -k \end{smallmatrix}\big] \,\frac{J^{\mu}_{k}}{k}\right)\,,\qquad s_\mu \coloneqq \min \big\{ k  > 0 \mid  F_{0,1}\big[\begin{smallmatrix} \mu \\ -k \end{smallmatrix}\big] \neq 0 \big\}\,,
\end{equation*}
which in the following are the central objects of study. Here ${}^{\sigma}W_{i,k}$ are the $\sigma$-twisted modes from \cref{TwistedWmodes}. Compared to \cref{HAllDilatonPolarize}, we do not introduce the coefficients $ F_{\frac{1}{2},1}$ and $F_{0,2}$ here yet, as it turns out these are not important for most of the proof. As in \cref{Wgenerating}, it will be useful to gather these in a generating function.
\begin{equation*}
{}^{\sigma}H_i(\tilde{x}) = \sum_{k \in \Z} \frac{{}^{\sigma}H_{i,k}}{\tilde{x}^k} \bigg( \frac{\dd \tilde{x}}{\tilde{x}} \bigg)^i = \hat{T} \cdot {}^\sigma W_i(\tilde{x}) \cdot \hat{T}^{-1}\,.
\end{equation*}
We will also recall from \cref{DefCurrents}
\begin{equation*}
J^{\mu,a} (\tilde{x}) := \sum_{k \in \frac{a}{r_\mu}+\mathbb{Z}} J^\mu_{r_\mu k} \frac{\dd \tilde{x}}{\tilde{x}^{k+1}}
\end{equation*}
and define $ J^{\mu,{\rm tot}}(\tilde{x}) = \sum_{a =0}^{r_\mu - 1} J^{\mu,a} (\tilde{x})$. To treat these currents more uniformly, we introduce $\tilde{C} = \bigsqcup_{\mu = 1}^d \tilde{C}_\mu$, the union of $d$ copies of a formal neighbourhood  $\tilde{C}_\mu$ of the origin in the complex plane. We use the notation $ \begin{psmallmatrix}\mu \\ z\end{psmallmatrix}$ for the coordinate in the $\tilde{C}_\mu$. We define a function $ \tilde{x} \colon \tilde{C} \to \C $ by $ \tilde{x} \begin{psmallmatrix} \mu \\ z \end{psmallmatrix} = z^{r_\mu}$. Cf. \cref{SCpart} for more on this viewpoint. We then define a unified current
\begin{equation*}
J\begin{psmallmatrix} \mu \\ z\end{psmallmatrix} := \frac{J^{\mu,{\rm tot}}(\tilde{x}(z))}{r_\mu} = \sum_{k \in \Z} J^\mu_k \frac{\dd \tilde{x}}{r_\mu \tilde{x}^{k/r_\mu +1}} = \sum_{k \in \Z} J^\mu_k \frac{\dd z}{z^{k +1}}\,.
\end{equation*}
The factor of $r_\mu$ in the denominator is a convention making \eqref{HfromJby01s} simpler. We sometimes omit $\mu$ from the notation and simply denote $z \in \tilde{C}$. With these notations, we see that the dilaton shift induces
\begin{equation*}
J (z) \to J (z) + \omega_{0,1} (z)\,, \qquad \omega_{0,1}\begin{psmallmatrix} \mu \\z \end{psmallmatrix} \coloneqq \sum_{k >0} F_{0,1}\big[\begin{smallmatrix} \mu \\ -k \end{smallmatrix}\big]\,z^{k - 1}\dd z\,.
\end{equation*}
We also use the shorthand notation $ t_\mu = -\frac{1}{r_\mu}\,F_{0,1} \begin{bsmallmatrix} \mu \\ -s_\mu \end{bsmallmatrix} $ for the leading coefficient.

\medskip

\subsection{The degree one condition}
\label{sec:deg_one_part}

\medskip

In this subsection and the next one we assume that all $s_{\mu}$ are finite. Let us begin by identifying the degree one component of ${}^{\sigma}H_i$. Let $\pi_1$ be the projection to degree one.

\begin{lemma}
	\label{DegOneSeries}
	For any $i\in [r]$
	\begin{equation}
	\label{HfromJby01s}
	\pi_1 \big({}^\sigma H_i (\xi ) \big) = \sum_{z \in \tilde{x}^{-1}(\xi )} \sum_{\substack{Z \subseteq \tilde{x}^{-1}(\xi ) \setminus \{ z\} \\ |Z|=i-1}} J(z) \prod_{z' \in Z} \omega_{0,1}(z')\,.
	\end{equation}
So we have
\begin{equation}
\pi_1 \big({}^\sigma H_{i,k} \big) = \sum_{\substack{\mu \in [d] \\ a \in \mathbb{Z}}} \mc{M}_{(i,k),(\mu,a)}\,J_a^\mu\,,
\end{equation}
with a matrix $\mc{M}_{(i,k),(\mu,a)}$ having the property that, for each $(\mu,a)$, there exists $K_{\mu,a}$ such that for any $i \in [r]$ and $k \geq K_{a,\mu}$ we have $\mc{M}_{(i,k),(\mu,a)} = 0$.
\end{lemma}
Note that the right-hand side is a symmetric function in the elements of $ \tilde{x}^{-1}(\xi)$, and therefore it contains only integral powers of $\xi$.
\begin{proof}
	From \cref{Wikz}, we see that $ {}^\sigma W_i (\xi )$ is a linear combination of currents
	\begin{equation*}
	\hbar^{\sum_{\mu} j_\mu} ~ \prod_{\mu\in M} \norder{ \prod_{l = 2j_\mu + 1}^{i_\mu} J^{\mu,a^\mu_l}(\xi) }\,.
	\end{equation*}
	After the dilaton shift, these will only contribute to the degree one component of $ {}^\sigma H_i (\xi) $ if $j_\mu=0$ for all $\mu\in M$. We have
	\begin{equation}
	\label{pi1efnun}
	\pi_1\big({}^{\sigma}W_{i}(\xi)\big) = \pi_1\left(\sum_{M \subseteq [d]} \,\,\sum_{\substack{i_{\mu} \in [r_{\mu}]\,\,\mu \in M \\  \sum_{\mu} i_{\mu} = i}} \prod_{\mu \in M} \frac{1}{r_\mu^{i_\mu}} \bigg(\sum_{a^\mu_1,\ldots,a^\mu_{i_\mu} \in [0,r_\mu)} \Psi^{(0)}_{r_\mu}(a^\mu_{1},\ldots,a^\mu_{i_\mu})\,\, \norder{ \prod_{\mu \in M} \prod_{l = 1}^{i_\mu} J^{\mu,a^\mu_l}(\xi)}\,\,\bigg)\right)\,,
	\end{equation}
	and therefore on the right-hand side we need to take the contribution of the shifts in all factors but one. The definition \eqref{eq:psidef} of the $\Psi^{(0)}$ is nothing but a sum over subsets of Galois conjugates of the function $ \tilde{x}$, so we obtain
	\begin{equation*}
	\sum_{a^\mu_1,\ldots,a^\mu_{i_\mu} \in [0,r_{\mu})} \Psi^{(0)}_{r_\mu}(a^\mu_{1},\ldots,a^\mu_{i_\mu})\,\norder{\prod_{l = 1}^{i_\mu} J^{\mu,a^\mu_l}(\xi)} = \sum_{\substack{Z \subseteq \tilde{x}^{-1} (\xi ) \cap \tilde{C}_\mu \\|Z| = i_\mu}} \norder{\prod_{z' \in Z} J\big(\begin{smallmatrix} \mu \\ z'\end{smallmatrix}\big)}\,.
	\end{equation*}
	The sum over $ M \subseteq [d]$ in \eqref{pi1efnun} `globalises'  this sum of subsets from one component $\tilde{C}_\mu$ to all of $\tilde{C}$. As stated before, the degree one projection extracts the dilaton shifts of all but one (the choice of $z$) of these factors, which proves \eqref{HfromJby01s}. We obtain the matrix $\mc{M}$ by expanding this equation in $\xi$ and collecting the contributions of $J_a^\mu$. The vanishing property comes from the fact that $\omega_{0,1}$ contains only nonnegative positive powers of $z$.
\end{proof}

We will restrict the range of indices on both sides and show that the matrix $\mathcal{M}$ is invertible in order to bring the differential operators into the normal form of an Airy structure. But first, we would like to see that this matrix is invertible without restricting it to any subspace yet.
Define
\begin{equation*}
{}^\sigma H(\xi, u) := \sum_{i=1}^r {}^\sigma H_i (\xi )\, u^{r-i}\,.
\end{equation*}
Then
\begin{equation}
\label{DegOneDoubleGenFun}
\pi_1 \big({}^\sigma H (\xi, u) \big) = \sum_{i=1}^r \sum_{z \in \tilde{x}^{-1}(\xi )} \sum_{\substack{Z \subseteq \tilde{x}^{-1}(\xi ) \setminus \{ z\} \\ |Z|=i-1}} \!\! J(z) \prod_{z' \in Z} \omega_{0,1}(z') \, u^{r-i}  = \!\! \sum_{z \in \tilde{x}^{-1} (\xi)} J(z) \!\!\prod_{z' \in \tilde{x}^{-1}(\xi) \setminus \{ z\} }\!\!\!\!\!\! \big( u+ \omega_{0,1}(z')\big) \,.
\end{equation}
We are looking for the inverse to this operation. Remember from \eqref{F01tum} that for any $\mu \in [d]$ we write $t_\mu \coloneq -\frac{1}{r_\mu}\,F_{0,1}\big[\begin{smallmatrix} \mu \\ -s_\mu \end{smallmatrix}\big]$.

\begin{lemma}
	\label{InverseJfromH}
	Assume ${\rm gcd}(r_{\mu},s_{\mu}) = 1$ for all $\mu \in [d]$ and $ t_\mu^{r_\mu} \neq t_\nu^{r_\nu} $ for any distinct $\mu,\nu$ such that $ \frac{r_\mu}{s_\mu} = \frac{r_\nu}{s_\nu}$. Then the currents can be recovered from $\pi_1\big({}^\sigma H(\xi, u)\big)$ as follows:
	\begin{equation}
	\label{CurrentFromH}
	J(z) = \Res_{u = - \omega_{0,1}(z)}  \frac{\pi_1\big( {}^\sigma H(\tilde{x}(z), u)\big)\,\dd u}{\prod_{z' \in \tilde{x}^{-1}(\tilde{x}(z))} (u+ \omega_{0,1}(z'))}\,.
	\end{equation}
\end{lemma}
\begin{proof}
	If we plug \cref{DegOneDoubleGenFun} back in \cref{CurrentFromH}, we get
	\begin{align*}
	&\Res_{u = - \omega_{0,1}(z)} \sum_{\zeta \in \tilde{x}^{-1}( \tilde{x}(z))} J(\zeta ) \!\!\prod_{\zeta' \in \tilde{x}^{-1}(\tilde{x}(z)) \setminus \{ \zeta \} }\!\!\!\!\!\! \big( u+ \omega_{0,1}(\zeta')\big) \prod_{z' \in \tilde{x}^{-1}(\tilde{x}(z))} \big( u+ \omega_{0,1}(z')\big)^{-1} \dd u \\
	&= \Res_{u = - \omega_{0,1}(z)} \sum_{\zeta \in \tilde{x}^{-1}(\tilde{x}(z))} \frac{J(\zeta )\, \dd u}{u+ \omega_{0,1}(\zeta )} 
	\end{align*}
	Because we took $ z $ in a small neighbourhood of zero (but of course not zero itself), the conditions of the lemma ensure that all $ \omega_{0,1}(\zeta )$ for $ \zeta $ in the same fibre have different values. Therefore, the only contribution to the residue comes from $ \zeta = z$.
\end{proof}

\begin{remark}
	\label{OtherInverse}
	The other inverse relationship,
	\begin{equation*}
	\pi_1 \big({}^\sigma H (\xi, u) \big) = \!\! \sum_{z \in \tilde{x}^{-1} (\xi)} \prod_{z' \in \tilde{x}^{-1}(\xi) \setminus \{ z\} }\!\!\!\!\!\! \big( u+ \omega_{0,1}(z')\big) \Res_{u' = - \omega_{0,1}(z)} \frac{\pi_1\big( {}^\sigma H(\tilde{x}(z), u')\big)\,\dd u'}{\prod_{z' \in \tilde{x}^{-1}(\tilde{x}(z))} ( u'+ \omega_{0,1}(z'))}\,,
	\end{equation*}
	follows immediately from Lagrange interpolation since $ \pi_1 \big({}^\sigma H(\tilde{x}(z),u)\big)$ is polynomial in $u$. 
\end{remark}

\begin{remark}
\label{rem:t_cond}
For \cref{InverseJfromH,OtherInverse} to work, all we really need is that $ \omega_{0,1}$ takes distinct values on all elements of the fibre of $ \tilde{x} $ near the ramification point $ \xi = 0$, i.e. the map $ (\tilde{x},\omega_{0,1}) \colon \tilde{C} \to T^*\P^1 $ is an embedding on a punctured neighbourhood of the ramification point. The conditions $ \text{gcd}(r_\mu,s_\mu) = 1$ and $t_\mu^{r_\mu} \neq t_\nu^{r_\nu}$ for $ \mu \neq \nu $ such that $ \frac{r_\mu}{s_\mu} = \frac{r_\nu}{s_\nu}$ ensure this. Though in general not necessary at this point, it will become crucial to impose these stricter conditions in \Cref{prop:deg_one_cond_arb_autom}. See \Cref{CurveEmbeddingDeformed} for this.\par
 In the undeformed case, the setting of \cref{thm:W_gl_Airy_arbitrary_autom}, i.e. monomial $ \omega_{0,1}$, this is in fact necessary as well as sufficient. Suppose we are in the setting of \cref{thm:W_gl_Airy_arbitrary_autom} and there were a $ \mu $ such that $ \text{gcd}(r_\mu,s_\mu) = d_\mu > 1$, let $ z \in \tilde{C}_\mu$, and let $ \theta $ be a primitive $ r_\mu$th root of unity. Then $ x(z) = x(\theta^{r_\mu/d_\mu} z)$ and $ \omega_{0,1}(z) = \omega_{0,1}(\theta^{r_\mu/d_\mu} z)$. So in \cref{DegOneDoubleGenFun}, the dependence on $J^\mu_* $ is given by
\begin{equation*}
\prod_{ z' \in \tilde{x}^{-1}(\xi)} \big(u+ \omega_{0,1}(z')\big) \sum_{i =1}^{r_\mu} \frac{J \begin{psmallmatrix}\mu \\ \theta^i z \end{psmallmatrix}}{u + \omega_{0,1}(\theta^i z)} = \prod_{ z' \in \tilde{x}^{-1}(\xi)} \big(u+ \omega_{0,1}(z')\big) \sum_{i =1}^{r_\mu/d_\mu} \frac{ \sum_{j=1}^{d_\mu}  J \begin{psmallmatrix}\mu \\ \theta^{i+jr_\mu/d_\mu} z \end{psmallmatrix}}{u + \omega_{0,1}(\theta^i z)}\,.
\end{equation*}
Performing the $ j$-sum, we see that $ \pi_1\big({}^\sigma H(\xi,u )\big)$ only depends on $ J^\mu_a $ for $ d_\mu \mid a$. Therefore, it is impossible to retrieve all $ J^\mu_a$ from $\pi_1({}^\sigma H)$.\par
In the case $ \mu \neq \nu$ such that $ \frac{r_\mu}{s_\mu} = \frac{r_\nu}{s_\nu}$ and $ t_\mu^{r_\mu} = t_\nu^{r_\nu}$, the situation is similar. First note that by redefining the local coordinate on $\tilde{C}_\mu$, we may actually assume $ t_\mu = t_\nu$. Writing the local coordinates as $ z \in \tilde{C}_\mu$ and $ z' \in \tilde{C}_\nu$, we then have $ \omega_{0,1}(\theta^i z) = \omega_{0,1}( \theta^i z')$ for all $ i \in [r_\mu ]$. By the same argument as above, $ \pi_1(H)$ then depends on $ J^\mu_a$ and $ J^\nu_b$ only in the combination $ J^\mu_a + J^\nu_a$, so again we can never retrieve the individual $ J^\mu_a $ and $J^\nu_a$.
\end{remark}

Since the $\omega_{0,1}$ are power series with exponents bounded below (by $ s_\mu$ on component $ \tilde{C}_\mu$), we can see that for any $i$, the set of $ k$ such that $ {}^\sigma H_{i,k} $ gets a contribution from $ J^\mu_a$ with $ a \leq 0$ in \cref{HfromJby01s} is bounded from above. Therefore, the following definition makes sense.

\begin{definition}
	\label{def:kLowerBound}
	For $ i \in [r]$, we define $ k^{>}_{\mathrm{min}}(i)$ to be the smallest $ K$ such that  for all $ k \geq K$, $ \pi_1 ({}^\sigma H_{i,k})$  given by \cref{HfromJby01s} lies in the linear span of $J^\mu_a $ with $\mu\in [d]$ and $a > 0$ solely.\par
	Analogously, for $ i \in [r]$ we choose $ k^{\geq}_{\mathrm{min}}(i)$ to be the lowest bound for which $\pi_1 ({}^\sigma H_{i,k})$ with $k\geq k^{\geq}_{\mathrm{min}}(i)$ only features $J^\mu_a $ with $\mu\in [d]$ and $a \geq 0$.
\end{definition}

It turns out that $k^{>}_{\mathrm{min}}(i)$ can be approximated as follows.

\begin{lemma}
	\label{prop:k_min}
	If $\frac{r_1}{s_1} \geq \cdots \geq \frac{r_d}{s_d}$ and $\gcd(r_\mu, s_\mu)=1$ for all $\mu\in [d]$ we have $k^{>}_{\mathrm{min}}(i) \leq \kLowerBound{i}{r}{s}$ where
	\begin{equation}
	\label{eq:k_min}
	\kLowerBound{i}{r}{s} \coloneqq \begin{cases}
	-\left\lfloor\tfrac{s_1(i-1)}{r_1}\right\rfloor + \delta_{i,1}\,, & 1\leq i \leq r_1\\[0.5em]
	-\left\lfloor\tfrac{s_2(i-r_1-1)}{r_2}\right\rfloor - s_1  + \delta_{i,1+r_1}\,, & r_1 < i \leq \rsumind{2}\\[0.3em]
	\qquad \qquad \vdots & \qquad \quad \vdots \\[0.3em]
	-\left\lfloor\tfrac{s_d(i-\rsumind{d-1}-1)}{r_d}\right\rfloor - \ssumind{d-1}  + \delta_{i,1+\rsumind{\mu-1}}\,,& \rsumind{d-1} < i \leq r\,
	\end{cases}\,,
	\end{equation}
	where for a subset $M \subseteq [d]$ we denoted $\mathbf{r}_M := \sum_{\mu \in M} r_\mu$ and $\mathbf{s}_M := \sum_{\mu \in M} s_\mu$.
\end{lemma}

\begin{proof}
	Recall that $k$ is the exponent in
	\begin{equation}
	\label{eq:kexpchar}
	\frac{\pi_1 ({}^\sigma H_{i,k})}{\xi^k} \left( \frac{\dd \xi}{\xi} \right)^i\,.
	\end{equation}
	As we also get $i$ factors of $ \frac{\dd \xi}{\xi} $ on the right-hand side of \cref{HfromJby01s} (one from $ J$ and $i-1$ from the $ \omega_{0,1}$, we will concentrate on the remaining powers of $ \xi$.\par
	From \cref{DegOneSeries}, we see that in order to determine an upper bound $k_{\mathrm{max}} \in \Q $ for the exponents $k$ in \eqref{eq:kexpchar} for which we get a non-vanishing contribution from $J^\mu_a$ with  $a\leq 0$ it suffices to inspect $a= 0$ and to take only the leading order of all $\omega_{0,1}$s into account. More specifically, to make $k$ in \eqref{eq:kexpchar} as large as possible, we need to choose the $\omega_{0,1}$s efficiently: an $\omega_{0,1} $ on branch $\nu$ has leading order $(z')^{s_\nu} \frac{\dd z'}{z'}  = \xi^{\frac{s_\nu}{r_\nu}} \frac{\dd \xi}{\xi}$. So, to minimise the power of $\xi$, we need to take $\tfrac{s_\nu}{r_\nu}$ minimal, i.e. $\tfrac{r_\nu}{s_\nu}$ maximal, i.e. $\nu$ minimal. From this it follows that the maximal $k$ such that $J_a^*$, $a \leq 0$ can contribute, is found by first taking all $r_1 $ factors $\omega_{0,1}$ with arguments on component $\nu = 1$, then the $r_2$ on component $ \nu = 2$, up until we get to $i-1$ factors.\par
	If $i \in [r_1]$, we get $-k_{\mathrm{max}} = \tfrac{s_1}{r_1}(i-1)$ from this, while for $i \in (\rsumind{\mu-1},\rsumind{\mu}]$, we similarly get $-k_{\mathrm{max}} =  (i - \rsumind{\mu-1}-1 )\tfrac{s_\mu}{r_\mu} + \sum_{\nu =1}^{\mu-1} \tfrac{s_\nu}{r_\nu} r_\nu$. Of course, it might still happen that the coefficient of the power $\xi^{-k_{\mathrm{max}}}$ vanishes due to cancelling contributions of different combinations of $\omega_{0,1}$s. However, in any case $k_{\mathrm{max}} \in \Q$ is an upper bound for $k \in \Q$ for which $J^*_a$ with $a \leq 0$ can still contribute. Therefore
	\begin{equation}
	\label{eq:dfrac_rs_min}
	k^{>}_{\mathrm{min}}(i) \leq \min \Big\{ K \in \Z \quad \Big|\quad  K  > -\big(\ssumind{\mu-1}  + (i - \rsumind{\mu-1}-1 )\tfrac{s_\mu}{r_\mu} \big)\Big\}\qquad i \in (\rsumind{\mu-1},\rsumind{\mu}]\,.
	\end{equation}
	Finally, let us argue that the right-hand side of the above equation is nothing but $\kLowerBound{i}{r}{s}$ as defined in \eqref{eq:k_min}. If  $ i \in (\rsumind{\mu-1} +1,\rsumind{\mu}]$, this is clearly true, as then $ (i-\rsumind{\mu-1})\tfrac{s_\mu}{r_\mu} \notin \Z$, and taking the integral part implies strict inequality. If $ i= \rsumind{\mu-1} +1 $, we need to add $1$ to obtain strict inequality, explaining the Kronecker symbol in \eqref{eq:k_min}.
\end{proof}

\begin{remark}
	For generic $\omega_{0,1}$, i.e. generic values for dilaton shifts, we even have $k^{>}_{\mathrm{min}}(i) =  \kLowerBound{i}{r}{s}$ for all $i$. Indeed, it should be clear from the proceeding discussion that the case $k^{>}_{\mathrm{min}}(i) <  \kLowerBound{i}{r}{s}$ can only occur if the leading order inspected in the proof of \cref{prop:k_min} vanishes. This can only happen if $\xi^{-k_{\mathrm{max}}}$ gets contributions from several combinations of $\omega_{0,1}$s, which however is only possible if in case $ i \in (\rsumind{\mu-1}, \rsumind{\mu}]$ we have a $\nu\neq \mu$ with $\tfrac{r_\nu}{s_\nu}=\tfrac{r_\mu}{s_\mu}$. This can be compared in \cref{SCpart} to the results of \cref{lemYA} and \cref{ALEtoMLE}.
\end{remark}
The lower bound $k^{\geq}_{\mathrm{min}}$ can be approximated in the same way.

\begin{lemma}
	\label{prop:k_min_geq}
	If $\frac{r_1}{s_1} \geq \cdots \geq \frac{r_d}{s_d}$ and $\gcd(r_\mu, s_\mu)=1$ for all $\mu\in [d]$ we have $k^{\geq}_{\mathrm{min}}(i) \leq \mathfrak{d}^{\geq}_{\mathbf{r},\mathbf{s}}(i)$ where
	\begin{equation}
		\label{eq:k_min_geq}
		\mathfrak{d}^{\geq}_{\mathbf{r},\mathbf{s}}(i) \coloneqq -\left\lfloor\tfrac{s_\mu(i-\rsumind{\mu-1}-1)}{r_\mu}\right\rfloor - \ssumind{\mu-1} \qquad i \in (\rsumind{\mu-1},\rsumind{\mu}]\,.
	\end{equation}
\end{lemma}

\begin{proof}
	By the same arguments as in the proof of \Cref{prop:k_min} we have
	\begin{equation*}
		k^{\geq}_{\mathrm{min}}(i) \leq \min \Big\{ K \in \Z \quad \Big|\quad  K  \geq -\big(\ssumind{\mu-1}  + (i - \rsumind{\mu-1}-1 )\tfrac{s_\mu}{r_\mu} \big)\Big\}\qquad i \in (\rsumind{\mu-1},\rsumind{\mu}]\,.
	\end{equation*}
	The right-hand side of this equation is of course nothing but \eqref{eq:k_min_geq}.
\end{proof}

From now on we will always assume that $\frac{r_1}{s_1} \geq \ldots \geq \frac{r_d}{s_d}$ if not stated otherwise. The modes selected with the help of $\kLowerBound{i}{r}{s}$ via
\begin{equation*}
{}^{\sigma}H_{i,k} \qquad i\in [r],\quad k\geq \kLowerBound{i}{r}{s}
\end{equation*}
shall be shown to be an Airy structure for some certain $(r_\mu,s_\mu)_{\mu = 1}^d$. For future reference let us therefore define the following index set.
\begin{definition}
	We define the index set $I^{>}_{\mathbf{r},\mathbf{s}}$ to be
	\begin{equation}
	\label{eq:index_set_defn}
	I^{>}_{\mathbf{r},\mathbf{s}} \coloneqq  \big\{(i,k)\in [r] \times \mathbb{Z} \quad |\quad k\geq \kLowerBound{i}{r}{s}\big\}
	\end{equation}
	with $\kLowerBound{i}{r}{s}$ as in \eqref{eq:k_min}. Analogously, we set $I^{\geq}_{\mathbf{r},\mathbf{s}} \coloneqq  \big\{(i,k)\in [r] \times \mathbb{Z} ~ |\ ~ k\geq \mathfrak{d}^{\geq}_{\mathbf{r},\mathbf{s}}(i)\big\}$.
\end{definition}
With the fact that $k^{>}_{\mathrm{min}}(i)\leq \kLowerBound{i}{r}{s}$ and \cref{eq:k_min} we have two different characterisations of $I^{>}_{\mathbf{r},\mathbf{s}}$. The first property tells us that for $(i,k)\in I^{>}_{\mathbf{r},\mathbf{s}}$ the degree one projection of ${}^{\sigma}H_{i,k}$ is a linear combination of $J^\mu_a$ with $a>0$ only while the characterisation in \eqref{eq:k_min} will be important later in order to check whether the modes satisfy the  subalgebra condition. In the following we will need yet another characterisation of $I^{>}_{\mathbf{r},\mathbf{s}}$.
\begin{lemma}
	\label{prop:index_set_rewritten}
	For $(i,k) \in [r] \times \mathbb{Z}$ we have  $(i,k) \notin I^{>}_{\mathbf{r},\mathbf{s}}$ if and only if for all $\mu\in [d]$ we have
	\begin{equation}
		\label{eq:Pi_Irs_characterisation}
		r_\mu (k+\ssumind{\mu-1}) + s_\mu (i-\rsumind{\mu-1}-1)\leq 0 \,.
	\end{equation}
\end{lemma}
\begin{proof}
	First, let us take $(i,k) \in [r] \times \mathbb{Z}$ for which equation \eqref{eq:Pi_Irs_characterisation} holds for all $\mu\in[d]$. Then especially it holds for $\lambda\in[d]$ chosen so that $i\in(\rsumind{\lambda-1},\rsumind{\lambda}]$. In this case \eqref{eq:Pi_Irs_characterisation} can be rewritten as
	\begin{equation}
		\label{eq:Pi_Irs_characterisation_lam}
		k \leq  - \tfrac{s_\lambda}{r_\lambda} (i-\rsumind{\lambda-1}-1) - \ssumind{\lambda-1}
	\end{equation}
	which comparing with \eqref{eq:dfrac_rs_min} means nothing but $k < \kLowerBound{i}{r}{s}$ and thus $(i,k) \notin I^{>}_{\mathbf{r},\mathbf{s}}$.
	
	For the other direction suppose $(i,k) \notin I^{>}_{\mathbf{r},\mathbf{s}}$. Then again if we choose $\lambda$ so that $i\in(\rsumind{\lambda-1},\rsumind{\lambda}]$ the statement that $k < \kLowerBound{i}{r}{s}$ translates into \eqref{eq:Pi_Irs_characterisation_lam}. Hence, for arbitrary $\mu\in[d]$ we can use the inequality in order to obtain
	\begin{equation*}
		\begin{split}
			&r_\mu (k+\ssumind{\mu-1}) + s_\mu (i-\rsumind{\mu-1}-1)\\
			&\quad\leq \left(s_\mu - \frac{r_\mu s_\lambda}{r_\lambda}\right)(i-\rsumind{\lambda-1}-1) + s_\mu (\rsumind{\lambda-1} - \rsumind{\mu-1}) - r_\mu (\ssumind{\lambda-1} - \ssumind{\mu-1})\,.
		\end{split}
	\end{equation*}
	Now suppose $\mu\leq\lambda$. Then by assumption $s_\mu - \frac{r_\mu s_\lambda}{r_\lambda}\leq 0$ and therefore
	\begin{equation*}
		r_\mu (k+\ssumind{\mu-1}) + s_\mu (i-\rsumind{\mu-1}-1) \leq s_\mu \rsum{[\mu,\lambda)} - r_\mu \ssum{[\mu,\lambda)} \leq 0
	\end{equation*}
	where the last inequality follows from our assumption that $\frac{r_\mu}{s_\mu}\geq\frac{r_\nu}{s_\nu}$ for all $\nu\geq \mu$. The case $\mu>\lambda$ can be treated similarly. Note that in this case $s_\mu - \frac{r_\mu s_\lambda}{r_\lambda}\geq 0$ and thus
	\begin{equation*}
		r_\mu (k+\ssumind{\mu-1}) + s_\mu (i-\rsumind{\mu-1}-1) \leq \left(s_\mu - \frac{r_\mu s_\lambda}{r_\lambda}\right)r_\lambda + r_\mu \ssum{[\lambda,\mu)} - s_\mu \rsum{[\lambda,\mu)} \leq 0
	\end{equation*}
	where the last inequality follows from our assumption that $\frac{r_\mu}{s_\mu}\leq\frac{r_\nu}{s_\nu}$ for all $\nu\leq \mu$.
\end{proof}

Remember that we defined the matrix $\degOneMat_{(i,k), (\mu,a)}$ to be the collection of coefficients
\begin{equation*}
\degProj{1}{{}^{\sigma}H_{i,k}} = \sum_{\mu\in [d],\, a \in \mathbb{Z}} \degOneMat_{(i,k), (\mu,a)} ~ J^\mu_a
\end{equation*}
of the projection to degree $1$. It is given abstractly in \cref{DegOneSeries}. So far we found out that $\degOneMat$ admits a two-sided inverse and moreover we characterised those modes featuring only derivatives in degree one. The next step is now to combine both results in order to bring the operators $({}^{\sigma}H_{i,k})_{(i,k)\in I^{>}_{\mathbf{r},\mathbf{s}}}$ into the normal form of an Airy structure.

\begin{lemma}
	\label{prop:deg_one_cond_arb_autom} 
	Assume ${\rm gcd}(r_\mu,s_\mu) = 1$  for all $\mu\in[d]$ and $t_\mu^{r_\mu}\neq t_\nu^{r_\nu}$ for any distinct $\mu,\nu$ with $\frac{r_\mu}{s_\mu} = \frac{r_\nu}{s_\nu}$. Then, there exists matrices $\mathcal{N}_{(\mu,a),(i,k)}$ and $\mc{\mathcal{M}}_{(i,k),(\mu,a)}$ indexed by $(\mu,a) \in [d] \times \mathbb{Z}_{> 0}$ and $(i,k) \in I^{>}_{\mathbf{r},\mathbf{s}}$ (see \eqref{eq:index_set_defn}) obeying the vanishing properties in \cref{deffilter}, that are inverse to each other, and such that
	\[ 
	{}^{\sigma}\tilde{H}_{\mu,a} = \sum_{(i,k) \in I^{>}_{\mathbf{r},\mathbf{s}}} \mc{N}_{(\mu,a),(i,k)}\,{}^{\sigma}H_{i,k} \qquad (\mu,a) \in [d] \times \mathbb{Z}_{> 0}
	\]  
satisfy the degree one condition.
\end{lemma}

\begin{proof}		
	From \cref{InverseJfromH,OtherInverse} we know that the matrix $\degOneMat$ encoding the coefficients of the degree one projection of the modes ${}^{\sigma}H_{i,k}$ admits a two-sided inverse, if we keep the full range of indices $ (i,k) \in [r] \times \mathbb{Z}$ and $(\mu,a) \in [d] \times \mathbb{Z}$. However, we actually need such a relation between the semi-infinite column vectors
\[
\mathbf{H}_- := \big(\pi_1({}^\sigma H_{i,k})\big)_{(i,k) \in I^{>}_{\mathbf{r},\mathbf{s}}} \qquad {\rm and} \qquad \mathbf{J}_- := \big(J^\mu_a\big)_{(\mu,a) \in [d] \times \mathbb{Z}_{> 0}}\,.
\]
For this, let us also introduce the semi-infinite column vectors
\[
\mathbf{H}_+ := \big(\pi_1({}^\sigma H_{i,k})\big)_{(i,k) \notin I^{>}_{\mathbf{r},\mathbf{s}}} \qquad {\rm and}\qquad \mathbf{J}_+ := \big(J^\mu_a\big)_{(\mu,a) \in [d] \times \mathbb{Z}_{\leq 0}}\,,
\] 
and the infinite column vectors
\[
\mathbf{H} \coloneqq \begin{psmallmatrix} \mathbf{H}_- \\ \mathbf{H}_+ \end{psmallmatrix}\,,\qquad  \mathbf{J} \coloneqq \begin{psmallmatrix} \mathbf{J}_- \\ \mathbf{J}_+ \end{psmallmatrix}\,.
\]
We can then write \cref{DegOneDoubleGenFun} symbolically as $ \mathbf{H} = \mc{M}\cdot \mathbf{J}$. The vanishing properties of the matrix $\mc{M}$ guarantee that this product is well-defined (i.e. the evaluation of each entry involves only finite sums). By definition of $\mathbf{I}_{\mathbf{r},\mathbf{s}}$, this splits as
	\begin{equation*}
	\begin{pmatrix} \mathbf{H}_- \\ \mathbf{H}_+ \end{pmatrix} = \begin{pmatrix} \mc{M}_{--} & 0 \\ \mc{M}_{-+} & \mc{M}_{++} \end{pmatrix} \begin{pmatrix} \mathbf{J}_- \\ \mathbf{J}_+ \end{pmatrix}\,.
	\end{equation*}
	Writing $\mc{N}$ for the inverse of $\mc{M}$, the relation $ \mc{M}\cdot \mc{N} = \id $ implies that $ \mc{M}_{--}\cdot \mc{N}_{--} = \id_-$ (with obvious notation). As these are  semi-infinite matrices, we cannot conclude yet that $\mc{N}_{--}\cdot \mc{M}_{--} = \id_-$. However, using the same arguments as above to prove this it is sufficient to show that $\mc{N}_{+-}=0$. For this let us inspect $\mc{N}$ further. According to \cref{InverseJfromH}, seen $\mc{N}$ as a linear operator we have
	\begin{equation}
	\label{Ntheform}\mc{N} \colon \Upsilon(\xi,u) \mapsto \Res_{u = - \omega_{0,1}(z)}  \frac{\Upsilon(\tilde{x}(z), u)\,\dd u}{\prod_{z' \in \tilde{x}^{-1}(\tilde{x}(z))} ( u+ \omega_{0,1}(z'))}\,.
	\end{equation}
	By a straightforward calculation, we see that, if $ z \in \tilde{C}_\mu $ and $ z' \in \tilde{C}_\nu$, then, as $ z \to 0$,
	\begin{equation*}
	\omega_{0,1}(z') - \omega_{0,1}(z) = \begin{cases}\mc{O} ( z^{s_\mu - 1} \dd z ) & \nu\geq\mu \\ \mc{O} (z^{\frac{r_\mu s_\nu}{r_\nu}-1} \dd z ) & \nu < \mu\end{cases}\,.
	\end{equation*}
	More precisely, a thorough analysis analogous to the one in \cref{lemYA} shows that
	\begin{equation}
		\label{Ndenom}
		\begin{split}
		\prod_{z' \in \tilde{x}^{-1}(\tilde{x}(z))\setminus\{z\}} \!\!\!\!\! \big(- \omega_{0,1}(z) + \omega_{0,1}(z')\big) & = \mf{t}_\mu z^{\mf{u}_\mu} (\dd z)^{r-1} + \mc{O} (z^{\mf{u}_\mu+1} (\dd z)^{r-1} )\,, \\
		\mf{u}_\mu & \coloneqq-r-s_\mu+\sum_{\nu\leq \mu} r_\mu s_\nu + \sum_{\nu>\mu} r_\nu s_\mu +1 
		\end{split}
	\end{equation}
	for some non-vanishing $\mf{t}_\mu$. Thus, considering the term related to $ {}^{\sigma}H_{i,k}$, we see after some elementary addition of exponents that
	\begin{align*}
	\mc{N} \bigg( u^{r-i} \frac{1}{\xi^k} \Big( \frac{\dd \xi}{\xi}\Big)^i \bigg) &\in \mc{O} \bigg( z^{-r_\mu (k+\ssumind{\mu-1}) - s_\mu (i-\rsumind{\mu-1}-1)} \Big( \frac{\dd z}{z}\Big)\bigg) \,.
	\end{align*}
	Since by \Cref{prop:index_set_rewritten} for $(i,k)\notin I^{>}_{\mathbf{r},\mathbf{s}}$ we have
	\begin{equation*}
	r_\mu (k+\ssumind{\mu-1}) + s_\mu (i-\rsumind{\mu-1}-1) \leq 0
	\end{equation*}
	this transformation maps the vector $\mathbf{H}_+$ to $\mathbf{J}_+$ and therefore indeed $\mc{N}_{+-}=0$ implying the sought-after relation $\mc{N}_{--}\cdot \mc{M}_{--} = \id_-$.
	
	From \eqref{Ntheform}, one can check the desired vanishing property for the entries of $\mc{N}_{--}$ and therefore, applying the matrix $\mc{N}_{--}$ to the semi-infinite column vector with entries
	\begin{equation*}
	\degProj{1}{{}^{\sigma}H_{i,k}} = \sum_{\substack{\mu\in [d] \\  a > 0}} \degOneMat_{(i,k), (\mu,a)} ~ \hbar \partial_{x^\mu_a} \qquad (i,k)\in I^{>}_{\mathbf{r},\mathbf{s}}
	\end{equation*}
	is well-defined and gives the semi-infinite column vector $\mathbf{J}_-$.
	
	To completely prove the degree one property, we need to check that for any $(i,k)\in I^{>}_{\mathbf{r},\mathbf{s}}$, the degree zero component of ${}^{\sigma}H_{i,k}$ is vanishing. Following the proof of \cref{DegOneSeries} one finds that the projection to degree zero of ${}^{\sigma}H_{i,k}$ for arbitrary $i$ and $k$ is
	\begin{equation}
	\label{eq:HfromJby01s_zerodeg}
	\pi_0 \big({}^\sigma H_i (\xi ) \big) = \sum_{\substack{Z \subseteq \tilde{x}^{-1}(\xi ) \\ |Z|=i}}  \prod_{z \in Z} \omega_{0,1}(z)\,.
	\end{equation}
	Similarly to \cref{prop:k_min}, we then see that in order to get a non-vanishing contribution to $\pi_0 ({}^\sigma H_{i,k})$, we need 
	\begin{equation*}
	k \leq -\big(\ssumind{\mu-1}  + (i - \rsumind{\mu-1})\tfrac{s_\mu}{r_\mu} \big)\,,
	\end{equation*}
	where the difference with that \namecref{prop:k_min} is the substitution of $ i-1$ by $i$. 
	By the proof of that lemma,
	\begin{equation}
	\label{eq:degzerovanishes}
	\kLowerBound{i}{r}{s} > -\big(\ssumind{\mu-1}  + (i - \rsumind{\mu-1}-1)\tfrac{s_\mu}{r_\mu} \big) > -\big(\ssumind{\mu-1}  + (i - \rsumind{\mu-1})\tfrac{s_\mu}{r_\mu} \big) \,.
	\end{equation}
	Hence $\pi_0 ( {}^\sigma H_{i,k}) = 0$ for $(i,k) \in I^{>}_{\mathbf{r},\mathbf{s}}$, and this concludes the proof.
\end{proof}

\begin{remark}
\label{CurveEmbeddingDeformed}
Comparing with \cref{lemYA}, we see that the assumptions ${\rm gcd}(r_\mu,s_\mu) = 1$  for all $\mu\in[d]$ and $t_\mu^{r_\mu}\neq t_\nu^{r_\nu}$ for any distinct $\mu,\nu$ with $\frac{r_\mu}{s_\mu} = \frac{r_\nu}{s_\nu}$ are crucial for the vanishing of $\mc{N}_{+-}$ as they ensure that \eqref{Ndenom} starts with the correct order in $z$. So in case these conditions are not satisfied we cannot expect the differential operators to be an Airy structure in general. A case in which this is indeed failing is provided by the example
\begin{equation*}
	x\begin{psmallmatrix} \mu \\ z \end{psmallmatrix} = \tfrac{z^2}{2}\,, \qquad y\begin{psmallmatrix} 1 \\ z \end{psmallmatrix} = \tfrac{1}{z} \,, \qquad y \begin{psmallmatrix} 2 \\ z \end{psmallmatrix} = \tfrac{1}{z} + z
\end{equation*}
with $\omega_{0,1} = y \dd x$ presented in the vein of \cref{sec:free_energ_arb_autom}. If we tried to compute the free energies associated to this input data recursively we would find a non-symmetric $F_{0,3}$. Indeed, using a correspondence which will be established in \cref{SAIRPS} we can associate a multidifferential $\omega_{0,3}$ to $F_{0,3}$ which can be computed explicitly and shown to be non-symmetric as
\begin{equation*}
	\omega_{0,3}\begin{psmallmatrix} 2 & 1 & 1 \\ z_1& z_2& z_3 \end{psmallmatrix} = 0 \neq -\frac{\dd z_1\dd z_2\dd z_3}{z_1^2 z_2^2 z_3^2} = \omega_{0,3}\begin{psmallmatrix} 1 & 1 & 2 \\ z_3& z_2& z_1 \end{psmallmatrix}\,.
\end{equation*}
Hence, the differential operators associated with the above input data cannot form an Airy structure.
\end{remark}

\medskip

\subsection{The subalgebra condition}
\label{sec:match_lie_subalg}

\medskip

Having proven the degree one condition for the modes
\begin{equation*}
	{}^{\sigma}H_{i,k} \qquad (i,k)\in I^{>}_{\mathbf{r},\mathbf{s}}
\end{equation*}
with index set $I^{>}_{\mathbf{r},\mathbf{s}}$ as defined in \eqref{eq:index_set_defn} we need to check whether these modes form a graded Lie subalgebra as demanded for Airy structures. In order to do so we use \cref{prop:desc_part_Lie_subalg} stating that if $I$ is \define{induced by a descending partition} then $({}^{\sigma}H_{i,k})_{(i,k)\in I}$ generate a graded Lie subalgebra. By this we mean that there exists $\lambda_1 \geq \cdots \geq \lambda_{\ell}$ with $\sum_{j = 1}^{\ell} \lambda_j = r$ such that $I =I_\lambda $ where
\begin{equation*}
	I_{\lambda} \coloneqq  \big\{ (i,k)\in [r] \times \mathbb{Z} \quad | \quad  \forall i \in [r],\quad  \lambda(i)+ k > 0 \big\}\,,
\end{equation*}
and we set
\begin{equation*}
	\lambda(i) \coloneqq \min \bigg\{m \quad \bigg| \quad \sum_{j = 1}^m \lambda_j \geq i\bigg\}\,.
\end{equation*}
In our case at hand we want to check whether $I^{>}_{\mathbf{r},\mathbf{s}} \cup \{(1,0)\}$ is induced by a descending partition. We will later then exclude ${}^{\sigma}H_{1,0}$ from the associated mode set by setting this mode to zero. Explicitly, this means that we want to classify the cases in which there exists a descending partition $\lambda\descPart r$ such that $\lambda(i) = 1 - \kLowerBound{i}{r}{s} + \delta_{i,1}$. Writing $i = i' + \rsumind{\nu-1}$ for $i'\in [r_{\nu}]$ again assuming $\frac{r_1}{s_1}\geq\cdots\geq \frac{r_d}{s_d}$ we can write out $\kLowerBound{i}{r}{s}$ using \cref{prop:k_min} and obtain
\begin{equation}
	\label{eq:partition_k_min_relation}
	\lambda(i) = 1 + \left\lfloor \frac{s_\nu (i' - 1)}{r_\nu} \right\rfloor + \ssumind{\nu-1} - \delta_{i',1} \, \delta_{\nu>1}\,.
\end{equation}
The case $d=1$ was -- up to a small addition we will need later -- studied in \cite{BBCCN18}, resulting in the following correspondence.
\begin{lemma}
	\label{prop:k_min_partition_one_cycle}
	Let $r,s\geq 1$ be coprime. Then there exists a descending partition $\lambda=(\lambda_1,\ldots,\lambda_\ell)$ such that
	\begin{equation}
		\label{eq:partition_k_min_relation_one_cycle}
		\lambda(i')=1 + \left\lfloor \frac{s (i' - 1)}{r} \right\rfloor \qquad i' \in [r]
	\end{equation}
	if and only if $r = \pm 1 \,\,{\rm mod}\,\, s$. In this case $\lambda$ is given by
	\begin{equation}
		\label{eq:partition_one_cycle}
		\lambda_1 =\cdots = \lambda_{r''} = r'+ 1, \qquad \lambda_{r''+1} =\cdots = \lambda_{\ell} = r',\quad \qquad \ell=s-\delta_{s,r+1}\,,
	\end{equation}
	writing $r=r's+r''$ with $r''\in\{1,s-1\}$. In particular we have $\lambda=(1)$ for $r=1$, $\lambda=(r)$ for $s=1$, and $\lambda=(1^r)$ for $s=r+1$.
\end{lemma}
\begin{proof}
	In case $r=1$ the statement trivially holds and is independent of the choice of $s$. So let $r>1$ now. The case where $s \in [r+1]$ was already discussed in \cite[Proposition B.1]{BBCCN18} and leads to the above classification of cases in which we can find a descending partition satisfying \eqref{eq:partition_k_min_relation_one_cycle}.
	
	So what is left to prove is that in case $s>r+1$ there is no partition $\lambda$ for which \eqref{eq:partition_k_min_relation_one_cycle} holds. Suppose the opposite is true. Then since $\lambda(i'+1)-\lambda(i')\in\{0,1\}$ it follows that
	\begin{equation*}
		\left\lfloor \frac{s i'}{r} \right\rfloor - \left\lfloor \frac{s (i' - 1)}{r} \right\rfloor \leq 1
	\end{equation*}
	for all $i'\in[r-1]$. Now writing $s=s' r + s''$ for some $s''\in [0,r-1]$ we see that the above is equivalent to
	\begin{equation*}
		s' + \left\lfloor \frac{s'' i'}{r} \right\rfloor - \left\lfloor \frac{s'' (i' - 1)}{r} \right\rfloor \leq 1\,.
	\end{equation*}
	This inequality can only be satisfied if $s'=1$ and
	\begin{equation*}
		\left\lfloor \frac{s'' i'}{r} \right\rfloor - \left\lfloor \frac{s'' (i' - 1)}{r} \right\rfloor \leq 0
	\end{equation*}
	for all $i'\in[r-1]$. Remember that we assumed $s>r+1$ which implies $s''>1$. However, in this case we can always find an $i'\in[r-1]$ for which $\lfloor \tfrac{s'' i'}{r} \rfloor - \lfloor \tfrac{s'' (i' - 1)}{r} \rfloor = 1$ which is a contradiction.
\end{proof}
For later reference, if $r$ and $s$ satisfy the properties of \Cref{prop:k_min_partition_one_cycle} let us write $\lambda^{r,s}$ for the associated partition $\lambda$ and $\ell^{r,s}\coloneqq s - \delta_{s,r+1}$ for its length.\\
To extend the statement of the last lemma to the case $d>1$ we need to make a few observations. First, notice that for a partition $\lambda\vdash r$ we have
\begin{equation*}
	\lambda(i+1) - \lambda(i) = 1 \qquad \text{if and only if}\qquad i=\sum_{j=1}^m \lambda_j
\end{equation*}
for some $m>0$ and $\lambda(i+1) - \lambda(i)=0$ otherwise. Therefore a descending partition $\lambda$ satisfying \eqref{eq:partition_k_min_relation} encodes the length of the intervals for which the right-hand side of the equation stays constant. Now notice that for $i\in (\rsumind{\nu-1}+1,\rsumind{\nu}]$ the right-hand side of \eqref{eq:partition_k_min_relation} is exactly of the form as in \eqref{eq:partition_k_min_relation_one_cycle} up to a constant shift, i.e.\
\begin{equation*}
	\lambda(i) = \lambda^{r_\nu,s_\nu} (i') + \ssumind{\nu-1} 
\end{equation*}
for $i = i' + \rsumind{\nu-1}$ with $i'\in (1,r_{\nu}]$. Hence for $\lambda$ to be descending $r_\nu$ and $s_\nu$ need to satisfy the properties from \Cref{prop:k_min_partition_one_cycle}. Moreover, if we analyse the jumping behaviour of $\lambda(i)$ at the transition $i\in [\rsumind{\nu-1},\rsumind{\nu-1}+2]$ we obtain a full description of $\lambda$. For this observe that because of the Kronecker delta
\begin{align*}
	\lambda(\rsumind{\nu-1}+1) - \lambda(\rsumind{\nu-1}) & = 0\\
	\lambda(\rsumind{\nu-1}+2) - \lambda(\rsumind{\nu-1}+1) & = 1 
\end{align*}
assuming $r_\nu>1$ in the second line. Therefore, a partition $\lambda$ describing the right-hand side of \eqref{eq:partition_k_min_relation} must be of the form
\begin{equation}
	\begin{split}
		\label{eq:lambda_decomposed}
		\lambda = \left(\begin{array}{ccccccc} &\lambda^{r_1,s_1}_1 &, \lambda^{r_1,s_1}_2 &, \ldots,  &,\lambda^{r_1,s_1}_{\ell^{r_1,s_1}-1} &,\lambda^{r_1,s_1}_{\ell^{r_1,s_1}}+1 &,\\[0.2em]
			&\lambda^{r_2,s_2}_1-1 &,\lambda^{r_2,s_2}_2 &,\ldots &,\lambda^{r_2,s_2}_{\ell^{r_2,s_2}-1} &, \lambda^{r_2,s_2}_{\ell^{r_2,s_2}} + 1 &,\\[0.2em]
			&\lambda^{r_3,s_3}_1-1 &,\lambda^{r_3,s_3}_2 &,\ldots &, \lambda^{r_3,s_3}_{\ell^{r_3,s_3}-1} &, \lambda^{r_3,s_3}_{\ell^{r_3,s_3}} + 1 &,\\
			& & & \vdots & & \\
			&\lambda^{r_d,s_d}_1-1 &,\lambda^{r_d,s_d}_2 &,\ldots &,\lambda^{r_d,s_d}_{\ell^{r_d,s_d}-1} &, \lambda^{r_d,s_d}_{\ell^{r_d,s_d}} &
		\end{array}\right)
	\end{split}
\end{equation}
where in  case $s_\mu=1$ the $\mu$th line
\begin{equation*}
	\big(\ldots,\lambda^{r_\mu,s_\mu}_1-1,\lambda^{r_\mu,s_\mu}_2 ,\ldots , \lambda^{r_\mu,s_\mu}_{\ell^{r_\mu,s_\mu} - 1} , \lambda^{r_\mu,s_\mu}_{\ell^{r_\mu,s_\mu}} + 1 ,\ldots \big)
\end{equation*}
must be replaced with just $(\ldots, r_\mu, \ldots)$. This partition is almost just the concatenation \mbox{$(\lambda^{r_1,s_1},\ldots, \lambda^{r_d,s_d})$} but it has one box moved from the first row of $\lambda^{r_{\mu+1},s_{\mu+1}}$ to the last row of $\lambda^{r_{\mu},s_{\mu}}$ for each $\mu\in[d-1]$. A classification of the cases in which $\lambda$ is descending is now given as follows.
\begin{lemma}
	\label{prop:Lie_subalg_cond_arb_autom}
	Let $d \geq 2$. Given $\frac{r_1}{s_1}\geq\cdots\geq \frac{r_d}{s_d}$ with  $r_\mu$ and $s_\mu$ coprime for all $\mu$, there exists a descending partition $\lambda = (\lambda_1,\ldots,\lambda_\ell)$ of $r=r_1+\cdots+r_d$ such that \eqref{eq:partition_k_min_relation} is satisfied if and only if the following holds
	\begin{enumerate}
		\item\label{item:s_1_condtion} $r_1 = -1\,\,{\rm mod}\,\, s_1$ ; 
		
		\item\label{item:s_mu_condtion} $s_\mu=1$ for all $\mu\in (1,d)$ ;
		
		\item\label{item:s_d_condtion} $r_d = +1\,\,{\rm mod}\,\, s_d$.
	\end{enumerate}
	In this case $\lambda$ is given by
	\begin{equation}
		\label{eq:arb_autom_corresp_partition}
		\lambda =\begin{cases}
			\big({(r_{1}'+1)}^{s_{1}} , r_2 , r_3 , \ldots , r_{d-1} , {r_{d}'}^{s_{d}}\big) & r_d\neq 1 \\
			\big({(r_{1}'+1)}^{s_{1}} , r_2 , r_3 , \ldots , r_{d-1}\big) & r_d=1
		\end{cases}\,,
	\end{equation}
	where $r_\mu'\coloneqq  \lfloor r_\mu / s_\mu \rfloor$.
\end{lemma}

\begin{proof}
	Suppose property \labelcref{item:s_1_condtion}--\labelcref{item:s_d_condtion} are satisfied. Then inserting the explicit expressions for the partitions $\lambda^{r_\mu,s_\mu}$ into \eqref{eq:lambda_decomposed} immediately tells us that $\lambda$ is of the form \eqref{eq:arb_autom_corresp_partition} which is a descending partition as claimed.
	
	Conversely, in order to argue that \labelcref{item:s_1_condtion}--\labelcref{item:s_d_condtion} are necessary for $\lambda$ to be a descending partition first note that every row in \eqref{eq:lambda_decomposed} has to be descending individually which forces $\lambda^{r_\mu,s_\mu}$ to be descending for all $\mu$. Hence, \Cref{prop:k_min_partition_one_cycle} tells us that necessarily $r_\mu = \pm 1 \,\,{\rm mod}\,\, s_\mu$. To further constrain the choice for $r_\mu$ and $s_\mu$ note that in case $s_\mu\geq 3$ the requirement that $\lambda^{r_\mu,s_\mu}_1 - 1 \geq \lambda^{r_\mu,s_\mu}_2$ implies that $r_\mu = 1 \,\,{\rm mod}\,\, s_\mu$ and $\lambda^{r_\mu,s_\mu}_{\ell^{r_\mu,s_\mu}-1} \geq \lambda^{r_\mu,s_\mu}_{\ell^{r_\mu,s_\mu}}+1$ forces $r_\mu = - 1 \,\,{\rm mod}\,\, s_\mu$. This explains \labelcref{item:s_1_condtion} and \labelcref{item:s_d_condtion} and tells us that $s_\mu\leq 2$ for $\mu\in(1,d)$. Suppose now $s_\mu= 2$ for some $\mu\in(1,d)$. Then since we assume $\lambda$ to be descending we find that $\big\lceil\tfrac{r_\mu}{2}\big\rceil - 1 = \lambda^{r_\mu,2}_1 - 1 \geq \lambda^{r_\mu,2}_1 + 1 = \big\lfloor\tfrac{r_\mu}{2}\big\rfloor + 1$ which is a contradiction. Thus, the only possibility we are left with is indeed $s_\mu=1$ for all $\mu\in(1,d)$.
\end{proof}

We now have everything at hand to prove the standard case of \cref{prop:AiryAllDilatonPolarize}.

\begin{proof}[Proof of \cref{prop:AiryAllDilatonPolarize}, standard case]
	Recall that the situation of the theorem is as follows:
	\begin{align*}
	{}^{\sigma}H_{i,k} &\coloneqq \hat{\Phi}\hat{T} \cdot {}^{\sigma}W_{i,k} \cdot \hat{T}^{-1}\hat{\Phi}^{-1}\,,\\
	\hat{T} &\coloneqq\exp\left(\sum_{\mu \in [d]} \sum_{k > 0} \Big(\hbar^{-1} F_{0,1}\big[\begin{smallmatrix} \mu \\ -k \end{smallmatrix}\big] + \hbar^{-\frac{1}{2}} F_{\frac{1}{2},1}\big[\begin{smallmatrix} \mu \\ -k \end{smallmatrix}\big]\Big)\,\frac{J^{\mu}_{k}}{k}\right)
	= \hat{T}_2\hat{T}_1\,,\\
	\hat{T}_1 &\coloneqq\exp\left(\frac{1}{\hbar} \sum_{\mu \in [d]} \sum_{k \geq s_\mu} F_{0,1}\big[\begin{smallmatrix} \mu \\ -k \end{smallmatrix}\big]\,\frac{J^\mu_k}{k}\right)\,, \qquad \hat{T}_2  \coloneqq\exp\left( \frac{1}{\hbar^{\frac{1}{2}}} \sum_{\mu \in [d]} \sum_{k > 0} F_{\frac{1}{2},1}\big[\begin{smallmatrix} \mu \\ -k \end{smallmatrix}\big]\,\frac{J^\mu_k}{k}\right)\,,\\
	\hat{\Phi} &\coloneqq \exp\bigg( \frac{1}{2\hbar} \sum_{\substack{\mu,\nu \in [d] \\ k,l > 0}}  F_{0,2}\big[\begin{smallmatrix} \mu & \nu \\ -k & -l \end{smallmatrix}\big]\,\frac{J_k^\mu J_l^\nu}{kl} \bigg)\,.
	\end{align*} 
	Up to now, we have only considered the conjugation with $ \hat{T}_1$, so let us finish the argument for that case first.
	
	The selected modes
	\begin{equation}
	\label{eq:modes_airy_struct_arb_autom_pf}
	{}^{\sigma}H_{i,k} \coloneqq \hat{T_1} \cdot {}^{\sigma}W_{i,k} \cdot \hat{T_1}^{-1} \,\qquad i\in [r],\qquad k\geq i - \lambda(i) + \delta_{i,1}
	\end{equation}
	with the partition
	\begin{equation*}
	\lambda =\begin{cases}
	\big({(r_{1}'+1)}^{s_{1}} , r_2 , r_3 , \ldots , r_{d-1} , {r_{d}'}^{s_{d}}\big) & ,\,r_d\neq 1 \\
	\big({(r_{1}'+1)}^{s_{1}} , r_2 , r_3 , \ldots , r_{d-1}\big) &,\,r_d=1
	\end{cases}
	\end{equation*}
	exactly correspond to the modes $({}^{\sigma}H_{i,k})_{(i,k)\in I^{>}_{\mathbf{r},\mathbf{s}}}$ where $I^{>}_{\mathbf{r},\mathbf{s}}$ is the index set defined in \eqref{eq:index_set_defn} by performing the identification of index sets via \cref{prop:Lie_subalg_cond_arb_autom}. Thus, \cref{prop:deg_one_cond_arb_autom} tells us that after a change of basis the modes \eqref{eq:modes_airy_struct_arb_autom_pf} satisfy the degree one condition. Since by assumption
	\begin{equation*}
	{}^{\sigma}H_{1,0} = {}^{\sigma}W_{1,0} = J^1_0 + \cdots +J^d_0 = \hbar^{\frac{1}{2}}\,(Q_1+\cdots+Q_d) = 0\,,
	\end{equation*}
	the modes \eqref{eq:modes_airy_struct_arb_autom_pf} satisfy the subalgebra condition if the modes $({}^{\sigma}H_{i,k})_{(i,k)\in I_\lambda}$ do. Here $I_\lambda$ is defined as in \eqref{eq:part_to_index_set}. Now using that ${}^{\sigma}H_{i,k}$ is obtained from ${}^{\sigma}W_{i,k}$ via conjugation the claim immediately follows from \cref{prop:desc_part_Lie_subalg}.
	
	For the general case, conjugating also with $ \hat{T}_2 $ and $ \hat{\Phi}$, note first of all that conjugation preserves commutation relations, so the subalgebra condition still holds. For the degree one condition, note that conjugation by $\hat{T}_2$ gives the shifts
	\begin{equation*}
	J^\mu_{-k} \longrightarrow J^\mu_{-k} + \hbar^{\frac{1}{2}}\,F_{\frac{1}{2},1}\big[\begin{smallmatrix} \mu \\ -k \end{smallmatrix}\big], 
	\end{equation*}
	which preserves degrees, and only acts on $ J^\mu_k$ with $  k< 0$, which do not occur in $ \pi_1\big(\hat{T}_1\cdot {}^{\sigma}W_{i,k}\cdot \hat{T}_1^{-1})$ by the previous parts of the computation. Likewise, conjugation by $ \hat{\Phi}$ acts as in \cref{PolChange}, which again preserves degrees and only affects $ J^\mu_k $ with $ k < 0$, so it also preserves the degree one condition.
\end{proof}

\medskip

\subsection{The exceptional case}
\label{sec:fixed_pt_no_shift}
\medskip

Contrary to the case considered before let us now allow $s_\mu = \infty$ for $\mu\in [d]$. Let us write
\[
\tilde{C}_{+}\coloneqq \bigsqcup_{\mu \in [d],\, s_\mu\neq \infty} \tilde{C}_{\mu}\,,
\]
and $\tilde{C}_-$ for the collection of all components $\tilde{C}_{\mu}$ on which $s_\mu = \infty$. 

\begin{lemma}
	\label{DegOneSeriesExcep}
	For any $i\in [r]$
	\begin{equation}
	\label{HfromJby01sExcep}
	\pi_1 \big({}^\sigma H_i (\xi ) \big) = \sum_{z \in \tilde{x}^{-1}(\xi )} \sum_{\substack{Z \subseteq \tilde{x}^{-1}(\xi ) \setminus \{ z\} \cap \tilde{C}_{+} \\ |Z|=i-1}} J(z) \prod_{z' \in Z} \omega_{0,1}(z')\,.
	\end{equation}
\end{lemma}

\begin{proof}
	The proof of this \namecref{DegOneSeriesExcep} is verbatim to the one of \cref{DegOneSeries} taking into account that $\omega_{0,1}\begin{psmallmatrix} \mu \\z \end{psmallmatrix}=0$ for all $\mu\in[d]$ with $s_\mu = \infty$.
\end{proof}

\begin{remark}
	\label{rem:generality_number_shifts}
	Let us make two important observations. First notice that if we write $r_+ \coloneqq \sum_{\mu \in [d],\, s_\mu\neq \infty} r_{\mu}$ then
	\begin{equation*}
	\pi_1 \big({}^\sigma H_{r_+ + 1} (\xi ) \big) = \sum_{z \in \tilde{x}^{-1}(\xi )\cap \tilde{C}_-} J(z) \prod_{z' \in \tilde{C}_{+}} \omega_{0,1}(z')
	\end{equation*}
	and moreover that for all $i>r_+ + 1$ we have
	\begin{equation*}
	\pi_1 \big({}^\sigma H_i (\xi ) \big) = 0\,.
	\end{equation*}
	Especially, from the last identity we deduce that, in order to end up with an Airy structure, it is necessary to have at most one $\mu\in[d]$ for which $s_\mu = \infty$. Moreover, necessarily for this $\mu$ we need $r_\mu=1$. Otherwise, there is no hope to obtain an Airy structure.
\end{remark}

Motivated by \cref{rem:generality_number_shifts} in the following we will assume that only $(r_d,s_d) = (1,\infty)$ while for all other $\mu\in[d-1]$ we have $s_\mu\neq\infty$, what we call the exceptional case  in \cref{thm:W_gl_Airy_arbitrary_autom}. Moreover, let us assume that as before $\frac{r_1}{s_1} \geq \ldots \geq \frac{r_{d-1}}{s_{d-1}}$. Rather than working with expression \eqref{HfromJby01sExcep} we will mainly use that by \eqref{eq:twist_mode_arbitr} we have
\begin{equation*}
{}^\sigma H_{i,k} = {}^\sigma H_{i,k}' + \sum_{a\in\mathbb{Z}} {}^\sigma H_{i-1,k-a}' \, J^d_a \qquad i\in[r],\qquad k\in\mathbb{Z}
\end{equation*}
where ${}^\sigma H_{i,k}'$ is obtained from ${}^\sigma H_{i,k}$ by formally setting $J^d_*$ equal to zero. Of course, ${}^\sigma H_{i,k}'$ may be computed via \eqref{HfromJby01s} replacing $d$ with $d-1$, i.e. these are modes of the standard case. Therefore, as the modes ${}^\sigma H_{i,k}$ are build up from modes considered in \cref{HfromJby01s} and an additional factor $J^d_*$, we can use the analysis of the standard case from the previous section in order to prove the degree one condition for these operators.

Let us select the following modes. For $i<r$ we define $\kLowerBound{i}{r}{s}$ exactly as in \eqref{eq:k_min} and set $\kLowerBound{r}{r}{s}\coloneqq 1 - \ssumind{d-1}$. Again, we define $I^{>}_{}$ as in \eqref{eq:index_set_defn} to be the index set associated to this choice of $\mf{d}^{>}_{\mathbf{r}, \mathbf{s}}$.

\begin{lemma}
	\label{prop:DegOneCondExcep}
	Assume ${\rm gcd}(r_\mu,s_\mu) = 1$  for all $\mu\in[d-1]$ and $t_\mu^{r_\mu}\neq t_\nu^{r_\nu}$ for any distinct $\mu,\nu$ with $\frac{r_\mu}{s_\mu} = \frac{r_\nu}{s_\nu}$. Then there exists an invertible matrix $\mc{N}$ such that
	\[ 
	{}^{\sigma}\tilde{H}_{\mu,a} = \sum_{(i,k) \in I^{>}_{\mathbf{r},\mathbf{s}}} \mc{N}_{(\mu,a),(i,k)}\,{}^{\sigma}H_{i,k} \qquad (\mu,a) \in [d] \times \mathbb{Z}_{> 0}
	\]  
	satisfy the degree one condition.
\end{lemma}
\begin{proof}
	First, let us argue that
	\begin{equation}
		\label{eq:Hexcepdegone}
		\pi_1\big({}^\sigma H_{i,k}\big) = \pi_1\big({}^\sigma H_{i,k}'\big) + \sum_{a\in\mathbb{Z}} \pi_0\big({}^\sigma H_{i-1,k-a}'\big) \, J^d_a
	\end{equation}
	for $k\geq \kLowerBound{i}{r}{s}$ is a linear combination of $J^\nu_a$s with $a>0$ only. Since for $i<r$ we defined $\kLowerBound{i}{r}{s}$ exactly as in \eqref{eq:k_min} it follows from \Cref{prop:k_min} that the first term $\pi_1\big({}^\sigma H_{i,k}'\big)$ is a linear combination of $J^\mu_a$s with $a>0$ only. That the second term $\sum_{a\in\mathbb{Z}} \pi_0\big({}^\sigma H_{i-1,k-a}'\big) \, J^d_a$ is a linear combination of $(J^d_a)_{a>0}$ only follows from the fact that $\pi_0\big({}^\sigma H_{i-1,k-a}'\big) = 0$ unless $k-a < \kLowerBound{i}{r}{s}$ as observed in \eqref{eq:degzerovanishes}. Hence,
	\begin{equation*}
		\pi_1\big({}^\sigma H_{i,k}\big) = \pi_1\big({}^\sigma H_{i,k}'\big) + \sum_{a>k-\kLowerBound{i}{r}{s}} \pi_0\big({}^\sigma H_{i-1,k-a}'\big) \, J^d_a \qquad i\in[r],\qquad k\geq \kLowerBound{i}{r}{s}
	\end{equation*}
	indeed lies in the linear span of $(J^\mu_a)_{(\mu,a) \in [d] \times \mathbb{Z}_{> 0}}$.
	
	In order to bring the operators into normal form, let us make use of our observation made earlier in \cref{rem:generality_number_shifts} that
	\begin{equation*}
		\pi_1\big({}^\sigma H_{r,k}\big) = \sum_{a>k-\kLowerBound{r}{r}{s}} \pi_0\big({}^\sigma H_{r-1,k-a}'\big) \, J^d_a\,.
	\end{equation*}
	This can be rephrased in the sense that $\pi_1\big({}^\sigma H_{r,k}\big) = \sum_{a > 0} \mc{A}_{k - \kLowerBound{r}{r}{s} + 1,a} \, J^d_a$ where $\mc{A}$ is an upper triangular matrix whose diagonal entries
	\begin{equation*}
		\mc{A}_{a,a} = \pi_0\big({}^\sigma H_{r-1, \kLowerBound{r}{r}{s} - 1}'\big) = (-1)^{r+d} \prod_{\mu = 1}^{d-1} \big(F_{0,1}\big[\begin{smallmatrix} \mu \\ -s_\mu \end{smallmatrix}\big] \big)^{r_\mu}\neq 0
	\end{equation*}
	may be read off from \eqref{eq:HfromJby01s_zerodeg} by taking only the leading order contributions of the $\omega_{0,1}$s into account. Thus, one can find a two-sided inverse of $\mc{A}$, and applying it to the semi-infinite vector $({}^{\sigma}H_{r,k})_{k\geq \kLowerBound{r}{r}{s}}$ we get $({}^{\sigma}\tilde{H}_{r,k})_{k\geq \kLowerBound{r}{r}{s}}$ for which
	\begin{equation*}
		\pi_1\big({}^\sigma \tilde{H}_{r,k}\big) = J^d_{k- \kLowerBound{r}{r}{s} + 1}\,.
	\end{equation*}
	By taking again suitable linear combinations, we can use the above modes in order to eliminate all $J^d_a$ from $\pi_1\big({}^\sigma H_{i,k}\big)$ for $i<r$, i.e. get operators $(\overline{H}_{i,k})_{(i,k) \in I^{>}_{\mathbf{r},\mathbf{s}},\, i<r}$ that satisfy
	\begin{equation*}
		\forall i \in  [r)\,,\,\,k\geq \kLowerBound{i}{r}{s}\,,\qquad \pi_1\big({}^{\sigma}\overline{H}_{i,k}\big) = \pi_1\big({}^{\sigma}H_{i,k}'\big)\,.
	\end{equation*}
	This expression is exactly the degree one projection of the operators considered in \cref{DegOneSeries} where we shifted in all cycles. From \cref{prop:deg_one_cond_arb_autom} we know that they can be brought in normal form, provided we can argue at last that the degree zero projection of these modes is vanishing. This is indeed the case: since $\pi_0\big({}^\sigma H_{i,k}\big) = \pi_0\big({}^\sigma H_{i,k}'\big)$ and we know by \eqref{eq:degzerovanishes} that the degree zero projection of $H_{i,k}'$ vanishes as long as $k\geq \kLowerBound{i}{r}{s}$, we see the same holds for ${}^\sigma H_{i,k}$.
\end{proof}

This provides us with everything we need to prove \cref{thm:W_gl_Airy_arbitrary_autom} in the exceptional case.

\begin{proof}[Proof of \cref{prop:AiryAllDilatonPolarize}, exceptional case]
	As in the standard case we know that conjugation with
	\begin{equation*}
	\hat{T}_2  \coloneqq\exp\left( \frac{1}{\hbar^{\frac{1}{2}}} \sum_{\mu \in [d]} \sum_{k > 0} F_{\frac{1}{2},1}\big[\begin{smallmatrix} \mu \\ -k \end{smallmatrix}\big]\,\frac{J^\mu_k}{k}\right) \,, \qquad \hat{\Phi} \coloneqq \exp\bigg( \frac{1}{2\hbar} \sum_{\substack{\mu,\nu \in [d] \\ k,l > 0}}  F_{0,2}\big[\begin{smallmatrix} \mu & \nu \\ -k & -l \end{smallmatrix}\big]\,\frac{J_k^\mu J_l^\nu}{kl} \bigg)
	\end{equation*}
	preserves the Airy structure conditions. It is hence sufficient to prove that
	\begin{equation*}
	{}^{\sigma}H_{i,k}\coloneqq \hat{T_1} \cdot {}^{\sigma}W_{i,k} \cdot \hat{T_1}^{-1} \,\qquad i\in [r]\,, \quad  k\geq 1 - \lambda(i) + \delta_{i,1} \,, \qquad \hat{T}_1 \coloneqq\exp\Bigg(\frac{1}{\hbar} \sum_{\substack{\mu \in [d] \\ s_\mu \neq \infty}} \sum_{k \geq s_\mu} F_{0,1}\big[\begin{smallmatrix} \mu \\ -k \end{smallmatrix}\big]\,\frac{J^\mu_k}{k}\Bigg)\,,
	\end{equation*}
	where we chose the partition
	\begin{equation*}
	\lambda = \big({(r_{1}'+1)}^{s_{1}} , r_2 , r_3 , \ldots , r_{d-1}\big)\,,
	\end{equation*}
	form an Airy structure. Indeed, as we have
	\begin{equation*}
	{}^{\sigma}H_{1,0} = \hbar^{\frac{1}{2}} (Q_1+\cdots + Q_d) = 0\,,
	\end{equation*}
	which vanishes by assumption, the selected modes already satisfy the subalgebra condition using \cref{prop:desc_part_Lie_subalg}. On the other hand, since $\lambda$ is chosen so that $1-\lambda(i)+\delta_{i,1} = \kLowerBound{i}{r}{s}$ for all $i\in [r]$, the selected modes also satisfy the degree one condition by \eqref{prop:DegOneCondExcep} as required.
\end{proof}

\medskip
\subsection{Approximate solution to the PDEs}
\label{sec:genuszerosolutionexists}
\medskip

Having analysed the cases in which $({}^\sigma H_{i,k})_{(i,k)\in I^{>}_{\mathbf{r},\mathbf{s}}}$ is an Airy structure and hence giving rise to a partition function $Z$ solving the associated system of differential equations to all orders in $\hbar$, we now turn our attention to the question in which case the associated differential equations are solved at leading order in the $\hbar$ expansion. For this we will use a technical lemma that holds for arbitrary families of differential operators satisfying the degree one condition.

\medskip
\subsubsection{The existence of a partition function}
\label{sec:genuszerosolutionexists2}
\medskip
Suppose we we are in the setting of \Cref{sec:AiryStructsFinDim}, i.e.\ $E$ is a finite-dimensional vector space over $\C$ with $(x_a)_{a\in A}$ a basis of $E^*$. Moreover, Let $\pi_k$ denote the projection to the $k$-graded piece of $\WeylAlg{E}{\hbar}$. This section is devoted to the question in which case a family $(H_a)_{a\in A}$ of differential operators satisfying the degree one condition admits a solution $F_0=\sum_{n\geq 3} \frac{1}{n!}F_{0,n}$ where $F_{0,n} \in {\rm Sym}^n(E^*)$ to the system of differential equations
\begin{equation}
	\label{eq:Hi_diff_eq_const_ord}
	\forall a\in A, \qquad e^{-\hbar^{-1} F_0} \, H_a \, e^{\hbar^{-1} F_0} \cdot 1 = o(\hbar^{\frac{1}{2}}) \,.
\end{equation}
We remark that due to the degree one condition such an $F_0$ is unique in case it exists and the free energies $F_{0,n}$ obey a recursion in $n$. Regarding the question of the existence of such a solution we have the following lemma.
\begin{lemma}
	\label{lem:suf_cond_genus0_soln}
	Suppose there is a family $(H_a)_{a\in A}$ of elements in $\WeylAlg{E}{\hbar}$ that (after normalisation) satisfies the degree one condition \eqref{eq:deg_one_cond} of an Airy structure and there is a second family $(H_j)_{j\in J}$ whose elements satisfy
	\begin{equation*}
		\pi_0(H_j)=0\,, \qquad \pi_1(H_j)\in\C\langle\hbar\partial_{x_*},\hbar^{\frac{1}{2}}\rangle
	\end{equation*}
	so that the combined family $(H_i)_{i\in A\cup J}$ satisfies the subalgebra condition \eqref{eq:gr_lie_subalg_cond} up to order $\hbar^{\frac{3}{2}}$, i.e. there exist $f^{i_3}_{i_1,i_2} \in \WeylAlg{E}{\hbar}$ so that for all $i_1,i_2\in A \cup J$
	\begin{equation*}
		[H_{i_1},H_{i_2}] - \hbar \sum_{i_3 \in A\cup J} f^{i_3}_{i_1,i_2} H_{i_3} = o(\hbar^{\frac{3}{2}})\,.
	\end{equation*}
	Then $(H_a)_{a\in A}$ admits a solution $F_0$ to the associated system of differential equations \eqref{eq:Hi_diff_eq_const_ord}.	
\end{lemma}
\begin{proof}
	After a change of basis we can assume that the family $(H_a)_{a\in A}$ is normalised, which means the operators take the form
	\begin{align*}
		H_a &= \hbar \partial_{x_a} + x^{\geq 2} + x^{\geq 1} (\hbar \partial)^{\geq 1} + (\hbar \partial)^{\geq 2} + o(\hbar^{\frac{1}{2}})  \qquad a\in A\,,\\
		H_j &= \sum_{b\in A} \mc{A}_{j,b} \hbar \partial_{x_b} + x^{\geq 2} + x^{\geq 1} (\hbar \partial)^{\geq 1} + (\hbar \partial)^{\geq 2} + o(\hbar^{\frac{1}{2}}) \qquad j\in J
	\end{align*}
	for some matrix $\mc{A}$. Here we used the same notation as in  \cite[Section 2.4]{KSTR} to represent terms in which certain symbols may appear with a fixed order. In the following we will construct free energies $F_{0,n}$ inductively in $n \geq 3$ so that
	\begin{equation}
		\label{eq:Hi_diff_eq_const_ord_induction}
		e^{-\hbar^{-1} \sum_{k=3}^n F_{0,k}} \, H_i \, e^{\hbar^{-1} \sum_{k=3}^n F_{0,k}} (1) = x^{\geq n} +  \hbar^{\geq \frac{1}{2}} x^{\geq 0}
	\end{equation}
	For all $i\in A\cup J$. Then taking $F \coloneqq \lim_{n \to \infty} \sum_{k=3}^n F_{0,k}$ we obtain a free energy that satisfies \eqref{eq:Hi_diff_eq_const_ord} as claimed.\\
	Note that since $H_i \cdot 1 = x^{\geq 2} +  \hbar^{\geq \frac{1}{2}} x^{\geq 0}$ the first non-trivial case is $n=3$. For this let us write $H_i = (\hbar \partial)^{=1} - H_{i;0,2} + x^{\geq 3} + x^{\geq 1} (\hbar \partial)^{\geq 1} + (\hbar \partial)^{\geq 2} + o(\hbar^{\frac{1}{2}})$, i.e.\ we write $- H_{i;0,2}$ for the terms $\hbar^{0} x^{=2} (\hbar \partial)^{=0}$. Then for $a_1,a_2\in A$ we have
	\begin{equation*}
		[H_{a_1},H_{a_2}] = \hbar \big(\partial_{a_2} H_{a_1;0,2} - \partial_{a_1} H_{a_2;0,2} + x^{\geq 2} + x^{\geq 0} (\hbar \partial)^{\geq 1} + o(\hbar^{\frac{1}{2}})\big)
	\end{equation*}
	and since by assumption the commutator has to lie in the left ideal generated by $\{H_i\}_{i\in A\cup J}$ we get $\partial_{a_2} H_{a_1;0,2} - \partial_{a_1} H_{a_2;0,2}=0$. As the last equality holds for all $a_1,a_2\in A$ there exists a unique homogenous polynomial $F_{0,3}$ of degree three in $x_*$ such that $\partial_a F_{0,3} = H_{a;0,2}$. With this choice \eqref{eq:Hi_diff_eq_const_ord_induction} is indeed satisfied for all $i\in A$ with $n=3$. So far we only used the exact same arguments as in \cite[Theorem 2.4.2]{KSTR}. Now we need to prove the compatibility of our choice for $F_{0,3}$ with $(H_j)_{j\in J}$. For this, notice that from
	\begin{equation*}
		[H_{j},H_{a}] = \hbar \big(\partial_{a} H_{j;0,2} - \sum_{b\in A}\mc{A}_{j,b}\partial_{b} H_{a;0,2} + x^{\geq 2} + x^{\geq 0} (\hbar \partial)^{\geq 1} + o(\hbar^{\frac{1}{2}})\big)
	\end{equation*}
	which holds for all $j\in J$ and $a\in A$ we can deduce that
	\begin{equation*}
		\partial_{a} \left(H_{j;0,2} - \sum_{b\in A}\mc{A}_{j,b}\partial_{b} F_{0,3}\right) = 0
	\end{equation*}
	which in turn implies that $H_{j;0,2} = \sum_{b\in A}\mc{A}_{j,b}\partial_{b} F_{0,3}$ since the term in brackets is a homogenous polynomial of order two. Therefore, we see that for $n=3$ equation \eqref{eq:Hi_diff_eq_const_ord_induction} is indeed satisfied for all $i\in J$ as well.
	
	Now for the induction step, suppose we have already constructed $\sum_{k=3}^n F_{0,k}$ solving \eqref{eq:Hi_diff_eq_const_ord_induction} for all $i\in A\cup J$. Then after conjugation $\overline{H}_i \coloneqq e^{-\hbar^{-1}\sum_{k=3}^n F_{0,k}} H_i e^{\hbar^{-1}\sum_{k=3}^n F_{0,k}}$ the operators take the form
	\begin{align*}
		\overline{H}_a &= \hbar \partial_{x_a} - \overline{H}_{a;0,n} + x^{\geq n+1} + x^{\geq 1} (\hbar \partial)^{\geq 1} + (\hbar \partial)^{\geq 2} + o(\hbar^{\frac{1}{2}})   && a\in A\,,\\
		\overline{H}_j &= \sum_{b\in A} \mc{A}_{j,b} \hbar \partial_{x_b} - \overline{H}_{j;0,n} + x^{\geq n+1} + x^{\geq 1} (\hbar \partial)^{\geq 1} + (\hbar \partial)^{\geq 2} + o(\hbar^{\frac{1}{2}}) && j\in J
	\end{align*}
	where we wrote $-\overline{H}_{i;0,n}$ for the piece $\hbar^{0}x^{=n} (\hbar\partial)^{=0}$ in $\overline{H}_i$. We use the exact same argument as for the base case and deduce from
	\begin{equation*}
		[H_{a_1},H_{a_2}] = \hbar \big(\partial_{a_2} H_{a_1;0,n} - \partial_{a_1} H_{a_2;0,n} + x^{\geq 2} + x^{\geq 0} (\hbar \partial)^{\geq 1} + o(\hbar^{\frac{1}{2}})\big)
	\end{equation*}
	that $\partial_{a_2} H_{a_1;0,n} - \partial_{a_1} H_{a_2;0,n}=0$. Hence, there is a homogenous polynomial $F_{0,n+1}$ of degree $n+1$ satisfying $\partial_a F_{n+1} = \overline{H}_{a;0,n}$ as required. The compatibility with the operators $(\overline{H}_j)_{j\in J}$ similarly follows from inspecting
	\begin{equation*}
		[H_{j},H_{a}] = \hbar \big(\partial_{a} H_{j;0,n} - \sum_{b\in A}\mc{A}_{j,b}\partial_{b} H_{a;0,n} + x^{\geq 2} + x^{\geq 0} (\hbar \partial)^{\geq 1} + o(\hbar^{\frac{1}{2}})\big)\qquad j\in J, \, a\in A
	\end{equation*}
	from which we deduce that
	\begin{equation*}
		\partial_{a} \left(H_{j;0,n} - \sum_{b\in A}\mc{A}_{j,b}\partial_{b} F_{0,n+1}\right) = 0\,.
	\end{equation*}
	For this to be satisfied $H_{j;0,n} - \sum_{b\in A}\mc{A}_{j,b}\partial_{b} F_{0,n+1} =0$ must vanish for itself. So we see that with our choice for $F_{0,n+1}$ we indeed have
	\begin{equation*}
		e^{-\hbar^{-1} \sum_{k=3}^{n+1} F_{0,k}} \, H_i \, e^{\hbar^{-1} \sum_{k=3}^{n+1} F_{0,k}} \cdot 1 = e^{-\hbar^{-1} F_{0,n+1}} \, \overline{H}_i \, e^{\hbar^{-1} F_{0,n+1}} \cdot 1  = x^{\geq n+1} +  \hbar^{\geq \frac{1}{2}} x^{\geq 0}
	\end{equation*}
	for all $i\in A\cup J$ which proves the induction step.
\end{proof}

\begin{remark}
	For expositional reasons we chose to prove \Cref{lem:suf_cond_genus0_soln} in the setting of finite dimensional vector spaces but we stress that the statement of the lemma carries over to the case where $E$ is a filtered vector space and $(H_a)_{a\in A }$ a filtered family of elements in $\WeylAlgComp{E}{\hbar}$. We leave it to the reader to check the details.
\end{remark}
\begin{remark}
	In the language of \cite{KSTR,BBCCN18} \Cref{lem:suf_cond_genus0_soln} is only a statement about the \textit{classical limit} of the family $(H_a)_{a\in A}$ since all higher order terms in $\hbar^{\frac{1}{2}}$ essentially play no role. We chose to state the statement in the \textit{quantum} setting nevertheless since all differential operators considered in this paper come with a natural $\hbar$-refinement.
\end{remark}

\medskip
\subsubsection{The proof of \Cref{thm:genus0_soln_admissible_Hik}}
\label{sec:genus0_soln_admissible_Hik}
\medskip
Let us now apply \Cref{lem:suf_cond_genus0_soln} to the family of modes $({}^\sigma H_{i,k})_{(i,k)\in I^{>}_{\mathbf{r},\mathbf{s}}}$ from \Cref{sec:AllDilatonPolarize} in order to work out sufficient conditions under which the operators admit a solution to their associated system of differential equations up to corrections in $\hbar^{\frac{1}{2}}$. In \Cref{prop:deg_one_cond_arb_autom} and \Cref{prop:DegOneCondExcep} we already analysed the cases in which this family satisfies the degree one condition. So in order to apply \Cref{lem:suf_cond_genus0_soln} we need to find a set $J\subseteq [r]\times\Z$ so that the combined family $({}^\sigma H_{i,k})_{(i,k)\in I_{\mathbf{r},\mathbf{s}}^{>} \cup J}$ satisfies the subalgebra condition and the operators $({}^\sigma H_{i,k})_{(i,k)\in J}$ feature no terms involving $x^*_*$ in degree one. As shown in \Cref{prop:k_min_geq} the latter requirement forces that we choose $J\subseteq I_{\mathbf{r},\mathbf{s}}^{\geq} \setminus I_{\mathbf{r},\mathbf{s}}^{>}$. Therefore, again by making use of \Cref{prop:desc_part_Lie_subalg} the question about when the operators $({}^\sigma H_{i,k})_{(i,k)\in I_{\mathbf{r},\mathbf{s}}^{>}}$ admit a solution to the associated system of differential equations to leading order in $\hbar^{\frac{1}{2}}$ is reduced to the question whether we can find a descending partition $\lambda\vdash r$ for which $I_{\mathbf{r},\mathbf{s}}^{>} \subseteq I_\lambda \subseteq I_{\mathbf{r},\mathbf{s}}^{\geq}$ because if such a partition exists we can simply take $J = I_\lambda\setminus I_{\mathbf{r},\mathbf{s}}^{>}$ and apply \Cref{lem:suf_cond_genus0_soln}. A  characterisation when this is possible is presented in the following.
\begin{lemma}
	\label{lem:mode_set_extendable}
	Given $\frac{r_1}{s_1}\geq\cdots\geq \frac{r_d}{s_d}$ with  $r_\mu$ and $s_\mu$ coprime integers for all $\mu$, there exists a descending partition $\lambda$ such that
	\begin{equation}
		\label{eq:mode_set_extension}
		I^{>}_{\mathbf{r},\mathbf{s}} \subseteq I_\lambda \subseteq I^{\geq}_{\mathbf{r},\mathbf{s}}
	\end{equation}
	if the following points are satisfied.
	\begin{enumerate}
		\item\label{item:r_eq_pm1mods1}  $r_\mu = \pm 1 \,\,{\rm mod}\,\, s_\mu$ for all $\mu\in[d]$.
		
		\item\label{item:fract_constr} For all $\mu_1 \neq \mu_2$ with $s_{\mu_i}>2$ such that either 
		\begin{equation*}
			r_{\mu_1} = 1 \,\,{\rm mod}\,\, s_{\mu_1} \text{ and } r_{\mu_2} = 1 \,\,{\rm mod}\,\, s_{\mu_2}
		\end{equation*}
		or
		\begin{equation*}
			r_{\mu_1} = -1 \,\,{\rm mod}\,\, s_{\mu_1} \text{ and } r_{\mu_2} = -1 \,\,{\rm mod}\,\, s_{\mu_2}
		\end{equation*}
		one has $\big\lfloor \tfrac{r_{\mu_1}}{s_{\mu_1}} \big\rfloor \neq \big\lfloor\tfrac{ r_{\mu_2}}{s_{\mu_2}}\big\rfloor$.
		
		\item\label{item:three_equal_constr} If there are pairwise distinct $\mu_1,\mu_2,\mu_3\in[d]$ with $\big\lfloor\frac{r_{\mu_1}}{s_{\mu_1}}\big\rfloor = \big\lfloor\frac{r_{\mu_2}}{s_{\mu_2}}\big\rfloor = \big\lfloor\frac{r_{\mu_3}}{s_{\mu_3}}\big\rfloor$, then there is an $m\in \{1,2,3\}$ for which $s_{\mu_m}=1$.
	\end{enumerate}
\end{lemma}

\begin{proof}
	First, let us describe what a possibly non-descending partition $\lambda'$ looks like that satisfies $I_{\lambda'} = I^{\geq}_{\mathbf{r},\mathbf{s}}$. On the intervals $i=i'+\rsumind{\nu-1}\in (\rsumind{\nu-1},\rsumind{\nu}]$ we have
	\begin{equation*}
		1 - \mf{d}_{\mathbf{r}, \mathbf{s}}(i) = \lambda^{r_\nu,s_\nu}(i') + \ssumind{\nu-1} 
	\end{equation*}
	where $\lambda^{r_\nu,s_\nu}$ is the partition described in \Cref{prop:k_min_partition_one_cycle}. To analyse the transitions between these intervals let us for a moment assume that $r_\nu>1$ and $s_\nu<r_\nu$ for all $\nu\in d$. These special cases will be treated separately later. Under these assumptions
	\begin{equation*}
		\big(1 - \mf{d}_{\mathbf{r}, \mathbf{s}}^{\geq}(\rsumind{\nu}+1)\big) -  \big(1 -\mf{d}_{\mathbf{r}, \mathbf{s}}^{\geq}(\rsumind{\nu-1})\big) = (1 + \ssumind{\nu}) - \ssumind{\nu} = 1 
	\end{equation*}
	and thus the partition $\lambda'=(\lambda^{r_1,s_1},\ldots,\lambda^{r_d,s_d})$ satisfies $I_{\lambda'} = I^{\geq}_{\mathbf{r},\mathbf{s}}$. Since by assumption \labelcref{item:r_eq_pm1mods1} each block $\lambda^{r_\nu,s_\nu}$ in $\lambda'$ is descending we only need to check whether it is descending at the transitions, i.e.\ whether $\lambda^{r_\nu,s_\nu}_{\ell^{r_\nu,s_\nu}} \geq \lambda^{r_{\nu+1},s_{\nu+1}}_{1}$ for all $\nu\in[d-1]$, and if not whether we can modify $\lambda'$ into a descending partition $\lambda$ satisfying \eqref{eq:mode_set_extension}. For this note that $\lambda^{r_\nu,s_\nu}_{\ell^{r_\nu,s_\nu}} \geq \lambda^{r_{\nu+1},s_{\nu+1}}_{1}$ holds if and only if
	\begin{equation*}
		\left\lfloor \frac{r_{\nu}}{s_{\nu}} \right\rfloor \geq \left\lceil \frac{ r_{\nu+1}}{s_{\nu+1}} \right\rceil\,.
	\end{equation*}
	This means in case $\big\lfloor \frac{r_{\nu}}{s_{\nu}} \big\rfloor > \big\lfloor\frac{ r_{\nu+1}}{s_{\nu+1}}\big\rfloor$ for all $\nu\in [d-1]$ the partition $\lambda'$ is indeed descending and so we can choose $\lambda=\lambda'$. If however there is a $\mu\in[d]$ for which $\big\lfloor \frac{r_{\mu}}{s_{\mu}} \big\rfloor = \big\lfloor\frac{ r_{\mu+1}}{s_{\mu+1}}\big\rfloor$ we need to be more careful. The case where at least one of $s_\mu$ or $s_{\mu+1}$ equals $1$ is unproblematic since then $\big\lfloor \frac{r_{\mu}}{s_{\nu}} \big\rfloor \geq \big\lceil \frac{ r_{\nu+1}}{s_{\mu+1}} \big\rceil$ is automatically satisfied. However, if $s_{\mu},s_{\mu+1}>1$ we have $\big\lfloor \frac{r_{\mu}}{s_{\nu}} \big\rfloor < \big\lceil \frac{ r_{\nu+1}}{s_{\mu+1}} \big\rceil$. The lack of $\lambda'$ being descending at this transition can be cured as follows. First, notice that \labelcref{item:fract_constr} forces that $r_{\mu} = -1 \,\,{\rm mod}\,\, s_{\mu}$ and $r_{\mu+1} = 1 \,\,{\rm mod}\,\, s_{\mu+1}$. One should remark here that the second case $r_{\mu} = 1 \,\,{\rm mod}\,\, s_{\mu}$ and $r_{\mu+1} = -1 \,\,{\rm mod}\,\, s_{\mu+1}$ allowed by \labelcref{item:fract_constr} can only occur if $s_{\mu} = s_{\mu+1} =2$ and is hence contained in the first case. Pictorially, we are in the following setting
	\begin{align}
		\lambda' = ~ ~ ~ &
		\Scale[0.6]{\begin{ytableau}%
				\none[\vdots] & \none & \none & \none & \none\\
				\, & \, & \none[\dots] & \, & \, \\
				\, & \, & \none[\dots] & \, & \, \\
				\, & \, & \none[\dots] & \, & \, \\
				\none[\vdots] & \none[\vdots] & \none[\ddots] & \none[\vdots] & \none[\vdots] \\
				\, & \, & \none[\dots] & \, & \, \\
				\, & \, & \none[\dots] & \, \\
				\, & \, & \none[\dots] & \, & \, \\
				\, & \, & \none[\dots] & \, \\
				\, & \, & \none[\dots] & \, \\
				\none[\vdots] & \none[\vdots] & \none[\ddots] & \none[\vdots] \\
				\, & \, & \none[\dots] & \, \\
				\, & \, & \none[\dots] & \, \\
				\none[\vdots] & \none & \none & \none & \none
		\end{ytableau}}
		\begin{array}{ll}
			\left.\phantom{\Scale[0.6]{\ydiagram{1,1,1,1,1,1}}}\right\} &= \lambda^{r_{\mu},s_{\mu}} \\[-0.1em]
			\left.\phantom{\Scale[0.6]{\ydiagram{1,1,1,1,1,1}}}\right\} &= \lambda^{r_{\mu+1},s_{\mu+1}}
		\end{array}%
		\label{eq:problemcases1}
	\end{align}
	We will now make use of the fact that there is a certain ambiguity in choosing $\lambda$ so that $I_\lambda$ lies between $I^{>}_{\mathbf{r},\mathbf{s}}$ and $I^{\geq}_{\mathbf{r},\mathbf{s}}$. So far we focused on the partition $\lambda'$ satisfying $I_{\lambda'}=I^{\geq}_{\mathbf{r},\mathbf{s}}$. Now notice that deleting an element $(\rsumind{\mu}+1,-\ssumind{\mu})\in I^{\geq}_{\mathbf{r},\mathbf{s}} \setminus I^{>}_{\mathbf{r},\mathbf{s}}$ affects the partition describing the index set in the sense that one box from the first row of $\lambda^{r_{\mu+1},s_{\mu+1}}$ gets shifted to the last row of $\lambda^{r_{\mu},s_{\mu}}$, i.e. the index sets associated to the partitions
	\begin{align*}
		\lambda' & = (\lambda^{r_1,s_1}_1,\ldots,\lambda^{r_\mu,s_\mu}_{\ell^{r_\mu,s_\mu}-1},\lambda^{r_\mu,s_\mu}_{\ell^{r_\mu,s_\mu}},\lambda^{r_{\mu+1},s_{\mu+1}}_{1},\lambda^{r_{\mu+1},s_{\mu+1}}_{2},\ldots,\lambda^{r_{d},s_{d}}_{\ell^{r_{d},s_{d}}})\,,\\
		\lambda'' & =(\lambda^{r_1,s_1}_1,\ldots,\lambda^{r_\mu,s_\mu}_{\ell^{r_\mu,s_\mu}-1},\lambda^{r_\mu,s_\mu}_{\ell^{r_\mu,s_\mu}} + 1 ,\lambda^{r_{\mu+1},s_{\mu+1}}_{1}-1,\lambda^{r_{\mu+1} ,s_{\mu+1}}_{2},\ldots,\lambda^{r_{d},s_{d}}_{\ell^{r_{d},s_{d}}})
	\end{align*}
	are related by $I_{\lambda''}=I_{\lambda'}\setminus\{(\rsumind{\mu}+1,-\ssumind{\mu})\}$. The index set $I_{\lambda''}$ then still satisfies $I^{>}_{\mathbf{r},\mathbf{s}} \subseteq I_{\lambda''} \subseteq I^{\geq}_{\mathbf{r},\mathbf{s}}$ as required. We notice now that if we perform such a box shift in the case displayed in \eqref{eq:problemcases1} we indeed obtain a partition $\lambda''$ that is descending at the transition between the $\mu$th and $(\mu+1)$th block. This solves the problem of having two subsequent blocks with \mbox{$\big\lfloor\frac{r_\mu}{s_\mu}\big\rfloor = \big\lfloor\frac{r_{\mu+1}}{s_{\mu+1}}\big\rfloor$}. However, if there is a sequence \mbox{$\big\lfloor\frac{r_\mu}{s_\mu}\big\rfloor = \ldots = \big\lfloor\frac{r_{\mu+a}}{s_{\mu+a}}\big\rfloor$} with $a>1$ then by \ref{item:three_equal_constr} we must already have $s_{\mu+2}=\ldots=s_{\mu+a}=1$. This means that if also $s_{\mu+1}=1$ there is nothing to worry about and in case $s_{\mu+1}>1$ we have
	\begin{align*}
		\lambda' = ~ ~ ~ &
		\Scale[0.6]{\begin{ytableau}%
				\none[\vdots] & \none & \none & \none & \none\\
				\, & \, & \none[\dots] & \, & \, \\
				\, & \, & \none[\dots] & \, & \, \\
				\, & \, & \none[\dots] & \, & \, \\
				\none[\vdots] & \none[\vdots] & \none[\ddots] & \none[\vdots] & \none[\vdots] \\
				\, & \, & \none[\dots] & \, & \, \\
				\, & \, & \none[\dots] & \, \\
				\, & \, & \none[\dots] & \, & \, \\
				\, & \, & \none[\dots] & \, \\
				\, & \, & \none[\dots] & \, \\
				\none[\vdots] & \none[\vdots] & \none[\ddots] & \none[\vdots] \\
				\, & \, & \none[\dots] & \, \\
				\, & \, & \none[\dots] & \, \\
				\, & \, & \none[\dots] & \, \\
				\none[\vdots] & \none[\vdots] & \none & \none[\vdots]\\
				\, & \, & \none[\dots] & \, \\
				\none[\vdots] & \none & \none & \none
		\end{ytableau}}
		\begin{array}{ll}
			&\\[-0.5em]
			\left.\phantom{\Scale[0.6]{\ydiagram{1,1,1,1,1,1}}}\right\} &= \lambda^{r_{\mu},s_{\mu}} \\[-0.1em]
			\left.\phantom{\Scale[0.6]{\ydiagram{1,1,1,1,1,1}}}\right\} &= \lambda^{r_{\mu+1},s_{\mu+1}}\\[-0.1em]
			&= \lambda^{r_{\mu+2},s_{\mu+2}}\\[-0.4em]
			&\qquad \Scale[0.6]{\vdots} \\[-0.2em]
			&= \lambda^{r_{\mu+a},s_{\mu+a}}
		\end{array}%
	\end{align*}
	and so we can perform a box-shift in order to construct a descending partition $\lambda$ from $\lambda'$.	This finishes the construction of a partition $\lambda$ satisfying \eqref{eq:mode_set_extension} if $s_\nu < r_\nu$ and $r_\nu>1$ for all $\nu\in[d]$.\\
	Hence, we are left with the cases in which possibly $r_\mu=1$ or $s_\mu=r_\mu+1$ for a $\mu\in[d]$. So suppose there is a $\mu$ with $s_\mu=r_\mu=1$ and assume $\mu\in[d]$ is minimal with this property. Then from our assumption that $\frac{r_{\nu_1}}{s_{\nu_1}} \geq \frac{r_{\nu_2}}{s_{\nu_2}}$ for $\nu_1 \leq \nu_2$ we deduce that necessarily $s_\nu\leq r_\nu$ and $r_\nu>1$ for all $\nu<\mu$. Hence, we know from our previous analysis that we can find a partition $\lambda\vdash\rsumind{\mu-1}$ whose associated index set $I_{\lambda}$ satisfies
	\begin{equation*}
		I^{>}_{\mathbf{r},\mathbf{s}} \cap \{(i,k) ~\rvert~ i \leq \rsumind{\mu-1} \} \subseteq I_{\lambda} \subseteq I^{\geq}_{\mathbf{r},\mathbf{s}} \cap\{(i,k) ~\rvert~ i \leq \rsumind{\mu-1} \}\,.
	\end{equation*}
	Now let $\delta\geq \mu$ be maximal with $r_\nu=s_\nu=1$ for all $\nu\in[\mu,\delta]$. Then we compute
	\begin{equation*}
		\big(1-\mf{d}^{\geq}_{\mathbf{r},\mathbf{s}}(\rsumind{\nu-1}+1)\big) - \big(1-\mf{d}^{\geq}_{\mathbf{r},\mathbf{s}}(\rsumind{\nu-1})\big) = 1
	\end{equation*}
	for all $\nu\in[\mu,\delta]$ and thus
	\begin{equation*}
		I^{>}_{\mathbf{r},\mathbf{s}} \cap\{(i,k) ~\rvert~ i \leq \rsumind{\delta} \} \subseteq I_{\lambda'} \subseteq I^{\geq}_{\mathbf{r},\mathbf{s}} \cap\{(i,k) ~\rvert~ i \leq \rsumind{\delta} \}
	\end{equation*}
	if we choose $\lambda'= (\lambda,1^{\delta-\mu+1})$ which is descending. Since we have chosen $\delta$ maximal we must have $s_\nu>r_\nu$ for all $\nu>\delta$ which however implies that $\delta\in\{d-2,d-1\}$ by \ref{item:three_equal_constr}.  First let us assume $\delta=d-1$. Then
	\begin{equation*}
		-\mf{d}^{\geq}_{\mathbf{r},\mathbf{s}}(\rsumind{d-1}+1) + \mf{d}^{\geq}_{\mathbf{r},\mathbf{s}}(\rsumind{d-1}) = 1
	\end{equation*}
	and thus
	\begin{equation*}
		I^{>}_{\mathbf{r},\mathbf{s}} \subseteq I_{(\lambda',1^{r_d})} \subseteq I^{\geq}_{\mathbf{r},\mathbf{s}} \,.
	\end{equation*}
	Now suppose $\delta=d-2$. In this case we need to be more careful. First, let us assume $r_{d-1}>1$ which due to \ref{item:r_eq_pm1mods} forces $s_{d-1}=r_{d-1}+1$ and thus also $r_{d}=1$ due to \ref{item:fract_constr}. We find that between the two transition points
	\begin{equation*}
		-\mf{d}^{\geq}_{\mathbf{r},\mathbf{s}}(\rsumind{d-2}+1) + \mf{d}^{\geq}_{\mathbf{r},\mathbf{s}}(\rsumind{d-2}) = 1\,, \qquad -\mf{d}^{\geq}_{\mathbf{r},\mathbf{s}}(\rsumind{d-1}+1) + \mf{d}^{\geq}_{\mathbf{r},\mathbf{s}}(\rsumind{d-1}) = 2\,.
	\end{equation*}
	The fact that $\mf{d}^{\geq}_{\mathbf{r},\mathbf{s}}$ jumps by two between $\rsumind{d-1} \rightarrow \rsumind{d-1}+1=r$ is only a mild problem that can be cured by not including $(\rsumind{d},-\ssumind{d-1})$ into our index set. Indeed, if we choose $\lambda''=(\lambda',1^{r_{d-1}},1)$ then the associated index set indeed has the desired property that $I^{>}_{\mathbf{r},\mathbf{s}} \subseteq I_{\lambda''} \subseteq I^{\geq}_{\mathbf{r},\mathbf{s}}$.\\
	Lastly, let us assume that $r_{d-1}=1$. In this case \ref{item:r_eq_pm1mods} and \ref{item:fract_constr} only leave us with $s_{d-1}=2$ and $r_d=1$. Recall here that we treated the case $s_{d-1}=1$ earlier. This time we find
	\begin{equation*}
		-\mf{d}^{\geq}_{\mathbf{r},\mathbf{s}}(\rsumind{d-2}+1) + \mf{d}^{\geq}_{\mathbf{r},\mathbf{s}}(\rsumind{d-2}) = 1\,, \qquad -\mf{d}^{\geq}_{\mathbf{r},\mathbf{s}}(\rsumind{d-1}+1) + \mf{d}^{\geq}_{\mathbf{r},\mathbf{s}}(\rsumind{d-1}) = 2\,.
	\end{equation*}
	Again we cure the fact that $\mf{d}^{\geq}_{\mathbf{r},\mathbf{s}}$ jumps by the value of two by deleting $(\rsumind{d},-\ssumind{d-1})$ from the index set and finally as in the case before we find that $(\lambda',1,1)$ satisfies the desired property $I^{>}_{\mathbf{r},\mathbf{s}} \subseteq I_{(\lambda',1,1)} \subseteq I^{\geq}_{\mathbf{r},\mathbf{s}}$.
\end{proof}

\begin{remark}
	The statement of \Cref{lem:mode_set_extendable} can easily be extended to the exceptional case. If $r_d=1$ then $\mf{d}^{>}_{\mathbf{r},\mathbf{s}}$ and $\mf{d}^{\geq}_{\mathbf{r},\mathbf{s}}$ actually do not depend on the choice of $s_d$ for finite $s_d$. Therefore, we can extend the definition of $\mf{d}^{>}_{\mathbf{r},\mathbf{s}}$ and $\mf{d}^{\geq}_{\mathbf{r},\mathbf{s}}$ to the case $s_d=\infty$ by choosing $\mf{d}^{>}_{\mathbf{r},\mathbf{s}}(r) = \mf{d}^{\geq}_{\mathbf{r},\mathbf{s}}(r)+1 = -\ssumind{d-1}+1$ as in the finite setting. Thus, all the results of \Cref{lem:mode_set_extendable} easily translate to the exceptional case if we adopt the convention that we say $r_d = 1\,\,{\rm mod}\,\, s_d$ if $(r_d,s_d)=(1,\infty)$ and interpret $\frac{r_d}{s_d}$ as being zero.
\end{remark}

As motivated in the introduction of this subsection we can now easily use \Cref{lem:mode_set_extendable} to prove  \Cref{thm:genus0_soln_admissible_Hik}.

\begin{proof}[Proof of \Cref{thm:genus0_soln_admissible_Hik}]
	The partition $\lambda$ in the statement of the \namecref{thm:genus0_soln_admissible_Hik} is chosen so that $1-\lambda(i)+\delta_{i,1}=\mf{d}^{>}_{\mathbf{r},\mathbf{s}}(i)$ for all $i\in [r]$. Thus, we already know from \Cref{prop:deg_one_cond_arb_autom,prop:DegOneCondExcep} that the selected modes satisfy the degree one condition. Moreover, \Cref{lem:mode_set_extendable} tells us that we can choose a partition $\lambda'$ so that the modes $({}^\sigma H_{i,k})_{(i,k)\in J}$ associated to the index set $J=I_{\lambda'} \setminus I^{>}_{\mathbf{r},\mathbf{s}}$ satisfy the properties required in \Cref{lem:suf_cond_genus0_soln}. Therefore, the \namecref{lem:suf_cond_genus0_soln} ensures the sought-after existence of a leading order solution to the system of differential equations associated to $({}^\sigma H_{i,k})_{(i,k)\in I_\lambda}$.
\end{proof}

\newpage

\part{Spectral curve descriptions}
\label{SCpart}

In this second part, we translate the differential constraints coming from the $\mathcal{W}$-algebra representations of \cref{S1} into constraints on the order of poles of certain combinations of multidifferentials $\omega_{g,n}$ on a spectral curve built from the coefficients $F_{g,n}$. The latter constraints are called ``abstract loop equations''. In a second step, we show that the unique solution to the abstract loop equations is provided by an adaptation of Bouchard--Eynard topological recursion to the setting of singular spectral curves. In fact, this provides us with the right definition of the topological recursion \`a la Chekhov--Eynard--Orantin in this setting, together with the proof that it is well-defined.

\section{From Airy structures to local spectral curves}
\label{SAIRPS}
\subsection{Fields for a single cycle}
\label{onecyc}

\medskip

We will start by reconsidering \cref{sec:twisted_modules} in the case of $ \sigma$ consisting of a single cycle, of length $ r$. In this case, we can omit all $ \mu$-indices, and consider
\begin{equation*}
J_{k} = \left\{\begin{array}{lll} \hbar \partial_{x_k} & & {\rm if}\,k > 0 \\ \hbar^{\frac{1}{2}}Q & & {\rm if}\,\,k = 0 \\ -kx_{-k} & & {\rm if}\,\,k < 0 \end{array}\right. \,,
\end{equation*}
the standard representation of the Heisenberg algebra of $\mathfrak{gl}_{r}$. It is useful to write $\tilde{x} = z^r$. We split the current as follows:
\begin{equation*}
\begin{split}
J(z) & = \sum_{k \in \mathbb{Z}} \frac{J_{k}\,\dd z}{z^{k + 1}} = J_+(z) + J_-(z) + \hbar^{\frac{1}{2}}\,\frac{Q\dd z}{z}\,, \\
J_+(z) & = \sum_{k > 0} J_{-k}\,z^{k - 1}\,\dd z\,, \\
J_{-}(z) & = \sum_{k > 0} \frac{J_{k}\,\dd z}{z^{k + 1}}\,,
\end{split}
\end{equation*}
Choose a primitive $r$th root of unity $ \theta $ and let $\mathfrak{f}(z) = \{z,\theta z,\ldots,\theta^{r - 1}z\}$. Set
\[
\omega_{0,2}^{{\rm std}}(z_1,z_2) = \frac{\dd z_1 \dd z_2}{(z_1 - z_2)^2}\,.
\]
We can rewrite \eqref{Wikz} as
\begin{equation}\label{WiAsCorrelators}
W_i(\tilde{x}) = \frac{1}{r} \! \sum_{Z\,:\,[i] \hookrightarrow \mathfrak{f}(z)} \sum_{\substack{0 \leq j \leq \lfloor i/2 \rfloor \\ A_0 \sqcup A_+ \sqcup A_- = [2j + 1,i]}} \!\!\!\!\! \frac{\hbar^{j}}{2^{j}j!(i - 2j)!} \prod_{l = 1}^{j} \omega_{0,2}^{{\rm std}}(Z_{2l - 1},Z_{2l}) \prod_{l \in A_0} \frac{Q \dd z}{z} \prod_{l \in A_+} J_+(Z_l) \prod_{l \in A_-} J_-(Z_l)\,,
\end{equation}
As in \cref{sec:AllDilatonPolarize}, let us apply a general dilaton shift and change of polarisation to these operators. We take
\begin{equation*}
\begin{split}
\hat{T} & = \exp\bigg(\sum_{k > 0} \big(\hbar^{-1}\,F_{0,1}[-k] + \hbar^{-\frac{1}{2}}\,F_{\frac{1}{2},1}[-k]\big)\frac{J_{k}}{k}\bigg)\,, \\
\hat{\Phi} & = \exp\bigg(\frac{1}{2\hbar} \sum_{k_1,k_2 > 0} F_{0,2}[-k_1,-k_2]\,\frac{J_{k_1}J_{k_2}}{k_1k_2} \bigg)\,,
\end{split}
\end{equation*}
in which we can always assume that $F_{0,2}[-k_1,-k_2] = F_{0,2}[-k_2,-k_1]$, and introduce
\[
H_{i}(\tilde{x}) \coloneqq  \hat{\Phi}\hat{T} \cdot W_i(\tilde{x})  \cdot \hat{T}^{-1}\hat{\Phi}^{-1} \,.
\]
The effect of the dilaton shift $\hat{T}$ in \eqref{WiAsCorrelators} is to replace $J_+(z)$ with $J_+(z) + \hbar^{\frac{1}{2}}\big(\omega_{\frac{1}{2},1}(z) - Q\tfrac{\dd z}{z}\big)+ \omega_{0,1}(z)$, where
\begin{equation*}
\begin{split}
\omega_{0,1}(z) & \coloneqq  \sum_{k > 0} F_{0,1}[-k]\,z^{k - 1}\dd z\,,\\
\omega_{\frac{1}{2},1}(z) & \coloneqq  Q\,\frac{\dd z}{z} + \sum_{k > 0} F_{\frac{1}{2},1}[-k]\,z^{k - 1}\dd z\,.
\end{split}
\end{equation*}
Using the Baker--Campbell--Hausdorff formula, it is easy to see that the net effect of the change of polarisation $\hat{\Phi}$ is to replace $\omega_{0,2}^{{\rm std}}$ with
\[
\omega_{0,2}(z_1,z_2) \coloneqq  \frac{\dd z_1 \dd z_2}{(z_1 - z_2)^2} + \sum_{k_1,k_2 > 0} F_{0,2}[-k_1,-k_2]\,z_1^{k_1 - 1}z_2^{k_2 - 1}\,\dd z_1 \dd z_2\,,
\]
and to replace $J_-(z)$ with\footnote{More precisely, the conjugation by $\hat{\Phi}$ leaves $J_-(z)$ invariant and add extra terms with positive $J$s in $J_+(z)$. We collect all terms with negative (resp. positive) $J$s in $\mathcal{J}_+$ (resp $\mathcal{J}_-$), thus leading to \eqref{Jmoinsz}.}
\begin{equation}
\label{Jmoinsz}
\mathcal{J}_-(z) \coloneqq  \sum_{k > 0} J_k\,\dd \xi_{-k}(z)\,,
\end{equation}
where for $k > 0$
\[ 
\dd \xi_{-k}(z) \coloneqq  \frac{\dd z}{z^{k + 1}} + \sum_{l > 0} \frac{F_{0,2}[-k,-l]}{k}\,\dd\xi_{l}(z) =  \Res_{z' = 0} \bigg(\int_{0}^{z'} \omega_{0,2}(\cdot,z)\bigg)\frac{\dd z'}{(z')^{k + 1}}\,.
\] 
For uniformity we also define for $k \geq 0$
\[
\dd \xi_{k}(z) \coloneqq  z^{k - 1}\dd z\,.
\]
So, we can rewrite
\begin{equation*}
\begin{split}
\omega_{0,1}(z) & = \sum_{k > 0} F_{0,1}[-k]\,\dd \xi_{k}(z)\,, \\
\omega_{\frac{1}{2},1}(z) & = Q\,\dd\xi_{0}(z) + \sum_{k > 0} F_{\frac{1}{2},1}[-k]\,\dd \xi_{k}(z)\,, \\
\omega_{0,2}(z_1,z_2) & = \frac{\dd z_1 \, \dd z_2}{(z_1 - z_2)^2} + \sum_{k_1,k_2 > 0} F_{0,2}[-k_1,-k_2]\,\dd\xi_{k_1}(z_1) \dd \xi_{k_2}(z_2)\,.
\end{split}
\end{equation*}
so that we have
\[
\mathcal{J}_+(z) \coloneqq  J_+(z) = \sum_{k > 0} kx_{k}\,\dd\xi_{k}(z)\,.
\]
We then obtain
\begin{equation}
\begin{split}
\label{Hix0} & \\
rH_i(\tilde{x})  =  \!\!\!\!\!\!\! \sum_{\substack{Z\,:\,[i] \hookrightarrow \mathfrak{f}(z) \\ 0 \leq j \leq \lfloor i/2 \rfloor \\ A_{0} \sqcup A_{\frac{1}{2}} \sqcup A_+ \sqcup A_- = [2j + 1,i]}} \!\!\!\!\!\!\!\!\!\!\!\!\! \frac{\hbar^{j + \frac{1}{2}|A_{\frac{1}{2}}|}}{2^{j}j!(i - 2j)!} \prod_{l = 1}^{j} \omega_{0,2}(Z_{2l - 1},Z_{2l})\prod_{l \in A_{0}} & \omega_{0,1}(Z_{l}) \prod_{l \in A_{\frac{1}{2}}} \omega_{\frac{1}{2},1}(Z_l) \\
& \qquad \prod_{z' \in A_+} \mathcal{J}_+(z') \prod_{z' \in A_-} \mathcal{J}_-(z')\,.
\end{split}
\end{equation}
We prefer to convert this expression into a sum over subsets $Z \subseteq \mathfrak{f}(z)$ of cardinality $i$. Then, we have to sum over partitions $B_1 \sqcup \ldots \sqcup B_j \sqcup A_0 \sqcup A_{\frac{1}{2}} \sqcup A_+ \sqcup A_- = Z$ where $|B_l| = 2$ for any $l \in [j]$, and it can arise in exactly $2^{j}(i - 2j)!$ terms in \eqref{Hix0} corresponding to the choice of an order within each pair $B_j$, and the choice of a labelling by $[2j + 1,i]$ for the elements in $A_0 \sqcup A_{\frac{1}{2}} \sqcup A_+ \sqcup A_-$. So only the factor $1/j!$ remains. It can also be erased by forgetting the ordering of $B_1,\ldots,B_j$. More precisely, introducing the set $\mathcal{P}(\mathfrak{f}(z))$ whose elements are sets of disjoint pairs in $ \mf{f}(z)$, and writing $\sqcup \mathbf{P} \coloneqq  \bigsqcup_{P \in \mathbf{P}} P$ if $\mathbf{P} \in \mathcal{P}(\mathfrak{f}(z))$, we obtain
\begin{equation*}
\begin{split}
 rH_i(\tilde{x}) = \sum_{\substack{Z \subseteq \mathfrak{f}(z) \\ |Z| = i}} \sum_{\substack{(\sqcup \mathbf{P}) \sqcup A_{0} \sqcup A_{\frac{1}{2}} \sqcup A_+ \sqcup A_- = Z \\ \mathbf{P} \in \mathcal{P}(\mathfrak{f}(z))}} \hbar^{|\mathbf{P}| + \frac{1}{2}|A_{\frac{1}{2}}|} \!\! \prod_{\{z',z''\} \in \mathbf{P}} \!\! \omega_{0,2}(z',z'')\prod_{z' \in A_{0}} &\omega_{0,1}(z') \prod_{z' \in A_{\frac{1}{2}}} \omega_{\frac{1}{2},1}(z') \\
& \prod_{z' \in A_{+}} \mathcal{J}_{+}(z')\prod_{z' \in A_-} \mathcal{J}_{-}(z')\,.
\end{split} 
\end{equation*} 

\medskip
  
\subsection{Fields for an arbitrary twist}
\label{Ftwis}

\medskip

We now return to the general situation of \cref{sec:twisted_modules}. Let $\sigma$ be a permutation of $[r]$ with cycles of  lengths $r_{\mu}$ labelled by $\mu \in [d]$. For each $\mu \in [d]$, we have generators $J^{\mu}_{k}$ of the Heisenberg algebra of $\mathfrak{gl}_{r_{\mu}}$:
\[
J_{k}^{\mu} = \left\{\begin{array}{lll} \hbar \partial_{x_k^{\mu}} & & {\rm if}\,k > 0 \\ \hbar^{\frac{1}{2}}Q_{\mu} & & {\rm if}\,\,k = 0 \\ -kx^{\mu}_{-k} & & {\rm if}\,\,k < 0 \end{array}\right.\,,
\]
whose currents we split as $J^\mu_+(z)$ and $ J^\mu_-(z)$ in the same way as in \cref{onecyc}. We obtain modes $W^{\mu}_{i_{\mu},k_{\mu}}$ indexed by $i_{\mu} \in [r_{\mu}]$ and $k_{\mu} \in \mathbb{Z}$ for a representation of the $\mc{W}(\mathfrak{gl}_r)$ algebra given by \eqref{eq:twist_mode_coxeter}.
To match \cref{onecyc}, we introduce for each $\mu \in [d]$ formal variables $z$ such that $\tilde{x} = z^{r_{\mu}}$. These $z$ thus depend on $\mu$, but they will appear in generating series with superscript $\mu$ so that one can infer directly from the formula which power $r_{\mu}$ one should use to relate it to the global variable $\tilde{x}$.

At this stage we are naturally led to introduce a curve which is the union of copies of a formal disk for each $\mu \in [d]$:
\[
\tilde{C} = \bigsqcup_{\mu = 1}^{d} \tilde{C}_{\mu},\qquad \tilde{C}_{\mu} \coloneqq  {\rm Spec}\,\mathbb{C}\llbracket z \rrbracket\,.
\]
When necessary to avoid confusion, points in $\tilde{C}$ will be denoted $\big(\begin{smallmatrix} \mu \\ z \end{smallmatrix}\big)$ to indicate in which copy of the formal disk we consider them. One can consider $\tilde{x}$ as a branched cover $\tilde{C} \longrightarrow V \coloneqq  {\rm Spec} \,\mathbb{C} \llbracket X \rrbracket$ given by $z \mapsto z^{r_{\mu}}$ on the $\mu$th copy of the $z$-formal disk. The smooth (but reducible) curve $\tilde{C}$ is in fact the normalisation $\pi \colon \tilde{C} \rightarrow C$ of the singular curve
\[
C = {\rm Spec}\,\mathbb{C} \llbracket x,z \rrbracket/ \big( \prod_{\mu=1}^d (x - z^{r_\mu})\big)\,.
\]
The branched cover $\tilde{x}\colon \tilde{C} \rightarrow V$ factors through $x \colon C \rightarrow V$. This is the local picture we will globalise later in \cref{globalsp} by considering more general branched covers
\[  
\tilde{C} \mathop{\longrightarrow}^{\pi} C \mathop{\longrightarrow}^{x} V,\qquad \tilde{x} = x \circ \pi\,,
\]
where $V,\tilde{C}$ are regular curves and $C$ is a possibly singular curve whose normalisation is $\tilde{C}$. For the moment we stick to the local setting.

Let us again consider a general dilaton shift and change of polarisation
\begin{equation*}
\begin{split} 
\hat{T} & = \exp\Bigg(\sum_{\substack{\mu \in [d] \\ k > 0}} \Big(\hbar^{-1}\,F_{0,1}\big[\begin{smallmatrix} \mu \\ -k \end{smallmatrix}\big] + \hbar^{-\frac{1}{2}}\,F_{\frac{1}{2},1}\big[\begin{smallmatrix} \mu \\ - k \end{smallmatrix}\big]\Big)\,\frac{J^{\mu}_{k}}{k}\Bigg)\,, \\
\hat{\Phi} & = \exp\Bigg( \frac{1}{2\hbar} \sum_{\substack{\mu,\nu \in [d] \\ k,l > 0}}  F_{0,2}\big[\begin{smallmatrix} \mu & \nu \\ -k & -l \end{smallmatrix}\big]\,\frac{J_{k}^{\mu} J_{l}^{\nu}}{kl} \Bigg)\,,
\end{split}   
\end{equation*} 
and the conjugated operator
\begin{equation*}
H_i(\tilde{x}) = \hat{\Phi}\,\hat{T} \cdot W_i(\tilde{x}) \cdot \hat{T}^{-1}\,\hat{\Phi}^{-1} = \sum_{k \in \mathbb{Z}} \frac{H_{i,k}(\dd \tilde{x})^i}{\tilde{x}^{k + i}}\,.
\end{equation*}
To express $H_i(\tilde{x})$, we introduce the basis of meromorphic $1$-forms $\dd \xi^{\mu}_{k}$ on $\tilde{C}$, indexed by $\mu \in [d]$ and $k \in \mathbb{Z}$. It is defined by
\begin{equation*} 
\begin{split}
k \geq 0\qquad :& \qquad \dd\xi^{\mu}_k\big(\begin{smallmatrix} \nu \\ z\end{smallmatrix}\big) = \delta_{\mu,\nu}\,z^{k - 1}\,\dd z\,, \\
k > 0 \qquad :& \qquad \dd \xi_{-k}^{\mu}\big(\begin{smallmatrix} \nu \\ z \end{smallmatrix}\big) = \delta_{\mu,\nu}\,\frac{\dd z}{z^{k + 1}} + \sum_{l > 0} \frac{F_{0,2}\big[\begin{smallmatrix} \mu & \nu \\ -k & -l \end{smallmatrix}\big]}{k}\,z^{l - 1}\,\dd z\,. 
\end{split}
\end{equation*} 
We also introduce the meromorphic forms $\omega_{0,1}$, $ \omega_{\frac12, 1}$ and bidifferential $\omega_{0,2}$ on $\tilde{C}$:
\begin{equation*}
\begin{split} 
\omega_{0,1} & = \sum_{\substack{\mu \in [d] \\ k > 0}} F_{0,1}\big[\begin{smallmatrix} \mu \\ -k \end{smallmatrix}\big]\,\dd\xi_{k}^{\mu}\,, \\
\omega_{\frac{1}{2},1} & = \sum_{\mu \in [d]} Q_{\mu}\dd\xi_{0}^{\mu} + \sum_{k > 0} F_{\frac{1}{2},1}\big[\begin{smallmatrix} \mu \\ -k \end{smallmatrix}\big]\,\dd \xi_{k}^{\mu}\,, \\
\omega_{0,2} & = \omega_{0,2}^{{\rm std}} + \sum_{\substack{\mu_1,\mu_2 \in [d] \\ k_1,k_2 > 0}}F_{0,2}\big[\begin{smallmatrix} \mu_1 & \mu_2 \\ -k_1 & -k_2 \end{smallmatrix}\big]\,\dd \xi_{k_1}^{\mu_1} \dd \xi_{k_2}^{\mu_2}\,,
\end{split}
\end{equation*}
where
\[
\omega_{0,2}^{{\rm std}}\big(\begin{smallmatrix} \nu_1 & \nu_2 \\ z_1 & z_2 \end{smallmatrix}\big) = \frac{\delta_{\nu_1,\nu_2}\,\dd z_1 \dd z_2}{(z_1 - z_2)^2}\,.
\]
For $k \in \mathbb{Z}$, we introduce the $1$-form on $\tilde{C}$
\[
\dd \xi_{k}^*\big(\begin{smallmatrix} \mu \\ z \end{smallmatrix}\big) = \dd\xi_{k}^{\mu}\big(\begin{smallmatrix} \mu \\ z \end{smallmatrix}\big)\,.
\]
We recall that the index $\mu \in [d]$ of the component to which a point $z' \in \tilde{C}$ belongs is implicit in the data of $z'$.

Similarly to \cref{onecyc}, the effect of the dilaton shift is to replace $J_+^{\mu}(z)$ with
\[
J_+^{\mu}(z) + \hbar^{\frac{1}{2}}\Big(\omega_{\frac{1}{2},1}\big(\begin{smallmatrix} \mu \\ z \end{smallmatrix}\big) - Q_{\mu}\tfrac{\dd z}{z}\Big) + \omega_{0,1}\big(\begin{smallmatrix} \mu \\ z \end{smallmatrix}\big)\,,
\]
while the effect of the change of polarisation is to replace
$\omega_{0,2}^{{\rm std}}(z_1,z_2)$ with $\omega_{0,2}\big(\begin{smallmatrix} \mu_1 & \mu_2 \\ z_1 & z_2 \end{smallmatrix}\big)$ and $J_-^{\mu}(z)$ with
\[
\mathcal{J}_{-}\big(\begin{smallmatrix} \mu \\ z \end{smallmatrix}\big) \coloneqq \sum_{k > 0} J_{k}^{\mu}\,\dd\xi_{-k}^{*}\big(\begin{smallmatrix} \mu \\ z \end{smallmatrix}\big)\,.
\] 
For uniformity we also set
\[
\mathcal{J}_{+}\big(\begin{smallmatrix} \mu \\ z \end{smallmatrix}\big) \coloneqq  J_+^{\mu}(z) = \sum_{k > 0} kx_{k}^{\mu}\,\dd\xi_{k}^{*}\big(\begin{smallmatrix} \mu \\ z \end{smallmatrix}\big)\,.
\]
We can repeat the argument of \cref{onecyc} with several $\mu$s,
 defining the fiber over $\tilde{x}$ in $\tilde{C}$
\[
\mathfrak{f}(z) \coloneqq  \bigsqcup_{\mu = 1}^{d} \mathfrak{f}_{\mu}(z)\,,
\]
and getting
\begin{equation*} 
\begin{split}
H_i(\tilde{x})  = \sum_{\substack{Z \subseteq \mathfrak{f}(z) \\ |Z| = i}} \sum_{\substack{(\sqcup \mathbf{P}) \sqcup A_{0} \sqcup A_{\frac{1}{2}} \sqcup A_{+} \sqcup A_{-} = Z \\ \mathbf{P} \in \mathcal{P}(\mathfrak{f}(z))}} \hbar^{|\mathbf{P}| + \frac{1}{2}|A_{\frac{1}{2}}|} \prod_{\{z',z''\} \in \mathbf{P}} \omega_{0,2}(z',z'') &\prod_{z' \in A_{0}} \omega_{0,1}(z') \prod_{z' \in A_{\frac{1}{2}}} \omega_{\frac{1}{2},1}(z') \\
& \prod_{z' \in A_{+}} \mathcal{J}_{+}(z')\prod_{z' \in A_{-}} \mathcal{J}_{-}(z')\,.
\end{split}
\end{equation*}

\medskip

\subsection{Action of the fields on the partition function}

\medskip

Given a formal function
\begin{equation}
\label{Fsum} F = \sum_{\substack{g \geq 0,\,\,n \geq 1 \\ 2g - 2 + n > 0}} \sum_{\substack{\mu_1,\ldots,\mu_n \in [d] \\ k_1,\ldots,k_n > 0}}\frac{\hbar^{g - 1}}{n!} F_{g,n}\big[\begin{smallmatrix}  \mu_1 & \cdots & \mu_n \\ k_1 & \cdots & k_n \end{smallmatrix}\big]\,\prod_{i = 1}^n x_{k_i}^{\mu_i}\,,
\end{equation}
let us compute $G_{i}(\tilde{x}) = e^{-F}H_i(\tilde{x})e^{F}\cdot 1$. The partition function $e^F$ is annihilated by the differential operators above a certain index in the $\mathcal{W}$ if and only if the $G_i$ satisfy certain bounds on their pole orders as $ \tilde{x} \to 0$. Because $F$ is a function (i.e. it does not contain a differential part), it commutes with $ \omega_{0,2}$, $\omega_{0,1}$, $\omega_{\frac{1}{2},1}$, and $\mc{J}_+$. The only non-trivial computation is
\begin{equation*}
e^{-F} \mc{J}_-(z') e^F = \mc{J}_-(z') + [\mc{J}_-(z'), F]\,,
\end{equation*}
where each term $ \mc{J}_-(z')$ obtained like this has to act on a later $ [\mc{J}_-(z''),F]$, as it annihilates $1$. The $ \mc{J}_-$ commute among each other, so we get a partition of $ A_-$ into sets of operators acting on a single copy of $F$.
We obtain
\begin{equation}
\label{Gitildex}\begin{split}
G_i(\tilde{x}) & = \sum_{\substack{Z \subseteq \mathfrak{f}(z) \\ |Z| = i}} \sum_{\substack{(\sqcup \mathbf{P}) \sqcup A_{0} \sqcup A_{\frac{1}{2}} \sqcup A_{+} \sqcup A_- = Z \\ \mathbf{P} \in \mathcal{P}(\mathfrak{f}(z))}}  \hbar^{|\mathbf{P}| + \frac{1}{2}|A_{\frac{1}{2}}|}\,\prod_{\{z',z''\} \in \mathbf{P}} \omega_{0,2}(z',z'') \prod_{z' \in A_{0}} \omega_{0,1}(z') \prod_{z' \in A_{\frac{1}{2}}} \omega_{\frac{1}{2},1}(z') \\
&\quad \prod_{z' \in A_+} \bigg( \sum_{\tilde{k}_{z'} \geq 1} \tilde{k}_{z'}x_{\tilde{k}_{z'}}^{\mu_{z'}}\,\dd \xi^*_{\tilde{k}_{z'}}(z') \bigg) \\
&  \quad \cdot \!\!\!\!\!\!\!\! \sum_{\substack{\mathbf{L} \vdash A_- \\ g_L,m_L \geq 0,\,\,L \in \mathbf{L} \\ 2g_L - 2 + m_L + |L| > 0}}\!\!\! \prod_{L \in \mathbf{L}} \sum_{\substack{\mathbf{\nu}_{L}\colon [m_L] \to [d] \\ \mathbf{\ell}_{L}\colon [m_L] \to \N^*\\ k_L \colon L \to \N^*}}  \! \bigg(\frac{\hbar^{g_L - 1 + |L|}}{m_{L}!}\,F_{g_L,|L| + m_L}\big[\begin{smallmatrix} \boldsymbol{\mu}_{|L} & \nu_{L,1} & \cdots & \nu_{L,m_L} \\ k_L & \ell_{L,1} & \cdots & \ell_{L,m_L}\end{smallmatrix}\big] \!\! \prod_{l \in [m_L]} \!\! x_{\ell_{L,l}}^{\nu_{L,l}}   \prod_{z' \in L} \dd\xi^*_{-k_{z'}}(z') \bigg)\,,
\end{split}
\end{equation} 
where $\mathbf{\mu}\colon A_+ \sqcup A_- \rightarrow [d]$ associates to $z' $ the index $\mu_{z'} \in [d]$ such that $z' \in \mathfrak{f}_{\mu_{z'}}(z)$, and we identified $\boldsymbol{\mu}_{|L}$ and $\mathbf{k}_{L}$ with the tuples $(\mu_{z'})_{z' \in L}$ and $(k_{z'})_{z' \in L}$.\par
We decompose $G_i$ in homogeneous terms with respect to the exponent of $\hbar$ and the number of $ x^\mu_k$:
\begin{equation*}
G_i(\tilde{x}) = \sum_{g,n \geq 0} \frac{\hbar^g}{n!} G_{i;g,n}(\tilde{x})\,.
\end{equation*}
In order to completely rephrase this in terms of spectral curves, we need to get rid of the $ x_k^\mu $ and replace them with $ \dd \xi $s. For every $n$, prepare a tuple $w_{[n]} = (w_j)_{j=1}^n $ of points on $ \tilde{C}$ and define
\begin{equation*}
\mc{E}^{(i)}_{g,n}(\tilde{x};w_{[n]})  \coloneqq \prod_{j=1}^n \ad_{\mc{J}_-(w_j)} G_{i;g,n}(\tilde{x})\,.
\end{equation*}
To compute it, we introduce the multidifferential forms for $g \geq 0$ and $n \geq 1$ such that $2g - 2 + n > 0$
\begin{equation}
\label{wgnd}
\omega_{g,n}(z_1,\ldots,z_n) \coloneqq  \sum_{\substack{\mu_1,\ldots,\mu_n \in [d] \\ k_1,\ldots,k_n > 0}} F_{g,n}\big[\begin{smallmatrix} \mu_1 & \cdots & \mu_n \\ k_1 & \cdots & k_n \end{smallmatrix}\big]\,\prod_{j = 1}^n \dd\xi_{-k_i}^{\mu_i}(z_i)\,.
\end{equation}
Besides, under this action, we get
\[
\ad_{\mc{J}_-(w)}\bigg(\sum_{k > 0} k\,x^{\nu}_{k}\,\dd\xi^\mu_{k}(z)\bigg) = \sum_{k > 0} k\,\dd\xi_{-k}^{\nu}(w)\,\dd\xi^{\mu}_{k}\big(z)\,,
\]
which is the series expansion of $\omega_{0,2}\big(\begin{smallmatrix} \nu & \mu  \\ w & z \end{smallmatrix}\big)$ with $|z| < |w|$.

We then notice that the sums over $\tilde{k}_{z'}$, $k_L$, $\nu_L$ and $\ell_L$ in \eqref{Gitildex} recombine into
\begin{equation*}
\begin{split} 
& \quad \mathcal{E}^{(i)}_{g,n}(\tilde{x},w_{[n]}) \\
& = \sum_{\substack{Z \subseteq \mathfrak{f}(z) \\ |Z| = i}} \,\,\,\sum_{\substack{(\sqcup \mathbf{P}) \sqcup A_0 \sqcup A_{\frac{1}{2}} \sqcup A_{+} \sqcup A_- = Z \\ \mathbf{P} \in \mathcal{P}(\mathfrak{f}(z))}} \sum_{\iota \colon A_+\hookrightarrow [n]} \sum_{\substack{\mathbf{L} \vdash A_- \\ M \vdash_{\mathbf{L}} [n] \setminus \iota (A_+)}} \sum_{\substack{g_L \geq 0,\,\,L \in \mathbf{L} \\ 2g_L - 2 + |L| + |M_L| > 0 \\ \frac{1}{2}|A_{\frac{1}{2}}| + |\mathbf{P}| + |A_-| + \sum_{L} (g_L - 1) = g}} \\ 
& \quad \prod_{\{z',z''\} \in \mathbf{P}} \omega_{0,2}(z',z'') \prod_{z' \in A_{0}} \omega_{0,1}(z') \prod_{z' \in A_{\frac{1}{2}}} \omega_{\frac{1}{2},1}(z') \prod_{z' \in A_+} \omega_{0,2}(w_{\iota (z')},z')  \prod_{L \in \mathbf{L}} \omega_{g,|L| + |M_L|}(L,w_{M_L})\,. 
\end{split} 
\end{equation*} 
We now observe that the factors $\omega_{0,1},\omega_{\frac{1}{2},1},\omega_{0,2}$ can be treated uniformly by summing over partitions $\mathbf{L} \vdash Z$ and allowing $(g_L,|L| + m_L) = (0,1),(\tfrac{1}{2},1),(0,2)$, which were exactly the terms for which $2g_L - 2 + |L| + m_L \leq 0$. We get
\begin{equation*}
\mathcal{E}^{(i)}_{g,n}(\tilde{x};w_{[n]}) = \sum_{\substack{Z \subseteq \mathfrak{f}(z) \\ |Z| = i}} \,\,\,\sum_{\substack{\mathbf{L} \vdash Z \\ N \vdash_{\mathbf{L}} [n]}} \sum_{\substack{g_L \geq 0,\,\,L \in \mathbf{L} \\ g = i + \sum_{L} (g_L - 1)}} \prod_{L \in \mathbf{L}} \omega_{g_L,|L| + |N_L|}(L,w_{N_L})\,.
\end{equation*}

\medskip

\subsection{From PDEs to abstract loop equations}
\label{SAIRNUSN} 

\medskip

\Cref{prop:AiryAllDilatonPolarize} gives sufficient conditions on the values of $(r_{\mu})_{\mu = 1}^{d}$, of positive integers $(s_{\mu})_{\mu = 1}^{d}$, of scalars $(t_{\mu})_{\mu = 1}^{d}$ and $(Q_{\mu})_{\mu = 1}^{d}$ to get a unique $F$ of the form \eqref{Fsum} such that for any $i \in [r]$ and $k \geq \mathfrak{d}_{\mathbf{r},\mathbf{s}}(i)$
\[
e^{-F}H_{i,k}e^{F} \cdot 1= 0\,.
\]
The translation of these differential constraints in terms of the \emph{correlators} $\boldsymbol{\omega} = (\omega_{g,n})_{g,n}$ defined in \eqref{wgnd} is called ``abstract loop equations''. It says that for any $n \geq 0$, we have
\begin{equation*}
\mathcal{E}^{(i)}_{g,n}(\tilde{x};w_{[n]}) \in o\big(\tilde{x}^{-\mathfrak{d}_{\mathbf{r},\mathbf{s}}(i)}\big)\cdot \bigg(\frac{\dd \tilde{x}}{\tilde{x}}\bigg)^i  ,\qquad \tilde{x} \rightarrow 0\,.
\end{equation*}
In other words, $\mathcal{E}^{(i)}_{g,n}(\tilde{x},w_{[n]})$ is meromorphic and has a pole of order strictly less than $\mf{d}_{\mathbf{r},\mathbf{s}}(i) + i$ at the point $\tilde{x} = 0$ in $V$. If we let $\tilde{\mathcal{E}}^{(i)}_{g,n}(z,w_{[n]})$ be its pullback to a meromorphic $i$-differential on $\tilde{C}$, this is tantamount to requiring that, for any $\mu \in [d]$
\[
\tilde{\mathcal{E}}^{(i)}_{g,n}\big(\begin{smallmatrix} \mu \\ z \end{smallmatrix};\,w_{[n]}\big) \in o \big(z^{-r_{\mu}\mathfrak{d}_{\mathbf{r},\mathbf{s}}(i)}\big) \cdot \bigg(\frac{\dd z}{z}\bigg)^i,\qquad z \rightarrow 0\,.
\] 

\medskip

\section{Topological recursion on global spectral curves}
\label{globalsp}
\medskip

We are going to formalise what we have found in the context of global, possibly singular spectral curves. This will lead us to define the appropriate notion of abstract loop equations in \cref{SecLoop}, and to show in \cref{SecTR} that its unique solution is given by an appropriate topological recursion \`a la Chekhov--Eynard--Orantin, that is by computing residues on the normalisation of the singular curve.

\medskip

\subsection{Spectral curves}
\label{SecSP}

\medskip

\begin{definition}
\label{defspc} A \emph{spectral curve} is a triple $\mc{C} = (C,x,y)$, where $ C$ is a reduced analytic curve over $\mathbb{C}$ and $x,y$ are meromorphic functions on $C$, such that all fibers of $ x$ are finite and $ \omega_{0,1} \coloneqq y \, \dd x $ has no poles at ramifications of $ x$.
\end{definition}

Note that $C$ is not necessarily connected, compact, or irreducible. We will work with its normalisation $\pi\,:\,\tilde{C} \rightarrow C$, which is a smooth curve. We have meromorphic functions $\tilde{x} = x \circ \pi$ and $\tilde{y} = y \circ \pi$ defined on $\tilde{C}$. Let $\mathfrak{b} \subset \C $ be the set of points $ b$ that have a neighbourhood $ U_b $ such that the cardinality of the fibre of $x$ is constant on $ U_b \setminus \{ b\}$ and strictly smaller at $b$ itself. It is the collection of branchpoints of $\tilde{x}$ and images of locally reducible points away from $\infty$. We also denote $ \mf{a} = x^{-1} (\mf{b})$ and $\tilde{\mathfrak{a}} = \tilde{x}^{-1}(\mathfrak{b})$. We assume that $\mathfrak{b}$ is finite. As a result, $\tilde{\mathfrak{a}}$ and $\mathfrak{a}$ are also finite. Note that, since we assumed that all fibers of $x$ are finite, the same is true of $\tilde{x}$ and there cannot be an irreducible component of $\tilde{C}$ where $\tilde{x}$ is constant.

If $\alpha \in \mathfrak{a}$, we let $U_{\alpha} \subset C$ be a small neighborhood of $\alpha$ that is invariant under local Galois transformations and
\[
\tilde{U}_{\alpha}\coloneqq  \pi^{-1}(U_{\alpha}),\qquad \tilde{U}_{\alpha}' = \tilde{U}_{\alpha} \setminus \pi^{-1}(\alpha),\qquad V_{\alpha} = x(U_{\alpha}),\qquad V_{\alpha}' = V_{\alpha} \setminus \{x(\alpha)\}\,.
\]
Without loss of generality we can assume that $V_\alpha \subset \mathbb{C}$. If $z \in \tilde{U}_{\alpha}$, we define
\[
\mathfrak{f}_\alpha(z) = \tilde{x}^{-1}(\tilde{x}(z)) \cap \tilde{U}_{\alpha},\qquad \mathfrak{f}_{\alpha}'(z) = \mathfrak{f}_{\alpha}(z) \setminus \{z\},\qquad \tilde{\mathfrak{a}}_{\alpha} = \pi^{-1}(\alpha)\,.
\]
Note that $\tilde{\mathfrak{a}}_{\alpha}$ is in bijection with the set of branches in $\tilde{C}$ above $\alpha$, and we denote $d_{\alpha} \coloneqq  |\tilde{\mathfrak{a}}_{\alpha}|$. For each $\mu \in \tilde{\mathfrak{a}}_{\alpha}$, we introduce a small neighborhood $\tilde{C}_{\mu}$ of $\mu$ in $\tilde{C}$, such that $\pi(\tilde{C}_{\mu}) = U_{\alpha}$, as well as $\tilde{C}_{\mu}' = \tilde{C}_{\mu} \setminus \{\mu\}$. We have of course
\[
\tilde{U}_{\alpha} = \bigsqcup_{\mu \in \tilde{\mathfrak{a}}_{\alpha}} \tilde{C}_{\mu}\,.
\]
By taking a smaller neighborhood, we can always assume that the $(\tilde{C}_{\mu})_{\mu \in \tilde{\mathfrak{a}}}$ are pairwise disjoint. As anticipated in \cref{Ftwis}, if we want to insist that a point $z \in \tilde{U}_{\alpha}$ belongs to $\tilde{C}_{\mu}$, we will denote it $\big(\begin{smallmatrix} \mu \\ z \end{smallmatrix}\big)$. The fibers can be decomposed
\[
\mathfrak{f}_{\alpha}(z) = \bigsqcup_{\mu \in \tilde{\mathfrak{a}}_{\alpha}} \mathfrak{f}_{\mu}(z),\qquad \mathfrak{f}_{\mu}(z) \coloneqq  \mathfrak{f}_{\alpha}(z) \cap \tilde{C}_{\mu}\,.
\]
We denote $r_{\mu} = |\mathfrak{f}_{\mu}(z)|$ which is independent of $z \in \tilde{C}_{\mu}'$ and $r_{\alpha} = |\mathfrak{f}_{\alpha}(z)|$ which is independent of $z \in \tilde{U}'_{\alpha}$.  In particular
\[
r_{\alpha} = \sum_{\mu \in \tilde{\mathfrak{a}}_{\alpha}} r_{\mu}\,.
\]

If $\gamma$ is a small loop in $V_{\alpha}$ around $x(\alpha)$, it induces a Galois transformation in the cover $\tilde{x}_{|\tilde{U}_{\alpha}}$, that is for each $z \in \tilde{U}_{\alpha}'$ a permutation $\sigma_{\alpha}$ of $\mathfrak{f}_{\alpha}(z)$, which on $\tilde{\mathfrak{f}}_{\mu}(z)$ restricts to a cyclic transformation of order $r_{\mu}$. This integer represents the order of ramification at $\mu \in \tilde{\mathfrak{a}}_{\alpha}$ of $\tilde{x}_{|\tilde{C}_{\mu}}$.

\begin{remark}
If $|\tilde{\mathfrak{a}}_{\alpha}| = 1$, $C$ is irreducible locally at $\alpha$, hence smooth at $\alpha$. We can then use the same symbol to denote the point $\alpha \in C$ and the unique point above it in $\tilde{C}$. If $|\tilde{\mathfrak{a}}_{\alpha}| > 1$, $C$ is reducible locally at $\alpha$, hence singular at $\alpha$. If $|\tilde{\mathfrak{a}}_{\alpha}| = 2$, $\alpha$ is a node. For $\mu \in \tilde{\mathfrak{a}}$, we have $r_{\mu} = 1$ if and only if $\mu$ is not a ramification point of $\tilde{x}$. We say that the spectral curve is smooth if all ramification points in $C$ are smooth.
\end{remark}

For each $\alpha \in \mathfrak{a}$ and $\mu \in \tilde{\mathfrak{a}}_\alpha$, there exists a local coordinate $\zeta$ on $\tilde{C}_{\mu}$ such that
\[
\tilde{x}\begin{psmallmatrix} \mu \\ z \end{psmallmatrix} = x(\alpha) + \zeta(z)^{r_{\mu}}\,.
\]
As in \cref{Ftwis}, when working with local coordinates it should be clear from the context which $\tilde{C}_{\mu}$ is involved.  Specifying such coordinates requires the choice of a $r_{\mu}$th root of unity for $\tilde{x} - x(\alpha)$. We assume such a choice is fixed. We also choose a primitive $r_{\mu}$th root of unity, denoted $\theta_{\mu}$. If $z \in \tilde{C}_{\mu}'$, the set of coordinates of the points in $\mathfrak{f}_{\alpha}(z)$ is
\[ 
\big\{\vartheta_{\nu}^{j}\zeta^{r_{\mu}/r_{\nu}} \quad | \quad \nu \in \tilde{\mathfrak{a}}_{\alpha},\,\,\,j \in [r_{\nu} ]\big\}\,.
\]

Let us write locally at $\mu \in \mathfrak{a}$ the Laurent series expansion of the function $\tilde{y}$
\[
\tilde{y}\big(\begin{smallmatrix} \mu \\ z \end{smallmatrix}\big) \sim \sum_{k >0} \frac{1}{r_{\mu}}\,F_{0,1}\big[\begin{smallmatrix} \mu \\ -k \end{smallmatrix}\big]\,\zeta^{k-r_{\mu}}\,,
\]
and define
\[
s_{\mu} \coloneqq  \min\big\{k \in \mathbb{Z} \quad \big| \quad F_{0,1}\big[\begin{smallmatrix} \mu \\ -k \end{smallmatrix}\big] \neq 0 \big\} \in \mathbb{Z} \cup \{+\infty\}\,.
\]
In particular, $s_{\mu} = +\infty$ if $y$ vanishes identically in the connected component of $\mu$ in $\tilde{C}$. If $s_{\mu}$ is finite, we introduce
\[
t_{\mu} \coloneqq  -\frac{1}{r_{\mu}}\,F_{0,1}\big[\begin{smallmatrix} \mu \\ -s_{\mu} \end{smallmatrix}\big]\,.
\]

We equip $\tilde{\mathfrak{a}}_{\alpha}$ with a total order $\preccurlyeq$ satisfying
\[
\mu \preccurlyeq \nu \qquad \Longrightarrow\qquad  \frac{r_{\mu}}{s_{\mu}} \geq \frac{r_{\nu}}{s_{\nu}}
\]
and denote $\prec$ the corresponding strict order. (Such orders exist.) Note that the inequality still makes sense for $\mu$s such that $s_{\mu} = +\infty$. Then, in agreement with \cref{sec:W_gl_airy_structs_proof}, for $\lambda,\mu \in \tilde{\mathfrak{a}}_{\alpha}$ we let
\begin{equation*}
\begin{split}
[\mu] & \coloneqq  \big\{\nu \in \tilde{\mathfrak{a}}_{\alpha}\quad | \quad \nu \preccurlyeq \mu\big\}\,, \\ 
[\mu ) &\coloneqq  \big\{\nu \in \tilde{\mathfrak{a}}_{\alpha} \quad |\quad \nu \prec \mu \big\}\,,  \\
[\lambda,\mu] & := \big\{\nu \in \mathfrak{a}_{\alpha} \quad | \quad \lambda \preccurlyeq \nu \preccurlyeq \mu \big\}\,,
\end{split}
\end{equation*}
and likewise for the open segments $[\lambda,\mu)$, $(\lambda,\mu]$, etc. For instance $[\mu ) = [\min \tilde{\mathfrak{a}}_{\alpha},\mu)$. If $M \subseteq \tilde{\mathfrak{a}}_{\alpha}$, we let
\[
\mathbf{r}_{M} \coloneqq  \sum_{\mu \in M} r_{\mu}\,,\qquad\qquad \mathbf{s}_{M} \coloneqq  \sum_{\mu \in M} s_{\mu}\,.
\]
For $\mu \in \tilde{\mathfrak{a}}_{\alpha}$ we define
\[
\Delta_{\mu} \coloneqq  \mathbf{r}_{[\mu ]}s_{\mu} - \mathbf{s}_{[\mu ]}r_{\mu}\,.
\]

\begin{definition}
\label{defYA} For $\alpha \in \mathfrak{a}$, we define a function of $z \in \tilde{U}_{\alpha}$ by
\[
Y_{\alpha}(z) \coloneqq  \prod_{z' \in \mathfrak{f}'_{\alpha}(z)} (\tilde{y}(z') - \tilde{y}(z))\,.
\]
\end{definition}
In the next paragraph we will need to study the order of vanishing of these functions at $\tilde{\mathfrak{a}}$. This is given by the following lemma.

\begin{lemma}
\label{lemYA} If one of the following conditions is satisfied
\begin{itemize}
\item[(i)] there exist distinct $\mu,\nu \in \tilde{\mathfrak{a}}_{\alpha}$ such that $s_{\mu} = s_{\nu} = +\infty$; or
\item[(ii)] there exists at least one $\mu \in \tilde{\mathfrak{a}}_{\alpha}$ such that $s_{\mu} = +\infty$ and $r_{\mu} > 1$,
\end{itemize}
then $Y_{\alpha}(z)$ vanishes identically on $\tilde{C}_{\mu}$ for the $\mu$ involved in these conditions. Otherwise, for any $\mu \in \tilde{\mathfrak{a}}_{\alpha}$, we have $Y_{\alpha}(z) \in \mc{O}(\zeta^{\mathfrak{v}_{\mu}})$ when $z \in \tilde{C}_{\mu}$ approaches $\mu$, where
\begin{equation}
\label{alphamu} \mathfrak{v}_{\mu} = (s_{\mu} - r_{\mu})(r_{\alpha} - 1) - \Delta_{\mu}\,.
\end{equation}
If furthermore either
\begin{itemize}
\item[(iii)] there exist distinct $\mu,\nu \in \tilde{\mathfrak{a}}_{\alpha}$ such that $s_{\mu},s_{\nu}$ are finite,  $\tfrac{r_{\mu}}{s_{\mu}} = \tfrac{r_{\nu}}{s_{\nu}}$ and $t_{\mu}^{r_{\nu}} = t_{\nu}^{r_{\nu}}$; or
\item[(iv)] there exists $\mu \in \tilde{\mathfrak{a}}_{\alpha}$ such that $s_{\mu}$ is finite and ${\rm gcd}(r_{\mu},s_{\mu}) > 1$,
\end{itemize}
then $Y_{\alpha}(z) \in \mc{O}(\zeta^{\mathfrak{v}_{\mu} + 1})$. If none of the above conditions are satisfied, then there exists a non-zero scalar $\mathfrak{t}_{\alpha,\mu}$ such that $Y_{\alpha}(z)- \mathfrak{t}_{\alpha,\mu} \zeta^{\mathfrak{v}_{\mu}} \in \mc{O}(\zeta^{\mf{v}_\mu+1})$.
\end{lemma}
\begin{proof}
If $s_{\mu} = +\infty$, $\tilde{y}$ is identically zero for $z \in \tilde{C}_{\mu}$. Conditions \emph{(i)} and \emph{(ii)} both imply there is a $z' \in \mf{f}_{\alpha}'(z)$ such that $ \tilde{y}(z') \equiv 0$ as well, so one of the factors in $Y_{\alpha}(z)$ vanishes identically.\par
We now assume that \emph{(i)} and \emph{(ii)} are not satisfied. Let us add for the moment the assumption that all $s_{\mu}$ are finite. We compute
\begin{equation}
\label{prodmu}
\prod_{z' \in \mathfrak{f}_{\mu}(z)} (\tilde{y}(z') - \tilde{y}(z)) = \bigg(\prod_{j = 1}^{r_{\mu} - 1} (\theta_{\mu}^{s_{\mu}j} - 1)\bigg) (-t_{\mu})^{r_{\mu} - 1}\,\zeta^{(s_{\mu} - r_{\mu})(r_{\mu} - 1)} + \mathcal{O}(\zeta^{(s_{\mu} - r_{\mu})(r_{\mu} - 1) + 1})\,,
\end{equation}
and observe that the scalar prefactor in the first term is non-zero if and only if $r_{\mu}$ and $s_{\mu}$ are coprime --- in that case it is equal to $r_{\mu}t_{\mu}^{r_{\mu} - 1}$. For $\nu \in \tilde{\mathfrak{a}}_{\alpha}$ distinct from $\mu$, we have
\begin{equation}
\label{prodnu}
\begin{split} \prod_{z' \in \mathfrak{f}_{\nu}(z)} (\tilde{y}(z') - \tilde{y}(z)) & = \prod_{j = 0}^{r_{\nu} - 1} (-t_{\nu}\theta_{\nu}^{j}\zeta^{r_{\mu}(s_{\nu} - r_{\nu})/r_{\nu}} +  t_{\mu}\zeta^{s_{\mu} - r_{\mu}}  + \cdots) \\
& = t_{\mu,\nu} \zeta^{\min(r_{\mu}s_{\nu},r_{\nu}s_{\mu}) - r_{\mu}r_{\nu}} + \mathcal{O}(\zeta^{\min(r_{\mu}s_{\nu},r_{\nu}s_{\mu}) -r_{\mu}r_{\nu} + 1})\,,
\end{split}
\end{equation}
where $\cdots$ are higher order terms, and
\[
t_{\mu,\nu} = \left\{\begin{array}{lll} -t_{\nu}^{r_{\nu}} & & {\rm if}\,\,r_{\mu}s_{\nu} < r_{\nu}s_{\mu} \\ t_{\mu}^{r_{\nu}} & & {\rm if}\,\,r_{\mu}s_{\nu} > r_{\nu}s_{\mu} \\ t_{\mu}^{r_{\nu}} - t_{\nu}^{r_{\nu}} & & {\rm if}\,\,r_{\mu}s_{\mu} = r_{\nu}s_{\mu} \end{array}\right. \,.
\]
We have $t_{\mu,\nu} = 0$ if and only if $(\mu,\nu)$ obey the condition \emph{(iii)}. Multiplying \eqref{prodmu} with the product of \eqref{prodnu} over all $\nu \neq \mu$, we deduce that $Y_{\alpha}(z) = \mathfrak{t}_{\alpha,\mu} \zeta^{\mathfrak{v}_{\mu}} + \mathcal{O}(\zeta^{\mathfrak{v}_{\mu} + 1})$ and $\mathfrak{t}_{\alpha,\mu} = 0$ if and only if the conditions \emph{(iii)} and \emph{(iv)} are satisfied, with the exponent
\begin{equation}
\label{alphacm}
\begin{split}
\mathfrak{v}_{\mu} & = (s_{\mu} - r_{\mu})(r_{\mu} - 1) + \sum_{\nu \neq \mu} \big(\min(r_{\mu}s_{\nu},r_{\nu}s_{\mu}) - r_{\mu}r_{\nu}\big) \\
& = r_{\mu} - s_{\mu} + \sum_{\nu \in \tilde{\mathfrak{a}}_{\alpha}} \big(\min(r_{\mu}s_{\nu},r_{\nu}s_{\mu}) - r_{\mu}r_{\nu}\big) \\
& = -r_{\mu}(r_{\alpha} - 1)  - s_{\mu} + \sum_{\nu \prec \mu} r_{\mu}s_{\nu} + \sum_{\nu \succcurlyeq \mu} r_{\nu}s_{\mu} \\
& = (s_{\mu}-r_{\mu})(r_{\alpha} - 1) - \Delta_{\mu}\,,
\end{split}
\end{equation}
as claimed. This concludes the proof in absence of an infinite $s$. 

Now let us assume there exists a unique $\mu_- \in \tilde{\mathfrak{a}}_{\alpha}$ such that $s_{\mu_-} = +\infty$. As we assume that \emph{(i)} and \emph{(ii)} are not satisfied, we must have $r_{\mu_-} = 1$. If $\mu \neq \mu_-$ and we take $z \in \tilde{C}_{\mu}$, we only need to pay attention to the factor \eqref{prodnu} for $\nu = \mu_-$, and in fact \cref{prodnu} remains valid, hence $Y_{\alpha}(z) = \mathfrak{t}_{\alpha,\mu}\zeta^{\mathfrak{v}_{\mu}} + \mathcal{O}(\zeta^{\mathfrak{v}_{\mu} + 1})$ with the same expression for $\mathfrak{v}_{\mu}$ and the same discussion for the (non-)vanishing of $\mathfrak{t}_{\alpha,\mu}$. Notice that by definition of the order we must have $\mu_- = \max(\tilde{\mathfrak{a}}_{\alpha})$ so $\mu_-$ does not appear in $\Delta_{\mu}$. If $z \in \tilde{C}_{\mu_-}$, the factor \eqref{prodmu} is absent in $Y_{\alpha}(z)$ and the other factors $\nu \neq \mu_-$ are as in \eqref{prodnu}. But, as $\mu_-$ only appears in $\mathfrak{v}_{\mu}$ via the first term of the first line of \eqref{alphacm}, which can be consistently set to $0$ since $r_{\mu_-} = 1$, the formula for $\mathfrak{v}_{\mu}$ remains valid.
\end{proof}

\begin{remark} Notice that if $(iv)$ does not hold, i.e. any finite $s_{\mu}$ is coprime to $r_{\mu}$, the condition $\frac{r_{\mu}}{s_{\mu}} = \frac{r_{\nu}}{s_{\nu}}$ is equivalent to $(r_{\mu},s_{\mu}) = (r_{\nu},s_{\nu})$. In that case, condition $(iii)$ can be replaced with a more symmetric one
\begin{itemize}
\item[$(iii)$'] there exist distinct $\mu,\nu \in \tilde{\mathfrak{a}}_{\alpha}$ such that $s_{\mu},s_{\nu}$ are finite, $(r_{\mu},s_{\mu}) = (r_{\nu},s_{\nu})$ and $t_{\mu}^{r_{\mu}} = t_{\nu}^{r_\nu}$.
\end{itemize}
\end{remark}

\medskip

\subsection{Correlators, master loop equations and topological recursion}
\label{SecLoop}
\label{Sec:42} 

\medskip

Let $\mc{C} = (C,x,y)$ be a spectral curve with normalisation $\pi\,:\,\tilde{C} \rightarrow C$.

\begin{definition}
A \emph{fundamental bidifferential of the second kind} on $\mc{C}$ is an element
\begin{equation*}
B \in H^0\big( \tilde{C} \times \tilde{C} ; K_{\tilde{C}}^{\boxtimes 2} (2\Delta )\big)^{\mf{S}_2}\,,
\end{equation*}
with biresidue $1$ on the diagonal $\Delta \subset \tilde{C} \times \tilde{C}$, where $K_{\tilde{C}}$ is the sheaf of differentials on $\tilde{C}$.\par
A \emph{crosscap differential} on $\mc{C}$ is the data of a (possibly empty) divisor $D$ on $\tilde{C} \setminus \tilde{\mathfrak{a}}$ and
\[
q \in H^0\big(\tilde{C};K_{\tilde{C}}(D + \tilde{\mathfrak{a}})\big)\,,
\] 
such that
\[
\forall \alpha \in \mathfrak{a},\qquad \sum_{\mu \in \tilde{\mathfrak{a}}_{\alpha}} \Res_{z = \mu} q(z) = 0\,.
\]
\end{definition}

\begin{definition}
\label{DEFCOR} A \emph{family of correlators} is a family of multidifferentials $\boldsymbol{\omega} = (\omega_{g,n})_{g \in \frac{1}{2}\N, n \geq 1}$ on $\tilde{C}$ such that $\omega_{0,1} = \tilde{y}\,\dd \tilde{x}$, $\omega_{0,2}$ is a fundamental bidifferential of the second kind on $\mc{C}$, $\omega_{\frac{1}{2},1}$ is a crosscap differential, and for $2g - 2 + n > 0$
\begin{equation*}
\omega_{g,n} \in H^0 \big(\tilde{C}^n ; (K_{\tilde{C}}(*\tilde{\mathfrak{a}}))^{\boxtimes n}\big)^{\mathfrak{S}_{n}}\,.
\end{equation*} 
It satisfies the \emph{projection property} if for $2g -2 + n > 0$,
\begin{equation}  
\label{omnorm}\omega_{g,n}(z_1,z_{[2,n]}) = \sum_{\mu \in \tilde{\mathfrak{a}}} \Res_{z = \mu} \bigg( \int_{\mu}^z \omega_{0,2}(\cdot, z_1) \bigg) \omega_{g,n}(z,z_{[2,n]})\,.
\end{equation} 
\end{definition}
Note that \eqref{omnorm} is automatically satisfied for $(g,n) = (0,2)$. Differentials satisfying the projection property cannot have residues, and if they are holomorphic, they must vanish. We can always assume by taking smaller neighborhoods that the divisor $D$ of the crosscap differential is supported outside $\sqcup_{\alpha \in \mathfrak{a}} \tilde{U}_{\alpha}$.

\begin{definition}\label{DisconnConnCor}
Let $\boldsymbol{\omega}$ be a family of correlators, and $i \geq 1$, $g \in \tfrac{1}{2}\N$ and $n \geq 0$. The genus $g$, $i$-disconnected, $n$-connected correlator is defined by 
\begin{equation*}
\mc{W}_{g,i,n} (z_{[i]}; w_{[n]}) \coloneqq  \sum_{\substack{\mathbf{L} \vdash [i]\\ \sqcup_{L \in \mathbf{L}} N_{L} = [n]\\ i + \sum_{L} (g_L - 1) = g}} \prod_{L \in \mathbf{L}} \omega_{g_{L},|L| + |N_L|} (z_{L} ,w_{N_L})\,.
\end{equation*}
We define $\mc{W}_{g,i,n}'$ by the same formula, but omitting any summand containing some $\omega_{0,1}$.
\end{definition}

If $i \in [r_{\alpha}]$, we let $\varkappa_i\,:\,\tilde{U}_{\alpha}^{(i)} \longrightarrow V_{\alpha}$ be the smooth curve obtained by taking the fibered product of $i$ copies of $\tilde{x}\,:\,\tilde{U}_{\alpha}' \rightarrow \mathbb{C}$, deleting the big diagonal $\Delta{(i)}$, and quotienting by the (free) action of $\mathfrak{S}_{i}$. Points in $\tilde{U}_{\alpha}^{(i)}$ are exactly subsets of cardinality $i$ of $\mf{f}_{\alpha}(z)$ for some $z \in \tilde{U}_{\alpha}'$. We have natural holomorphic maps
\[
(\tilde{U}_{\alpha}')^{i} \setminus \Delta{(i)} \mathop{\longrightarrow}^{\mathsf{q}_i} \tilde{U}_{\alpha}^{(i)} \mathop{\longrightarrow}^{\mathsf{x}_i} V_{\alpha}'\,,
\]
where $\mathsf{q}_i$ forgets the order of elements of an $i$-tuple and $\mathsf{x}_i(\{z_1,\ldots,z_n\}) = \tilde{x}(z_1) = \cdots = \tilde{x}(z_i)$. Let $\mathsf{I}_i\,:\,\tilde{U}_{\alpha}^{i} \setminus \Delta^{(i)} \rightarrow \tilde{C}^i$ be the natural inclusion. We introduce
\begin{equation*}
\begin{split}
\mc{E}_{\alpha;g,n}^{(i)} & \coloneqq  \frac{(\mathsf{x}_i\mathsf{q}_i)_*}{i!} \mathsf{I}_i^*(\mc{W}_{g,i,n}) \in H^0\big(V_{\alpha}' \times \tilde{C}^n;K_{V_{\alpha}'}^{\otimes i} \boxtimes K_{\tilde{C}}(*\tilde{\mathfrak{a}})^{\boxtimes n}\big)^{\mathfrak{S}_{n}} \\
\tilde{\mc{E}}_{\alpha;g,n}^{(i)} & \coloneqq  \tilde{x}^*\mc{E}_{\alpha;g,n}^{(i)}\,,
\end{split}
\end{equation*}
where all operations do not concern the last $n$ variables. More concretely,
\begin{equation}
\label{Qi} 
\begin{split}
\mc{E}_{\alpha;g,n}^{(i)}(x_0;z_{[n]}) & = \sum_{\substack{Z \subseteq \tilde{x}^{-1}(x_0) \cap \tilde{U}_{\alpha} \\ |Z|=i}} \mc{W}_{g,i,n} (Z;z_{[n]})\,, \\
\tilde{\mc{E}}_{\alpha;g,n}^{(i)}(z_0;z_{[n]}) & = \sum_{\substack{Z \subseteq \mathfrak{f}_{\alpha}(z_0) \\ |Z| = i}} \mc{W}_{g,i,n}(Z;z_{[n]})\,.
\end{split} 
\end{equation}
The symmetry factor $i!$ disappeared since $\mc{W}_{g,i,n}$ is symmetric in its $i$ first variables. Note that reading \eqref{Qi} in the local coordinate $\zeta_0$ of $z_0 \in \tilde{U}_{\mu}'$ each term may be multivalued --- i.e. fractional powers of $\zeta(z_0)$ could appear --- however the sum is single-valued as it is the pullback along $\tilde{x}$ of a $1$-form on $V_{\alpha}'$.

\begin{definition}
\label{MDef} We say that a family of correlators satisfies the \emph{master loop equations} if for any $g \in \tfrac{1}{2}\N$ and $n \geq 0$ such that $2g - 2 + (n + 1) > 0$, for any $\alpha \in \mathfrak{a}$ and $i \in [r_{\alpha}]$, any $\mu \in \tilde{\mathfrak{a}}_{\alpha}$, when $z_0 \in \tilde{C}_{\mu}'$ approaches $\mu$, we have
\[
\sum_{i = 1}^{r_{\alpha}} \tilde{\mathcal{E}}^{(i)}_{\alpha;g,n}(z_0;z_{[n]}) \big(-\omega_{0,1}(z_0)\big)^{r_{\alpha} - i} = \mc{O}\big(\zeta^{-1 + \mathfrak{v}_{\mu} + (r_{\mu} - 1)(r_{\alpha} - 1)}(\dd \zeta)^{r_{\alpha}}\big)\,.
\]
\end{definition}

The relevance of this notion comes from the fact that the master loop equation can be solved by the topological recursion.

\begin{proposition}\label{MLEtoTR}
Assume that none of the conditions (i), (ii), (iii), and (iv) appearing in \cref{lemYA} are satisfied. Then, if $\boldsymbol{\omega}$ is a family of correlators satisfying the master loop equation (\cref{MDef}) and the projection property (\cref{DEFCOR}), we must have for any $g \in \frac{1}{2}\N$ and $n \geq 0$ such that $2g - 2 + (n + 1) > 0$,
\begin{equation}
\label{TRsum} \omega_{g,n + 1}(z_0,z_{[n]}) = \sum_{\alpha \in \mathfrak{a}} \sum_{\mu \in \tilde{\mathfrak{a}}_{\alpha}} \Res_{z = \mu}\bigg( \sum_{ Z \subseteq \mathfrak{f}'_{\alpha}(z)} K_{\mu}^{(1+|Z|)}(z_0;z,Z) \mc{W}'_{g,1 + |Z|,n}(z,Z;z_{[n]}),\bigg)\,,
\end{equation}
where for $m \geq 2$ we have introduced the $m$th \emph{recursion kernel} for $ |Z| = i-1$
\begin{equation}\label{TRkernel}
K_{\mu}^{(i)} (z_0;z,Z) \coloneqq  - \frac{\int_{\mu}^z \omega_{0,2}(\cdot, z_0)}{\prod_{z' \in Z} \big((\tilde{y}(z') - \tilde{y}(z))\dd \tilde{x}(z)\big)}\,.
\end{equation}
\end{proposition}
\begin{proof}
The proof is similar to \cite[Theorem~7.6.5]{Kra19}, the only difference being the order of the pole in the master loop equation. For completeness, we include the argument here. By definition, $\mc{W}_{g,1,n}' = \omega_{g,n+1}$. By the projection property and \cref{defYA}
\begin{equation}
\label{trpre}
\begin{split}
& \quad \omega_{g,n+1}(z_0,z_{[n]})  \\
&= \sum_{\alpha \in \mathfrak{a}} \sum_{\mu \in \tilde{\mathfrak{a}}_{\alpha}} \Res_{z = \mu}  \bigg( \int_{\mu}^{z} \omega_{0,2}(\cdot,z_0) \bigg) \mc{W}'_{g,1,n}(z;z_{[n]}) \\
&= - \sum_{\alpha \in \mathfrak{a}} \sum_{\mu \in \tilde{\mathfrak{a}}_{\alpha}} \Res_{z = \mu} K_{\mu}^{(r_{\alpha})}(z_0;\mathfrak{f}_{\alpha}(z))\mc{W}'_{g,1,n}(z_0;z_{[n]}) \cdot Y_{\alpha}(z)\cdot \big(\dd \tilde{x}(z)\big)^{(r_{\alpha} - 1)}\,,
\end{split}
\end{equation}
where we noticed that
\begin{equation*} K_{\mu}^{(r_{\alpha})}(z_0;\mathfrak{f}_{\alpha}(z)) = -\frac{\int_{\mu}^{z} \omega_{0,2}(\cdot,z_0)}{Y_{\alpha}(z)\big(\dd \tilde{x}(z)\big)^{(r_{\alpha} - 1)}}
\end{equation*}
always considering $z \in \mathfrak{f}_{\alpha}(z)$ as the first element of the set. Let $\alpha \in \mathfrak{a}$, and use the combinatorial identity \cite[Lemma 7.6.4]{Kra19}, which states that
\begin{equation*}
\sum_{\{z\} \subseteq Z \subseteq \mathfrak{f}_{\alpha}(z)} \!\! \mc{W}'_{g,|Z|,n}(Z;z_{[n]}) \!\!\! \prod_{z' \in \mathfrak{f}_{\alpha}(z) \setminus Z} \! \big((\tilde{y}(z') - \tilde{y}(z))\dd \tilde{x}(z)\big) = \sum_{i = 1}^{r_{\alpha}} \tilde{\mathcal{E}}_{\alpha;g,n}^{(i)}(z;z_{[n]}) \big(-\omega_{0,1}(z)\big)^{r_{\alpha} - i}\,.
\end{equation*}
Isolating the term $Z = \{z\}$ and substituting in \eqref{trpre}, we obtain
\begin{equation}
\label{OMQ}
\begin{split}
\omega_{g,n+1}(z_0,z_{[n]}) &= \sum_{\alpha \in \mathfrak{a}} \sum_{\mu \in \tilde{\mathfrak{a}}_{\alpha}} \Res_{z = \mu} K_{\mu}^{(r_{\alpha})}(z_0;\mathfrak{f}_{\alpha}(z)) \bigg( - \sum_{i=1}^{r_{\alpha}} \tilde{\mc{E}}_{\alpha;g,n}^{(i)}(z;z_{[n]}) \big(-\omega_{0,1} (z) \big)^{r_{\alpha}-i} \\
& \qquad\qquad\qquad\qquad + \sum_{\{z\} \subset Z \subseteq \mathfrak{f}_{\alpha}(z)} \mc{W}'_{g,|Z|,n}(Z;z_{[n]}) \prod_{z' \in \mathfrak{f}_{\alpha}(z) \setminus Z} \big((\tilde{y}(z') - \tilde{y}(z))\dd\tilde{x}(z)\big)\bigg) \\
 &= \sum_{\alpha \in \mathfrak{a}} \sum_{\mu \in \tilde{\mathfrak{a}}_{\alpha}} \Res_{z = \mu} \bigg( \sum_{\{z \} \subset Z \subseteq \mathfrak{f}_{\alpha}(z)} K_{\mu}^{(|Z|)}(z_0;Z) \mc{W}'_{g,|Z|,n}(Z;z_{[n]})\bigg)\,. 
\end{split} 
\end{equation} 
By \cref{lemYA} and the assumption, we know that for $z \in \tilde{C}_{\mu}'$ approaching $\mu$
\[
Y_{\alpha}(z)(\dd \tilde{x}(z))^{r_{\alpha} - 1} \sim \mathfrak{t}_{\alpha,\mu} \zeta^{\mathfrak{v}_{\mu} + (r_{\mu} - 1)(r_{\alpha} - 1)}\,(\dd \zeta)^{r_{\alpha} - 1}
\]
for some non-zero scalar $\mathfrak{t}_{\alpha,\mu}$. Since the numerator of the recursion kernel vanishes at order $1$ at $z = \mu$, the master loop equation implies that the first term inside the bracket of \eqref{OMQ} is $\mc{O}(\dd \zeta)$ hence does not contribute to the residue. Besides, the contribution of the second sum can be simplified by observing that
\[
K_{\mu}^{(|Z|)}(z_0;Z) = K_{\mu}^{(r_{\alpha})}(z_0;\mathfrak{f}_{\alpha}(z)) \prod_{z' \in \mathfrak{f}_{\alpha}(z) \setminus Z} \big((\tilde{y}(z') - \tilde{y}(z))\dd \tilde{x}(z)\big)\,.
\]
Redefining $Z$ by removing $z$ from it, we obtain the desired formulas.
\end{proof}
\begin{remark} From the proof, we see that if one of the conditions \emph{(i)} and \emph{(ii)} appearing in \cref{lemYA} is satisfied, the recursion kernel is ill-defined as the denominator vanishes identically in the neighborhood of some $\mu$. Besides, if one of the conditions \emph{(iii)} or \emph{(iv)} is satisfied, the same thing could occur or at least the order of vanishing of the denominator is finite but higher than the one specified by \cref{MDef}. In the latter case, one can still ask for the analogue of \cref{MLEtoTR} simply by modifying the master loop equation to require that the first sum in \eqref{OMQ} is $\mc{O}(\dd \zeta)$.
\end{remark}

We note that the right-hand side of \eqref{TRsum} involves only $\omega_{g',n'}$ with $2g' - 2 + n' < 2g - 2 + (n + 1)$. For a fixed $\boldsymbol{\omega}^{{\rm un}} = (\omega_{0,1},\omega_{\frac{1}{2},1},\omega_{0,2})$, there exists at most one way to complete it into a system of correlators satisfying the master loop equation and the projection property: the $\omega_{g,n}$ are then determined by \eqref{TRsum} inductively on $2g - 2 + n > 0$. However, such a system of correlators may actually fail to exist at all. Indeed, \eqref{TRsum} gives a non-symmetric role to $z_0$ compared to $z_1,\ldots,z_n$, therefore the $\omega_{g,n + 1}(z_0,\ldots,z_n)$ that \eqref{TRsum} compute may fail to be symmetric, and so would not respect \cref{DEFCOR}.

\medskip

\subsection{Abstract loop equations}
\label{sec:absloop}
\medskip

We now address the aforementioned problem of existence of the solution to the master loop equations, thanks to the results obtained in \cref{SAIRPS}. We first introduce a seemingly different notion of ``abstract loop equations'' valid in the setting of \cref{SecSP}. It will turn out that they give the right generalisation of ``abstract loop equations'' proposed in \cite{BoSh17} for smooth spectral curves. We will show that, under admissibility conditions on the spectral curves that pertain to our constructions of Airy structures in \cref{S1}, the abstract loop equations have a solution satisfying the projection properties, and imply the master loop equation. Therefore, this solution must be given by the topological recursion  formula \eqref{TRsum}, and this proves a posteriori that this definition is well-posed, i.e. it produces inductively only multidifferentials that are symmetric under permutations of all their variables. A direct proof of symmetry by residue computations on $\tilde{C}$ seems rather elusive.

Let $\mc{C}$ be a spectral curve as in \cref{SecSP}. We introduce integers $\mathfrak{d}_{\alpha}(i)$ for each $\alpha \in \mathfrak{a}$ and $i \in [r_{\alpha}]$ matching \cref{prop:k_min}. If $i \in [r_{\alpha}]$, we first decompose it into $i = \mathbf{r}_{[\lambda)} + i'$ for the unique $\lambda \in \tilde{\mathfrak{a}}_{\alpha}$ such that $\mathbf{r}_{[\lambda)} < i \leq \mathbf{r}_{[\lambda]}$. Then, $i' \in [r_{\lambda}]$ and we have
\begin{equation}
\label{dadef} \mathfrak{d}_{\alpha}(i) \coloneqq - \Big\lfloor \frac{s_{\lambda}(i' - 1)}{r_{\lambda}} \Big\rfloor - \mathbf{s}_{[\lambda )} + \delta_{i',1}\,.
\end{equation}

\begin{definition}
\label{def:absloopgen}We say that a family of correlators satisfies the \emph{abstract loop equations} if for any $g \in \tfrac{1}{2}\N$ and $n \geq 0$ such that $2g - 2 + (n + 1) > 0$, for any $\alpha \in \mathfrak{a}$ and $i \in [r_{\alpha}]$ when $x_0 \rightarrow x(\alpha)$ we have
\begin{equation}
\label{abstrloopeq}
\mathcal{E}^{(i)}_{\alpha;g,n}(x_0;z_{[n]}) = \mc{O}\Big(x_0^{-(\mathfrak{d}_{\alpha}(i) - 1- \delta_{i,1})}\,\Big(\frac{\dd x_0}{x_0}\Big)^i\Big)\,.
\end{equation}
This condition is equivalent to the property that, for any $\mu \in \tilde{\mathfrak{a}}_{\alpha}$ we have when $z_0 \in \tilde{C}_{\mu}'$ approaches $\mu$
\[ 
\tilde{\mathcal{E}}^{(i)}_{\alpha;g,n}\big(z_0;z_{[n]}) = \mc{O}\Big(\zeta_0^{-r_{\mu}(\mathfrak{d}_{\alpha}(i) -1 - \delta_{i,1}) }\Big(\frac{\dd \zeta_0}{\zeta_0}\Big)^i\Big)\,.
\]
\end{definition}

\begin{proposition}\label{ALEtoMLE}
Assume that none of the conditions (i), (ii), (iii), (iv) appearing in \cref{lemYA} are satisfied. Then, the abstract loop equations imply the master loop equations.
\end{proposition}
\begin{proof}
We treat the case where $s_{\mu}$ is finite for all $\mu \in \mathfrak{a}$. The case where there could exist $\mu_{\alpha,-} \in \tilde{\mathfrak{a}}_{\alpha}$ (which is then unique) such that $s_{\mu_{\alpha,-}} = +\infty$ is left as exercise to the reader.

For each $\alpha \in \mathfrak{a}$ and $i \in [r_{\alpha}]$, the abstract loop equations imply that for any $\mu \in \tilde{\mathfrak{a}}_{\alpha}$ we have when $z_0 \in \tilde{U}_{\mu}'$ approaches $\mu$
\[
\tilde{\mathcal{E}}^{(i)}_{\alpha;g,n}(z_0;z_{[n]}) \big(-\omega_{0,1}(z)\big)^{r_{\alpha} - i} = O\big(\zeta^{-r_{\mu}(\mathfrak{d}_{\alpha}(i) - 1 -\delta_{i,1}) - i + (s_{\mu} - 1)(r_{\alpha} - i)} (\dd \zeta)^{r_{\alpha}}\big)\,.
\]
Comparing with \cref{MDef}, the result will be proved after we justify that
\[
\mathfrak{p}_{\mu}(i) \coloneqq  -r_{\mu}(\mathfrak{d}_{\alpha}(i) -1 - \delta_{i,1}) - i + (s_{\mu} - 1)(r_{\alpha} - i) - \big(-1 + \mathfrak{v}_{\mu} + (r_{\mu} - 1)(r_{\alpha} - 1)\big)
\]
is always nonnegative. We recall the definition of $\mathfrak{v}_{\mu}$ in \eqref{alphamu}
\[
\mathfrak{v}_{\mu} = (s_{\mu} - r_{\mu})(r_{\alpha} - 1) - \mathbf{r}_{[\mu)}s_{\mu} + \mathbf{s}_{[\mu)}r_{\mu}\,.
\]
We decompose $i = \mathbf{r}_{[\lambda)} + i'$ with the unique $\lambda \in \tilde{\mathfrak{a}}_{\alpha}$ such that $\mathbf{r}_{[\lambda)} < i \leq \mathbf{r}_{[\lambda]}$ and $i' \in [\mathbf{r}_{\lambda}]$, and we denote $\lambda_{\alpha} \coloneqq  \min \tilde{\mathfrak{a}}_{\alpha}$. Inserting the definition of $\mathfrak{d}_{\alpha}(i)$ from \eqref{dadef}, we obtain
\[
\mathfrak{p}_{\mu}(i) = s_{\mu}(1 - i + \mathbf{r}_{[\mu)}) + r_{\mu}\bigg(1 + \Big\lfloor \frac{s_{\lambda}(i' - 1)}{r_{\lambda}} \Big\rfloor + \mathbf{s}_{[\lambda)} - \mathbf{s}_{[\mu)} - \delta_{\lambda \succ \lambda_{\alpha}}\delta_{i',1}\bigg)\,.
\] 
We are going to use often the inequality
\begin{equation}
\label{ineq}
\lfloor x \rfloor > x - 1\,.
\end{equation}
Checking nonnegativity of $\mathfrak{p}_{\mu}(i)$ is done by a case discussion.

\begin{itemize}
\item If $\mu = \lambda$, this becomes
\[
\mathfrak{p}_{\mu}(i) = s_{\mu}(1 - i') + r_{\mu}\bigg(1 + \Big\lfloor \frac{s_{\mu}(i' - 1)}{r_{\mu}}\Big\rfloor -  \delta_{\mu \succ \lambda_{\alpha}}\delta_{i',1}\bigg)\,.
\] 
For $i' = 1$ and $\mu \succ \lambda_{\alpha}$, we get $\mathfrak{p}_{\mu}(i) = 0$. For $i' = 1$ and $\mu = \lambda_{\alpha}$, we get $\mathfrak{p}_{\mu}(i) = r_{\lambda_{\alpha}} > 0$. For $i' \geq 2$, using \eqref{ineq} yields directly $\mathfrak{p}_{\mu}(i) > 0$.
\item If $\mu \prec \lambda$, we have
\[
\mathfrak{p}_{\mu}(i) = s_{\mu}(1 - \mathbf{r}_{[\mu,\lambda )} - i') + r_{\mu}\bigg(1 + \Big\lfloor \frac{s_{\lambda}(i' - 1)}{r_{\lambda}} \Big\rfloor + \mathbf{s}_{[\mu,\lambda )} - \delta_{i',1}\bigg)\,.
\]
For $i' = 1$, this simplifies into
\[
\mathfrak{p}_{\mu}(i) = -s_{\mu}\mathbf{r}_{[\mu,\lambda )} + r_{\mu}\mathbf{s}_{[\mu,\lambda )}\,.
\]
By definition of the order, for all $\nu \in [\mu,\lambda )$ we have $\frac{r_{\mu}}{s_{\mu}} \geq \frac{r_{\nu}}{s_{\nu}}$, therefore $\mathfrak{p}_{\mu}(i) \geq 0$. For $i' \geq 2$, we can use $\frac{s_{\lambda}}{r_{\lambda}} \geq \frac{s_{\mu}}{r_{\mu}}$ and \eqref{ineq} and obtain $\mathfrak{p}_{\mu}(i) > 0$.
\item If $\mu \succ \lambda$, we rather have
\[
\mathfrak{p}_{\mu}(i) = s_{\mu}(1 + \mathbf{r}_{[\lambda,\mu )} - i') + r_{\mu}\bigg(1 + \Big\lfloor \frac{s_{\lambda}(i' - 1)}{r_{\lambda}} \Big\rfloor - \mathbf{s}_{[\lambda,\mu )} - \delta_{\lambda \succ \lambda_{\alpha}} \delta_{i',1}\bigg)\,.
\]
For $i' = 1$ and $\lambda = \lambda_{\alpha}$ this simplifies to
\[
\mathfrak{p}_{\mu}(i) = s_{\mu}\mathbf{r}_{[\mu )} - r_{\mu}\mathbf{s}_{[\mu )} + r_{\mu}
\]
and thanks to the inequality $\frac{r_{\nu}}{s_{\nu}} \geq \frac{r_{\mu}}{s_{\mu}}$ for all $\nu \in [\mu )$ we deduce $\mathfrak{p}_{\mu}(i) \geq r_{\mu} > 0$. For $i' = 1$ and $\lambda \succ \lambda_{\alpha}$, we have
\[
\mathfrak{p}_{\mu}(i) = s_{\mu}\mathbf{r}_{[\lambda,\mu )} - r_{\mu}\mathbf{s}_{[\lambda,\mu )}\,.
\]
Due to the inequality $\frac{r_{\nu}}{s_{\nu}} \geq \frac{r_{\mu}}{s_{\mu}}$ for all $\nu \in [\lambda,\mu )$ we have again $\mathfrak{p}_{\mu}(i) \geq 0$. For $i' \geq 2$, we use the inequality \eqref{ineq} to write
\begin{equation*}  
\begin{split}
\mathfrak{p}_{\mu}(i) & > s_{\mu}(1 + \mathbf{r}_{[\lambda,\mu )} - i') + r_{\mu}\bigg(\frac{s_{\lambda}(i' - 1)}{r_{\lambda}} - \mathbf{s}_{[\lambda,\mu )}\bigg) \\
& > r_{\mu}(i' - 1)\bigg(\frac{s_{\lambda}}{r_{\lambda}} - \frac{s_{\mu}}{r_{\mu}}\bigg) + \mathbf{r}_{[\lambda,\mu )}s_{\mu} - \mathbf{s}_{[\lambda,\mu )}r_{\mu} 
\end{split}
\end{equation*}
and due to the ordering we find again $\mathfrak{p}_{\mu}(i) > 0$.
\end{itemize}
\end{proof}

\begin{remark}
In the proof we see that for any $\alpha \in \mathfrak{a}$, if $ |\tilde{\mathfrak{a}}_\alpha | > 1$, there exists $\mu \in \tilde{\mathfrak{a}}_{\alpha}$ and $i \in [r_{\alpha}]$ such that $\mathfrak{p}_{\mu}(i) = 0$. Therefore, we do use all the vanishing provided by the abstract loop equations to derive the master loop equations.
\end{remark}

Combining with \cref{MLEtoTR}, we obtain the following result.
\begin{proposition}
\label{combinabsloop} Assume that none of the conditions $(i)$, $(ii)$, $(iii)$, $(iv)$ appearing in \cref{lemYA} are satisfied. For a fixed $(\omega_{0,1},\omega_{0,2},\omega_{1,\frac{1}{2}})$, the topological recursion \eqref{TRsum} gives the unique --- if it exists, i.e. if the result is symmetric in all variables --- solution to the abstract loop equations.
\end{proposition}

The notion of abstract loop equation was first introduced \cite{BEO15,BoSh17} for smooth curves with simple ramifications and was shown there to be a mechanism implying directly the topological recursion. This was extended to higher order ramifications on smooth curves having $\tilde{y}$ holomorphic near $\mathfrak{a}$ in~\cite{BoEy17,BBCCN18,Kra19}, and to the more general case where $y\dd x$ is holomorphic near $\mathfrak{a}$ in ~\cite{BBCCN18}. The novelty of \cref{MLEtoTR,ALEtoMLE} here is the treatment of possibly singular curves.

\medskip

\subsection{Topological recursion for admissible spectral curves}
\label{SecTR}

\medskip

In this paragraph, we express the abstract loop equations in a more algebraic way, that will make the bridge to Airy structures. The converse route was anticipated in \cref{SAIRPS}.

Let $\mc{C}$ be a spectral curve. We can attach to it a local spectral curve matching the definitions in \cref{Ftwis}. Namely, we let
\[
\tilde{C}^{{\rm loc}} = \bigsqcup_{\mu \in \tilde{\mathfrak{a}}} \tilde{C}^{{\rm loc}}_{\mu},\qquad \tilde{C}^{{\rm loc}}_{\mu} \coloneqq  {\rm Spec}\,\mathbb{C} \llbracket \zeta \rrbracket\,.
\]
We recall for each $\alpha \in \mathfrak{a}$ and ramification point $\mu \in \tilde{\mathfrak{a}}_{\alpha}$ above $\alpha$, we have a local coordinate $\zeta$ such that:
\[
\tilde{x}\begin{psmallmatrix} \mu \\ z \end{psmallmatrix} = x(\alpha) + \zeta(z)^{r_{\mu}}\,.
\]
For each $\mu \in \tilde{\mathfrak{a}}$, define $ \tilde{C}^{\rm loc'}_\mu \coloneqq \Spec \C (\!( \zeta )\!) $ and let
\[
\mathcal{L}_{\mu} \coloneqq  H^0\big(\tilde{C}^{{\rm loc'}}_{\mu};K_{\tilde{C}^{{\rm loc'}}_{\mu}}\big) \cong \mathbb{C}(\!(\zeta)\!).\dd\zeta
\]
be a copy of the space of formal Laurent series, and
\[
\mathcal{L} = H^0\big(\tilde{C}^{{\rm loc'}};K_{\tilde{C}^{{\rm loc'}}}\big) \cong \bigoplus_{\mu \in \tilde{\mathfrak{a}}} \mathcal{L}_{\mu}\,. 
\]
We denote by
\[ 
\mathsf{Loc}_{\mu}\,:\,\,H^0\big(\tilde{U};K_{\tilde{U}}(*\tilde{\mathfrak{a}})\big) \rightarrow \mathcal{L}_{\mu}
\] 
the linear map associating to a meromorphic differential its all-order Laurent series expansion near $\mu$ using the local coordinate $\zeta$ in $\tilde{U}_{\mu}$, and
\[
\mathsf{Loc} = \bigoplus_{\mu \in \tilde{\mathfrak{a}}} \mathsf{Loc}_{\mu}\,.
\]
We define elements $\dd\xi_{k}^{\mu} \in \mathcal{L}$, indexed by $\mu \in \tilde{\mathfrak{a}}$ and $k \geq 0$
\[
\dd\xi^{\mu}_{k}\big(\begin{smallmatrix} \nu \\ \zeta\end{smallmatrix}\big) = \delta_{\mu,\nu}\,\zeta^{k - 1}\,\dd \zeta\,.
\]
We introduce the standard bidifferential of the second kind on $\tilde{U}$, that is
\[ 
\omega_{0,2}^{{\rm std}}\big(\begin{smallmatrix} \mu_1 & \mu_2 \\ z_1 & z_2 \end{smallmatrix}\big) \coloneqq  \frac{\delta_{\mu_1,\mu_2}\,\dd \zeta(z_1) \dd \zeta(z_2)}{(\zeta(z_1) - \zeta(z_2))^2}\,.
\]

Let now $\boldsymbol{\omega}$ be a family of correlators on $\mc{C}$. We can encode the correlators $\omega_{g,n}$ with $2g - 2 + n \geq 0$ by the following Laurent series expansion
\begin{equation}
\label{locun} \begin{split}
\mathsf{Loc}(\omega_{0,1}) & = \sum_{\substack{\mu \in \tilde{\mathfrak{a}} \\ k > 0}} F_{0,1}\big[\begin{smallmatrix} \mu \\ -k \end{smallmatrix}\big]\,\dd\xi^{\mu}_k\,, \\
\mathsf{Loc}(\omega_{\frac{1}{2},1}) & = \sum_{\mu \in \tilde{\mathfrak{a}}} \bigg(Q_{\mu}\,\dd\xi^{\mu}_{0} + \sum_{k > 0} F_{\frac{1}{2},1}\big[\begin{smallmatrix} \mu \\ -k \end{smallmatrix}\big]\,\dd\xi^{\mu}_{k}\bigg)\,, \\
\mathsf{Loc}^{\otimes 2}(\omega_{0,2} - \omega_{0,2}^{{\rm std}}) & = \sum_{\substack{\mu_1,\mu_2 \in \tilde{\mathfrak{a}} \\ k_1,k_2 > 0}} F_{0,2}\big[\begin{smallmatrix} \mu_1 & \mu_2 \\ -k_1 & -k_2 \end{smallmatrix}\big]\,\dd\xi_{k_1}^{\mu_1} \dd\xi_{k_2}^{\mu_2}\,.
\end{split} 
\end{equation}
Using the fundamental bidifferential of the second kind, we introduce another family of differentials $\dd\xi_{-k}^{\mu}$, now globally defined on $\tilde{C}$ and indexed by $\mu \in \tilde{\mathfrak{a}}$ and $k > 0$
\begin{equation}
\label{xinegdef} \dd\xi^{\mu}_{-k}(z) = \Res_{z' = \mu} \bigg(\int_{\mu}^{z'} \omega_{0,2}(\cdot,z)\bigg) \frac{\dd \zeta(z')}{(\zeta(z'))^{k + 1}}\,.
\end{equation}
Notice that it is such that, for any $\mu,\nu \in \tilde{\mathfrak{a}}$
\[
\mathsf{Loc}_{\nu}(\dd\xi^{\mu}_{-k}) = \frac{\delta_{\mu,\nu}\dd \zeta}{\zeta^{k + 1}} + \sum_{l > 0} \frac{F_{0,2}\big[\begin{smallmatrix} \mu & \nu \\ -k & -l \end{smallmatrix}\big]}{k}\,\dd\xi^{\nu}_{l}\,.
\]
Assuming that $\boldsymbol{\omega}$ satisfies the projection property, by symmetry we can apply this property to each variable to obtain the existence of a finite decomposition for $2g - 2 + n > 0$
\begin{equation}
\label{omgnxp} \omega_{g,n}(z_1,\ldots,z_n) = \sum_{\substack{\mu_1,\ldots,\mu_n \in \tilde{\mathfrak{a}} \\ k_1,\ldots,k_n > 0}} F_{g,n}\big[\begin{smallmatrix} \mu_1 & \cdots & \mu_n \\ k_1 & \cdots & k_n \end{smallmatrix}\big] \prod_{i = 1}^n \dd \xi_{-k_i}^{\mu_i}(z_i)\,,
\end{equation}
where $F_{g,n}\big[\begin{smallmatrix} \boldsymbol{\mu} \\ \mathbf{k} \end{smallmatrix}\big]$ are scalars.

\begin{definition}
The \emph{partition function} $Z$ associated to a $\boldsymbol{\omega}$ satisfying the projection property is defined as
\[
Z \coloneqq e^F\,, \qquad F = \sum_{\substack{g \in \tfrac{1}{2}\N,\,\,n \geq 1  \\ 2g - 2 + n > 0}} \sum_{\substack{\mu_1,\ldots,\mu_n \in \tilde{\mathfrak{a}} \\ k_1,\ldots,k_n > 0}} \frac{\hbar^{g - 1}}{n!} F_{g,n}\big[\begin{smallmatrix} \mu_1 & \cdots & \mu_n \\ k_1 & \cdots & k_n \end{smallmatrix}\big]\,\prod_{i = 1}^n x_{k_i}^{\mu_i}\,.
\]
\end{definition}

We now would like to translate the abstract loop equations on $\boldsymbol{\omega}$ into constraints for its partition function. For this purpose, we introduce for each $\alpha \in \mathfrak{a}$ a copy $W_{\alpha;i,k}$ of the differential operators in \cref{eq:twist_mode_arbitr} indexed by $i \in [r_{\alpha}]$ and $k \in \mathbb{Z}$ forming a representation of the $\mc{W}(\mathfrak{gl}_{r_{\alpha}})$-VOA using as twists permutations $\sigma_{\alpha}$ which is a product of disjoint cycles of respective orders $(r_{\mu})_{\mu \in \tilde{\mathfrak{a}}_{\alpha}}$. They are described in terms of the Heisenberg generators indexed by $\mu \in \tilde{\mathfrak{a}}$ and $k \in \mathbb{Z}$
\[
J^{\mu}_{k} = \left\{\begin{array}{lll} \hbar\partial_{x_{k}^{\mu}} & & {\rm if}\,\,k > 0 \\[0.2ex] \hbar^{\frac{1}{2}}Q_{\mu} & & {\rm if}\,\,k = 0 \\[0.3ex] -kx_{-k}^{\mu} & & {\rm if}\,\,k < 0 \end{array}\right.\,.
\]
where we use
\[
Q_{\mu} = \Res_{z = \mu} \omega_{\frac{1}{2},1}(z)
\]
coming from the crosscap differential. Then, we construct the dilaton shift and the change of polarisation
\begin{equation*}
\begin{split}
\hat{T} & = \exp\left(\sum_{\substack{\mu \in \tilde{\mathfrak{a}} \\ k > 0}} \Big(\hbar^{-1}\,F_{0,1}\big[\begin{smallmatrix} \mu \\ -k \end{smallmatrix}\big] + \hbar^{-\frac{1}{2}}\,F_{\frac{1}{2},1}\big[\begin{smallmatrix} \mu \\ -k \end{smallmatrix}\big]\Big)\,\frac{J_{k}^{\mu}}{k}\right)\,, \\
\hat{\Phi} & = \exp\left(\frac{1}{2\hbar} \sum_{\substack{\mu,\nu \in \tilde{\mathfrak{a}} \\ k,l > 0}}F_{0,2}\big[\begin{smallmatrix} \mu & \nu \\ -k & -l \end{smallmatrix}\big]\,\frac{J_{k}^{\mu}J_{l}^{\nu}}{kl}\right)\,.
\end{split} 
\end{equation*}

\begin{definition}
\label{Airytocurve}To a spectral curve $\mc{C}$ equipped with a crosscap differential $\omega_{\frac{1}{2},1}$ and a fundamental bidifferential of the second kind $\omega_{0,2}$, we associate the system of differential operators  indexed by $\alpha \in \mathfrak{a}$, $i \in [r_{\alpha}]$ and $k \in \mathbb{Z}$
\[
H_{\alpha;i,k} \coloneqq  \hat{\Phi}\,\hat{T} \cdot {}^{\sigma_\alpha}W_{i,k}  \cdot \hat{T}^{-1}\,\hat{\Phi}^{-1}\,,
\]
where $ W^{\sigma}_{i,k}$ is as in \cref{eq:twist_mode_arbitr}, and $ \sigma_\alpha$ is the monodromy permutation at $\alpha$.
We also introduce the set
\[
\mathcal{I} \coloneqq  \Big\{(\alpha,i,k) \quad \Big| \quad \, \alpha \in \mathfrak{a},\,\, i \in [r_{\alpha}],\,\, k \geq \mathfrak{d}_{\alpha}(i) - \delta_{i,1}\Big\}\,.
\]
\end{definition}

\begin{proposition}\label{HASequivALE}
Assume that none of the conditions (i), (ii), (iii), (iv) appearing in \cref{lemYA} are satisfied, and let $\boldsymbol{\omega}$ be a system of correlators satisfying the projection property. Then, the abstract loop equations for $\boldsymbol{\omega}$ are equivalent to the following system of differential equations for its partition function:
\begin{equation}
\label{PDEtocorr}
\forall (\alpha,i,k) \in \mathcal{I},\qquad e^{-F} H_{\alpha;i,k} e^{F} \cdot 1 = 0\,.
\end{equation}
\end{proposition}
\begin{proof}
If $|\mathfrak{a}| = 1$ this is the computation done in \cref{SAIRNUSN}. Given the formalism that we introduced, it is straightforward to adapt it to handle several $\alpha$s, where the $H_{\alpha;i,k}$ now form a representation of the direct sum over $\alpha \in \mathfrak{a}$ of the $\mc{W}(\mathfrak{gl}_{r_{\alpha}})$-VOAs.
\end{proof}

It is now easy to combine the construction of Airy structures in \cref{thm:W_gl_Airy_arbitrary_autom} with \cref{MLEtoTR,ALEtoMLE} to obtain our second main result. We recall that we had defined
\begin{equation*}
t_{\mu} = -\frac{1}{r_{\mu}} F_{0,1}\big[\begin{smallmatrix} \mu \\ -s_{\mu} \end{smallmatrix}\big] \,.
\end{equation*}
  
\begin{definition}
\label{dereg} We say $\alpha \in \mathfrak{a}$ is \emph{regularly admissible} if
\begin{itemize}
\item $C$ is irreducible locally at $\alpha$, that is $|\tilde{\mathfrak{a}}_{\alpha}| = 1$.
\item $\tilde{y}$ is holomorphic near $\alpha$ and $\dd \tilde{y}(\alpha) \neq 0$.
\end{itemize}
\end{definition} 
In that case, in all the previous definitions and constructions in the neighborhood $\tilde{U}_{\alpha}$ we replace $\tilde{y}(z)$ with $\tilde{y}(z) - \tilde{y}(\alpha)$. In particular, we take $s_{\alpha} = r_{\alpha} + 1$, and the value of $\tilde{y}(\alpha)$ plays absolutely no role in all the results we have mentioned.

\begin{definition}
\label{deirreg}We say $\alpha \in \mathfrak{a}$ is \emph{irregularly admissible} if
\begin{itemize}
\item for any $\mu \in \tilde{\mathfrak{a}}_{\alpha}$ such that $ r_\mu > 1$, $\tilde{y}$ has a pole at $\mu$ but $\tilde{y}\dd \tilde{x}$ is regular at $\mu$. In particular, this imposes $s_{\mu} \in [1,r_{\mu})$.
\item for any distinct $\mu,\nu \in \tilde{\mathfrak{a}}_{\alpha}$ such that $(r_\mu,s_\mu) = (r_\nu,s_\nu)$, we have $t_\mu^{r_\mu} \neq t_\nu^{r_\nu}$.
\item if $ |\tilde{\mf{a}}_\alpha| >1$, there exist distinct $\mu_{+},\mu_- \in \tilde{\mathfrak{a}}_{\alpha}$ such that $r_{\mu_{\pm}} = \mp 1\,\,{\rm mod}\,\,s_{\mu_{\pm}}$ and $\tfrac{r_{\mu_+}}{s_{\mu_+}} \geq \tfrac{r_{\mu_-}}{s_{\mu_-}}$, and for any $\mu \in \tilde{\mathfrak{a}}_{\alpha}\setminus \{\mu_-,\mu_+\}$, we have $s_{\mu} = 1$ and $\tfrac{r_{\mu_+}}{s_{\mu_+}} \geq r_{\mu} \geq \tfrac{r_{\mu_-}}{s_{\mu_-}}$.
\item if $ |\tilde{\mf{a}}_\alpha| =1$ then $r_{\mu} = \pm 1\,\,{\rm mod}\,\,s_{\mu}$ for $\mu\in\tilde{\mf{a}}_\alpha$.
\end{itemize} 
\end{definition}
These conditions always imply that for any $\mu$, we have ${\rm gcd}(r_\mu,s_\mu) = 1$; in other words the plane curve $(\tilde{C},\tilde{x},\tilde{y})$ is locally irreducible at $\mu$. Here, the second condition avoids the pathology of \emph{(iv)} in \cref{lemYA} and the next results. The third condition is then equivalent to avoiding the pathology \emph{(iii)} in \cref{lemYA}, because $\frac{r_{\mu}}{s_{\mu}} = \frac{r_{\nu}}{s_{\nu}}$ and $(r_{\mu},s_{\mu})$ coprime, $(r_{\nu},s_{\nu})$ coprime imply that $(r_{\mu},s_{\mu}) = (r_{\nu},s_{\nu})$. The fourth and fifth conditions match those in \cref{thm:W_gl_Airy_arbitrary_autom} if $d >1$, and the case $d=1$ corresponds to \cref{thm:W_gl_Airy_Coxeter}.

\begin{definition}
\label{deexpreg} We say $\alpha \in \mathfrak{a}$ is \emph{exceptionally admissible} if
\begin{itemize}
\item there exists a unique $\mu_- \in \tilde{\mathfrak{a}}_{\alpha}$ such that $s_{\mu_-} = +\infty$, and it has $r_{\mu_-} = 1$.
\item the three first properties in \cref{deirreg} that do not involve $\mu_-$ are satisfied.
\item there exists $\mu_+ \in \tilde{\mathfrak{a}}_{\alpha} \setminus \{\mu_-\}$ such that $r_{\mu_+} = -1\,\,{\rm mod}\,\,s_{\mu_+}$.
\item for any $\mu \in \tilde{\mathfrak{a}}_{\alpha} \setminus \{\mu_-,\mu_+\}$, we have $s_{\mu} = 1$ and $\tfrac{r_{\mu_+}}{s_{\mu_+}} \geq r_{\mu}$.
\end{itemize}
\end{definition}
Allowing infinite $s$, the first condition guarantees that we avoid the pathologies \emph{(i)} and \emph{(ii)} in \cref{lemYA}, which make the denominator of the recursion kernel be identically zero in some open set. The last two conditions match those in \cref{thm:W_gl_Airy_arbitrary_autom}.

\begin{definition}
\label{dall} A spectral curve $\mc{C} = (C,x,y)$ is \emph{admissible} if all $\alpha \in \mathfrak{a}$ are either regularly, irregularly or exceptionally admissible. The tuple $(r_{\mu},s_{\mu})_{\mu \in \tilde{\mathfrak{a}}_{\alpha}}$ is called the type of the ramification point $\alpha \in \mathfrak{a}$. 
\end{definition}

\begin{theorem}
\label{mainth2}Let $\mc{C}$ be an admissible spectral curve equipped with a fundamental bidifferential of the second kind $\omega_{0,2}$ and with a crosscap differential $\omega_{\frac{1}{2},1}$. Then there exists a unique way to complete $(\omega_{0,1},\omega_{\frac{1}{2},1},\omega_{0,2})$ into a system of correlators $\boldsymbol{\omega}$ satisfying the projection property and the abstract loop equations (or the master loop equations). Moreover, $\omega_{g,n}$ is computed by the topological recursion \eqref{TRsum} by induction on $2g - 2 + n > 0$, and the result of this formula is symmetric in all its variables. The $F_{g,n}$ determined by its decomposition are the coefficients of expansion of the partition function of the Airy structure introduced in \cref{Airytocurve}.
\end{theorem}

For smooth curves with simple ramifications --- i.e. $\tilde{C} = C$, $|\tilde{\mf{a}}_{\alpha}| = 1$ and $r_{\alpha} = 2$ for all $\alpha \in \mathfrak{a}$ --- the symmetry is proved in \cite[Theorem~4.6]{EyOr07}. For admissible smooth curves, \cref{mainth2} is proved in \cite[Theorem 5.32]{BBCCN18}. For singular curves, the admissibility condition we have adopted is not far from being optimal for this formulation of the abstract loop equation/formulas like \eqref{TRsum}. It may not be impossible to define a topological recursion for more general spectral curves, but either the formula \eqref{TRsum} will have to be different or the exponents in the master loop equations/abstract loop equations should be increased, in a consistent way so that there still exist a unique symmetric solution.\par

We obtained \Cref{mainth2} by exploiting the dictionary established in \Cref{SAIRPS} and using \Cref{prop:AiryAllDilatonPolarize}. It is therefore natural to ask what the analogous statement of \Cref{thm:genus0_soln_admissible_Hik} (with the addition of \Cref{rem:AiryAllDilatonPolarizegenuszero}) in the setting of spectral curves should be. For this we make the following observation.
\begin{proposition}\label{HASequivALEgequal0}
	Again assume that none of the conditions (i), (ii), (iii), (iv) appearing in \cref{lemYA} are satisfied, and let $\boldsymbol{\omega}$ be a system of correlators satisfying the projection property. Then, $\boldsymbol{\omega}$ satisfies the abstract loop equations for $g=0$\footnote{With this we mean that for any $n \geq 2$, $\alpha \in \mathfrak{a}$ and $i \in [r_{\alpha}]$ equation \eqref{abstrloopeq} is satisfied for $g=0$ as $x_0 \rightarrow x(\alpha)$.} if and only if the associated partition function satisfies
	\begin{equation}
		\label{PDEtocorr_leadingord}
		\forall (\alpha,i,k) \in \mathcal{I},\qquad e^{-F} H_{\alpha;i,k} e^{F} \cdot 1 = o(\hbar^{\frac{1}{2}})\,.
	\end{equation}
\end{proposition}
From this it becomes clear how to translate \Cref{thm:genus0_soln_admissible_Hik}.
\begin{definition}
	\label{dequasi} We say $\alpha \in \mathfrak{a}$ is \emph{admissible in genus $0$} if
	\begin{itemize}
		\item There is at most one  $\mu_- \in \tilde{\mathfrak{a}}_{\alpha}$ for which $s_{\mu_-} = +\infty$. If such a $\mu_-$ exists then $r_{\mu_-} = 1$.
		
		\item for any distinct $\mu,\nu \in \tilde{\mathfrak{a}}_{\alpha}$ such that $(r_\mu,s_\mu) = (r_\nu,s_\nu)$, we have $t_\mu^{r_\mu} \neq t_\nu^{r_\nu}$.
		
		\item $r_{\mu} = \pm 1\,\,{\rm mod}\,\,s_{\mu}$ for all $\mu\in\tilde{\mf{a}}_\alpha$.
		
		\item for all $\mu_1 \neq \mu_2$ with $s_{\mu_i}>2$ such that either 
		\begin{equation*}
			r_{\mu_1} = 1 \,\,{\rm mod}\,\, s_{\mu_1} \text{ and } r_{\mu_2} = 1 \,\,{\rm mod}\,\, s_{\mu_2}
		\end{equation*}
		or
		\begin{equation*}
			r_{\mu_1} = -1 \,\,{\rm mod}\,\, s_{\mu_1} \text{ and } r_{\mu_2} = -1 \,\,{\rm mod}\,\, s_{\mu_2}
		\end{equation*}
		one has $\big\lfloor \frac{r_{\mu_1}}{s_{\mu_1}} \big\rfloor \neq \big\lfloor\frac{ r_{\mu_2}}{s_{\mu_2}}\big\rfloor$.
		
		\item if there are pairwise distinct $\mu_1,\mu_2,\mu_3\in\tilde{\mathfrak{a}}_{\alpha}$ with $\big\lfloor\frac{r_{\mu_1}}{s_{\mu_1}}\big\rfloor = \big\lfloor\frac{r_{\mu_2}}{s_{\mu_2}}\big\rfloor = \big\lfloor\frac{r_{\mu_3}}{s_{\mu_3}}\big\rfloor$, then there is an $m\in \{1,2,3\}$ for which $s_{\mu_m}=1$.
	\end{itemize}
	A spectral curve $\mc{C} = (C,x,y)$ is \emph{admissible in genus $0$} if all $\alpha \in \mathfrak{a}$ are.
\end{definition}
\begin{theorem}
	\label{mainth2gequal0}If $\mc{C}$ be a spectral curve admissible in genus $0$, equipped with a fundamental bidifferential of the second kind $\omega_{0,2}$ there exists a unique way to complete $(\omega_{0,1},\omega_{0,2})$ into a system of correlators $(\omega_{0,n})_{n\geq 1}$ satisfying the projection property and the abstract loop equations (or the master loop equations) for $g=0$. Moreover, $\omega_{0,n}$ is computed by the topological recursion \eqref{TRsum} by induction on $n > 0$, and the result of this formula is symmetric in all its variables. The $F_{0,n}$ determined by its decomposition are the free energies solving \eqref{PDEtocorr_leadingord}. 
\end{theorem}
Clearly, not every curve admissible in genus $0$ is admissible and there is indeed a good reason for this definition. Suppose we are in the setting of \Cref{mainth2gequal0}. Then given a crosscap differential $\omega_{\frac{1}{2},1}$ one could ask whether one can extend the family $(\omega_{0,1},\omega_{0,2},\omega_{0,3},\ldots,\omega_{\frac{1}{2},1})$ to a system of correlators $\boldsymbol{\omega}$ satisfying the abstract loop equations for all $g\geq 0$. Doing explicit calculations we will observe in \Cref{Sec63} that this is not always possible without at least imposing certain conditions on the choice of $\omega_{\frac{1}{2},1}$.

\medskip

\subsection{Decoupling of exceptional components}

\medskip

If we erase some or all of the components of an exceptionally admissible local spectral curve $\mc{C}$ indexed by the $\mu_- \in \mathfrak{a}$ such that $s_{\mu_-} = \infty$, we still obtain an admissible local spectral curve $\mc{C}'$. We prove below a decoupling result if $\omega_{0,2}$ has no cross-terms with these components and $\omega_{\frac{1}{2},1}$ vanishes on these components. This decoupling means that the computing $\omega_{g,n}$ on $ \mc{C}$ and restricting to $\mc{C}'$ gives the same result as restricting $(\omega_{0,1},\omega_{0,2},\omega_{\frac{1}{2},1})$ to $ \mc{C}'$ and then computing $ \omega_{g,n}$ by the topological recursion on $ \mc{C}'$.

\begin{proposition}
\label{propdecouple} Let $(C,x,y)$ be an exceptionally admissible spectral curve, equipped with a fundamental bidifferential of the second kind $\omega_{0,2}$ and with a crosscap differential $\omega_{\frac{1}{2},1}$.

Let $\mathfrak{a}'$ be a non-empty subset of exceptionally admissible ramification points, and denote $\tilde{\mathfrak{a}}'$ the set of $\mu_- \in \tilde{\mathfrak{a}}$ such that $\mu_- \in \tilde{\mathfrak{a}}_{\alpha}$ for some $\alpha \in \mathfrak{a}'$ and $s_{\mu_-} = \infty$. Assume that for any $\mu_- \in \tilde{\mathfrak{a}}'$ and $\nu \in \tilde{\mathfrak{a}} \setminus \tilde{\mathfrak{a}}'$ we have
\[
(\mathsf{Loc}_{\mu_-} \otimes \mathsf{Loc}_{\nu})(\omega_{0,2}) = 0\,,\qquad \mathsf{Loc}_{\mu_-}(\omega_{\frac{1}{2},1}) = 0\,.
\]

Then, if we denote $\boldsymbol{\omega}$ the outcome of topological recursion and
\[
\mathsf{Loc}' = \bigsqcup_{\mu \in \tilde{\mathfrak{a}}\setminus \tilde{\mathfrak{a}}'} \mathsf{Loc}_{\mu}\,,
\]
the $\mathsf{Loc}'$-projection of the system of correlators $\boldsymbol{\omega}$ obtained from $(C,x,y,\omega_{0,2},\omega_{1,\frac{1}{2}})$ by the topological recursion \eqref{TRsum}, satisfies the topological recursion on the local spectral curve
\[
(\tilde{C}^{{\rm loc}})' = \bigsqcup_{\mu \in \tilde{\mathfrak{a}} \setminus \tilde{\mathfrak{a}}'} \tilde{C}_{\mu}
\]
equipped with the restriction of $x,y,\omega_{0,2},\omega_{\frac{1}{2},1}$ onto $(\tilde{C}^{{\rm loc}})'$.
\end{proposition}

\begin{proof}
We detail the proof in the case of a single ramification point with  $s_{\mu_-} = \infty$. The general case follows because topological recursion is local. It suffices to work from the start with the normalised local spectral curve attached to $(C,x,y)$. We write $ \tilde{C}^{\rm loc}_- $ for the connected component of $ \tilde{C}^{\rm loc}$ associated to $ \mu_-$ and $ \tilde{C}^{\rm loc}_+$ for all other components.

First let us prove all $ \omega_{g,n}$ with exactly one argument in $\tilde{C}_-^{{\rm loc}}$ are zero. We will prove this by induction on the Euler characteristic. The base cases hold, as $ \omega_{0,1}$ and $ \omega_{\frac{1}{2},1}$ vanish on $\tilde{C}_-^{{\rm loc}}$ and $ \omega_{0,2}$ does not have cross-terms. For the induction step, let us recall the topological recursion formula \eqref{TRsum}.
\begin{align*}
\omega_{g,n + 1}(z_0,z_{[n]}) &= \sum_{\alpha \in \mathfrak{a}} \sum_{\mu \in \tilde{\mathfrak{a}}_{\alpha}} \Res_{z = \mu}\bigg( \sum_{\{z\} \subset Z \subseteq \mathfrak{f}_{\alpha}(z)} K_{\mu}^{(|Z|)}(z_0;Z) \mc{W}'_{g,|Z|,n}(Z;z_{[n]})\bigg)\,,\\
K_{\mu}^{(m)} (z_0;z_{[m]}) &\coloneqq  - \frac{\int_{\mu}^z \omega_{0,2}(\cdot, z_0)}{\prod_{l=2}^m \big((\tilde{y}(z_l) - \tilde{y}(z_1))\dd \tilde{x}(z_1)\big)}\,.
\end{align*}
In this case, $ \mf{a} = \{ 0\}$, and $ \tilde{\mf{a}}_0 = \{ 0 \in \tilde{C}_+^{{\rm loc}}, 0\in \tilde{C}_-^{{\rm loc}}\}.$\par
Let us analyse this formula with $ z_0 \in \tilde{C}_-^{{\rm loc}}$ and all $ z_i  \in \tilde{C}_+^{{\rm loc}}$ for $ i \in [n]$.  First, because the $ \omega_{0,2}$ does not contain cross-terms, we must have $ z \in C_-$, so also $ \mu = 0 \in \tilde{C}_-^{{\rm loc}}$. Because $\tilde{x}|_{\tilde{C}_-^{{\rm loc}}}$ is injective, $\mf{f}_0(z)$ contains no other point in $\tilde{C}^{{\rm loc}}_-$. Hence for all terms in the sum, $ Z \cap \tilde{C}^{{\rm loc}}_- = \{ z\}$, and by the induction hypothesis, all $\mc{W}'(Z;z_{[n]})$ are zero (they contain a factor $ \omega_{g',n'}$ with $ 2g' -2 +n' < 2g-2+n$ and one argument in $\tilde{C}^{{\rm loc}}_-$).\par
Now, we will prove the proposition using a similar induction. The base cases, $ \omega_{0,1}$ and $ \omega_{0,2}$ (and the trivial $ \omega_{\frac{1}{2},1}$) do indeed not mix several components.\par
For the induction step, we again look at the topological recursion formula. Let us look at the terms contributing in the case $ z_0$ and all $z_{[n]}$ are in $\tilde{C}^{{\rm loc}}_+$.  As before, $ z \in \tilde{C}^{{\rm loc}}_+$. Furthermore, $ \mf{f}_0(z) $ contains exactly one element, say $ \zeta$, in $\tilde{C}^{{\rm loc}}_-$, so any $ Z$ contains at most one such element. Any term in the sum not containing $\zeta $ also contributes to the topological recursion on $\tilde{C}^{{\rm loc}}_+$, and all terms including $ \zeta $ must vanish by the first part of this proof, as they have a factor $ \omega_{g',n'}$ with exactly one argument --- namely $\zeta $ --- in $\tilde{C}^{{\rm loc}}_-$. 
\end{proof}

\medskip

\section{Calculations for low \texorpdfstring{$2g - 2 + n$}{2g-2+n}.}
\label{sec:free_energ_arb_autom}
\medskip

In this section, we calculate some of the first correlators in the unique way of \cref{mainth2}, but for not necessarily admissible spectral curves. In this generality, there is no guarantee for the correlators to be symmetric, and we will find that indeed they are not for certain choices of parameters. As correlators coming from Airy structures are symmetric by construction, these calculations give necessary condition for our collections of differential operators to form Airy structures. These conditions are summarised in \cref{sec:NecCond}.
\medskip
\subsection{The standard case}
\label{sec:stndcasecorr}
\medskip

Let us consider the spectral curve with a unique ramification point $\alpha$ at which $d$ irreducible components labelled by $\tilde{\mf{a}}_\alpha=\tilde{\mf{a}}$ intersect, and defined for $\mu \in \tilde{\mf{a}}$ and $z$ on the $\mu$th component of the normalisation by the formulas
\begin{equation}
\label{eq:monoSC}
x\big(\begin{smallmatrix} \mu \\ z \end{smallmatrix}\big) = z^{r_{\mu}}\,,\qquad y\big(\begin{smallmatrix} \mu \\ z \end{smallmatrix}\big) = -t_{\mu}\,z^{s_{\mu} - r_{\mu}}\,.
\end{equation}
We equip it with the bidifferential and the crosscap differential
\[
\qquad \omega_{0,2}\big(\begin{smallmatrix} \mu_1 & \mu_2 \\ z_1 & z_2 \end{smallmatrix}\big) = \delta_{\mu_1,\mu_2}\,\frac{\dd z_1 \dd z_2}{(z_1 - z_2)^2}\,,\qquad  \omega_{\frac{1}{2},1}\big(\begin{smallmatrix} \mu \\ z \end{smallmatrix}\big) = \frac{Q_{\mu}\dd z}{z}\,.
\]
We assume that ${\rm gcd}(r_{\mu},s_{\mu}) = 1$, and that for $\frac{r_{\mu}}{s_{\mu}} = \frac{r_{\nu}}{s_{\nu}}$ and $\mu \neq \nu$ we must have $t_{\mu}^{r_\mu} \neq t_{\nu}^{r_\nu}$. The correlators for $\chi = 2 - 2g - n = -1$ are
\begin{equation}
\begin{split}  
\label{thefirstg}\omega_{0,3}(z_1,z_2,z_3) & = \sum_{\mu \in \tilde{\mf{a}}} \mathop{{\rm Res}}_{z = \mu} \sum_{z' \in \mathfrak{f}'(z)} K_\mu(z_1,z,z')\Big(\omega_{0,2}(z,z_2)\,\omega_{0,2}(z',z_3) + \omega_{0,2}(z,z_3)\,\omega_{0,2}(z',z_2)\Big)\,, \\ 
\omega_{\frac{1}{2},2}(z_1,z_2) & = \sum_{\mu \in \tilde{\mf{a}}}  \mathop{{\rm Res}}_{z = \mu} \sum_{z' \in \mathfrak{f}'(z)} K_\mu(z_1,z,z')\Big(\omega_{0,2}(z,z_2)\,\omega_{\frac{1}{2},1}(z') + \omega_{0,2}(z',z_2)\, \omega_{\frac{1}{2},1}(z)\Big)\,, \\ 
\omega_{1,1}(z_1) & = \sum_{\mu \in \tilde{\mf{a}}}  \mathop{{\rm Res}}_{z = \mu} \sum_{z' \in \mathfrak{f}'(z)} K_\mu(z_1,z,z')\Big(\omega_{0,2}(z,z') + \omega_{\frac{1}{2},1}(z)\,\omega_{\frac{1}{2},1}(z')\Big) \,,
\end{split}
\end{equation}
where  
\[
K_\mu(z_1,z,z') = \frac{\int_\mu^{z} \omega_{0,2}(\cdot,z_1)}{(y(z) - y(z'))\dd x(z)}\,.
\]

In this section, we compute these correlators, and show that the symmetry of $\omega_{0,3}$ and $\omega_{\frac{1}{2},2}$ poses constraints on the parameters $(r_{\mu},s_{\mu},t_{\mu},Q_{\mu})_{\mu \in \tilde{\mf{a}}}$. We also obtain similar constraints from partial calculation of $ \omega_{0,4} $. In the light of \Cref{HASequivALEgequal0} and \Cref{mainth2gequal0} the symmetry of $\omega_{0,3}$ and $ \omega_{0,4} $ is necessary so that the family of differential operators from admits a partition function solving the associated system of differential equations to leading order in $\hbar^\frac{1}{2}$. In the same way \Cref{HASequivALE} and \Cref{mainth2} tell us that additionally we need to take into account the symmetry conditions for $\omega_{\frac{1}{2},2}$ which are necessary constraints to obtain Airy structures from the construction presented in Section~\ref{S2}. These ideas will be used in \cref{secneccons} to prove \cref{prop:gen_zero_nec=suff,propneccons}.

\medskip

\subsubsection{Genus zero}
\label{GenusZeroConstraints}

\medskip

In this \namecref{GenusZeroConstraints}, we calculate $\omega_{0,3}$ and obtain constraints on the parameters of the spectral curve \eqref{eq:monoSC} that are necessary for the symmetry of the correlators $ \omega_{0,3} $ and $ \omega_{0,4}$.
\begin{proposition}
	\label{prop:symm_cond_omega03}
	Assume $r_\mu$ and $s_\mu$ are coprime for all $\mu\in\tilde{\mf{a}}$. Then $\omega_{0,3}$ is symmetric if and only if the following holds
	\begin{enumerate}
		\item\label{item:r_eq_pm1mods}  $r_\mu = \pm 1 \,\,{\rm mod}\,\, s_\mu$ for all $\mu\in \tilde{\mf{a}}$.
		\item\label{item:eq_fract_constr} For all $\mu_1 \neq \mu_2$ with $s_{\mu_i}>2$ such that either 
		\begin{equation*}
		r_{\mu_1} = 1 \,\,{\rm mod}\,\, s_{\mu_1} \text{ and } r_{\mu_2} = 1 \,\,{\rm mod}\,\, s_{\mu_2}
		\end{equation*}
		or
		\begin{equation*}
		r_{\mu_1} = -1 \,\,{\rm mod}\,\, s_{\mu_1} \text{ and } r_{\mu_2} = -1 \,\,{\rm mod}\,\, s_{\mu_2}
		\end{equation*}
		one has $\big\lfloor \tfrac{r_{\mu_1}}{s_{\mu_1}} \big\rfloor \neq \big\lfloor\tfrac{ r_{\mu_2}}{s_{\mu_2}}\big\rfloor$.
	\end{enumerate}
	When these conditions are satisfied, then $\omega_{0,3}$ is given by
	\begin{equation}\label{03explicit}
	\omega_{0,3}\big(\begin{smallmatrix} \mu_1 & \mu_2 & \mu_3 \\ z_1 & z_2 & z_3 \end{smallmatrix}\big) = \sum_{\mu}\frac{c_\mu \delta_{\mu_1,\mu_2,\mu_3,\mu}}{t_\mu r_\mu} \sum_{k_1,k_2,k_3 > 0} \frac{k_1k_2k_3\,\dd z_1  \dd z_2 \dd z_3}{z_1^{k_1 + 1}z_2^{k_2 + 1}z_3^{k_3 + 1}} \delta_{k_1 + k_2 + k_3,s_{\mu}}\,,
	\end{equation}
	where
	\begin{equation}
	\label{defcmu}
	c_\mu \coloneqq \begin{cases} - r'_\mu & r_\mu = r'_\mu s_\mu + 1\\ r'_\mu + 1 & r_\mu = r'_\mu s_\mu + s_\mu -1\end{cases}\,.
	\end{equation} 
\end{proposition}

\begin{proposition}
	\label{prop:symm_cond_omega04}
	Assume conditions \labelcref{item:r_eq_pm1mods} and \labelcref{item:eq_fract_constr} of \Cref{prop:symm_cond_omega03} hold. Then $\omega_{0,4}$ can only be symmetric if
	\begin{enumerate}
		\setcounter{enumi}{2}
		\item\label{item:threefracteq} In case there are $\mu_1,\mu_2,\mu_3$ distinct with $\big\lfloor\frac{r_{\mu_1}}{s_{\mu_1}}\big\rfloor = \big\lfloor\frac{r_{\mu_2}}{s_{\mu_2}}\big\rfloor = \big\lfloor\frac{r_{\mu_3}}{s_{\mu_3}}\big\rfloor$ then there is an $m\in\{1,2,3\}$ for which $s_{\mu_m}=1$.
	\end{enumerate}
\end{proposition}
Notice that the conditions from \Cref{prop:symm_cond_omega03,prop:symm_cond_omega04} are exactly the defining properties of a curve admissible in genus $0$ which are those curves for which we already know that all $(\omega_{0,n})_{n\geq 1}$ computed by the topological recursion will be symmetric.
\begin{corollary}\label{cor:genuszerosymclassification}
Let $\mc{C}$ be a spectral curve of type \eqref{eq:monoSC} equipped with the standard bidifferential. Then the topological recursion formula applied to $(\omega_{0,1},\omega_{0,2})$ yields a family of symmetric multidifferentials $(\omega_{0,n})_{n\geq 1}$ if and only if $\mc{C}$ is admissible in genus $0$.
\end{corollary}

Before we prove \Cref{prop:symm_cond_omega03,prop:symm_cond_omega04}, we first give some more general considerations for calculating the genus zero correlators.\par
Recalling \cref{TRsum,TRkernel}, we see that the recursion kernel can be split in factors coming from different irreducible components of the spectral curve. First, note that in $ K(z_1;z,Z)$, we always need $z_1$ and $z$ to lie in the same irreducible component, and then we have
\begin{equation*}
\int_\mu^z \omega_{0,2} (\plh, z_1 ) = \frac{z \dd z_1}{z_1 (z_1-z)}\,.
\end{equation*}
Because $ \omega_{0,2} $ in our current situation does not mix irreducible components, this shows immediately that to get symmetric correlators, we need
\begin{equation*}
\omega_{0,3}\big(\begin{smallmatrix} \mu & \nu & \lambda \\ z_1 & z_2 & z_3 \end{smallmatrix}\big) = 0
\end{equation*}
unless $ \mu = \nu = \lambda $: if one (say $\mu$) is different from the other two, we can use the recursion with respect to its variable, and get $z \in \tilde{C}_\mu$, so both terms in \eqref{thefirstg} would involve an $ \omega_{0,2}$ between two different irreducible components.\par
This vanishing can then be used to calculate $ \omega_{0,4}$ with arguments in exactly three different irreducible components. All terms involving $K^{(2)}$ would also involve a vanishing $ \omega_{0,3}$, while the same argument as for $ \omega_{0,3}$ above shows that the contribution of $K^{(3)}$ to such $\omega_{0,4}$ should vanish --- it also involves only $ \omega_{0,2}$. In fact, this same argument could be applied to $ \omega_{0,4} \big(\begin{smallmatrix} \mu & \nu & \nu & \nu \\ z_1 & z_2 & z_3 & z_4 \end{smallmatrix}\big)$ (we only need one irreducible component to be different from all others), but this turns out not to give a new constraint, so we omit it here.\par

\begin{remark}
	This argument can be used inductively to show that all $ \omega_{0,n}$ with exactly one argument on a given irreducible component must vanish. This is analogous to the proof of \cref{propdecouple}, but also uses that all arguments of the recursion kernel must couple to a different correlator, as we restrict to genus zero. However, in general the recursive computation of these correlators does require $K^{(m)}$ of order $m$ up to the degree of $\tilde{x}$, and therefore becomes quite complicated.
\end{remark}

For the remainder of this \namecref{GenusZeroConstraints}, we restrict to recursion kernels coupled to $ \omega_{0,2}$'s, which is sufficient for our calculations. We also assume $ z_1 \in \tilde{C}_\mu $. From the shape of the recursion kernel, we obtain several possible contributions (combining factors from the kernel and the correlators):
\begin{enumerate}[label=(\arabic*)]
	\item\label{AlwaysInKernel}There will always be one term
	\[
	\frac{z \dd z_1}{z_1 (z_1-z)} \frac{\dd z \dd z_m}{(z-z_m)^2}
	\]
	with $ z_m \in \tilde{C}_\mu $.
	\item\label{KernelSameComp} For any other $z_m \in \tilde{C}_\mu$, we get a contribution 
	\[
	\frac{\theta_\mu^{a_m}}{r_\mu t_\mu (\theta_\mu^{a_m s_\mu} -1) z^{s_\mu -1}} \frac{\dd z_m}{(\theta_\mu^{a_m} z - z_m)^2}\,,
	\]
	where $ \theta_\mu$ is a primitive $ r_\mu$th root of unity and we need to sum over all subsets $\{ a_m \} \subset [r_\mu - 1]$ of size determined by the number of $z_m \in \tilde{C}_\mu$.
	\item\label{KernelOtherComp} For any $ \nu \neq \mu$ and $ z_m \in \tilde{C}_\nu$, we get a contribution 
	\[
	\frac{\theta_\nu^{b_m} z^{r_\mu/r_\nu -1}}{r_\nu(t_\nu \theta_\nu^{b_m s_\nu} z^{r_\mu s_\nu/r_\nu - 1} - t_\mu z^{s_\mu -1})} \frac{\dd z_m}{(\theta_\nu^{b_m} z^{r_\mu /r_\nu} - z_m)^2}
	\]
	and we need to sum over all subsets $\{ b_m \} \subset [r_\nu ]$ of size determined by the number of $z_m \in \tilde{C}_\nu$.
\end{enumerate}
We need to take the series expansion of each of these near $ z= 0$. For \labelcref{AlwaysInKernel}, this is
\begin{equation*}
\frac{z \dd z_1}{z_1 (z_1-z)} \frac{\dd z \dd z_m}{(z-z_m)^2} = \dd z \dd z_1 \dd z_m \sum_{\ell,k \geq 1} k z_1^{-\ell-1} z_m^{-k-1} z^{\ell+k-1}\,.
\end{equation*}
The summations for cases \labelcref{KernelSameComp,KernelOtherComp} get quite complicated for general sizes of the subset, but for size one, they are computable. For case \labelcref{KernelSameComp}, we get
\begin{align*}
\sum_{a=1}^{r_\mu -1} \frac{\theta_\mu^a}{r_\mu t_\mu (\theta_\mu^{a s_\mu} -1) z^{s_\mu -1}} \frac{\dd z_m}{(\theta_\mu^{a} z - z_m)^2} 
&= \frac{\dd z_m}{r_\mu^2 t_\mu} \sum_{a = 1}^{r_\mu-1} \sum_{\ell=0}^{r_\mu -1} \ell \theta_\mu^{a s_\mu \ell} \sum_{k \geq 1} k \theta_\mu^{ak} z_m^{-k-1} z^{k-s_\mu}\\
&= \frac{\dd z_m}{r_\mu^2 t_\mu} \sum_{k \geq 1} \sum_{\ell=0}^{r_\mu -1} \ell\big(r_\mu \delta_{r_\mu | s_\mu \ell+ k} -1\big)  k z_m^{-k-1} z^{k-s_\mu}\\
&= \frac{\dd z_m}{r_\mu t_\mu} \sum_{k \geq 1}  \Big(\ell_\mu (k) - \frac{r_\mu -1}{2} \Big)  k z_m^{-k-1} z^{k-s_\mu}\,,
\end{align*}
where $ \ell_\mu (k)$ is the unique $\ell \in [0, r_\mu )$ such that $ r_\mu\,|\, s_\mu \ell + k $.\par
For case \labelcref{KernelOtherComp}, there are three different subcases, depending on the sign of $ \frac{r_\mu}{s_\mu} - \frac{r_\nu}{s_\nu} $:
\begin{enumerate}[label=(3\alph*)]
	\item\label{KernelOtherComp=} If $ \frac{r_\mu}{s_\mu} = \frac{r_\nu}{s_\nu}$, we get
	\begin{equation*}
	\sum_{b=1}^{r_\nu} \frac{\theta_\nu^{b_m} z^{r_\mu/r_\nu -1}}{r_\nu(t_\nu \theta_\nu^{b_m s_\nu} z^{r_\mu s_\nu/r_\nu - 1} - t_\mu z^{s_\mu -1})} \frac{\dd z_m}{(\theta_\nu^{b_m} z^{r_\mu /r_\nu} - z_m)^2} = \dd z_m \sum_{k \geq 1} \frac{t_\mu^{r_\mu - \ell_\mu(k)-1} t_\nu^{\ell_\mu(k)}}{t_\nu^{r_\mu} - t_\mu^{r_\mu}} \, k \,z_m^{-k-1} z^{k-s_\mu}\,.
	\end{equation*}
	\item\label{KernelOtherComp>} If $ \frac{r_\mu}{s_\mu} > \frac{r_\nu}{s_\nu}$, we get
	\begin{equation*}
	\sum_{b=1}^{r_\nu} \frac{\theta_\nu^{b_m} z^{r_\mu/r_\nu -1}}{r_\nu(t_\nu \theta_\nu^{b_m s_\nu} z^{r_\mu s_\nu/r_\nu - 1} - t_\mu z^{s_\mu -1})} \frac{\dd z_m}{(\theta_\nu^{b_m} z^{r_\mu /r_\nu} - z_m)^2} = - \frac{\dd z_m}{t_\mu} \!\!\!\!\!\!\sum_{\substack{k,\ell \geq 1\\r_\nu | s_\nu (\ell-1) + k}} \!\!\!\!\! \Big( \frac{t_\nu}{t_\mu} \Big)^{\ell-1} k \,z_m^{-k-1} z^{(s_\nu (\ell-1) + k)\frac{r_\mu}{r_\nu} -\ell s_\mu}\,.
	\end{equation*}
	\item\label{KernelOtherComp<} If $ \frac{r_\mu}{s_\mu} < \frac{r_\nu}{s_\nu}$, we get
	\begin{equation*}
	\sum_{b=1}^{r_\nu} \frac{\theta_\nu^{b_m} z^{r_\mu/r_\nu -1}}{r_\nu(t_\nu \theta_\nu^{b_m s_\nu} z^{r_\mu s_\nu/r_\nu - 1} - t_\mu z^{s_\mu -1})} \frac{\dd z_m}{(\theta_\nu^{b_m} z^{r_\mu /r_\nu} - z_m)^2} = \frac{\dd z_m}{t_\nu} \!\!\!\!\!\sum_{\substack{k,\ell \geq 1\\r_\nu | s_\nu \ell -k}} \!\!\! \Big( \frac{t_\mu}{t_\nu} \Big)^{\ell-1} k \,z_m^{-k-1} z^{(k-\ell s_\nu )\frac{r_\mu}{r_\nu} +(\ell-1) s_\mu}\,.
	\end{equation*}
\end{enumerate}

\begin{proof}[Proof of \cref{prop:symm_cond_omega03}]
	We have already argued that in the symmetric case $ \omega_{0,3}$ vanishes unless all arguments are on the same component. Let us therefore first calculate the value of $ \omega_{0,3}$ in exactly this case. If all arguments lie on $\tilde{C}_\mu$, we get a contribution from case \labelcref{AlwaysInKernel} and the one-argument version of case \labelcref{KernelSameComp}. Then taking the residue we obtain
	\begin{align*}
	\frac{\omega_{0,3} \begin{psmallmatrix} \mu & \mu & \mu\\ z_1 & z_2 & z_3 \end{psmallmatrix}}{\dd z_1 \dd z_2 \dd z_3} &= \Res_{z=0} \bigg( \dd z \sum_{\ell_2,k_2 \geq 1} k_2 z_1^{-\ell_2-1} z_2^{-k_2-1} z^{\ell_2+k_2-1} \frac{1}{r_\mu t_\mu} \sum_{k_3 \geq 1}  \Big(\ell_\mu (k_3) - \frac{r_\mu -1}{2} \Big)  k_3 z_3^{-k_3-1} z^{k_3-s_\mu} \\
	&\qquad \quad + \dd z \sum_{\ell_3,k_3 \geq 1} k_3 z_1^{-\ell_3-1} z_3^{-k_3-1} z^{\ell_3+k_3-1} \frac{1}{r_\mu t_\mu} \sum_{k_2 \geq 1}  \Big(\ell_\mu (k_2) - \frac{r_\mu -1}{2} \Big)  k_2 z_2^{-k_2-1} z^{k_2-s_\mu} \bigg) \\
	&= \frac{1}{r_\mu t_\mu} \Res_{z=0} \dd z  \sum_{\ell,k_2,k_3 \geq 1} k_2 k_3 \big(\ell_\mu (k_2) + \ell_\mu (k_3) - r_\mu +1 \big) z_1^{-\ell-1} z_2^{-k_2-1} z_3^{-k_3-1} z^{\ell+k_2+k_3-s_\mu-1} \\
	&= \frac{1}{r_\mu t_\mu}  \sum_{k_1,k_2,k_3 \geq 1} k_2 k_3 \big(\ell_\mu (k_2) + \ell_\mu (k_3) - r_\mu +1 \big) z_1^{-k_1-1} z_2^{-k_2-1} z_3^{-k_3-1} \, \delta_{k_1+k_2+k_3,s_\mu} \,.
	\end{align*}
	For $r_\mu=1$ this expression vanishes and hence is symmetric. In case  $r_\mu > 1$ we know from \cite[Proposition B.2]{BBCCN18} that for $s_\mu\in[r_\mu+1]$, this is symmetric if and only if $r_\mu = \pm 1 \,\,{\rm mod}\,\, s_\mu$. If however, $r_\mu > 1$ and $s_\mu>r_\mu+1$ it is easy to see that the correlators can never be symmetric. Let us assume $\omega_{0,3}$ is symmetric. Then for all $k_1,k_2,k_3>0$ satisfying $k_1+k_2+k_3=s_\mu$ we must have
	\begin{equation*}
	k_2 k_3 \big(\ell_\mu (k_2) + \ell_\mu (k_3) - r_\mu +1 \big) = k_1 k_3 \big(\ell_\mu (k_1) + \ell_\mu (k_3) - r_\mu +1 \big)\,.
	\end{equation*}
	Plugging $k_1=1$, $k_2=r_\mu$, and $k_3=s_\mu-r_\mu-1$ into the above equation the right-hand side vanishes and we obtain
	\begin{equation*}
	r_\mu \big(\ell_\mu (s_\mu-r_\mu-1) - r_\mu +1\big) = 0\,.
	\end{equation*}
	From this we deduce that $\ell_\mu (s_\mu-r_\mu-1) = r_\mu-1$. This means that there must exist an $m$ such that $m r_\mu = s_\mu-r_\mu-1 + s_\mu(r_\mu-1)$. This in turn implies that $r_\mu=1$ which contradicts our starting assumption. We therefore conclude that the symmetry condition for $\omega_{0,3} \begin{psmallmatrix} \mu & \mu & \mu\\ z_1 & z_2 & z_3 \end{psmallmatrix}$ is exactly captured by \labelcref{item:r_eq_pm1mods}.
	
	To show \cref{03explicit}, recall that $ \ell_\mu (k) $ is the unique $ \ell \in [0, r_\mu)$ such that $ r_\mu \mid s_\mu \ell + k$. In other words, there is  an $m$ such that $m (r'_\mu s_\mu + \varepsilon_\mu ) = s_\mu \ell_\mu (k) + k$. Viewing this formula modulo $ s_\mu$ shows that for $k < s_\mu$
	\begin{equation*}
	m = \begin{cases} k  & \text{if } \varepsilon_\mu = 1 \\  s_\mu -k &\text{if } \varepsilon_\mu = s_\mu -1 \end{cases}\,.
	\end{equation*}
	From this, it follows easily that $ (\ell_\mu (k_2 ) + \ell_\mu (k_3) -r_\mu +1) = c_\mu $ as $k_2, k_3 < s_\mu$.
	
	Now let us turn our attention to the case where the arguments $\omega_{0,3}$ lie on two different components. We will see that in this case the symmetry is controlled by condition \labelcref{item:eq_fract_constr}. Let us first show that this condition is indeed necessary for the symmetry of $\omega_{0,3}$. So suppose there are $\mu$ and $\nu$ violating \labelcref{item:eq_fract_constr}, i.e.\ $s_\mu,s_\nu>2$, $\big\lfloor\frac{r_\mu}{s_\mu}\big\rfloor =  \big\lfloor\frac{r_\nu}{s_\nu}\big\rfloor$ and either $r_{\mu} = 1 \,\,{\rm mod}\,\, s_{\mu}$ and $r_{\nu} = 1 \,\,{\rm mod}\,\, s_{\nu}$ or $r_{\mu} = -1 \,\,{\rm mod}\,\, s_{\mu}$ and $r_{\nu} = -1 \,\,{\rm mod}\,\, s_{\nu}$. After relabelling we may also assume that $\frac{r_\mu}{s_\mu} \geq \frac{r_\nu}{s_\nu}$. In the following we will show by explicit computation that in this case $\omega_{0,3}\begin{psmallmatrix} \mu & \mu & \nu \\ z_1 & z_2 & z_3 \end{psmallmatrix} \neq 0$. This however means that $\omega_{0,3}$ cannot be symmetric since $ \omega_{0,3}\begin{psmallmatrix} \nu & \mu & \mu \\ z_3 & z_2 & z_1 \end{psmallmatrix} $ is always vanishing.
	
	In case $ \frac{r_\mu}{s_\mu} = \frac{r_\nu}{s_\nu}$, $\omega_{0,3}\begin{psmallmatrix} \mu & \mu & \nu \\ z_1 & z_2 & z_3 \end{psmallmatrix}$ gets a contribution from case \labelcref{AlwaysInKernel}, and the one-argument version of case \labelcref{KernelOtherComp=}. Then taking the residue gives
	\begin{align*}
		\frac{\omega_{0,3} \begin{psmallmatrix}\mu & \mu & \nu \\ z_1 & z_2 & z_3 \end{psmallmatrix}}{\dd z_1 \dd z_2 \dd z_3} &=
		\Res_{z =0} \dd z  \sum_{k_1,k_2,k_3 \geq 1} k_2 k_3 z_1^{-k_1-1} z_2^{-k_2-1} z^{k_1+k_2-1} \, \frac{t_\mu^{r_\mu - \ell_\mu(k)-1} t_\nu^{\ell_\mu(k)}}{t_\nu^{r_\mu} - t_\mu^{r_\mu}} \, z_m^{-k_3-1} z^{k_3-s_\mu}\\
		&=  \sum_{k_1,k_2,k_3 \geq 1} \frac{t_\mu^{r_\mu - \ell_\mu(k_3)-1} t_\nu^{\ell_\mu(k_3)}}{t_\nu^{r_\mu} - t_\mu^{r_\mu}} \, k_2 k_3 \,\delta_{k_1+k_2+k_3,s_\mu} \, z_1^{-k_1-1} z_2^{-k_2-1} z_3^{-k_3-1}\,.
	\end{align*}
	This is clearly non-vanishing since we assumed that $s_\mu,s_\nu>2$.
	
	If now $\frac{r_\mu}{s_\mu} > \frac{r_\nu}{s_\nu}$ the correlator gets a contribution from \labelcref{AlwaysInKernel} and \labelcref{KernelOtherComp>} leading to
	\begin{align}
		\frac{\omega_{0,3} \begin{psmallmatrix}\mu & \mu & \nu \\ z_1 & z_2 & z_3 \end{psmallmatrix}}{\dd z_1 \dd z_2 \dd z_3} &= -\frac{1}{t_\mu} \Res_{z =0} \dd z  \!\!\! \sum_{\substack{k_1,k_2,k_3,\ell \geq 1\\r_\nu | s_\nu (\ell-1) + k_3}} \!\! k_2 k_3 \,z_1^{-k_1-1} z_2^{-k_2-1} z^{k_1+k_2-1} \Big( \frac{t_\nu}{t_\mu} \Big)^{\ell-1}  z_3^{-k_3-1} z^{(s_\nu (\ell-1) + k_3)\frac{r_\mu}{r_\nu} -\ell s_\mu}\nonumber\\
		&= -\frac{1}{t_\mu}  \!\!\! \sum_{\substack{k_1,k_2,k_3,\ell \geq 1\\r_\nu | s_\nu (\ell-1) + k_3}} \!\! \Big( \frac{t_\nu}{t_\mu} \Big)^{\ell-1} k_2 k_3 \,z_1^{-k_1-1} z_2^{-k_2-1} z_3^{-k_3-1} \,\delta_{k_1+k_2+(s_\nu (\ell-1) + k_3)\frac{r_\mu}{r_\nu}-\ell s_\mu,\, 0}\nonumber\\
		&= -\frac{1}{t_\mu}  \! \sum_{k_1,k_2,k_3 \geq 1} \sum_{\ell'  \geq 0} \Big( \frac{t_\nu}{t_\mu} \Big)^{\ell_\nu(k_3)+r_\nu \ell'} k_2 k_3\, z_1^{-k_1-1} z_2^{-k_2-1} z_3^{-k_3-1} \nonumber\\
		&\hspace{5cm} \cdot \delta_{k_1 + k_2 - a_{\mu,\nu}(k_3) + \ell'(r_\mu s_\nu - r_\nu s_\mu), 0}\,. \label{eq:omega03mumunu}
	\end{align}
	where we set
	\begin{equation}
		\label{eq:amunukdef}
		a_{\mu,\nu}(k)\coloneqq -(s_\nu \ell_\nu (k) + k)\frac{r_\mu}{r_\nu} + (\ell_\nu (k) +1) s_\mu\,.
	\end{equation}
	The key observation to make now is that $a_{\mu,\nu}(1)>1$ under our assumptions on $\mu$ and $\nu$. Indeed, in case $r_{\mu} = 1 \,\,{\rm mod}\,\, s_{\mu}$ and $r_{\nu} = 1 \,\,{\rm mod}\,\, s_{\nu}$ we find that
	\begin{equation}
		\label{eq:amunuk>}
		a_{\mu,\nu}(1) = - r_\mu + (r_\mu - 1 + s_\mu) = s_\mu-1 > 1
	\end{equation}
	where we used that $\ell_\nu(1)=\big\lfloor \frac{r_\nu}{s_\nu}\big\rfloor=\big\lfloor \frac{r_\mu}{s_\mu}\big\rfloor$. Similarly, if $r_{\mu} = -1 \,\,{\rm mod}\,\, s_{\mu}$ and $r_{\nu} = -1 \,\,{\rm mod}\,\, s_{\nu}$ an analogous calculation gives $a_{\mu,\nu}(1)=s_\nu-1>1$. Hence, we find that the residue
	\begin{align*}
		\Res_{z_1=0}\,\Res_{z_2=0}\,\Res_{z_3=0} ~ z_1^{a_{\mu,\nu}(1)-1}z_2 z_3 \, \omega_{0,3} \begin{psmallmatrix}\mu & \mu & \nu \\ z_1 & z_2 & z_3 \end{psmallmatrix} \,& = -\frac{1}{t_\mu}  \! \sum_{\ell'  \geq 0} \Big( \frac{t_\nu}{t_\mu} \Big)^{\ell_\nu(1)+r_\nu \ell'} ~ \delta_{ \ell'(r_\mu s_\nu - r_\nu s_\mu), 0}\\
		& = -\frac{1}{t_\mu}\Big( \frac{t_\nu}{t_\mu} \Big)^{\ell_\nu(1)}
	\end{align*}
	is non-vanishing, telling us that $\omega_{0,3} \begin{psmallmatrix}\mu & \mu & \nu \\ z_1 & z_2 & z_3 \end{psmallmatrix}\neq 0$. This shows that condition \labelcref{item:eq_fract_constr} is indeed necessary to ensure the symmetry of $\omega_{0,3}$.
	
	Let us now finally argue why \labelcref{item:r_eq_pm1mods} and \labelcref{item:eq_fract_constr} are also sufficient to ensure symmetry. Indeed, our spectral curve is admissible in genus $0$ if additionally the condition \labelcref{item:threefracteq} of \namecref{prop:symm_cond_omega04} is satisfied. In this case $\omega_{0,3}$ computed by the topological recursion must be symmetric by \Cref{mainth2gequal0}. Note now that condition \labelcref{item:threefracteq} is trivially satisfied in case $d\leq 2$. In case $d>2$ since $\omega_{0,3}\begin{psmallmatrix}\mu & \mu & \mu \\ z_1 & z_2 & z_3 \end{psmallmatrix}$ can as well be computed after restricting $x$, $y$, and $\omega_{0,2}$ to the component $\tilde{C}_\mu$ --- a one-component spectral curve for which \labelcref{item:threefracteq} is trivially satisfied --- the correlator must be symmetric under permutation of the arguments. Similarly, for $\mu\neq\nu$ we can compute $\omega_{0,3}\begin{psmallmatrix}\mu & \nu & \mu \\ z_1 & z_2 & z_3 \end{psmallmatrix}$, $\omega_{0,3}\begin{psmallmatrix}\mu & \mu & \nu \\ z_1 & z_2 & z_3 \end{psmallmatrix}$, etc.\ after restricting the input data to the two-component curve $\tilde{C}_\mu \sqcup \tilde{C}_\nu$ for which \labelcref{item:threefracteq} holds as well. Finally, if all arguments lie on different components $\mu$, $\nu$, $\lambda$ we always have $\omega_{0,3}\begin{psmallmatrix}\mu & \nu & \lambda \\ z_1 & z_2 & z_3 \end{psmallmatrix}=0$ which of course is symmetric. Thus, although we are not imposing condition \labelcref{item:threefracteq} of \Cref{prop:symm_cond_omega04} the correlator $\omega_{0,3}$ is already symmetric given that \labelcref{item:eq_fract_constr} and \labelcref{item:eq_fract_constr} hold.
	\end{proof}

\begin{proof}[Proof of \cref{prop:symm_cond_omega04}]
	Suppose $\omega_{0,4}$ is symmetric and there are distinct $\mu,\nu,\lambda$ with $\big\lfloor\frac{r_{\mu}}{s_{\mu}}\big\rfloor = \big\lfloor\frac{r_{\nu}}{s_{\nu}}\big\rfloor = \big\lfloor\frac{r_{\lambda}}{s_{\lambda}}\big\rfloor$ and $s_\mu,s_\nu,s_\lambda>1$. After relabelling we may assume that $\frac{r_{\mu}}{s_{\mu}}\geq \frac{r_{\nu}}{s_{\nu}}\geq \frac{r_{\lambda}}{s_{\lambda}}$. As argued before
	\begin{equation*}
		\omega_{0,4}\begin{psmallmatrix}\mu & \nu & \mu & \lambda \\ z_1 & z_3 & z_2 & z_4 \end{psmallmatrix} = 0
	\end{equation*}
	because $\mu\neq \nu$. So it suffices to show that $\omega_{0,4}\begin{psmallmatrix}\mu & \mu & \nu & \lambda \\ z_1 & z_2 & z_3 & z_4 \end{psmallmatrix}\neq 0$ in order to get a contradiction.
	
	In case $\frac{r_{\mu}}{s_{\mu}} > \frac{r_{\nu}}{s_{\nu}} \geq \frac{r_{\lambda}}{s_{\lambda}}$ the correlator $\omega_{0,4}\begin{psmallmatrix}\mu & \mu & \nu & \lambda \\ z_1 & z_2 & z_3 & z_4 \end{psmallmatrix}$ gets a contribution from \labelcref{AlwaysInKernel} and twice from \labelcref{KernelOtherComp>} giving us
	\begin{align*}
		\frac{\omega_{0,4}\begin{psmallmatrix}\mu & \mu & \nu & \lambda \\ z_1 & z_2 & z_3 & z_4 \end{psmallmatrix}}{\dd z_1 \dd z_2 \dd z_3 \dd z_4}\, & = -\frac{1}{t_\mu^2} \Res_{z =0} \dd z  \!\!\! \sum_{\substack{k_1,k_2,k_3,k_4,\ell_3,\ell_4 \geq 1\\r_\nu | s_\nu (\ell_3-1) + k_3\\r_\lambda | s_\lambda (\ell_4-1) + k_4}} \!\!\!\! k_2 k_3 k_4 \,z_1^{-k_1-1} z_2^{-k_2-1} z^{k_1+k_2-1} \Big( \frac{t_\nu}{t_\mu} \Big)^{\ell_3-1}  z_3^{-k_3-1} z^{(s_\nu (\ell_3-1) + k_3)\frac{r_\mu}{r_\nu} -\ell_3 s_\mu} \\[-1em]
		&\hspace{11.5em}\cdot \Big( \frac{t_\lambda}{t_\mu} \Big)^{\ell_4-1} z_4^{-k_4-1} z^{(s_\lambda (\ell_4-1) + k_4)\frac{r_\mu}{r_\lambda} -\ell_4 s_\mu}\\
		& = -\frac{1}{t_\mu^2}  \sum_{\substack{k_1,k_2,k_3,k_4 \geq 1 \\ \ell_3',\ell_4'\geq 0}} \!\!\!\! k_2 k_3 k_4 \Big( \frac{t_\nu}{t_\mu} \Big)^{\ell_\nu(k_3)+\ell_3' r_\nu} \Big( \frac{t_\lambda}{t_\mu} \Big)^{\ell_\lambda(k_4)+\ell_4' r_\lambda} \left(\prod_{m=1}^4 z_m^{-k_m-1}\right) \\
		&\hspace{10em} \cdot \delta_{k_1 + k_2 - a_{\mu,\nu}(k_3) +\ell_3'(r_\mu s_\nu - r_\nu s_\mu) - a_{\mu,\lambda}(k_4) +\ell_4'(r_\mu s_\lambda - r_\lambda s_\mu) ,\, 0}\,.
	\end{align*}
	where we defined $a_{\mu,\nu}(k_3)$ and $a_{\mu,\lambda}(k_4)$ as in \eqref{eq:amunukdef}. Again, the key observation to make is that $a_{\mu,\nu}(1)>0$. We have already shown this in \eqref{eq:amunuk>} for two cases. In the other two remaining cases a straight forward calculation shows that $a_{\mu,\nu}(1)=1$ in case $r_\mu = -1 \,\,{\rm mod}\,\, s_\mu$ and $r_\nu = 1 \,\,{\rm mod}\,\, s_\nu$ and $a_{\mu,\nu}(1)=(s_\mu-1)(s_\nu-1)$ if $r_\mu = 1 \,\,{\rm mod}\,\, s_\mu$ and $r_\nu = -1 \,\,{\rm mod}\,\, s_\nu$ confirming that indeed $a_{\mu,\nu}(1)$ is a strictly positive integer. Thus, the residue
	\begin{equation*}
		\begin{split}
			&\Res_{z_1=0} \cdots \Res_{z_4=0} ~~ z_1^{a_{\mu,\nu}(1)} z_2^{a_{\mu,\lambda}(1)} z_3 z_4\, \omega_{0,4}\begin{psmallmatrix}\mu & \mu & \nu & \lambda \\ z_1 & z_2 & z_3 & z_4 \end{psmallmatrix}\\
			&\qquad= -\frac{a_{\mu,\lambda}(1)}{t_\mu^2} \!\! \sum_{\ell_3',\ell_4'\geq 0} \!\!  \Big( \frac{t_\nu}{t_\mu} \Big)^{\ell_\nu(1)+\ell_3' r_\nu} \Big( \frac{t_\lambda}{t_\mu} \Big)^{\ell_\lambda(1)+\ell_4' r_\lambda} \delta_{\ell_3'(r_\mu s_\nu - r_\nu s_\mu) + \ell_4'(r_\mu s_\lambda - r_\lambda s_\mu) ,\, 0}\\
			&\qquad= -\frac{a_{\mu,\lambda}(1)}{t_\mu^2} \Big( \frac{t_\nu}{t_\mu} \Big)^{\ell_\nu(1)} \Big( \frac{t_\lambda}{t_\mu} \Big)^{\ell_\lambda(1)}
		\end{split}
	\end{equation*}
	is non-vanishing and hence $\omega_{0,4}\begin{psmallmatrix}\mu & \mu & \nu & \lambda \\ z_1 & z_2 & z_3 & z_4 \end{psmallmatrix}\neq 0$ giving us the desired contradiction.\\
	
	Now if $\frac{r_{\mu}}{s_{\mu}} = \frac{r_{\nu}}{s_{\nu}} > \frac{r_{\lambda}}{s_{\lambda}}$ we can proceed similarly. This time the correlator receives contributions from \labelcref{AlwaysInKernel}, \labelcref{KernelOtherComp=}, and \labelcref{KernelOtherComp>} leading to
	\begin{align*}
		\frac{\omega_{0,4}\begin{psmallmatrix}\mu & \mu & \nu & \lambda \\ z_1 & z_2 & z_3 & z_4 \end{psmallmatrix}}{\dd z_1 \dd z_2 \dd z_3 \dd z_4}\, & = -\frac{1}{t_\mu} \!\! \sum_{\substack{k_1,k_2,k_3,k_4 \geq 1 \\ \ell_4'\geq 0}} \!\!\!\! k_2 k_3 k_4 \frac{t_\mu^{r_\mu-\ell_\mu(k_3)-1} t_\nu^{\ell_\mu(k_3)}}{t_\mu^{r_\mu}-t_\nu^{r_\mu}} \Big( \frac{t_\lambda}{t_\mu} \Big)^{\ell_\lambda(k_4)+\ell_4' r_\lambda} \left(\prod_{m=1}^4 z_m^{-k_m-1}\right) \\
		&\hspace{10em} \cdot \delta_{k_1 + k_2 + k_3 - s_\mu - a_{\mu,\lambda}(k_4) +\ell_4'(r_\mu s_\lambda - r_\lambda s_\mu) ,\, 0}\,.
	\end{align*}
	Therefore, also in this case we find that $\omega_{0,4}\begin{psmallmatrix}\mu & \mu & \nu & \lambda \\ z_1 & z_2 & z_3 & z_4 \end{psmallmatrix}\neq 0$ as
	\begin{equation*}
		\Res_{z_1=0} \cdots \Res_{z_4=0} ~~ z_1^{s_\mu-1} z_2^{a_{\mu,\lambda}(1)} z_3 z_4\, \omega_{0,4}\begin{psmallmatrix}\mu & \mu & \nu & \lambda \\ z_1 & z_2 & z_3 & z_4 \end{psmallmatrix} = \frac{a_{\mu,\lambda}(1)}{t_\mu} \frac{t_\mu^{r_\mu-\ell_\mu(1)-1} t_\nu^{\ell_\mu(1)}}{t_\nu^{r_\mu}-t_\mu^{r_\mu}} \Big( \frac{t_\lambda}{t_\mu} \Big)^{\ell_\lambda(1)}\neq 0\,.
	\end{equation*}

	Similarly, if $\frac{r_{\mu}}{s_{\mu}} = \frac{r_{\nu}}{s_{\nu}} = \frac{r_{\lambda}}{s_{\lambda}}$ an analogous calculation gives
	\begin{equation*}
		\Res_{z_1=0} \cdots \Res_{z_4=0} ~~ z_1^{s_\mu-1} z_2^{s_\mu-1} z_3 z_4\, \omega_{0,4}\begin{psmallmatrix}\mu & \mu & \nu & \lambda \\ z_1 & z_2 & z_3 & z_4 \end{psmallmatrix} \neq 0\,.
	\end{equation*}
\end{proof}

\medskip

\subsubsection{The Bouchard--Eynard formula for \texorpdfstring{$(0,3)$}{(0,3)} does not hold in general}

\medskip

Let us compare $\omega_{0,3}$ with
\begin{equation*}
\begin{split}
\check{\omega}_{0,3}\big(\begin{smallmatrix} \mu_1 & \mu_2 & \mu_3 \\ z_1 & z_2 & z_3 \end{smallmatrix}\big) & = \sum_{\mu} \mathop{{\rm Res}}_{z = \mu}  \frac{\omega_{0,2}(z,z_1)\omega_{0,2}(z,z_2)\omega_{0,2}(z,z_3)}{\dd x(z)\dd y(z)} \\
& = -\sum_{\mu}  \delta_{\mu_1,\mu_2,\mu_3,\mu}\sum_{k_1,k_2,k_3 > 0} \frac{k_1k_2k_3\,\dd z_1 \dd z_2 \dd z_3}{z_1^{k_1 + 1}z_2^{k_2 + 1}z_3^{k_3 + 1}}\,\frac{\delta_{k_1 + k_2 + k_3,s_{\mu}}}{t_{\mu}r_{\mu}(s_{\mu} - r_{\mu})}\,.
\end{split}
\end{equation*}
According to \cite[Proposition 11]{BoEy13}, $ \check{\omega}_{0,3} = \omega_{0,3}$ for regularly admissible spectral curves. More generally, if the spectral curve is admissible, comparing with \cref{prop:symm_cond_omega03}, we see that $ \omega_{0,3} = \check{\omega}_{0,3}$ if and only if $ c_\mu = \frac{1}{s_\mu -r_\mu}$ or $ s_\mu \leq 2$ for each $ \mu \in \tilde{\mf{a}}$. The condition $ c_\mu = \frac{1}{s_\mu -r_\mu}$ implies that $ s_\mu = r_\mu +1$, which by property \labelcref{item:eq_fract_constr} of \cref{prop:symm_cond_omega03} can only hold for one $\mu$. Comparing with \cref{thm:W_gl_Airy_arbitrary_autom}, we see that e.g. $ (r_1,s_1,r_2,s_2)=(5,3,3,2)$ gives an Airy structure for which $ \omega_{0,3} \neq \check{\omega}_{0,3}$. We conclude that \cite[Proposition 11]{BoEy13} does not always hold in our general setup.

\medskip

\subsubsection{\texorpdfstring{$(g,n) = (\tfrac{1}{2},2)$}{(g,n)=(1/2,2)}}

\label{Sec63}
\medskip
In \Cref{SecTR} we already asked the question whether topological recursion applied to any curve admissible in genus $0$ --- which in this section we will continue to assume to be of the form \eqref{eq:monoSC} --- together with a choice of $\omega_{0,2}$ and $\omega_{\frac{1}{2},1}$ will yield symmetric $\omega_{g,n}$ not just for $g=0$ but also for $g>0$. By studying the symmetry of $\omega_{\frac{1}{2},2}$ we will see that this is not always possible unless the parameters $Q_\mu$ in $\omega_{\frac{1}{2},1}\big(\begin{smallmatrix} \mu \\ z \end{smallmatrix}\big) = \frac{Q_{\mu}\dd z}{z}$ satisfy certain compatibility conditions. Let us first look at two examples.

\begin{example}
	Let us consider the two-component curve of type $(r_1,s_1)=(r_2,s_2)=(3,2)$ and choose $t_1=-t_2=1$. Then the topological recursion gives us
	\begin{equation*}
		\omega_{\frac{1}{2},2}\big(\begin{smallmatrix} 1 & 1 \\ z_1 & z_2 \end{smallmatrix}\big) = -\bigg(\frac{Q_1}{3}+ \frac{Q_2}{2}\bigg) \frac{\dd z_1 \dd z_2}{z_1^2 z_2^2} \,,\qquad \omega_{\frac{1}{2},2}\big(\begin{smallmatrix} 2 & 2 \\ z_1 & z_2 \end{smallmatrix}\big) = \bigg(\frac{Q_1}{2}+ \frac{Q_2}{3}\bigg) \frac{\dd z_1 \dd z_2}{z_1^2 z_2^2}\,,
	\end{equation*}
	which are always symmetric under the exchange of arguments, while for arguments lying on different components we find that
	\begin{equation*}
		\omega_{\frac{1}{2},2}\big(\begin{smallmatrix} 1 & 2 \\ z_1 & z_2 \end{smallmatrix}\big) = \frac{Q_1}{2} \frac{\dd z_1 \dd z_2}{z_1^2 z_2^2} \,,\qquad \omega_{\frac{1}{2},2}\big(\begin{smallmatrix} 2 & 1 \\ z_1 & z_2 \end{smallmatrix}\big) = -\frac{Q_2}{2} \frac{\dd z_1 \dd z_2}{z_1^2 z_2^2}
	\end{equation*}
	which means that $\omega_{\frac{1}{2},2}$ can only be symmetric if $Q_1 + Q_2 = 0$.
\end{example}

\begin{example}
	Let us look at the two-component curve with $(r_1,s_1)=(8,3)$ and $(r_2,s_2)=(5,3)$. Furthermore, we choose $t_1=t_2=1$. In this case $\omega_{\frac{1}{2},2}$ vanishes for arguments lying on different components and
	\begin{equation*}
	\begin{split}
		\omega_{\frac{1}{2},2}\big(\begin{smallmatrix} 1 & 1 \\ z_1 & z_2 \end{smallmatrix}\big) & = -\frac{5 Q_1+8 Q_2}{4} \frac{\dd z_1 \dd z_2}{z_1^2 z_2^3} - \frac{Q_1+4 Q_2}{4}\frac{\dd z_1 \dd z_2}{z_1^3 z_2^2}\,,\\
		\omega_{\frac{1}{2},2}\big(\begin{smallmatrix} 2 & 2 \\ z_1 & z_2 \end{smallmatrix}\big)  & = - \frac{6 Q_2}{5} \frac{\dd z_1 \dd z_2}{z_1^2 z_2^3} - \frac{Q_2}{5} \frac{\dd z_1 \dd z_2}{z_1^3 z_2^2}\,.
		\end{split}
	\end{equation*}
	For the above to be symmetric under the exchange of $z_1$ and $z_2$ we need $Q_1=Q_2=0$ which means in this case there is no ambiguity in choosing the $Q_\mu$.
\end{example}

These two examples are special instances of the following classification.

\begin{proposition}
	\label{prop:symm_cond_omega1half2}
	Assume $(r_\mu,s_\mu)_{\mu\in\tilde{\mf{a}}}$ are integers satisfying condition \labelcref{item:r_eq_pm1mods,item:eq_fract_constr} from \cref{prop:symm_cond_omega03}. Then $\omega_{\frac{1}{2},2}$ is symmetric if and only if 
	\begin{enumerate}
		\setcounter{enumi}{3}
		\item\label{item:1half2cond1} If $s_\mu>2$ and $r_\mu = 1 \,\,{\rm mod}\,\, s_\mu$ for $\mu\in \tilde{\mf{a}}$ then
		\begin{equation*}
		\sum_{\substack{\nu\neq \mu \\ \frac{r_\mu}{s_\mu} > \frac{r_\nu}{s_\nu}}} Q_\nu = 0 \qquad\text{ and }\qquad \forall \ell\in [s_\mu-3],\quad \sum_{\substack{\nu\neq \mu \\ s_\nu = 1,\, \big\lfloor \frac{r_\mu}{s_\mu}\big\rfloor = r_\nu}} Q_\nu t_\nu^{r_\nu \ell} = 0\,.
		\end{equation*}
		
		\item\label{item:1half2cond2} If $s_\mu>2$ and $r_\mu = -1 \,\,{\rm mod}\,\, s_\mu$ for $\mu\in \tilde{\mf{a}}$ then
		\begin{equation*}
		Q_\mu + \sum_{\substack{\nu\neq \mu \\ \frac{r_\mu}{s_\mu} > \frac{r_\nu}{s_\nu}}} Q_\nu = 0 \qquad\text{ and }\qquad \forall \ell\in [s_\mu-3],\quad \sum_{\substack{\nu\neq \mu \\ s_\nu = 1,\, \big\lceil \frac{r_\mu}{s_\mu}\big\rceil = r_\nu}} Q_\nu t_\nu^{-r_\nu \ell} = 0\,.
		\end{equation*}
		
		\item\label{item:1half2cond4} For all $\mu\neq\nu$ with $\bigl\lfloor\frac{r_\mu}{s_\mu}\bigr\rfloor = \bigl\lfloor\frac{r_\nu}{s_\nu}\bigr\rfloor$,
		\begin{equation*}
		r_\mu = -1 \,\,{\rm mod}\,\, s_\mu \qquad \text{ and } \qquad r_\nu = 1 \,\,{\rm mod}\,\, s_\nu\,,
		\end{equation*}
		and $s_\mu,s_\nu > 1$ we have
		\begin{equation*}
		Q_\mu = -Q_\nu\,.
		\end{equation*}
	\end{enumerate}
	If these conditions are satisfied, we have if $r_\mu = r_\mu' s_\mu +1$
	\begin{equation}
	\label{eq:F1half2_mumu_sym1}
	\begin{split}
	&\omega_{\frac{1}{2},2}\big(\begin{smallmatrix} \mu & \mu \\ z_1 & z_2 \end{smallmatrix}\big)\\
	&\quad = \sum_{k_1,k_2 > 0} \frac{k_1k_2\,\dd z_1 \dd z_2}{z_1^{k_1 + 1}z_2^{k_2 + 1}} ~\bigg( -\frac{Q_\mu r_\mu'}{t_\mu r_\mu} \delta_{k_1 + k_2 ,s_{\mu}} + \sum_{\substack{\nu \neq \mu \\ \frac{r_{\mu}}{s_{\mu}} = \frac{r_{\nu}}{s_{\nu}}}} \frac{ t_\mu^{r_\mu-1} }{t_{\nu}^{r_{\mu}}-t_{\mu}^{r_{\mu}}} \, Q_{\nu} \,\delta_{k_1 + k_2,s_{\mu}} \, \delta_{s_{\mu},2}\\
	&\qquad\qquad - \sum_{\substack{\nu\neq \mu \\ s_\nu = 1,\, \big\lfloor \frac{r_\mu}{s_\mu}\big\rfloor = r_\nu}} \frac{Q_\nu}{t_\mu} \left(\frac{t_\nu}{t_\mu}\right)^{r_\nu(s_\mu-2)} \delta_{k_1 + k_2 ,2} \, \delta_{s_{\mu}>2} + \sum_{\substack{\nu \neq \mu:\, s_\nu=2, \\ \frac{r_{\mu}}{s_{\mu}} < \frac{r_{\nu}}{s_{\nu}},\, \left\lfloor\frac{r_\mu}{s_\mu}\right\rfloor = \left\lfloor\frac{r_\nu}{s_\nu}\right\rfloor}} \frac{Q_\nu}{t_\mu} \left(\frac{t_{\mu}}{t_{\nu}}\right)^{r_\nu}  \delta_{k_1 + k_2 ,2} \, \delta_{s_{\mu}>2}\bigg)\,,
	\end{split}
	\end{equation}
	and if $r_\mu = r_\mu' s_\mu + s_\mu  - 1$ and $s_\mu>2$, we have
	\begin{equation}
	\label{eq:F1half2_mumu_sym2}
	\begin{split}
	&\omega_{\frac{1}{2},2}\big(\begin{smallmatrix} \mu & \mu \\ z_1 & z_2 \end{smallmatrix}\big) \\
	&\quad = \sum_{k_1,k_2 > 0} \frac{k_1k_2\,\dd z_1 \dd z_2}{z_1^{k_1 + 1}z_2^{k_2 + 1}} ~\bigg( \frac{Q_\mu (r_\mu'+1)}{t_\mu r_\mu} \delta_{k_1 + k_2 ,s_{\mu}} + \sum_{\substack{\nu\neq \mu \\ s_\nu = 1,\, \big\lceil  \frac{r_\mu}{s_\mu}\big\rceil = r_\nu}} \frac{Q_\nu}{t_\mu} \left(\frac{t_\mu}{t_\nu}\right)^{r_\nu(s_\mu-2)} \delta_{k_1 + k_2 ,2} \, \delta_{s_{\mu}>2} \\
	& \qquad\qquad\qquad\qquad\qquad\qquad\qquad\qquad\qquad- \sum_{\substack{\nu\neq \mu:\, s_\nu=2, \\ \frac{r_\mu}{s_\mu} > \frac{r_\nu}{s_\nu},\, \left\lfloor\frac{r_\mu}{s_\mu}\right\rfloor = \left\lfloor\frac{r_\nu}{s_\nu}\right\rfloor }} \frac{Q_\nu}{t_\mu} \left(\frac{t_\nu}{t_\mu}\right)^{r_\nu} \delta_{k_1 + k_2 ,2} \, \delta_{s_{\mu}>2}\bigg)\,.
	\end{split}
	\end{equation}
	Moreover, for $\mu \neq \nu$ we have 
	\begin{equation}
	\label{eq:omega1half2offdiag}
	\omega_{\frac{1}{2},2}\big(\begin{smallmatrix} \mu & \nu \\ z_1 & z_2 \end{smallmatrix}\big) = \begin{cases}
	-\frac{\dd z_1 \dd z_2}{z_1^{2}z_2^{2}} \, \frac{Q_\mu}{t_\mu}\left(\frac{t_\nu}{t_\mu}\right)^{r_{\mu}'} & \text{if $(r_\mu, s_\mu)$, $(r_\nu, s_\nu)$ satisfy \labelcref{item:1half2cond4} and $\frac{r_\mu}{s_\mu}>\frac{r_\nu}{s_\nu}$} \\[1em]
	- \frac{\dd z_1 \dd z_2}{z_1^{2}z_2^{2}} \, \frac{Q_\mu t_\mu^{r_\mu'} t_\nu^{r_\nu'} }{t_\mu^{r_\mu} - t_\nu^{r_\mu}} & \text{if } r_\mu =r_\nu \text{ and } s_\mu =s_\nu =2 \\[1em]
	0 & \text{otherwise}
	\end{cases}\,.
	\end{equation}
\end{proposition}

Before proving \cref{prop:symm_cond_omega1half2}, we need the following technical fact.
\begin{lemma}
	\label{prop:symm_cond_omega1half2_tech} Let $r_1,r_2,s_1,s_2>0$ with $r_\mu$ and $s_\mu$ coprime and $\frac{r_1}{s_1}>\frac{r_2}{s_2}$. Further assume the integers satisfy all constraints from
	\cref{prop:symm_cond_omega03}. Then
	\begin{equation*}
	\Delta\coloneqq r_1 s_2 - r_2 s_1
	\end{equation*}
	takes the following values.
	\begin{itemize}
		\item $\Delta = 1$ if $s_1=1$, $r_1=\bigl\lfloor \frac{r_2}{s_2}\bigr\rfloor+1$, and $r_2 = -1 \,\,{\rm mod}\,\, s_2$ or we have $s_2=1$, $r_2=\bigl\lfloor\frac{r_1}{s_1}\bigr\rfloor$, and $r_1 = 1 \,\,{\rm mod}\,\, s_1$.
		
		\item $\Delta = \mr{max}\{s_1,s_2\}-2$ if $s_1 = 2$, $\bigl\lfloor\frac{r_1}{s_1}\bigr\rfloor = \bigl\lfloor\frac{r_2}{s_2}\bigr\rfloor$, and $r_2 = 1 \,\,{\rm mod}\,\, s_2$ or we have $s_2 = 2$, $\bigl\lfloor\frac{r_1}{s_1}\bigr\rfloor = \bigl\lfloor\frac{r_2}{s_2}\bigr\rfloor$, and $r_1 = -1 \,\,{\rm mod}\,\, s_1$.
		
		\item Otherwise $\Delta \geq \mr{max}\{s_1,s_2\} - 1$.
	\end{itemize}
\end{lemma}
\begin{proof}
	Let us write $r_\mu\coloneqq r_\mu' s_\mu + r_\mu''$ with $r_\mu''\in [0,s_\mu)$. Then due to \labelcref{item:r_eq_pm1mods} we know that $r_\mu''\in\{1,s_\mu-1\}$. We can thus rewrite
	\begin{equation}
	\label{eq:Delta_rew}
	\Delta= s_1 s_2 (r_1'-r_2') + r_1''  s_2 - r_2''  s_1\,.
	\end{equation}
	By simply plugging in the indicated values for $r_\mu$, $s_\mu$ it is easy to see that they indeed produce $\Delta \in \{1, \mr{max}\{s_1,s_2\}-2\}$. For instance setting $s_1=1$, $r_1=r_2'+1$, and $r_2'' = s_2 - 1$ we directly obtain $\Delta=1$ as was claimed.\\
	It is therefore only left to prove that in all other cases $\Delta \geq \mr{max}\{s_1,s_2\} - 1$. To do so let us first assume $r_1'=r_2'$ in which case \eqref{eq:Delta_rew} reduces to
	\begin{equation*}
	\Delta = r_1'' s_2 - r_2'' s_1 = \left\{\begin{array}{lll}
	s_1 - s_2 &\text{if }r_1''= s_1-1,& r_2'' = s_2-1 \\
	s_1s_2 - s_1 - s_2 &\text{if }r_1''= s_1-1,& r_2'' = 1 \\
	s_1 + s_2 - s_1s_2 &\text{if }r_1''= 1,& r_2'' = s_2-1 \\
	s_2 - s_1 &\text{if }r_1''= 1,& r_2'' = 1 \\
	\end{array}\right.\,.
	\end{equation*}
	Let us go through the cases separately. First assume $r_1''= s_1-1$ and $r_2'' = s_2-1$. In this case \labelcref{item:eq_fract_constr} forbids $s_1,s_2>2$. Notice now that since we assume that $\frac{r_1}{s_1}>\frac{r_2}{s_2}$ necessarily $\Delta>0$ and therefore the only allowed cases are $s_2=1$ with $s_1>1$ or $s_2=2$ with $s_1>2$. While $s_2=1$ and $s_1>1$ gives $\Delta=s_2-1=\mr{max}\{s_1,s_2\}-1$ the second case was considered before yielding $\Delta=s_2-2=\mr{max}\{s_1,s_2\}-2$.\\
	For our further analysis of the case $r_1'=r_2'$ we may assume that $s_1,s_2>2$. Then if $r_1''= s_1-1$ and $r_2'' = 1$ one has
	\begin{equation*}
	\Delta = s_1(s_2 - 1) - s_2 \geq 3(s_2 - 1) - s_2 = 2s_2 - 3 \geq s_2\,,
	\end{equation*}
	and similarly also $\Delta \geq s_1$. Now assume $r_1''= 1$ and $r_2'' = s_2-1$. Then due to $s_1,s_2>2$ one has
	\begin{equation*}
	\Delta =  s_2 - s_1(s_2 - 1) < s_2 - 2(s_2 - 1) = 2 - s_2 < 0\,,
	\end{equation*}
	which contradicts the assumption that $\frac{r_1}{s_1} > \frac{r_2}{s_2}$. Finally notice that the case $r_1''= 1$ and $r_2'' = 1$ for $s_1,s_2>2$ is forbidden by \labelcref{item:eq_fract_constr}.\\
	Now assume that $r_1' \neq r_2'$. Due to $\frac{r_1}{s_1} > \frac{r_2}{s_2}$ we thus know that $r_1' > r_2'$ which using \eqref{eq:Delta_rew} implies that
	\begin{align*}
	\Delta &\geq s_1 s_2 + r_1''  s_2 - r_2''  s_1 \\
	&= s_1 (s_2 - r_2'') + r_1''  s_2 \\
	&\geq s_1 + r_1''  s_2\,,
	\end{align*}
	which is always larger than $\mr{max}\{s_1,s_2\}-1$ unless $r_1''=0$. However, if $r_1''=0$ then necessarily $s_1=1$. Then either $r_1' = r_2' +1$ implying that
	\begin{equation*}
	\Delta = s_2 - r_2'' = \begin{cases}
	1 & \text{if }r_2'' = s_2 - 1 \\
	s_2 - 1 & \text{if }r_2'' = 1
	\end{cases}\,,
	\end{equation*}
	or we have $r_1' > r_2' +1$ yielding
	\begin{equation*}
	\Delta \geq 2s_2 - r_2'' \geq s_2 > \mr{max}\{s_1,s_2\}-1\,.
	\end{equation*}
	Note that the cases above in which $\Delta=1$ are exactly those considered in the beginning of the proof.
\end{proof}

\begin{proof}[Proof of \cref{prop:symm_cond_omega1half2}]
	Let us start by computing $\omega_{\frac{1}{2},2}\big(\begin{smallmatrix} \mu & \mu \\ z_1 & z_2 \end{smallmatrix}\big)$. Inspecting \cref{thefirstg} it should be clear that one may proceed as in the computation of $\omega_{0,n}$. While the second term in in the bracket in \eqref{thefirstg} gives a contribution
	\begin{align}
	\omega_{\frac{1}{2},2}\big(\begin{smallmatrix} \mu & \mu \\ z_1 & z_2 \end{smallmatrix}\big) &= \ldots + \mathop{{\rm Res}}_{z = 0} \sum_{a=1}^{r_\mu-1} \frac{z \dd z_1}{z_1 (z_1-z)} \frac{\theta_\mu^{a}}{r_\mu t_\mu (\theta_\mu^{a s_\mu} -1) z^{s_\mu -1}} \frac{\dd z_2}{(\theta_\mu^{a} z - z_2)^2} \, \frac{Q_\mu \dd z}{z} \nonumber \\
	&= \ldots + \sum_{k_1,k_2 > 0} \frac{\dd z_1 \dd z_2}{z_1^{k_1 + 1}z_2^{k_2 + 1}}\,k_2\, Q_{\mu}\, \delta_{k_1 + k_2,s_{\mu}} \frac{2\ell_{\mu}(k_2) -r_{\mu} + 1}{2r_{\mu}t_{\mu}}\,, \label{eq:omega1half2mumu1}
	\end{align}
	the first one gets a contribution of type \labelcref{AlwaysInKernel} and additionally a factor
	\begin{equation*}
	\frac{1}{r_\nu(t_\nu \theta_\nu^{b s_\nu} z^{r_\mu s_\nu/r_\nu - 1} - t_\mu z^{s_\mu -1})}\frac{Q_\nu \dd z}{z}\,,
	\end{equation*}
	where one sums over $b\in[r_\nu]$ if $\nu\neq\mu$ and over $b\in[r_\mu-1]$ if $\nu=\mu$. This gives
	\begin{equation}
	\begin{split}
	\label{eq:F1half2_mumu}
	\omega_{\frac{1}{2},2}\big(\begin{smallmatrix} \mu & \mu \\ z_1 & z_2 \end{smallmatrix}\big) &= \mathop{{\rm Res}}_{z = 0} \sum_{\nu\in \tilde{\mathfrak{a}}} \sum_{b=1}^{r_\mu-\delta_{\mu,\nu}} \frac{z \dd z_1}{z_1 (z_1-z)} \frac{\dd z \dd z_2}{(z-z_2)^2} \frac{Q_\nu}{r_\nu(t_\nu \theta_\nu^{b s_\nu} z^{r_\mu s_\nu/r_\nu - 1} - t_\mu z^{s_\mu -1})}\frac{\dd z}{z} + \ldots \\
	&= \sum_{k_1,k_2 > 0} \frac{\dd z_1 \dd z_2}{z_1^{k_1 + 1}z_2^{k_2 + 1}} \bigg\{k_2 Q_{\mu} \delta_{k_1 + k_2,s_{\mu}}\,\frac{ \ell_{\mu}(k_2) -r_{\mu} + 1}{r_{\mu}t_{\mu}} \\
	& \qquad \qquad \qquad \qquad+ k_2 \,\sum_{\substack{\nu \neq \mu \\ \frac{r_{\mu}}{s_{\mu}} = \frac{r_{\nu}}{s_{\nu}}}} \frac{ t_\mu^{r_\mu-1} }{t_{\nu}^{r_{\mu}}-t_{\mu}^{r_{\mu}}} Q_{\nu} \delta_{k_1 + k_2,s_{\mu}} \\
	& \qquad \qquad \qquad \qquad- k_2 \,\sum_{\substack{\nu \neq \mu \\ \frac{r_{\mu}}{s_{\mu}} > \frac{r_{\nu}}{s_{\nu}}}}  \sum_{\ell' \geq 0} \frac{t_{\nu}^{r_\nu\ell'}}{t_{\mu}^{r_\nu\ell' + 1}} Q_{\nu} \delta_{k_1 + k_2 + \ell' (r_{\mu} s_{\nu} - r_{\nu} s_{\mu}) , s_{\mu}}\\
	& \qquad \qquad \qquad \qquad+ k_2 \,\sum_{\substack{\nu \neq \mu \\ \frac{r_{\mu}}{s_{\mu}} < \frac{r_{\nu}}{s_{\nu}}}}  \sum_{\ell' > 0} \frac{t_{\mu}^{r_\nu \ell' - 1}}{t_{\nu}^{r_\nu \ell'}} Q_{\nu} \delta_{k_1 + k_2 - \ell' (r_{\mu} s_{\nu} - r_{\nu} s_{\mu}), s_{\mu}}\bigg\}\,,
	\end{split}
	\end{equation}
	where in the second line we also included the contribution from \eqref{eq:omega1half2mumu1}.
	Let us further on use the notation
	\begin{equation*}
	\omega_{\frac{1}{2},2}\big(\begin{smallmatrix} \mu_1 & \mu_2 \\ z_1 & z_2 \end{smallmatrix}\big) = \sum_{k_1,k_2 > 0} \frac{\dd z_1 \dd z_2}{z_1^{k_1 + 1}z_2^{k_2 + 1}} ~ F_{\frac{1}{2},2}\big[\begin{smallmatrix} \mu_1 & \mu_2 \\ k_1 & k_2 \end{smallmatrix}\big]\,.
	\end{equation*}
		Inspecting \eqref{eq:F1half2_mumu} we notice that for $s_\mu\leq 2$ the components $F_{\frac{1}{2},2}\big[\begin{smallmatrix} \mu & \mu \\ k_1 & k_2 \end{smallmatrix}\big]$ are always symmetric under the exchange of arguments. To be more precise, for $s_\mu=1$ the components all vanish and for $s_\mu=2$ they get a contribution from the first and second line of \eqref{eq:F1half2_mumu} solely. Hence, using that $\ell_\mu(1)=r_\mu'$ for $s_\mu=2$ we see that in this case $F_{\frac{1}{2},2}\big[\begin{smallmatrix} \mu & \mu \\ k_1 & k_2 \end{smallmatrix}\big]$ is indeed given by \eqref{eq:F1half2_mumu_sym1}.\\
	Now let us assume that $s_\mu>2$. Then due to property \labelcref{item:eq_fract_constr} of \cref{prop:symm_cond_omega03} the contribution $\frac{r_\mu}{s_\mu}=\frac{r_\nu}{s_\nu}$ in \eqref{eq:F1half2_mumu} has to vanish and one ends up with
	\begin{equation}
	\label{eq:F1half2_mumu_sgr2}
	\begin{split}
	F_{\frac{1}{2},2}\big[\begin{smallmatrix} \mu & \mu \\ k_1 & k_2 \end{smallmatrix}\big] & = -\delta_{k_1 + k_2,s_{\mu}}\, k_2 \,\bigg(\frac{Q_{\mu} (r_{\mu} - 1 - \ell_{\mu}(k_2))}{r_{\mu} t_{\mu}} + \sum_{\substack{\nu \neq \mu \\ \frac{r_{\mu}}{s_{\mu}} > \frac{r_{\nu}}{s_{\nu}}}} \frac{Q_{\nu}}{t_{\mu}}\bigg)  \\
	& \quad + \frac{k_2}{t_\mu}\sum_{\ell' > 0} \bigg(-\sum_{\substack{\nu \neq \mu \\ \frac{r_{\mu}}{s_{\mu}} > \frac{r_{\nu}}{s_{\nu}}}} \!\! \left(\frac{t_{\nu}}{t_{\mu}}\right)^{r_\nu \ell'}\!\! Q_{\nu} \delta_{k_1 + k_2 + \ell' (r_{\mu} s_{\nu} - r_{\nu} s_{\mu}) , s_{\mu}} + \!\! \sum_{\substack{\nu \neq \mu \\ \frac{r_{\mu}}{s_{\mu}} < \frac{r_{\nu}}{s_{\nu}}}} \! \left(\frac{t_{\mu}}{t_{\nu}}\right)^{r_\nu \ell'} \!\! Q_{\nu} \delta_{k_1 + k_2 - \ell' (r_{\mu} s_{\nu} - r_{\nu} s_{\mu}), s_{\mu}} \bigg)\,.
	\end{split} 
	\end{equation}
	Since $r_{\mu} s_{\nu} - r_{\nu} s_{\mu} \neq 0$ for $\frac{r_{\mu}}{s_{\mu}} \neq \frac{r_{\nu}}{s_{\nu}}$ we can analyse the symmetry constraints coming from the first and second line of \eqref{eq:F1half2_mumu_sgr2} individually.\\
	First, assume $r_\mu = r_\mu' s_\mu + 1$. Then if $k_1 + k_2 = s_\mu$ we may use that
	\begin{equation*}
	r_{\mu} - 1 - \ell_{\mu}(k_2) = r_{\mu} - 1 - r_\mu' k_2 = r_\mu' k_1
	\end{equation*}
	to find that
	\begin{equation}
	\label{eq:F1half2_mumu_eqs_pl1}
	F_{\frac{1}{2},2}\big[\begin{smallmatrix} \mu & \mu \\ k_1 & k_2 \end{smallmatrix}\big] = -k_1 k_2 \, \frac{Q_{\mu} r_\mu'}{r_{\mu} t_{\mu}} - k_2 \sum_{\substack{\nu \neq \mu \\ \frac{r_{\mu}}{s_{\mu}} > \frac{r_{\nu}}{s_{\nu}}}} \frac{Q_{\nu}}{t_{\mu}}\,.
	\end{equation}
	This expression is symmetric under the exchange of $k_1$ and $k_2$ if and only if
	\begin{equation}
	\label{eq:F1half2_mumu_eqs_pl1sym1}
	\sum_{\substack{\nu \neq \mu \\ \frac{r_{\mu}}{s_{\mu}} > \frac{r_{\nu}}{s_{\nu}}}} Q_{\nu} = 0 \,.
	\end{equation}
	Let us proceed by analysing the symmetry constraints coming from coefficients with $k_1 + k_2< s_\mu$, i.e.\ we consider the second line of \eqref{eq:F1half2_mumu_sgr2}. Clearly the case $s_\mu=3$ is symmetric. Therefore now assume that $s_\mu>3$. In this case we can only expect a non-vanishing $F_{\frac{1}{2},2}\big[\begin{smallmatrix} \mu & \mu \\ k_1 & k_2 \end{smallmatrix}\big]$ if $|r_{\mu} s_{\nu} - r_{\nu} s_{\mu}| \in [s_\mu -3]$ for some $\nu\neq\mu$. Due to \cref{prop:symm_cond_omega1half2_tech} we know that for all $\nu$ either $|r_{\mu} s_{\nu} - r_{\nu} s_{\mu}| =1$ or $|r_{\mu} s_{\nu} - r_{\nu} s_{\mu}|  \geq s_\mu -2$ which limits the cases in which $F_{\frac{1}{2},2}\big[\begin{smallmatrix} \mu & \mu \\ k_1 & k_2 \end{smallmatrix}\big] \neq 0$ for $k_1 + k_2< s_\mu$ extremely. To be more precise, \cref{prop:symm_cond_omega1half2_tech} tells us that $|r_{\mu} s_{\nu} - r_{\nu} s_{\mu}| =1$ if and only if $s_\nu = 1$ and $r_\nu = \big\lfloor \tfrac{r_\mu}{s_\mu}\big\rfloor$. Therefore for $2< k_1 + k_2 < s_\mu$ we have
	\begin{align}
	F_{\frac{1}{2},2}\big[\begin{smallmatrix} \mu & \mu \\ k_1 & k_2 \end{smallmatrix}\big] &= -\frac{k_2}{t_\mu} \sum_{\ell'>0} \sum_{\substack{\nu \neq \mu \\ s_\nu = 1 ,\, r_\nu = \big\lfloor \frac{r_\mu}{s_\mu}\big\rfloor}}  \left(\frac{t_{\nu}}{t_{\mu}}\right)^{r_\nu \ell'} Q_{\nu}\,\delta_{k_1 + k_2 + \ell' , s_\mu} \nonumber \\
	&= -\frac{k_2}{t_\mu} \sum_{\substack{\nu \neq \mu \\ s_\nu = 1 ,\, r_\nu = \big\lfloor \frac{r_\mu}{s_\mu} \big\rfloor}}  \left(\frac{t_{\nu}}{t_{\mu}}\right)^{r_\nu (s_\mu-k_1-k_2)} Q_{\nu}\,, \label{eq:F1half2_mumu_eqs_min2}
	\end{align}
	which is symmetric if and only if for all $\ell\in [s_\mu-3]$ we have
	\begin{equation*}
	\sum_{\substack{\nu \neq \mu \\ s_\nu = 1 ,\, r_\nu = \big\lfloor \frac{r_\mu}{s_\mu} \big\rfloor}}  \left(\frac{t_{\nu}}{t_{\mu}}\right)^{r_\nu \ell} Q_{\nu} = 0\,.
	\end{equation*}
	Together with \eqref{eq:F1half2_mumu_eqs_pl1sym1} this explains symmetry condition \labelcref{item:1half2cond1}. Regarding the formula for $F_{\frac{1}{2},2}\big[\begin{smallmatrix} \mu & \mu \\ k_1 & k_2 \end{smallmatrix}\big]$ stated in \eqref{eq:F1half2_mumu_sym1} notice that for $s_\mu=2$ we have
	\begin{equation*}
	F_{\frac{1}{2},2}\big[\begin{smallmatrix} \mu & \mu \\ k_1 & k_2 \end{smallmatrix}\big] = -\frac{Q_{\mu} r_{\mu}'}{r_{\mu}t_{\mu}}\,\delta_{k_1 + k_2,2} + \sum_{\substack{\nu \neq \mu \\ \frac{r_{\mu}}{s_{\mu}} = \frac{r_{\nu}}{s_{\nu}}}} \frac{ t_\mu^{r_\mu-1} }{t_{\nu}^{r_{\mu}}-t_{\mu}^{r_{\mu}}} Q_{\nu} \delta_{k_1 + k_2,s_{\mu}}
	\end{equation*}
	while for $s_\mu>2$ one has
	\begin{equation*}
	\begin{split}
	& \quad F_{\frac{1}{2},2}\big[\begin{smallmatrix} \mu & \mu \\ k_1 & k_2 \end{smallmatrix}\big] \\
	& = -\frac{Q_{\mu} r_{\mu}'}{r_{\mu}t_{\mu}}\,\delta_{k_1 + k_2,s_\mu} - \sum_{\substack{\nu\neq \mu \\ s_\nu = 1,\, \big\lfloor \frac{r_\mu}{s_\mu}\big\rfloor = r_\nu}} \frac{Q_\nu}{t_\mu} \left(\frac{t_\nu}{t_\mu}\right)^{r_\nu(s_\mu-2)} \delta_{k_1 + k_2 ,2} + \sum_{\substack{\nu \neq \mu:\, s_\nu=2, \\ \frac{r_{\mu}}{s_{\mu}} < \frac{r_{\nu}}{s_{\nu}},\, \left\lfloor\frac{r_\mu}{s_\mu}\right\rfloor = \left\lfloor\frac{r_\nu}{s_\nu}\right\rfloor}} \frac{Q_\nu}{t_\mu} \left(\frac{t_{\mu}}{t_{\nu}}\right)^{r_\nu}  \delta_{k_1 + k_2 ,2}\,.
	\end{split}
	\end{equation*}
	In order to obtain the first two terms in the above expression one applies the symmetry constraint \labelcref{item:1half2cond1} on \eqref{eq:F1half2_mumu_eqs_pl1} and \eqref{eq:F1half2_mumu_eqs_min2}. The third term is due to contributions $\nu\neq\mu$ in \eqref{eq:F1half2_mumu_sgr2} for which $|r_{\mu} s_{\nu} - r_{\nu} s_{\mu}| = s_\mu -2$.  \Cref{prop:symm_cond_omega1half2_tech} tells us that this is the case for exactly those $\nu\neq \mu$ with $s_\nu=2$, $\frac{r_\mu}{s_\mu} < \frac{r_\nu}{s_\nu}$, and $\big\lfloor \tfrac{r_\mu}{s_\mu}\big\rfloor = \big\lfloor \tfrac{r_\nu}{s_\nu}\big\rfloor$. This now explains the origin of all terms occurring in \eqref{eq:F1half2_mumu_sym1} which closes the analysis of the case $r_\mu = r_\mu' s_\mu + 1$.
	
	Now assume that $r_\mu = r_\mu' s_\mu + s_\mu - 1$ and $s_\mu>2$. First let us inspect $F_{\frac{1}{2},2}\big[\begin{smallmatrix} \mu & \mu \\ k_1 & k_2 \end{smallmatrix}\big]$ given by \eqref{eq:F1half2_mumu_sgr2} for $k_1+k_2=s_\mu$. In this case we may use that
	\begin{equation*}
	r_{\mu} - 1 - \ell_{\mu}(k_2) = r_{\mu} - k_1 (r_\mu'+1)
	\end{equation*}
	in order to find that
	\begin{equation*}
	F_{\frac{1}{2},2}\big[\begin{smallmatrix} \mu & \mu \\ k_1 & k_2 \end{smallmatrix}\big] = k_1 k_2 \, \frac{Q_{\mu} (r_\mu'+1)}{r_{\mu} t_{\mu}} - k_2 \bigg(\frac{Q_{\mu}}{t_{\mu}} + \sum_{\substack{\nu \neq \mu \\ \frac{r_{\mu}}{s_{\mu}} > \frac{r_{\nu}}{s_{\nu}}}} \frac{Q_{\nu}}{t_{\mu}}\bigg)
	\end{equation*}
	for $k_1 + k_2 = s_\mu$. This is symmetric if and only if
	\begin{equation}
	\label{eq:F1half2_mumu_eqs_pl2sym1}
	Q_{\mu} + \sum_{\substack{\nu \neq \mu \\ \frac{r_{\mu}}{s_{\mu}} > \frac{r_{\nu}}{s_{\nu}}}} Q_{\nu} = 0 \,.
	\end{equation}
	Now let us turn to the case $k_1 + k_2 < s_\mu$. Again, we may only expect a contribution from $\nu\neq\mu$ in the second line of \eqref{eq:F1half2_mumu_sgr2} possibly leading to a non-symmetric term if  $|r_{\mu} s_{\nu} - r_{\nu} s_{\mu}| =1$. Now \cref{prop:symm_cond_omega1half2_tech} says that $|r_{\mu} s_{\nu} - r_{\nu} s_{\mu}| =1$ if and only if $s_\nu = 1$ and $r_\nu = \big\lceil \tfrac{r_\mu}{s_\mu}\big\rceil$. Thus, for $2< k_1 + k_2 < s_\mu$ we find that
	\begin{equation*}
	F_{\frac{1}{2},2}\big[\begin{smallmatrix} \mu & \mu \\ k_1 & k_2 \end{smallmatrix}\big] =  \frac{k_2}{t_\mu} \sum_{\substack{\nu \neq \mu \\ s_\nu = 1 ,\, r_\nu = \big\lceil \frac{r_\mu}{s_\mu} \big\rceil}}  \left(\frac{t_{\mu}}{t_{\nu}}\right)^{r_\nu (s_\mu-k_1-k_2)} Q_{\nu}\,,
	\end{equation*}
	which is symmetric if and only if
	\begin{equation*}
	\sum_{\substack{\nu \neq \mu \\ s_\nu = 1 ,\, r_\nu = \big\lceil \frac{r_\mu}{s_\mu}\big\rceil }}  \left(\frac{t_{\mu}}{t_{\nu}}\right)^{r_\nu \ell} Q_{\nu} = 0
	\end{equation*}
	for all $\ell\in [s_\mu-3]$. This condition together with \eqref{eq:F1half2_mumu_eqs_pl2sym1} are of course nothing but \labelcref{item:1half2cond2}. Since the derivation of formula \eqref{eq:F1half2_mumu_sym2} is in line with the one of \eqref{eq:F1half2_mumu_sym1} we omit a further discussion.
	
	Now let us consider $F_{\frac{1}{2},2}\big[\begin{smallmatrix} \mu & \nu \\ k_1 & k_2 \end{smallmatrix}\big]$ for $\nu\neq\mu$. First, assume $\frac{r_\mu}{s_\mu} = \frac{r_\nu}{s_\nu}$. Note that in this case property \labelcref{item:eq_fract_constr} of \cref{prop:symm_cond_omega03} forces $s_\mu=s_\nu\leq 2$. Using the same approach as in the derivation of \eqref{eq:F1half2_mumu} one finds that
	\begin{equation}
	\label{eq:F1half2_munueqfr}
	F_{\frac{1}{2},2}\big[\begin{smallmatrix} \mu & \nu \\ k_1 & k_2 \end{smallmatrix}\big] = - k_2 \, \frac{Q_\mu t_\mu^{\ell_\mu(k_1)} t_\nu^{\ell_\mu(k_2)} }{t_\mu^{r_\mu} - t_\nu^{r_\nu}} \, \delta_{k_1 + k_2, s_\mu}\,.
	\end{equation}
	For $s_\mu=s_\nu\ = 1$ this expression is always vanishing and hence symmetric. Conversely, for $s_\mu=s_\nu\ = 2$ we have $F_{\frac{1}{2},2}\big[\begin{smallmatrix} \mu & \nu \\ k_1 & k_2 \end{smallmatrix}\big] = F_{\frac{1}{2},2}\big[\begin{smallmatrix} \nu & \mu \\ k_2 & k_1 \end{smallmatrix}\big]$ if and only if $Q_\mu = - Q_\nu$, which is captured in condition \labelcref{item:1half2cond4}.\\
	Now let us consider the case $\frac{r_\mu}{s_\mu} \neq \frac{r_\nu}{s_\nu}$. Without loss of generality we may assume that $\frac{r_\mu}{s_\mu} > \frac{r_\nu}{s_\nu}$. In this case we have
	\begin{align}
	F_{\frac{1}{2},2}\big[\begin{smallmatrix} \mu & \nu \\ k_1 & k_2 \end{smallmatrix}\big] 
	&= -k_2 \,\frac{Q_{\mu}}{t_\mu}\sum_{\ell' \geq 0} \bigg(\frac{t_{\nu}}{t_{\mu}}\bigg)^{r_\nu \ell' + \ell_\nu(k_2)} \delta_{k_1 + (k_2 + s_{\nu} \ell_\nu(k_2))\frac{r_{\mu}}{r_{\nu}} - s_\mu\ell_\nu(k_2) + \ell'(r_{\mu}s_{\nu} - r_{\nu} s_{\mu}) , s_{\mu}} \,. \label{eq:F1half2_munu_neq}
	\end{align}
	Let us first simplify this expression before turning to $F_{\frac{1}{2},2}\big[\begin{smallmatrix} \nu & \mu \\ k_2 & k_1 \end{smallmatrix}\big]$. Comparing \eqref{eq:F1half2_munu_neq} with \eqref{eq:omega03mumunu} we can deduce that property \labelcref{item:eq_fract_constr} of \cref{prop:symm_cond_omega03} ensures that $F_{\frac{1}{2},2}\big[\begin{smallmatrix} \mu & \nu \\ k_1 & k_2 \end{smallmatrix}\big]=0$ unless $k_1=1$. And moreover also $F_{\frac{1}{2},2}\big[\begin{smallmatrix} \mu & \nu \\ k_1 & k_2 \end{smallmatrix}\big]=0$ unless
	\begin{equation}
	\label{eq:F1half2_munu_neqz_cond}
	1 + (k_2 + s_{\nu} \ell_\nu(k_2))\tfrac{r_{\mu}}{r_{\nu}} - s_\mu\ell_\nu(k_2) = s_{\mu}\,.
	\end{equation}
	This relation cannot be satisfied for $k_2 \geq s_\nu$. Indeed, if $r_\nu = 1 \,\,{\rm mod}\,\, s_\nu$ then
	\begin{equation*}
		\begin{split}
			(k_2 + s_{\nu} \ell_\nu(k_2))\tfrac{r_{\mu}}{r_{\nu}} - s_\mu\ell_\nu(k_2) & = k_2(r_\mu - r_\nu' s_\mu) - \left\lfloor\frac{k_2 r_\nu'}{r_\nu} \right\rfloor (r_\mu s_\nu - r_\nu s_\mu) \\
			& = \frac{s_\mu}{s_\nu} \, k_2 + \bigg( \frac{k_2}{s_\nu} - \left\lfloor \frac{k_2 r_\nu'}{r_\nu}\right\rfloor\bigg)(r_\mu s_\nu - r_\nu s_\mu) \\ 
			& \geq \frac{s_\mu}{s_\nu} \, k_2  \geq s_\mu
		\end{split}
	\end{equation*}
	for all $k_2\geq s_\nu$ where we used that
	\begin{equation*}
		\ell_\nu(k) = \begin{cases}
			k r_\nu' - \left\lfloor\frac{k r_\nu'}{r_\nu} \right\rfloor \,r_\nu & \text{if } r_\nu = r_\nu' s_\nu + 1\\[0.5em]
			-k (r_\nu' +1) + \left\lceil\frac{k (r_\nu'+1)}{r_\nu} \right\rceil \,r_\nu & \text{if } r_\nu = r_\nu' s_\nu + s_\nu - 1
		\end{cases}\,,
	\end{equation*}
	and that $\big\lfloor \frac{k r_\nu'}{r_\nu}\big\rfloor \leq \frac{k}{s_\nu}$. The case $r_\nu = -1 \,\,{\rm mod}\,\, s_\nu$ can be covered using similar arguments. On the other hand for $1\leq k_2< s_\nu$ plugging in the explicit expression for $\ell_\nu(k_2)$ we find that \eqref{eq:F1half2_munu_neqz_cond} is equivalent to
	\begin{equation}
	\label{eq:F1half2_munu_neqz_cond2}
	s_\mu-1 = \begin{cases}
	k_2 (r_\mu - s_\mu r_\nu') & \text{if }r_\nu = r_\nu' s_\nu + 1 \\
	s_\mu + (s_\nu - k_2) (r_\mu - s_\mu r_\nu') & \text{if }r_\nu = r_\nu' s_\nu + s_\nu - 1
	\end{cases}\,.
	\end{equation}
	Let us first analyse the case $r_\nu = r_\nu' s_\nu + 1$. Since
	\begin{equation}
		\label{eq:F1half2_munu_neqz_cond2simpl}
		r_\mu - s_\mu r_\nu' = r_\mu'' + s_\mu \left(\left\lfloor\frac{r_\mu}{s_\mu} \right\rfloor - \left\lfloor\frac{r_\nu}{s_\nu} \right\rfloor + \delta_{s_\nu,1}\right)
	\end{equation}
	equation \eqref{eq:F1half2_munu_neqz_cond2} can only be satisfied if $s_\nu>1$ and $\big\lfloor\frac{r_\mu}{s_\mu} \big\rfloor = \big\lfloor\frac{r_\nu}{s_\nu} \big\rfloor$. Suppose the latter is true then since we assumed that $\frac{r_\mu}{s_\mu} > \frac{r_\nu}{s_\nu}$ we have $s_\nu r_\mu'' >s_\mu$. Therefore if we had $r_\mu''=1$ and $s_\mu>2$ this would violate property \labelcref{item:eq_fract_constr} of \cref{prop:symm_cond_omega03} as in this case $s_\nu,s_\mu>2$. Hence, we can assume that $r_\mu''=s_\mu-1$. Using \eqref{eq:F1half2_munu_neqz_cond2simpl} this means that \eqref{eq:F1half2_munu_neqz_cond2} is satisfied if and only if $k_2=1$.\\
	The case where $r_\nu = r_\nu' s_\nu + s_\nu - 1$ can be treated similarly. We can assume that $s_\nu>2$ since $s_\nu\leq 2$ is a special case of the one considered before. Now since $r_\mu - s_\mu r_\nu'\geq 0$ we always have
	\begin{equation*}
	s_\mu + (s_\nu - k_2) (r_\mu - s_\mu r_\nu') \geq s_\mu\,,
	\end{equation*}
	implying that \eqref{eq:F1half2_munu_neqz_cond2} cannot be satisfied for any $k_2>0$. Thus, $F_{\frac{1}{2},2}\big[\begin{smallmatrix} \mu & \nu \\ k_1 & k_2 \end{smallmatrix}\big]=0$ for all $k_i>0$ if $s_\nu>2$ and $r_\nu = r_\nu' s_\nu + s_\nu - 1$. To sum up, assuming that \labelcref{item:r_eq_pm1mods,item:eq_fract_constr} from \cref{prop:symm_cond_omega03} hold we have
	\begin{equation}
	\label{eq:F1half2_neq1}
	F_{\frac{1}{2},2}\big[\begin{smallmatrix} \mu & \nu \\ k_1 & k_2 \end{smallmatrix}\big]= \begin{cases} -
	\frac{Q_\mu}{t_\mu} \left(\frac{t_\nu}{t_\mu}\right)^{\ell_\nu(1)} & \text{if }k_1=k_2=1,\, s_\nu>1,\, \bigl\lfloor \frac{r_\mu}{s_\mu} \bigr\rfloor = \bigl\lfloor \frac{r_\nu}{s_\nu} \bigr\rfloor ,\, r_\mu = -1 \,\,{\rm mod}\,\, s_\mu ,\, r_\nu =1 \,\,{\rm mod}\,\, s_\nu\\
	0 & \text{otherwise.}
	\end{cases}
	\end{equation}
	One can obtain a similar expression for $F_{\frac{1}{2},2}\big[\begin{smallmatrix} \nu & \mu \\ k_2 & k_1 \end{smallmatrix}\big]$. First, one finds that
	\begin{equation*}
	F_{\frac{1}{2},2}\big[\begin{smallmatrix} \nu & \mu \\ k_2 & k_1 \end{smallmatrix}\big] = k_1 \sum_{\ell \geq 0} \frac{t_{\nu}^{\ell}}{t_{\mu}^{\ell + 1}} Q_{\nu} \delta_{r_{\mu}|k_1 - (\ell+1) s_{\mu}} \delta_{k_2 + (k_1 - (\ell+1) s_{\mu})\frac{r_{\nu}}{r_{\mu}} + \ell s_{\nu}, 0}\,.
	\end{equation*}
	Then comparing the above expression with \eqref{eq:omega03mumunu} one can simplify the expression for $F_{\frac{1}{2},2}\big[\begin{smallmatrix} \nu & \mu \\ k_2 & k_1 \end{smallmatrix}\big]$ as we did with $F_{\frac{1}{2},2}\big[\begin{smallmatrix} \mu & \nu \\ k_1 & k_2 \end{smallmatrix}\big]$ before. One finds that
	\begin{equation}
	\label{eq:F1half2_neq2}
	F_{\frac{1}{2},2}\big[\begin{smallmatrix} \nu & \mu \\ k_2 & k_1 \end{smallmatrix}\big] = \begin{cases}
	\frac{Q_\nu}{t_\nu} \left(\frac{t_\nu}{t_\mu}\right)^{r_\mu - \ell_\mu(1)} & k_1=k_2=1,\, s_\nu>1,\, \bigl\lfloor \frac{r_\mu}{s_\mu} \bigr\rfloor = \bigl\lfloor \frac{r_\nu}{s_\nu} \bigr\rfloor ,\, r_\mu = -1 \,\,{\rm mod}\,\, s_\mu ,\, r_\nu =1 \,\,{\rm mod}\,\, s_\nu\\
	0 & \text{otherwise}
	\end{cases}\,.
	\end{equation}
	Now comparing \eqref{eq:F1half2_neq1} with \eqref{eq:F1half2_neq2} it is clear that $F_{\frac{1}{2},2}\big[\begin{smallmatrix} \mu & \nu \\ k_1 & k_2 \end{smallmatrix}\big] = F_{\frac{1}{2},2}\big[\begin{smallmatrix} \nu & \mu \\ k_2 & k_1 \end{smallmatrix}\big]$ if and only if
	\begin{equation}
	\label{eq:1half2cond4}
	\frac{Q_\mu}{t_\mu} \left(\frac{t_\nu}{t_\mu}\right)^{\ell_\nu(1)} = - \frac{Q_\nu}{t_\nu} \left(\frac{t_\nu}{t_\mu}\right)^{r_\mu - \ell_\mu(1)}
	\end{equation}
	in case $r_\mu = r_\mu' s_\mu + s_\mu  -1$, $r_\nu =r_\nu' s_\nu + 1$, $s_\nu>1$, and  $r_\mu' = r_\nu'$. Notice now that due to $\ell_\nu(1) = r_\nu'$ and $\ell_\mu(1) = r_\mu - r_\mu' - 1$ equation \eqref{eq:1half2cond4} is equivalent to $Q_\mu = - Q_\nu$, which is nothing but condition \labelcref{item:1half2cond4}. Finally, notice that the equations \eqref{eq:F1half2_munueqfr} and \eqref{eq:F1half2_neq1} directly imply the formula for $\omega_{\frac{1}{2},2}\big(\begin{smallmatrix} \mu & \nu \\ z_1 & z_2 \end{smallmatrix}\big)$ stated in  \eqref{eq:omega1half2offdiag}. This finishes the proof.
\end{proof}

\medskip

\subsubsection{\texorpdfstring{$(g,n) = (1,1)$}{(g,n)=(1,1)}}

\medskip

We study separately the two types of terms in \eqref{thefirstg}
\[
\omega_{1,1}\big(\begin{smallmatrix} \mu_1 \\ z_1 \end{smallmatrix}\big) = \omega_{1,1}^{{\rm I}}\big(\begin{smallmatrix} \mu_1 \\ z_1 \end{smallmatrix}\big) + \omega_{1,1}^{{\rm II}}\big(\begin{smallmatrix} \mu_1 \\ z_1 \end{smallmatrix}\big)\,,
\]
where $\omega_{1,1}^{{\rm I}}$ is the contribution from $\omega_{0,2}$ and $\omega_{1,1}^{{\rm II}}$ the one from $\omega_{\frac{1}{2},1} \cdot \omega_{\frac{1}{2},1}$.

\begin{lemma}
	Let $\mu \in \tilde{\mf{a}}$. Then
	\begin{equation*}
	\omega_{1,1}^{{\rm I}}\big(\begin{smallmatrix} \mu \\ z \end{smallmatrix}\big) = \frac{r_{\mu}^2 - 1}{24 r_{\mu} t_{\mu}} \,\frac{\dd z}{z^{s_{\mu} + 1}}\,. 
	\end{equation*}
\end{lemma}
\begin{proof}
	We have
	\begin{equation}
	\label{tniugniun} 
	\omega_{1,1}^{{\rm I}} \big(\begin{smallmatrix} \mu \\ z_1 \end{smallmatrix}\big) =  \mathop{{\rm Res}}_{z = \mu} \sum_{a = 1}^{r_{\mu} - 1} K_{\mu}\big(\begin{smallmatrix} \mu & \mu & \mu \\ z_1 & z & \vartheta_{\mu}^{a}z \end{smallmatrix}\big)\, \frac{\vartheta_{\mu}^{a}}{(1 - \vartheta_{\mu}^{a})^2}\,\frac{(\dd z)^2}{z^2} 
	\end{equation} 
	where the presence of $\omega_{0,2}(z,z')$ forces the three variables in the recursion kernel to belong to the same component $\tilde{C}_{\mu}$ of $\tilde{C}$. To handle the sum, we use the following identity for $a$ not divisible by $r_{\mu}$
	\begin{equation}
	\label{doubleid}
	\frac{\theta_{\mu}^{a}}{(1 - \theta_{\mu}^{a})^2} = \sum_{m = 0}^{r_{\mu} - 1} \frac{m(r_{\mu} - m)}{2r_{\mu}}\,\vartheta_{\mu}^{am}\,.
	\end{equation}
	The contribution of each term of the right-hand side of \eqref{doubleid} to \eqref{tniugniun} can be computed similarly to \cref{KernelOtherComp=}. After computation of the residue we find
	\begin{equation}
	\label{thefunnysum}\omega_{1,1}^{{\rm I}}\big(\begin{smallmatrix} \mu \\ z_1 \end{smallmatrix}\big) = - \sum_{m = 0}^{r_{\mu} - 1} \frac{m(r_{\mu} - m)(r_{\mu} - 1 - 2\ell_{\mu}(m))}{4r_{\mu}^2t_{\mu}}\,\frac{\dd z_1}{z_1^{s_{\mu} + 1}}\,,
	\end{equation}
	and it is enough to sum over $m \in [r_{\mu} - 1]$. Recall that $\ell_{\mu}(m)$ is defined as the unique integer in $[0,r_{\mu})$ such that $\ell_{\mu}(m) = -mc_{\mu}\,\,{\rm mod}\,\,r_{\mu}$, where $ c_\mu $ is given by \cref{defcmu}
		
	In the case $r_\mu = r'_\mu s_\mu +1$, we can decompose $m = m's_{\mu} + m''$ with $m'' \in [s_{\mu}]$. It puts $m \in [r_{\mu} - 1]$ in bijection with $(m',m'') \in [0,r_{\mu}') \times [s_{\mu}]$, and we get
	\[
	\ell_{\mu}(m) = r_{\mu}'m - r_{\mu}m' = r_{\mu}'m'' - m' \,.
	\]
	Then, we have
\begin{equation*}
\begin{split}
	\sum_{m = 1}^{r_{\mu} - 1} m(r_{\mu} - 1)(r_{\mu} - 1 - 2\ell_{\mu}(m)) & = \sum_{m' = 0}^{r_{\mu} - 1} \sum_{m'' = 1}^{s_{\mu}} (m's_{\mu} + m'')(r_{\mu} - m's_{\mu} - m'')(r_{\mu} - 1 - 2r'_{\mu}m'' + 2m') \\
	& = -\frac{r_{\mu}(r_{\mu}^2 - 1)}{6}\,,
\end{split}
\end{equation*}
	whence
	\[
	\omega_{1,1}^{{\rm I}}\big(\begin{smallmatrix} \mu \\ z_1 \end{smallmatrix}\big) = \frac{r_{\mu}^2 - 1}{24r_{\mu}t_{\mu}}\,\frac{\dd z_1}{z_1^{s_{\mu} + 1}}\,.
	\]
	
	In the case $r_\mu = r'_\mu s_\mu + s_\mu - 1$, we decompose $m = m's_{\mu} + m''$ with $m'' \in [0,s_{\mu})$. It puts $m \in [r_{\mu}]$ in bijection with $(m',m'') \in [0,r_{\mu}'] \times [0,s_{\mu})$ and gives
	\[
	\ell_{\mu}(m) = -(r'_{\mu} +1)m + (m' + 1)r_{\mu} = r_{\mu} - m' - (r'_{\mu} +1)m''\,.
	\]
	Then, we have
\begin{equation*}
\begin{split}
\sum_{m = 0}^{r_{\mu}} m(r_{\mu} - 1)(r_{\mu} - 1 - 2\ell_{\mu}) & = \sum_{m' = 0}^{r_{\mu}'} \sum_{m'' = 1}^{s_{\mu} - 1} (m's_{\mu} + m'')(r_{\mu} - m's_{\mu} - m'')(2m' + 2(r'_{\mu} +1)m'' -r_{\mu} - 1) \\ 
& = -\frac{r_{\mu}(r_{\mu}^2 - 1)}{6}\,.
\end{split}
\end{equation*}
	So, we obtain the same formula for $\omega_{1,1}^{{\rm I}}$ in both cases.
\end{proof}

\begin{remark} This computation can also be done directly from the Airy structure, as
	\[
	\omega_{1,1}^{{\rm I}}\big(\begin{smallmatrix} \mu \\ z \end{smallmatrix}\big) = \sum_{k > 0} F_{1,1}[k]|_{Q = 0}\,\frac{\dd z}{z^{k + 1}} 
	\] 
	and $F_{1,1}[k]$ is the constant term of order $\hbar$ in the unique operator of the Airy structure of the form $\hbar\partial_{x_k} + \higherorderterms{2}$. In fact, as \eqref{tniugniun} coincides with the topological recursion on the sole component $\tilde{C}_{\mu}$ of the spectral curve, the value of $F_{1,1}[k]|_{Q = 0}$ must coincide at $t_{\mu} = \frac{1}{r_{\mu}}$ with the one computed in \cite[Lemma B.3]{BBCCN18}. This is indeed the case, but we note the calculation by this other method involves the sum
	\[
	\Psi^{(1)}(\emptyset) = -\frac{1}{2} \sum_{m = 0}^{r_{\mu} - 1} \frac{m(r_{\mu} - m)}{2r_{\mu}} = -\frac{r_{\mu}(r_{\mu}^2 - 1)}{24}\,,
	\]
	which is much simpler than \eqref{thefunnysum} although it leads to the same result.
\end{remark}

To obtain $\omega_{1,1}^{{\rm II}}$ we may split the contribution of the various $\nu$ as we did in the derivation of \cref{eq:F1half2_mumu}. The details are omitted and we only give the outcome:
\begin{equation*}
\begin{split}
\omega_{1,1}^{{\rm II}}\big(\begin{smallmatrix} \mu \\ z_1 \end{smallmatrix}\big) & = \dd z_1 Q_\mu \bigg\{-z_1^{-(s_{\mu} + 1)}\,Q_{\mu}\,\frac{r_{\mu} - 1}{2r_{\mu}t_{\mu}} \\
& \quad \quad - \sum_{\substack{\nu \neq \mu \\ \frac{r_{\mu}}{s_{\mu}} = \frac{r_{\nu}}{s_{\nu}}}} \frac{t_{\mu}^{r_{\mu} - 1}}{t_{\mu}^{r_{\mu}} - t_{\nu}^{r_{\mu}}}\,\frac{Q_{\nu}}{z_1^{s_{\mu} + 1} } \\
& \quad\quad - \sum_{\substack{\nu \neq \mu \\ \frac{r_{\mu}}{s_{\mu}} > \frac{r_{\nu}}{s_{\nu}}}} \sum_{\substack{k_1 > 0 \\ \ell' \geq 0}} \frac{t_{\nu}^{r_{\nu}\ell'}}{t_{\mu}^{r_{\nu}\ell' + 1}}\,\frac{Q_{\nu}}{z_1^{k_1 + 1}}\delta_{k_1 + \ell'(s_{\nu} r_{\mu} - s_{\mu}r_{\nu}),s_{\mu}} \\
& \quad \quad + \sum_{\substack{\nu \neq \mu \\ \frac{r_{\mu}}{s_{\mu}} < \frac{r_{\nu}}{s_{\nu}}}} \sum_{\substack{k_1 > 0 \\ \ell' > 0}} \frac{t_{\mu}^{r_{\nu}\ell' - 1}}{t_{\nu}^{r_{\nu}\ell'}}\,\frac{Q_{\nu}}{z_1^{k_1 + 1}}\delta_{k_1+ \ell'(s_{\mu}r_{\nu} - s_{\nu}r_{\mu}),s_{\mu}} \bigg\}\,.
\end{split}
\end{equation*}

\medskip
\subsection{The exceptional case}
\label{sec:exccasecorr}

\medskip

We proceed considering a spectral curve with only one branchpoint. Let us only slightly change the setting and add a component --- the exceptional component --- indexed by $\mu_-\in\tilde{\mathfrak{a}}$ on which
\[
x\big(\begin{smallmatrix} \mu_- \\ z \end{smallmatrix}\big) = z,\qquad y\big(\begin{smallmatrix} \mu_- \\ z \end{smallmatrix}\big) = 0\,,
\]
i.e.\ $r_{\mu_-}=1$ and $s_{\mu_-}=\infty$. On all other components $\mu\in\tilde{\mathfrak{a}}\setminus\{\mu_-\}$ we still take
\[
x\big(\begin{smallmatrix} \mu \\ z \end{smallmatrix}\big) = z^{r_{\mu}},\qquad y\big(\begin{smallmatrix} \mu \\ z \end{smallmatrix}\big) = -t_{\mu}\,z^{s_{\mu} - r_{\mu}}\,,
\]
and we equip the curve with the standard bidifferential and the crosscap differential
\[
\omega_{\frac{1}{2},1}\big(\begin{smallmatrix} \mu \\ z \end{smallmatrix}\big) = \frac{Q_{\mu}\dd z}{z}\,.
\]
As before we only require that ${\rm gcd}(r_{\mu},s_{\mu}) = 1$, and that for $\tfrac{r_{\mu}}{s_{\mu}} = \tfrac{r_{\nu}}{s_{\nu}}$ and $\mu \neq \nu$ we must have $t_{\mu}^{r_\mu} \neq t_{\nu}^{r_\nu}$. In the following we will consider and partially compute the correlators $\omega_{0,3}$, $\omega_{0,4}$ and $\omega_{\frac{1}{2},2}$ and characterise the cases in which they are symmetric. These correlators are still given by formula \eqref{thefirstg} and we will learn that in the exceptional case they mostly behave as in the standard case which was discussed in the preceding sections.
\begin{lemma}
	\label{prop:symm_cond_exc}
	The correlators $(\omega_{0,n})_{n\geq 1}$ are symmetric if and only if the conditions \labelcref{item:r_eq_pm1mods}, \labelcref{item:eq_fract_constr} of \cref{prop:symm_cond_omega03} and condition \labelcref{item:threefracteq} of \Cref{prop:symm_cond_omega04} are satisfied. Assuming the above conditions, $\omega_{\frac{1}{2},2}$ is symmetric if and only if \labelcref{item:1half2cond1}, \labelcref{item:1half2cond2}, and \labelcref{item:1half2cond4} of \cref{prop:symm_cond_omega1half2} hold. In this case the correlators $\omega_{0,3}$ and $\omega_{\frac{1}{2},2}$ are still given by the formulas stated in \cref{prop:symm_cond_omega03} and \cref{prop:symm_cond_omega1half2}, where any Kronecker delta involving $ s_{\mu_-} = \infty $ evaluates to $0$.
\end{lemma}

\begin{remark}
	One should remark that all expressions occurring in \labelcref{item:r_eq_pm1mods} to \labelcref{item:1half2cond4} make sense even with $s_{\mu_-}=\infty$ for a single $\mu_-\in\tilde{\mathfrak{a}}$ if we understand $r_{\mu_-} = 1 \,\,{\rm mod}\,\, s_{\mu_-}$ and interpret $\frac{r_{\mu_-}}{s_{\mu_-}}$ as being zero.
\end{remark}

\begin{proof}[Proof of \cref{prop:symm_cond_exc}]
	Let us first focus on the correlators $(\omega_{0,n})_{n\geq 1}$. Regarding those \Cref{mainth2gequal0} already tells us that \labelcref{item:r_eq_pm1mods}, \labelcref{item:eq_fract_constr}, and \labelcref{item:threefracteq} are sufficient conditions for these correlators to be symmetric. Conversely, we can argue that these conditions are necessary for the symmetry too by looking at $\omega_{0,3}$ and $\omega_{0,4}$.\\
	First of all we notice that in case $\mu_1,\mu_2,\mu_3$ satisfy $\mu_i\neq\mu_-$ the correlator $\omega_{0,3}\big(\begin{smallmatrix} \mu_1 & \mu_2 & \mu_3 \\ z_1 & z_2 & z_3 \end{smallmatrix}\big)$ is computed as in the standard case which was considered in \cref{prop:symm_cond_omega03} since the bidifferential does not mix the components. So first we can deduce that restriction of $\omega_{0,3}$ to the non-exceptional components is symmetric if and only if $(r_\mu,s_\mu)_{\mu\neq\mu_-}$ satisfy \labelcref{item:r_eq_pm1mods} and \labelcref{item:eq_fract_constr} and second we know that in this case $\omega_{0,3}$ is given by \eqref{03explicit}. The same argument applied to $\omega_{0,4}$ tells us that $(r_\mu,s_\mu)_{\mu\neq\mu_-}$ have to satisfy \labelcref{item:threefracteq} for the correlator to be symmetric.\\
	Hence, we only need to argue that \labelcref{item:r_eq_pm1mods}--\labelcref{item:threefracteq} involving the component $\tilde{C}_{\mu_-}$ are necessary to impose. However, condition \labelcref{item:r_eq_pm1mods} is satisfied by our convention that $r_{\mu_-} = 1 \,\,{\rm mod}\,\, s_{\mu_-}$. Regarding property \labelcref{item:eq_fract_constr}, suppose there is a $\mu\neq\mu_-$ for which $s_\mu>2$, $r_{\mu} = 1 \,\,{\rm mod}\,\, s_{\mu}$, and $\big\lfloor\frac{r_{\mu}}{s_{\mu}}\big\rfloor=0$. Then of course we must already have $r_{\mu_2}=1$. In this case the correlator
	\begin{align}
		\frac{\omega_{0,3}\big(\begin{smallmatrix} \mu & \mu & \mu_- \\ z_1 & z_2 & z_3 \end{smallmatrix}\big)}{\dd z_1 \dd z_2 \dd z_3} &= \frac{1}{\dd z_1}\mathop{{\rm Res}}_{z = \mu} K_{\mu}\Big(\begin{smallmatrix} \mu & \mu & \mu_- \\ z_1 & z & z \end{smallmatrix}\Big)\, \sum_{k_2,k_3 \geq 1} k_2k_3\, z^{k_2+ k_3 -2} z_2^{-k_2-1} z_3^{-k_3-1}\, (\dd z)^2 \nonumber\\
		&= -\frac{1}{t_\mu} \sum_{k_1,k_2,k_3 \geq 1} k_2k_3\,  z_1^{-k_1-1} z_2^{-k_2-1} z_3^{-k_3-1}\, \delta_{k_1+k_2+k_3 , s_\mu}\,, \label{eq:omega03_1exc}
	\end{align}
	where we used that
	\begin{equation*}
		K_{\mu}\Big(\begin{smallmatrix} \mu & \mu & \mu_- \\ z_1 & z & z \end{smallmatrix}\Big) = - \frac{\dd z_1}{t_\mu \dd z}\sum_{k_1 \geq 1} z^{k_1+1-s_\mu} z_1^{-k_1-1}\,.
	\end{equation*}
	This is clearly non-zero unlike $\omega_{0,3}\begin{psmallmatrix} \mu & \mu_- &\mu \\ z_1 & z_3 & z_2 \end{psmallmatrix}=0$ and therefore we see that \labelcref{item:eq_fract_constr} must indeed be satisfied for all $(r_\mu,s_\mu)_{\mu\in\tilde{\mathfrak{a}}}$ for $\omega_{0,3}$ to be symmetric.\\
	Regarding \labelcref{item:threefracteq} we proceed similarly. Suppose there are $\mu\neq \nu$ distinct from $\mu_-$ with $s_\mu,s_\nu>1$ and $\big\lfloor\frac{r_{\mu}}{s_{\mu}}\big\rfloor=\big\lfloor\frac{r_{\nu}}{s_{\nu}}\big\rfloor=0$. Then \labelcref{item:r_eq_pm1mods} and \labelcref{item:eq_fract_constr} only leave us with $r_\mu=r_\nu=1$ and $s_\mu=s_\nu=2$. We compute
	\begin{align*}
		\frac{\omega_{0,4}\big(\begin{smallmatrix} \mu & \mu & \nu & \mu_- \\ z_1 & z_2 & z_3 & z_4 \end{smallmatrix}\big)}{\dd z_1 \dd z_2 \dd z_3 \dd z_4} &= \frac{1}{\dd z_1}\mathop{{\rm Res}}_{z = \mu} K_{\mu}\Big(\begin{smallmatrix} \mu & \mu & \nu & \mu_- \\ z_1 & z & z & z \end{smallmatrix}\Big)\, \sum_{k_2,k_3,k_4 \geq 1} k_2k_3k_4\, z^{k_2+ k_3 + k_4 -3} z_2^{-k_2-1} z_3^{-k_3-1} z_4^{-k_4-1}\, (\dd z)^3 \nonumber\\
		&= -\frac{1}{t_\mu(t_\mu-t_\nu)} \frac{1}{z_1^{2} z_2^{2} z_3^{2} z_4^{2}}
	\end{align*}
	which is non-vanishing where
	\begin{equation*}
		K_{\mu}\Big(\begin{smallmatrix} \mu & \mu & \nu & \mu_- \\ z_1 & z & z & z \end{smallmatrix}\Big) = - \frac{\dd z_1}{(t_\mu - t_\nu) t_\mu (\dd z)^2} \sum_{k_1\geq 1} z^{k_1-2} z_1^{-k_1-1}\,.
	\end{equation*}
	Using that always $\omega_{0,4}\begin{psmallmatrix} \mu & \nu & \mu & \mu_- \\ z_1 & z_3 & z_2 & z_4 \end{psmallmatrix}=0$ this shows that \labelcref{item:eq_fract_constr} is indeed necessary to impose.

	Now let us turn to $\omega_{\frac{1}{2},2}$. In the following set $\omega_{\frac{1}{2},2}'\coloneqq\omega_{\frac{1}{2},2}|_{Q_{\mu_-=0}}$, i.e.\ with $\omega_{\frac{1}{2},2}'$ we denote the restriction of $\omega_{\frac{1}{2},2}$ to the non-exceptional components. Then for distinct $\mu,\nu\in\tilde{\mathfrak{a}}\setminus\{\mu_-\}$ we have
	\begin{equation*}
	\omega_{\frac{1}{2},2}\big(\begin{smallmatrix} \mu & \nu \\ z_1 & z_2 \end{smallmatrix}\big) = \omega_{\frac{1}{2},2}'\big(\begin{smallmatrix} \mu & \nu \\ z_1 & z_2 \end{smallmatrix}\big)
	\end{equation*}
	as the bidifferential is not mixing the components. Thus, from the analysis done in \cref{prop:symm_cond_omega1half2} we deduce that the condition ensuring the symmetry of this correlator is \labelcref{item:1half2cond4} and in the symmetric case it may be evaluated using formula \eqref{eq:omega1half2offdiag}. However, if the arguments lie on the same component $\mu\neq\mu_-$ we obtain an additional contribution
	\begin{align*}
	\omega_{\frac{1}{2},2}\big(\begin{smallmatrix} \mu & \mu \\ z_1 & z_2 \end{smallmatrix}\big) &= \omega_{\frac{1}{2},2}'\big(\begin{smallmatrix} \mu & \mu \\ z_1 & z_2 \end{smallmatrix}\big) + r_\mu Q_{\mu_-} \mathop{{\rm Res}}_{z = \mu} K_{\mu}\Big(\begin{smallmatrix} \mu & \mu & \mu_- \\ z_1 & z & z^{r_{\mu}} \end{smallmatrix}\Big)\, \sum_{k_2 \geq 1} k_2 \, z^{k_2 -2} z_2^{-k_2-1}\, (\dd z)^2 \dd z_2\\
	&= \omega_{\frac{1}{2},2}'\big(\begin{smallmatrix} \mu & \mu \\ z_1 & z_2 \end{smallmatrix}\big) - \sum_{k_1,k_2 \geq 1} \frac{\dd z_1 \dd z_2}{z_1^{k_1 + 1}z_2^{k_2 + 1}} \,  k_2 \frac{Q_{\mu_-}}{t_\mu}\, \delta_{k_1+k_2, s_\mu}
	\end{align*}
	coming from the first term in the bracket in \eqref{thefirstg}. This is symmetric for $s_\mu\leq 2$ while for $s_\mu>3$ we may use that $\omega_{\frac{1}{2},2}'\big(\begin{smallmatrix} \mu & \mu \\ z_1 & z_2 \end{smallmatrix}\big)$ is computed via \eqref{eq:F1half2_mumu_sgr2} in order to get
	\begin{equation*}
	\begin{split}
	&\omega_{\frac{1}{2},2}\big(\begin{smallmatrix} \mu & \mu \\ z_1 & z_2 \end{smallmatrix}\big) = \sum_{k_1,k_2 \geq 1} \frac{\dd z_1 \dd z_2}{z_1^{k_1 + 1}z_2^{k_2 + 1}} \Bigg\{-\delta_{k_1 + k_2,s_{\mu}}\, k_2 \,\bigg(\frac{Q_{\mu} (r_{\mu} - 1 - \ell_{\mu}(k_2))}{r_{\mu} t_{\mu}} + \sum_{\substack{\nu\in\tilde{\mathfrak{a}}\setminus\{\mu\} \\ \frac{r_{\mu}}{s_{\mu}} > \frac{r_{\nu}}{s_{\nu}}}} \frac{Q_{\nu}}{t_{\mu}}\bigg) \\
	& \qquad\qquad\qquad\qquad\qquad\qquad\quad + \frac{k_2}{t_\mu}\sum_{\ell' > 0} \bigg(-\sum_{\substack{ \nu\in\tilde{\mathfrak{a}}\setminus\{\mu,\mu_-\} \\ \frac{r_{\mu}}{s_{\mu}} > \frac{r_{\nu}}{s_{\nu}}}} \!\! \left(\frac{t_{\nu}}{t_{\mu}}\right)^{r_\nu \ell'}\!\! Q_{\nu} \delta_{k_1 + k_2 + \ell' (r_{\mu} s_{\nu} - r_{\nu} s_{\mu}) , s_{\mu}} \\
	&\qquad\qquad\qquad\qquad\qquad\qquad\qquad\qquad\quad\:\:+\sum_{\substack{\nu\in\tilde{\mathfrak{a}}\setminus\{\mu,\mu_-\} \\ \frac{r_{\mu}}{s_{\mu}} < \frac{r_{\nu}}{s_{\nu}}}} \! \left(\frac{t_{\mu}}{t_{\nu}}\right)^{r_\nu \ell'} \!\! Q_{\nu} \delta_{k_1 + k_2 - \ell' (r_{\mu} s_{\nu} - r_{\nu} s_{\mu}), s_{\mu}} \bigg) \Bigg\}\,.
	\end{split} 
	\end{equation*}
	Comparing \eqref{eq:F1half2_mumu_sgr2} with the above equation and following the characterisation of the symmetry of \eqref{eq:F1half2_mumu_sgr2} it should be clear that the condition for the symmetry of $\omega_{\frac{1}{2},2}\big(\begin{smallmatrix} \mu & \mu \\ z_1 & z_2 \end{smallmatrix}\big)$ above is encoded in \labelcref{item:1half2cond1} and \labelcref{item:1half2cond2}.
	
	Now let us turn to the computation of correlators with arguments lying on the exceptional component. It is a straightforward calculation to find that
	\begin{equation}
	\label{eq:F1half2_mumu_exc}
	\omega_{\frac{1}{2},2}\big(\begin{smallmatrix} \mu_- & \mu_- \\ z_1 & z_2 \end{smallmatrix}\big) = \frac{\dd z_1 \dd z_2}{z_1^{2}z_2^{2}} \!\! \sum_{\substack{\nu\neq\mu_- \\ (r_\nu,s_\nu) = (1,2)}} \!\! \frac{Q_\mu}{t_\mu}\,,
	\end{equation}
	which is always symmetric. Notice at this point that this is in accordance with \labelcref{item:1half2cond1} and \labelcref{item:1half2cond2} as the two conditions do not imply any constraints for $\mu_-$. Now let us briefly argue that formula \eqref{eq:F1half2_mumu_sym1} exactly produces the result for $\omega_{\frac{1}{2},2}\big(\begin{smallmatrix} \mu_- & \mu_- \\ z_1 & z_2 \end{smallmatrix}\big)$ we obtained above. Since $r_{\mu_-}'=\bigl\lfloor\frac{r_{\mu_-}}{s_{\mu_-}}\bigr\rfloor = 0$ the first term in \eqref{eq:F1half2_mumu_sym1} vanishes. The two terms after are empty sums, hence vanishing. Thus the last term in  \eqref{eq:F1half2_mumu_sym1} is the only one potentially leading to a non-zero contribution and indeed one finds that the sum condition in this sum coincides with the one in \eqref{eq:F1half2_mumu_exc}. We therefore deduce that formula \eqref{eq:F1half2_mumu_sym1} even applies for the exceptional case.
	
	It is also straightforward to compute
	\begin{equation}
	\label{eq:omega1half2_exc_cross}
	\omega_{\frac{1}{2},2}\big(\begin{smallmatrix} \mu & \mu_- \\ z_1 & z_2 \end{smallmatrix}\big) = -\delta_{r_\mu+1,s_\mu} \frac{Q_\mu}{t_\mu}\,\frac{\dd z_1 \dd z_2}{z_1^{2}z_2^{2}}, \qquad \omega_{\frac{1}{2},2}\big(\begin{smallmatrix} \mu & \mu_- \\ z_1 & z_2 \end{smallmatrix}\big) = \delta_{r_\mu+1,s_\mu} \frac{Q_{\mu_-}}{t_\mu}\,\frac{\dd z_1 \dd z_2}{z_1^{2}z_2^{2}}
	\end{equation}
	for $\mu\neq\mu_-$. We therefore deduce the symmetry condition that for all $\mu\neq\mu_-$ with $s_\mu=r_\mu+1$ necessarily $Q_{\mu}=-Q_{\mu_-}$. This condition is however covered in \labelcref{item:1half2cond4}. Notice moreover that formula \eqref{eq:omega1half2offdiag} applied to the case at hand indeed produces \eqref{eq:omega1half2_exc_cross}.
	
	At last, let us show that the statement of \cref{prop:symm_cond_omega04} is true in the exceptional case as well. For this notice that if we choose $\mu,\nu,\lambda\in\tilde{\mathfrak{a}}$ pairwise distinct with $ \frac{r_\mu}{s_\mu} = \frac{r_\nu}{s_\nu} = \frac{r_\lambda}{s_\lambda} $ then already $s_\mu<\infty$. Thus,  $\omega_{0,4}\big(\begin{smallmatrix} \mu & \mu & \nu & \lambda \\ z_1 & z_2 & z_3 & z_4 \end{smallmatrix}\big)$ gets the same contributions as in the standard case which means that the further discussion of the symmetry of this correlator must be in line with the proof of \cref{prop:symm_cond_omega04}.
\end{proof}

\medskip

\subsection{Necessary conditions for symmetry}
\label{secneccons}
\medskip

In \Cref{cor:genuszerosymclassification} we showed that the topological recursion formula outputs symmetric $(\omega_{0,n})_{n\geq 1}$ if and only if the considered spectral curve is admissible in genus $0$. Using \Cref{HASequivALEgequal0} this means that the associated differential operators $H_{\alpha;i,k}$, which we introduced in \Cref{Airytocurve}, admit a partition function solving the associated system of differential equations to leading order in $\hbar^{\frac{1}{2}}$ in the sense of equation \eqref{eq:Hik_diff_eq_const_ord}. This on the other hand implies that for the differential operators of the type considered in \Cref{sec:W_gl_Airy_structs_all_shift} --- which are exactly those corresponding to spectral curves of the form \eqref{eq:monoSC} --- the reverse direction of \Cref{thm:genus0_soln_admissible_Hik} is true as well. This is exactly the statement of \Cref{prop:gen_zero_nec=suff} which we have hereby proven. Note that by \Cref{prop:symm_cond_exc} the statement of \Cref{cor:genuszerosymclassification} holds in the exceptional case too.\par

The next natural step is to investigate the necessary conditions for the operators $H_{\alpha;i,k}$ to be an Airy structure. Suppose they are, then we know from \Cref{mainth2} that all $\omega_{g,n}$ have to be symmetric. So especially the symmetry conditions for $\omega_{\frac{1}{2},2}$ will impose further necessary conditions on the $H_{\alpha;i,k}$ to be an Airy Structure. We will analyse how they restrict the choice of $(r_\mu,s_\mu)_{\mu\in\mu\in\tilde{\mathfrak{a}}}$ in case we ask $\omega_{\frac{1}{2},2}$ to be symmetric for all $Q_{\mu}\in\mathbb{C}$, $\mu \in \tilde{\mathfrak{a}}$ only being subject to the constraint
\begin{equation}
	\label{Qcond}
	\sum_{\mu \in \tilde{\mathfrak{a}}} Q_{\mu} = 0\,.
\end{equation}
Before we begin our analysis let us remark that during our calculation in \cref{sec:stndcasecorr,sec:exccasecorr} we were always assuming that ${\rm gcd}(r_{\mu},s_{\mu}) = 1$, 
\begin{equation}
	\label{tcond}
	t_{\mu}^{r_\mu} \neq t_{\nu}^{r_\nu} \text{ whenever } (r_{\mu},s_{\mu}) = (r_{\nu},s_{\nu}) \text{ and } \mu \neq \nu\,,
\end{equation}
and that $s_\mu=\infty$ for at most one $\mu\in\tilde{\mathfrak{a}}$. However, we already argued in \cref{rem:t_cond,rem:generality_number_shifts} that these assumptions are indeed necessary so that the differential operators are an Airy structure. In order to be closer to the notation of \cref{sec:W_gl_Airy_structs_all_shift}, let us choose a lexicographic ordering $\mu:[d] \rightarrow \tilde{\mathfrak{a}}$ satisfying $\mu_{j}\prec\mu_{j+1}$.

\begin{proof}[Proof of \cref{propneccons}]
Due to \labelcref{item:r_eq_pm1mods} of \cref{prop:symm_cond_omega03} we know that necessarily $r_\nu = \pm 1 \,\,{\rm mod}\,\, s_\nu$ for all $\nu\in\tilde{\mathfrak{a}}$. If $s_{\mu_1} > 2$ then \labelcref{item:eq_fract_constr} tells us that we must have $\frac{r_{\mu_1}}{s_{\mu_1}}>\frac{r_{\mu_2}}{s_{\mu_2}}$. Therefore,
\begin{equation*}
\sum_{\substack{\nu\neq \mu_1 \\ \frac{r_{\mu_1}}{s_{\mu_1}} > \frac{r_\nu}{s_\nu}}} Q_\nu = \sum_{j=2}^d Q_{\mu_j}
\end{equation*}
and hence --- since we assume $\omega_{\frac{1}{2},2}$ to be symmetric for all choices of $Q_\mu$ satisfying \eqref{Qcond} --- condition \labelcref{item:1half2cond1} of \cref{prop:symm_cond_omega1half2} forbids that $r_{\mu_1} = 1 \,\,{\rm mod}\,\, s_{\mu_1}$ and $s_{\mu_1}>2$ as we can always choose the $Q$s such that $\sum_{j\in(d]} Q_{\mu_j} \neq 0$. Thus, we are left with $r_{\mu_1} = -1 \,\,{\rm mod}\,\, s_{\mu_1}$.\\
Now let us see which values the symmetry conditions allow $(r_{\mu_j},s_{\mu_j})$ to take in case $j\notin \{1,d\}$. First, let us assume $\frac{r_{\mu_1}}{s_{\mu_1}}=\frac{r_{\mu_j}}{s_{\mu_j}}$. Then by \labelcref{item:eq_fract_constr} necessarily $s_{\mu_j}\in\{1,2\}$. If however $\frac{r_{\mu_1}}{s_{\mu_1}}>\frac{r_{\mu_j}}{s_{\mu_j}}$ then condition \labelcref{item:1half2cond1} and \labelcref{item:1half2cond2} also force $s_{\mu_j}\in\{1,2\}$ using that we can choose the $Q$s arbitrary except that they must satisfy \eqref{Qcond}.

As one can use similar arguments in order to show that also $r_{\mu_d} = 1 \,\,{\rm mod}\,\, s_{\mu_d}$ we omit the further discussion of this case.\\
Now let us analyse the implications coming from condition \labelcref{item:1half2cond4} in case of generic $Q$s. If $d=2$ condition \labelcref{item:1half2cond4} is always satisfied since the requirement that $Q_\mu=-Q_\nu$ for $\mu\neq\nu$ is nothing but property \eqref{Qcond}. Therefore, let us assume that $d>2$. In this case \labelcref{item:1half2cond4} exactly forbids that $r_\mu=r_\nu$ and $s_\mu=s_\nu=2$ for $\mu\neq\nu$. Hence, taking into account the other constraints we have already derived for $(r_\mu,s_\mu)_{\mu\in\tilde{\mathfrak{a}}}$ we deduce that for $d>2$ necessarily $s_\mu=s_\nu=1$ whenever $(r_\mu,s_\mu)=(r_\nu,s_\nu)$ for some $\mu\neq\nu$.

Note that due to \cref{prop:symm_cond_exc} the above discussion covers both the standard and the exceptional case.
\end{proof}

\medskip

\part{Applications to intersection theory}
\label{part:appl}

We expect that the coefficients $F_{g,n}$ of the partition function for all basic Airy structures (those of \cref{sec:ninrign} and \cref{sec:W_gl_Airy_structs_all_shift}) and the $\omega_{g,n}$ of the corresponding topological recursion can be represented as integrals of distinguished cohomology classes on $\overline{\mathcal{M}}_{g,n}$ or its cousins. As soon as such a representation is known for one spectral curve admitting a single ramification point  and whose type belongs to a certain set, it is relatively easy to extend it to any spectral curve having ramification points whose type belong to this set. The reason is that the corresponding Airy structure is built from the basic Airy structures by direct sums, change of polarisations and further dilaton shifts, cf. \cref{sec:AllDilatonPolarize}. In terms of partition functions, this is sometimes called ``Givental decomposition''.

We develop this idea for two types for which the link to $\overline{\mathcal{M}}_{g,n}$ is already known:
\begin{itemize}
\item the type $(r,s) = (r,r + 1)$ is related to Witten $r$-spin theory. The $r = 2$ subcase is Eynard's formula \cite{Eyn11,Eyn14}, and by generalising it to any $r$ we answer a question of Shadrin to the first-named author.
\item the type $(r_1,s_1,r_2,s_2) = (r,r+1,1,\infty)$ is related to open $r$-spin intersection theory as discussed in \cref{sec:open}. 
\end{itemize}
The type $(r,s)$ for other $s$ is discussed in \cref{Ssmonfg} assuming the existence of a special class on $\overline{\mathcal{M}}_{g,n}$ which has only been constructed for $(r,s) = (2,1)$ so far \cite{Nor17,CheNor19}. It turns out that Laplace-type integrals play an important role in such representations, and we first study them in the preliminary \cref{SecLap}. The method is general: if in the future an enumerative interpretation is found for a larger set of types, it is rather automatic to follow the strategy at work in these two examples and extend our representations to any spectral curve having ramifications whose types belong to this larger set. 
\medskip
\section{Representation of correlators via intersection theory}

\label{Represe}

\medskip

\subsection{The Laplace isomorphisms}
\label{SecLap}
\medskip

Let $r \geq 1$. If $m \geq 1 - r$ is an integer, we introduce the $r$-fold factorial, either by induction
\begin{equation}
\label{rfoldfac} m!^{(r)} = \begin{cases} 1 &  \text{if } 1-r \leq m \leq 0 \\ m\cdot (m - r)!^{(r)} &  \text{if } m > 0 \end{cases}\,,
\end{equation}
or equivalently in terms of the Gamma function
\[
m!^{(r)} = \frac{r^{\frac{m}{r} + 1} \Gamma\big(\frac{m}{r} + 1\big)}{r^{\langle \frac{m}{r} \rangle} \Gamma\big(\big\langle \frac{m}{r} \big\rangle\big)}\,,
\]
where we have defined $\big\langle \tfrac{dr + a}{r} \big\rangle = \tfrac{a}{r}$ if $d \in \mathbb{Z}$ and $a \in [r]$.

\begin{definition}
\label{Laplacedef} We introduce two isomorphisms.
\begin{equation*}
\begin{split}
& \begin{array}{lllll}
\mathfrak{L}^{+} & : & \mathbb{C}\llbracket \zeta \rrbracket \dd \zeta & \longrightarrow & u^{\frac{1}{r}}\mathbb{C} \llbracket u^{\frac{1}{r}} \rrbracket \\
& & k\zeta^{k - 1}\dd \zeta & \longrightarrow & k!^{(r)}\,u^{\frac{k}{r}} \end{array}\,, \\
& \begin{array}{lllll} 
\mathfrak{L}^{-} & : & \mathbb{C}[\zeta^{-1}]\,\zeta^{-2}\dd \zeta & \longrightarrow & \epsilon^{\frac{1}{r}}\mathbb{C}[\epsilon^{\frac{1}{r}}] \\ & &  k!^{(r)}\,\zeta^{-(k + 1)}\dd \zeta  & \longmapsto & \epsilon^{\frac{k}{r}}  \end{array}\,.
\end{split}
\end{equation*}
Abusing notation slightly we also write $ \mf{L}^\pm $ for the maps extended to the domain $ \C (\!( \zeta )\!)\dd \zeta $ by defining them to be zero on any monomial not in the original domain of definition.
\end{definition}

The first map can be realised by integrating over paths from $0$ to $\infty$ in the $x$-plane.

\begin{lemma}
\label{LaplacePlus}
We have
\[
\mathfrak{L}^{+} = \sum_{j = 0}^{r - 1} \beth_j \int_{e^{\frac{{2\rm i}\pi j}{r}}\mathbb{R}_{+}} e^{-\frac{\zeta^r}{ur}} \cdot
\]
where the constants are
\begin{equation}
\label{bethformula}\beth_j \coloneqq \sum_{a = 1}^{r} \frac{e^{-\frac{2{\rm i}\pi j a}{r}}}{r^{\frac{a}{r}} \Gamma\big(\frac{a}{r}\big)}\,.
\end{equation}
\end{lemma}
\begin{proof}
Let $\beta_j$ be from $0$ to $\infty$ in the angular direction $e^{\frac{2{\rm i}\pi j}{r}}$. Consider integration along a formal combination of paths $\beta = \sum_{j = 0}^{r - 1} \beth_j\beta_j$ for some $\beth_j \in \mathbb{C}$. For $k = dr + a$ with $d \geq 0$ and $a \in [r]$, we compute
\begin{equation*}
\begin{split}
\int_{\beta} e^{-\frac{\zeta^r}{ur}}\,k\zeta^{k - 1}\dd \zeta & = \sum_{j = 0}^{r - 1} \beth_j\,e^{\frac{2{\rm i}\pi j a}{r}} (ur)^{\frac{k}{r}}\,\frac{k}{r} \int_{\mathbb{R}_{+}} e^{-\tilde{x}}\,\tilde{x}^{\frac{k}{r} - 1}\,\dd \tilde{x} \\
& = r\hat{\beth}_{a}\,u^{\frac{k}{r}}\,r^{\frac{k}{r}}\,\Gamma\big(\tfrac{k}{r} + 1\big) \\
& =  \hat{\beth}_{a}\,u^{\frac{k}{r}}\,r^{\frac{a}{r}}\, \Gamma\big(\tfrac{a}{r}\big)\,k!^{(r)}
\end{split}
\end{equation*}
with the change of variable $\tilde{x} = \zeta^r/ur$ and the discrete Fourier transform for $a \in \mathbb{Z}$
\[
\hat{\beth}_a = \frac{1}{r} \sum_{j = 0}^{r - 1} e^{\frac{2{\rm i}\pi ja}{r}} \beth_j \,.
\]
We get
\[
\int_{\rho} e^{-\zeta^r/ur}\,k\zeta^{k - 1}\dd \zeta = \mathfrak{L}_{+}[\zeta^{k - 1}\dd \zeta]
\]
provided we choose
\[
\hat{\beth}_{a} = \frac{1}{r^{\frac{a}{r}}\Gamma\big(\frac{a}{r}\big)}\,.
\]
This entails the result by inverse discrete Fourier transform.
\end{proof}

The second map can be realised by contour integration. To this end, let $\gamma$ be the Hankel contour giving the Gamma function representation (Figure~\ref{Hankelfig})
\[
\frac{1}{\Gamma(\alpha)} = \int_{\gamma} e^{x} x^{-\alpha}\dd x\qquad \alpha \in \mathbb{C}\,,
\] 
that is, $\gamma$ goes from $-\infty - {\rm i}0$ to $-{\rm i}0$, then round the origin to $+{\rm i}0$ and ends in $-\infty + {\rm i}0$. Under the branched covering $\zeta \mapsto x(\zeta) = \zeta^{r}$, this contour has $r$ different lifts $(\gamma_j)_{j = 1}^{r}$, which we label so that $\gamma_j$ comes from the asymptotic direction $e^{-\frac{{\rm i}\pi}{r}(2j + 1)}(+\infty + {\rm i}0)$, approaches the origin and then ends in the asymptotic direction $e^{-\frac{{\rm i}\pi}{r}(2j - 1)}(+\infty - {\rm i}0)$. These contours belong to the lattice of rank $(r - 1)$ of Lefschetz thimbles
\[
V \coloneqq H_1(\mathbb{C},\mathbb{S}_{M}^-;\mathbb{Z})\,,
\]
where $\mathbb{S}_{M}^- = \{z \in \mathbb{C} \,\, | \,\, {\rm Re}\,x < -M \}$ for some large $M > 0$. The homology class $\sum_{j = 1}^r \gamma_j$ is trivial and omitting one $\gamma_j$ we get a basis of $V$.
\begin{center}
\begin{figure}
\includegraphics[width=0.7\textwidth]{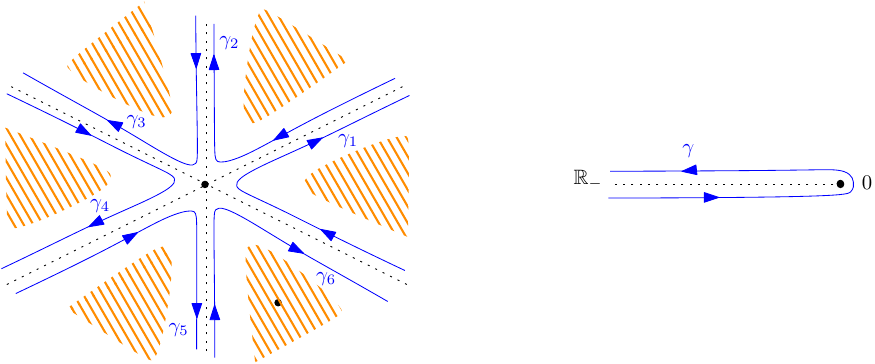}
\caption{\label{Hankelfig} Left panel: contours in the $\zeta$-plane for $r = 6$. The striped regions correspond to ${\rm Re}\,x > M$. Right panel: Hankel contour in the $x$ plane.}
\end{figure}
\end{center}

\begin{lemma}
\label{LaplaceMinus}
We have
\[
\mathfrak{L}^{-} = \sum_{j = 1}^{r} \frac{\beth_j}{2{\rm i}\pi r} \int_{\gamma_j} e^{\frac{\epsilon \zeta^r}{r}}\cdot 
\]
with the constants already appearing in \eqref{bethformula}.
\end{lemma}
\begin{proof}
Let $\epsilon > 0$, $j \in [r]$ and consider the integral
\begin{equation}
\label{Ibdef}
\mathcal{I}_j(\alpha) \coloneqq \frac{1}{2{\rm i}\pi} \int_{\gamma_j} e^{\frac{\epsilon\zeta^r}{r}} \zeta^{\alpha}\dd \zeta
\end{equation}  
for complex $\alpha \in \mathbb{C}$. Here, $\zeta \mapsto \zeta^{\alpha}$ is defined in the usual way as the analytic function on $\mathbb{C}\setminus \mathbb{R}_{-}$ such that for $\zeta \in \mathbb{R}_{-}^*$, we have
\[
\lim_{\epsilon \rightarrow 0^+} (\zeta \pm {\rm i}\epsilon)^{\alpha} =  e^{\pm {\rm i}\pi \alpha}|\zeta|^{\alpha}\,.
\]

Let us first consider ${\rm Re}\,\alpha > 0$.  In this case, as the integrand is regular we can squeeze the Hankel contour to the half-axes of angles $e^{-\frac{{\rm i}\pi(2k + 1)}{r}}$ for $k \in [r]$. After a change of variable $\tilde{x} = -\tfrac{\epsilon \zeta^r}{r}$, we find
\begin{equation}
\begin{split}
\label{Ibundef}
\mathcal{I}_j(\alpha) & = \bigg(\frac{\epsilon}{r}\bigg)^{-\frac{\alpha + 1}{r}}\,\big(-e^{-\frac{{\rm i}\pi (2j + 1)(\alpha + 1)}{r}} + e^{-\frac{{\rm i}\pi(2j - 1)(\alpha + 1)}{r}}\big) \cdot \frac{1}{r} \int_{\mathbb{R}_{+}} e^{-\tilde{x}}\tilde{x}^{\frac{\alpha + 1}{r} - 1}\,\dd\tilde{x} \\
& = \bigg(\frac{\epsilon}{r}\bigg)^{-\frac{\alpha + 1}{r}}\,e^{-\frac{2{\rm i}\pi j(\alpha + 1)}{r}} \cdot \frac{2{\rm i}}{r}\,\sin\bigg(\frac{\pi(\alpha + 1)}{r}\bigg)\,\Gamma\bigg(\frac{\alpha + 1}{r}\bigg)\,.
\end{split}
\end{equation}
We note that by definition in \eqref{Ibdef}, $\mathcal{I}_j(\alpha)$ is an entire function of $\alpha \in \mathbb{C}$. This is also true for the right-hand side of \eqref{Ibundef}: the Gamma function has simple poles when $\tfrac{\alpha + 1}{r} \in 2\mathbb{Z}_{\leq 0}$ which are compensated by a zero in the prefactor. Then, by analytic continuation the identity \eqref{Ibundef} remains true for all $\alpha \in \mathbb{C}$. 

We apply it to $\alpha = -(k + 1)$ where $k$ is a positive integer that we decompose as $k = rd + a$ with $a \in [r]$ and $d \geq 0$. Then
\begin{equation*}
\begin{split} 
\mathcal{I}_j(-(k + 1)) & = \bigg(\frac{\epsilon}{r}\bigg)^{\frac{k}{r}}\,e^{\frac{2{\rm i}\pi jk}{r}}\,\frac{2{\rm i}}{r}\,\sin\bigg(-\frac{\pi j k}{r}\bigg)\,\Gamma\bigg(-\frac{k}{r}\bigg) \\
& = \epsilon^{\frac{k}{r}}\,\frac{2{\rm i}\pi}{r}\,e^{\frac{2{\rm i}\pi ja}{r}}\,\frac{1}{r^{\frac{k}{r}}\Gamma\big(\frac{k}{r} + 1\big)} \\
& = \epsilon^{\frac{k}{r}}\,e^{\frac{2{\rm i}\pi ja}{r}}\,\frac{2{\rm i}\pi}{r^{\frac{a}{r}}\,\Gamma\big(\tfrac{a}{r}\big)\,k!^{(r)}}\,.
\end{split}
\end{equation*}
Coming back to the definition \eqref{bethformula} of $\beth_j$, we get
\[
\sum_{j = 1}^{r} \frac{\beth_j}{2{\rm i}\pi r} \int_{\gamma_j} e^{\frac{\epsilon x}{r}}\,\frac{k!^{(r)}\dd \zeta}{\zeta^{k + 1}} = \mathfrak{L}^{-}\bigg[\frac{k!^{(r)}\,\dd\zeta}{\zeta^{k + 1}}\bigg]\,.
\]
\end{proof}

\medskip

\subsection{Laplace transform on curves}
\label{Laplacesection}
\medskip

\subsubsection{Total Laplace transform}
\label{TotalLaplace}

\medskip

Let $C$ be a curve with normalisation $\pi\,:\,\tilde{C} \rightarrow C$, equipped with a meromorphic $1$-form $\dd x$. We shall rely on the notations introduced in \cref{SecTR}. We do not require that $\tilde{C}$ comes from a spectral curve in the sense of \cref{defspc}, neither that $\dd \tilde{x}$ is the differential of a meromorphic function on $\tilde{C}$. In particular, $\dd \tilde{x}$ could have simple poles, in which case a local primitive $\tilde{x}$ has logarithmic singularities and would be multivalued on $\tilde{C}$.

Recall that we write $\dd\tilde{x} = \pi^*\dd x$, that $\mathfrak{a}$ is the set of zeroes of $\dd x$ (including singular points) and $\tilde{\mathfrak{a}} = \pi^{-1}(\mathfrak{a})$. Denote the order of a zero of $\dd x$ at $\mu \in \tilde{\mathfrak{a}}$ by $r_{\mu} - 1$ (this could be zero if $\pi(a)$ is singular). Around each $\mu \in \tilde{\mathfrak{a}}$, we have a local coordinate such that
\[
\int_{\mu}^{\cdot} \dd \tilde{x} = \zeta^{r_{\mu}}\,.
\]

We define the vector spaces
\[
V = \bigoplus_{\mu \in \tilde{\mf{a}}} V_{\mu}\,,\qquad V_{\mu} \coloneqq \bigoplus_{l = 1}^{r_{\mu}} \mathbb{C}.e_{\mu,l}\,,
\]
and equip $V$ with the pairing
\begin{equation}
\label{pairingwithr} \eta(e_{\mu,l} \otimes e_{\nu,m}) = \frac{\delta_{\mu,\nu} \delta_{l+m,r_\mu}}{r_{\mu}}\,.
\end{equation}
\begin{remark} The fact that $e_{\mu,r_\mu}$ are null vectors for $\eta$  is a convention: it simplifies the formulas in \cref{higheret} but has no effect elsewhere. The factor $r_{\mu}$ is a choice that will simplify formulas in \cref{higheret} and \cref{ELSVwgnsec} but has no effect elsewhere.
\end{remark}
When needed, we shall decompose integers $k \in \mathbb{Z}$ as
\[
k = \overline{k}r_{\mu} + \hat{k},\qquad \hat{k} \in [r_{\mu}]
\]
and the index $\mu \in \tilde{\mf{a}}$ that one should use will be clear from the context. In order to get rid of fractional powers of the Laplace variable, we introduce the isomorphism
\[
\begin{array}{llccc}
E_{\mu}^*& : & \epsilon^{\frac{1}{r_{\mu}}}\mathbb{C}[\epsilon^{\frac{1}{r_{\mu}}}] &  \longrightarrow & V_{\mu}^*[\epsilon] \\ & & \epsilon^{\frac{k}{r}} & \longmapsto & e_{\mu,\hat{k}}^*\,\epsilon^{\overline{k}} \end{array}\,.
\]
where $e_{\mu,l}^*$ is the dual basis with respect to the pairing $\eta$. Importing \cref{Laplacedef}, for each $\mu \in \tilde{\mf{a}}$, we consider the local Laplace map
\[
\mf{L}^-_{\mu} = E_{\mu} \circ \mf{L}^{-} \circ \mathsf{Loc}_{\mu}\,:\,H^0\big(\tilde{C},K_{\tilde{C}}(*\tilde{\mf{a}})\big) \to V_{\mu}^*[\epsilon ]
\]
and the total map
\[
\mf{L}^-_{{\rm tot}} = \bigg(\bigoplus_{\mu \in \tilde{\mf{a}}} \mf{L}^{-}_{\mu}\bigg)\,:\,H^0\big(\tilde{C},K_{\tilde{C}}(*\tilde{\mf{a}})\big) \rightarrow V^*[\epsilon]\,.
\]
We define in a similar way
\begin{equation*}
\begin{split}
E_\mu & \,:\, u^{\frac{1}{r_\mu}}\mathbb{C}[\![\epsilon^{\frac{1}{r_{\mu}}}]\!] \longrightarrow V_{\mu}[\![u]\!]\,, \\
\mf{L}^+_{\mu} & \,:\, H^0\big(\tilde{C},K_{\tilde{C}}(*\tilde{\mf{a}})\big) \longrightarrow V_{\mu}[\![u]\!]\,, \\
\mf{L}^+_{{\rm tot}} & \,:\, H^0\big(\tilde{C},K_{\tilde{C}}(*\tilde{\mf{a}})\big) \longrightarrow V[\![u]\!]\,.
\end{split}
\end{equation*}
where the role of $e_{\mu,k}^*$ is now played by $e_{\mu,k}$.

\begin{remark}
\label{ConventionsLaplaceDuals}
From \cref{LaplacePlus,LaplaceMinus}, we see that the natural variables Laplace dual to $ x = \zeta^r$ are not $ u$ and $ \epsilon^{-1}$, but $ ru$ and $ r\epsilon^{-1}$. Moreover, by a different choice of constants (replacing multiple factorials by values of the $\Gamma $ function at rational numbers) the transforms could have been given by a single integral. We give the Laplace transform in this way to conform to conventions in the literature using $r$-fold factorials for the case of a smooth spectral curve, and its relation to the Witten $r$-spin class.
\end{remark}

\medskip

\subsubsection{Two generating series}

\medskip

Assume that we are given a holomorphic $1$-form $\omega_{0,1}$ in a neighboorhood of $\tilde{\mathfrak{a}}$ in $\tilde{C}$.

\begin{definition}
\label{Tudefn}We introduce $\mathsf{T}(u) \in V[\![u]\!]$ given by the formula
\begin{equation*}
\begin{split}
\mathsf{T}(u) & \coloneqq \sum_{\mu \in \tilde{\mathfrak{a}}} \mathsf{T}_{\mu,k} e_{\mu,\hat{k}} \,u^{\overline{k}}\\
&\coloneqq  \bigg(\sum_{\substack{\mu \in \tilde{\mathfrak{a}} \\ s_{\mu} \neq \infty}} e_{\mu,\hat{s}_{\mu}}u^{\overline{s}_{\mu}}\bigg) + \mathfrak{L}_{{\rm tot}}^+[\omega_{0,1}](u) \\ 
& = \sum_{\substack{\mu \in \tilde{\mathfrak{a}} \\ s_{\mu} \neq \infty}} \bigg(e_{\mu,\hat{s}_{\mu}} u^{\overline{s}_{\mu}} + \sum_{k > 0} (k - r_{\mu})!^{(r_{\mu})} F_{0,1}\big[\begin{smallmatrix} \mu \\ -k \end{smallmatrix}\big]\,e_{\mu,\hat{k}}\,u^{\overline{k}}\bigg) \,.
\end{split}
\end{equation*}
\end{definition}
 
Assume that we are given a fundamental bidifferential of the second kind $\omega_{0,2}$ on the smooth curve $\tilde{C}$.

\begin{definition}
\label{Buvdef} We introduce $\mathsf{B}(u,v) \in V^{\otimes 2}[\![u,v]\!]$, given by the formula
\begin{equation}
\label{Bu1u2} \begin{split}
\mathsf{B}(u,v) & \coloneqq (\mathfrak{L}^+_{{\rm tot}})^{\otimes 2}\big(\omega_{0,2} - \omega_{0,2}^{{\rm std}}) \\
& = \sum_{\substack{\mu,\nu \in \mathfrak{a} \\ k,l > 0}} (k - r_{\mu})!^{(r_{\mu})}\,(l - r_{\nu})!^{(r_{\nu})}\,F_{0,2}\big[\begin{smallmatrix} \mu & \nu \\ -k & -l \end{smallmatrix}\big]\,e_{\mu,\hat{k}} \otimes e_{\nu,\hat{l}}\,u^{\overline{k}} v^{\overline{l}}\,,
\end{split}
\end{equation}
where the second line follows from the decomposition \eqref{locun}.
\end{definition}

These definitions in particular apply to admissible spectral curves equipped with a fundamental bidifferential of the second kind, but make sense in this greater generality. Their relevance will become clear in \cref{theregucnu}.

\medskip

\subsubsection{Factorisation property for $\mathsf{B}$}
\label{higheret}

\medskip

We prove in this section a factorisation property for $\mathsf{B}$ when $\tilde{C}$ is smooth compact connected curve and $\dd \tilde{x}$ is a meromorphic $1$-form. Here we only need the data of a smooth compact curve $\tilde{C}$, of a meromorphic $1$-form that we denote $\dd \tilde{x}$, and of a fundamental bidifferential of the second kind on $\tilde{C}$, which defines $\mathsf{B}(u,v)$ through \cref{Buvdef}. We do not need the full data of a spectral curve, nor $\tilde{C}$ to be the normalisation of a curve. In light of \cref{ConventionsLaplaceDuals}, and for this section only, the `right' Laplace variables $ r u$ and $ r v$ (where $ r$ depends on the branch point) will play a role, so that the result will be best formulated for the following object.
\begin{definition}
Let $\mathsf{B}_{\mu,\nu} (u,v) $ be the projection of $\mathsf{B}(u,v)$ into $ (V_\mu \otimes V_\nu )[\![u,v]\!]$ and define
\begin{equation*}
\bar{\mathsf{B}}(u,v) \coloneqq \sum_{\mu, \nu \in \tilde{\mf{a}}} \mathsf{B}_{\mu,\nu}(u/r_\mu ,v/r_\nu )\,.
\end{equation*}
\end{definition}

\begin{proposition}
\label{compactfacto} Assume that $\tilde{C}$ is a smooth compact connected curve, $\dd \tilde{x}$ is a meromorphic $1$- form meromorphic, and $\omega_{0,2}$ is a fundamental bidifferential of the second kind on $\tilde{C}$. Then
\begin{equation}
\label{thefirstfacto}
\bar{\mathsf{B}}(u,v) = \frac{1}{u + v}\Big(u \bar{\mathsf{B}}(u,0) + v \bar{\mathsf{B}}(0,v) - u\,v\,\bar{\mathsf{B}}(u,0) \star \bar{\mathsf{B}}(0,v)\Big)\,,
\end{equation}
where $A \star B = (\id \otimes \eta \otimes \id)(A \otimes B)$. Besides, we have the compatibility relation
\begin{equation}
\label{thesecondfacto}
\bar{\mathsf{B}}(u,0) - \bar{\mathsf{B}}(0,-u) +u \bar{\mathsf{B}}(u,0) \star \bar{\mathsf{B}}(0,-u) =0 \,.
\end{equation}
\end{proposition} 
\begin{remark}
Such a property appeared in the case where $C$ is smooth and $\dd x$ has simple zeroes in \cite[Appendix B]{Eyn14}.
\end{remark}
\begin{proof}
We have introduced in \eqref{xinegdef} the family of meromorphic $1$-forms
\begin{equation}
\label{renewdefx}\dd\xi_{-k}^{\mu}(z) = \mathop{{\rm Res}}_{z' = \mu} \bigg(\int_{\mu}^{z'} \omega_{0,2}(z,\cdot)\bigg)\frac{\dd \zeta (z')}{\zeta (z')^{k + 1}} \in H^0\Big(\tilde{C},K_{\tilde{C}}\big((k + 1)\mu\big)\Big)
\end{equation}
indexed by $\mu \in \tilde{\mathfrak{a}}$ and $k > 0$. We have for $\nu \in \tilde{\mf{a}}$
\begin{equation}
\label{regexpLoc}
\mathsf{Loc}_{\nu}(\dd\xi_{-k}^{\mu}) = \frac{\delta_{\mu,\nu}\dd\zeta}{\zeta^{k + 1}} + \sum_{l > 0} \frac{F_{0,2}\big[\begin{smallmatrix} \mu & \nu \\ -k & -l \end{smallmatrix}\big]}{k}\,\zeta^{l - 1}\dd \zeta \,.
\end{equation}
The idea of the proof is to derive a recursion for these forms using the action of $\dd\big(\tfrac{\cdot}{\dd \tilde{x}}\big)$ --- this is \cref{idemeromor} below. We will first prove it for the polar part near the ramification points, and use that $\tilde{C}$ is smooth and compact and $\dd \tilde{x}$ meromorphic on $\tilde{C}$ to get an equality of globally defined meromorphic forms. This implies a recursion for the regular part of the expansion \eqref{regexpLoc}, i.e. the coefficients $F_{0,2}$, and this will imply the desired relation for $\mathsf{B}(u_1,u_2)$.

Since $\dd \tilde{x} = r_\nu\zeta^{r_{\nu} - 1}\dd\zeta$ near the ramification point $\nu$, we have
\[
\mathsf{Loc}_{\nu}\bigg[-\dd\bigg(\frac{\dd\xi_{-k}^{\mu}}{\dd \tilde{x}}\bigg)\bigg] = \frac{k + r_{\nu}}{r_\nu}\,\frac{\delta_{\mu,\nu}\dd \zeta}{\zeta^{k + r_{\nu} + 1}} + \sum_{l = 1}^{r_{\nu} - 1} \frac{F_{0,2}\big[\begin{smallmatrix} \mu & \nu \\ -k & -l \end{smallmatrix}\big]}{k}\,\frac{r_{\nu} - l}{r_{\nu}} \frac{\dd \zeta}{\zeta^{r_{\nu} - l + 1}} + \mc{O}(\dd \zeta)\,.
\]
As $\dd \tilde{x}$ is meromorphic, the $1$-form
\begin{equation}
\label{rdsdix} \dd\bigg(\frac{\dd \xi_{-k}^{\mu}}{\dd \tilde{x}}\bigg) + \frac{k + r_{\mu}}{r_{\mu}}\dd\xi_{-(k + r_\mu)}^{\mu} + \sum_{\nu \in \tilde{\mf{a}}} \sum_{l = 1}^{r_{\nu} - 1} \frac{F_{0,2}\big[\begin{smallmatrix} \mu & \nu \\ -k & -l \end{smallmatrix}\big]}{k}\,\frac{r_{\nu} - l}{r_{\nu}}\,\dd\xi_{l-r_\nu}^{\nu}
\end{equation}
is holomorphic on $\tilde{C}$. Since $\tilde{C}$ is compact and smooth, $H_1(\tilde{C},\C)$ is a finite-dimensional symplectic space equipped with the intersection pairing, and
\[
\mathsf{K} \coloneqq \bigg\{\gamma \in H_1(\tilde{C},\C) \quad \bigg| \quad \int_{\gamma} \omega_{0,2}(\cdot,z) = 0 \bigg\}
\]
is a Lagrangian subspace. From the definition \eqref{renewdefx}, we see that integrating any $\dd\xi_{-k}^{\mu}$ for $k > 0$ and $\mu \in \tilde{\mf{a}}$ along a cycle in $\mathsf{K}$ gives zero. The same is true for the first term in \eqref{rdsdix} since it is an exact form. As the period map induces a non-degenerate pairing $H^0(\tilde{C},K_{\tilde{C}}) \otimes \mathsf{K} \to \C$ and \eqref{rdsdix} is sent to $0$, we deduce the identity between meromorphic forms
\begin{equation}
\label{idemeromor} \dd\bigg(\frac{\dd \xi_{-k}^{\mu}}{\dd \tilde{x}}\bigg) + \frac{k + r_{\mu}}{r_{\mu}}\,\dd\xi_{-(k + r_{\mu})}^{\mu} + \sum_{\nu \in \tilde{\mf{a}}} \sum_{l = 1}^{r_{\nu} - 1} \frac{F_{0,2}\big[\begin{smallmatrix} \mu & \nu \\ -k & -l \end{smallmatrix}\big]}{k}\,\frac{r_{\nu} - l}{r_{\nu}} \dd\xi_{l-r_\nu}^{\nu} = 0\,.
\end{equation} 

We now would like to apply the local Laplace transform $\mathfrak{L}^+_{\nu}[\cdot](v)$ to this relation. Recall that $\mathfrak{L}^+$ by definition kills the polar part. So, by direct computation on the basis elements with \cref{Laplacedef}, we have
\[
\forall \phi \in H^0\big(\tilde{C},K_{\tilde{C}}(*\tilde{\mf{a}})\big)\qquad \mathfrak{L}^+_{\rho}\bigg[\dd\bigg(\frac{\phi}{\dd \tilde{x}}\bigg)\bigg](v) =  \frac{1}{r_{\rho}}\,\big[v^{-1}\mathfrak{L}^+_{\rho}[\phi](v)\big]_{+}\,,
\]
where $[\cdots]_{+}$ keeps only the nonnegative powers of $v$. The expansion \eqref{regexpLoc} implies for any meromorphic form 
\[
\mathfrak{L}^+_{\nu}[\dd\xi_{-k}^{\mu}](v) = \sum_{l > 0} \frac{F_{0,2}\big[\begin{smallmatrix} \mu & \nu \\ -k & -l \end{smallmatrix}\big]}{k}\,(l - r_{\nu})!^{(r_{\nu})}\,e_{\nu,\hat{l}}\,v^{\overline{l}}\,.
\]
By comparison with \eqref{Bu1u2}, $\mathsf{B}_{\mu,\rho}(u,v)$ can be obtained by the generating series
\[
\mathsf{B}_{\mu,\rho}(u,v) = \sum_{k > 0} k!^{(r_{\mu})} e_{\mu,\hat{k}} u^{\overline{k}} \otimes \mathfrak{L}^+_{\rho}[\dd\xi_{-k}^{\mu}](v)\,.
\]
Applying $\mathfrak{L}^+_{\rho}[\cdot](v)$ to \eqref{idemeromor}, multiplying by $k!^{(r_{\mu})}e_{\mu,\hat{k}} u^{\overline{k}}$, and summing over $k > 0$ then yields
\begin{equation*}
\begin{split}
 (r_\rho v)^{-1}(\mathsf{B}_{\mu,\rho}(u,v) - \mathsf{B}_{\mu,\rho}(u,0)\big) + (r_\mu u)^{-1}\big(\mathsf{B}_{\mu,\rho}(u,v) - \mathsf{B}_{\mu,\rho}(0,v)\big) & \\ + \sum_{\nu \in \tilde{\mf{a}}} \frac{1}{r_{\nu}} \sum_{l = 1}^{r_{\nu} - 1} (\id \otimes e_{\nu,l}^*)\big[\mathsf{B}_{\mu,\nu}(u,0)\big] \, (e^*_{\nu,r_{\nu} - l} \otimes \id \big)\big[\mathsf{B}_{\nu,\rho}(0,v)\big] & = 0\,.
\end{split}
\end{equation*}
Noticing that $\eta(e_{\nu,l} \otimes e_{\nu,r_{\nu} - m}) = r_{\nu}^{-1} \delta_{l,m}$ for $l,m \in [r_{\nu} - 1]$ and $\eta(e_{\mu,r_\mu},-) = 0$, this can be rewritten
\[
(r_\mu u + r_\rho v)\mathsf{B}_{\mu,\rho}(u,v)  - r_\mu u \mathsf{B}_{\mu,\rho}(u,0) - r_\rho v \mathsf{B}_{\mu,\rho}(0,v) + r_\mu r_\rho uv \sum_{\nu \in \tilde{\mf{a}}} (\id \otimes \eta \otimes \id)\big[\mathsf{B}_{\mu,\nu}(u,0) \otimes\mathsf{B}_{\nu,\rho}(0,v)\big]  = 0\,.
\]
Replacing $ u $ and $v$ by $ u/r_\mu$ and $v/r_\rho$, respectively, and summing over $\mu,\rho \in \tilde{\mf{a}}$, this implies the desired formula \eqref{thefirstfacto}. Since the left-hand side is a formal power series in $u,v$, it can be specialised at $v=-u$. This is possible on the right-hand side if and only if \eqref{thesecondfacto} is satisfied.
\end{proof}

For each $\mu \in \tilde{\mathfrak{a}}$, we have an orthogonal direct sum $V_{\mu} \cong {}^{{\rm s}}V_\mu \oplus {}^0V_\mu$ with
\[
{}^{{\rm s}}V_{\mu}  = \bigoplus_{l = 1}^{r_{\mu} - 1} \mathbb{C}.e_{\mu,l}\,,\qquad {}^0V_\mu = \mathbb{C}.e_{\mu,r_{\mu}}\,,
\]
and we decompose $V \cong {}^0V \oplus {}^{{\rm s}}V$ with
\[
{}^{\text{s}}V \coloneqq \bigoplus_{\mu \in \tilde{\mathfrak{a}}} {}^{{\rm s}}V_\mu\,,\qquad  {}^0V \coloneqq \bigoplus_{\mu \in \tilde{\mathfrak{a}}} {}^0V_\mu\,.
\]

\begin{definition} \label{Rudef} 
We introduce two generating series
\begin{equation*}
\begin{split}
\mathsf{R}(u) & \coloneqq \eta^{\bot} - u\, \bar{\mathsf{B}}(u,0) \in (V^{\text{s}})^{\otimes 2}\llbracket u\rrbracket\,,  \\
\mathsf{R}^{\bot}(u) & \coloneqq \eta^{\bot} - u\, \bar{\mathsf{B}}(0,u) \in (V^{\text{s}})^{\otimes 2}\llbracket u\rrbracket\,,
\end{split}
\end{equation*}
where $\eta^{\bot} \in ({}^{\text{s}}V)^{\otimes 2} \subset V^{\otimes 2}$, is induced by the pairing $\eta$, i.e.
\begin{equation}
\eta^\bot = \sum_{\mu \in \tilde{\mf{a}}} \Big( \sum_{l=1}^{r_\mu -1} r_{\mu} e_{\mu,l} \otimes e_{\mu, r_\mu-l} \Big)\,,
\end{equation}
We write $ ^{00}\mathsf{B}$, $^{0\text{s}}\mathsf{B}$, $ ^{\text{s}0}\mathsf{B}$, and $ ^{\text{ss}}\mathsf{B}$ for the projections of $\bar{\mathsf{B}}$ on the adequate subspaces in both arguments, and likewise for $\mathsf{R}$.
\end{definition}
\begin{corollary}
\label{compactfactospin}
Under the assumptions of \cref{compactfacto},  we have:
\begin{equation*}
\begin{split}
^{\textup{ss}}\bar{\mathsf{B}}(u,v) &= \frac{1}{u+v}\Big(\eta^\bot - {}^{\textup{ss}}\mathsf{R}(u) \star {}^{\textup{ss}}\mathsf{R}^\bot(v)\Big)\,,\\
^{\textup{s}0}\bar{\mathsf{B}}(u,v) &= \frac{-1}{u+v}\Big({}^{\textup{s}0}\mathsf{R}(u) + {}^{\textup{ss}}\mathsf{R}(u) \star {}^{\textup{s}0}\mathsf{R}^\bot(v)\Big)\,,\\
^{0\textup{s}}\bar{\mathsf{B}}(u,v) &= \frac{-1}{u+v}\Big({}^{0\textup{s}}\mathsf{R}^\bot(v) + {}^{0\textup{s}}\mathsf{R}(u) \star {}^{\textup{ss}}\mathsf{R}^\bot(v)\Big)\,,\\
^{00}\bar{\mathsf{B}}(u,v) &= \frac{-1}{u+v}\Big({}^{00}\mathsf{R}(u) + {}^{00}\mathsf{R}^\bot(v) +  {}^{0\textup{s}}\mathsf{R}(u) \star {}^{\textup{s}0}\mathsf{R}^\bot(v) \Big)\,,\\
\end{split}
\end{equation*}
and the following compatibility relations hold
\begin{equation*}
\begin{split}
{}^{\textup{ss}}\mathsf{R}(u) \star {}^{\textup{ss}}\mathsf{R}^{\bot}(-u) &= \eta^\bot\,, \\
{}^{\textup{ss}}\mathsf{R}(u) \star {}^{\textup{s}0}\mathsf{R}^\bot (-u) &= - {}^{\textup{s}0}\mathsf{R}(u)\,,\\
{}^{0\textup{s}}\mathsf{R}(u) \star {}^{\textup{ss}} \mathsf{R}^\bot(-u) &= - {}^{0\textup{s}}\mathsf{R}^\bot(-u) \,,\\
^{0\textup{s}} \mathsf{R}(u) \star ^{\textup{s}0}\mathsf{R}^\bot(-u) &= - {}^{00} \mathsf{R}(u) - {}^{00} \mathsf{R}^\bot(-u)\,.
\end{split}
\end{equation*}
\end{corollary}
The structure of the first line is familiar from Givental formalism and from \cite[Appendix B]{Eyn14}, but the other three are new.

\medskip

\subsubsection{Primary differentials and their descendants}
\label{primsec}

\medskip

\begin{definition} We consider a new family of meromorphic differentials on $\tilde{C}$, indexed by $\mu \in \tilde{\mathfrak{a}}$ and $k > 0$:
\begin{equation}
\label{changeba}
\dd\widehat{\xi}_{-k}^{\,\mu} := \mathcal{D}^{\overline{k}}(\dd \xi_{-\hat{k}}^{\mu}),\qquad {\rm where}\quad \mathcal{D}(\varphi) := -\dd\big(\varphi/\dd \tilde{x})\,.
\end{equation}
For $k \in [r_{\mu}]$ they are equal to $\dd \xi_{-k}^{\mu}$,  while the $k > r_{\mu}$ ones are called descendants. Borrowing the language of Frobenius manifolds, we call $(\dd\xi_{-k}^{\mu})_{\mu.k}$ the canonical basis and $(\dd \widehat{\xi}^{\,\mu}_{-k})_{\mu,k}$ the flat basis.
\end{definition}

The aim of this section is to compute the change of basis, from the canonical basis to the flat basis of differentials, when $\tilde{C}$ is smooth compact connected as in the previous section. This generalises computations done in \cite{Eyn14} for simple ramifications. As we shall see in \Cref{ELSVwgnsec}, the correspondence between topological recursion and intersection theory takes a particularly nice form in the $\widehat{\xi}$-basis.

By an easy recursion from \eqref{idemeromor}, passing from $\xi$ to $\widehat{\xi}$ is change of basis with triangular structure. Thus, there exist scalars $A\big[\begin{smallmatrix} \mu & \nu \\ -k & - l \end{smallmatrix}\big]$ indexed by $\mu,\rho \in \tilde{\mathfrak{a}}$ and $k,m> 0$, such that
\begin{equation}
\label{basishatxi}\dd\xi_{-k}^{\mu} = \sum_{\substack{\rho \in \tilde{\mathfrak{a}} \\ m > 0}} A\big[\begin{smallmatrix} \mu & \rho \\ -k & - m \end{smallmatrix}\big]\,\dd\widehat{\xi}^{\,\rho}_{-m}\,,
\end{equation}
where the sum on the right-hand side is finite.

\begin{definition}
We introduce:
\begin{equation*}
\begin{split}
\mathsf{A}_{\mu,\nu}(u,v) & = \bigg(\sum_{k,l > 0} A\big[\begin{smallmatrix} \mu & \nu \\ -k & - l \end{smallmatrix}\big] \, k!^{(r_{\mu})} e_{\mu,\hat{k}} \otimes e_{\nu,\hat{l}}\, u^{\overline{k}}v^{\overline{l}}\bigg) \in V_{\mu} \otimes V_{\nu} [\![u,v]\!]\,\qquad \mu,\nu \in \tilde{\mathfrak{a}}\,, \\
\bar{\mathsf{A}}(u,v) & = \bigg(\sum_{\mu,\nu \in \tilde{\mathfrak{a}}} \mathsf{A}_{\mu,\nu}(u/r_{\mu},v)\bigg) \in V^{\otimes 2}[\![u,v]\!]\,, \\
\eta^{\vee} & = \sum_{\mu \in \tilde{\mathfrak{a}}} \sum_{k = 1}^{r_{\mu}} k\,e_{\mu,k} \otimes e_{\mu,k}\,.
\end{split}
\end{equation*}
Note that in $\bar{\mathsf{A}}$ we do not rescale the variable $v$, unlike for $\bar{\mathsf{B}}$. Due to the definition of the primary differentials, we have $\bar{\mathsf{A}}(0,v) = \bar{\mathsf{A}}(0,0) = \eta^{\vee}$. We use as in Definition~\ref{Rudef} the notation ${}^{{\rm ss}}$ for the projection to  $(V^{{\rm s}})^{\otimes 2}$.
\end{definition}

\begin{proposition}
\label{ABuv} Assume that $\tilde{C}$ is a smooth compact connected curve, $\dd \tilde{x}$ is a meromorphic $1$-form on $\tilde{C}$ and $\omega_{0,2}$ is a fundamental bidifferential of the second kind on $\tilde{C}$. We have:
\[
\overline{\mathsf{A}}(u,v) = \frac{\eta^{\vee} - u \bar{\mathsf{B}}(u,0) \star \eta^{\vee}}{1 - uv}\,.
\]
In particular
\[
{}^{{\rm ss}}\overline{\mathsf{A}}(u,v) = \frac{\mathsf{R}(u) \star \eta^{\vee}}{1 - uv}\,.
\]
\end{proposition}
\begin{proof}
Inserting the change of basis \eqref{basishatxi} into the relation \eqref{idemeromor} and identifying the coefficient of $\dd\widehat{\xi}^{\rho}_{-m}$ yields, for any $\mu,\rho \in \tilde{\mathfrak{a}}$ and $k,m > 0$
\[
\mathsf{A}\big[\begin{smallmatrix} \mu & \rho \\ -k & -m + r_{\rho}\end{smallmatrix}\big] = \frac{k + r_{\mu}}{r_{\mu}}\, \mathsf{A}\big[\begin{smallmatrix} \mu & \rho \\ -(k + r_{\mu}) & - m \end{smallmatrix}\big] + F_{0,2}\big[\begin{smallmatrix} \mu & \rho \\ -k & m - r_{\rho} \end{smallmatrix}\big]\,\frac{m}{kr_{\rho}}\,\delta_{m < r_{\rho}}\,,
\]
where by convention $\mathsf{A}$ with nonnegative lower indices vanish. Multiplying by $k!^{(r_{\mu})}\,e_{\mu,\hat{k}} \otimes e_{\rho,\hat{m}} u^{\overline{k}}v^{\overline{m}}$ and summing over $k,m > 0$ we obtain
\begin{equation*}
\begin{split}
& \quad v\mathsf{A}_{\mu,\rho} (u,v) \\ 
& = (r_{\mu} u)^{-1}\big(\mathsf{A}_{\mu,\rho}(u,v) - \mathsf{A}_{\mu,\rho}(0,v)\big) + \sum_{k > 0} \sum_{m = 1}^{r_{\rho} - 1} F_{0,2}\big[\begin{smallmatrix} \mu & \rho \\ -k & m - r_{\rho} \end{smallmatrix}\big] (k - r_{\mu})!^{(r_{\mu})}\,\frac{m}{r_{\rho}}\,e_{\mu,\hat{k}} \otimes e_{\rho,m}\,u^{\overline{k}} \\
& = (r_{\mu} u)^{-1}\big(\mathsf{A}_{\mu,\rho}(u,v) - \mathsf{A}_{\mu,\rho}(0,v)\big) + \bigg(\sum_{k,l > 0} F_{0,2}\big[\begin{smallmatrix} \mu & \rho \\ -k & -l \end{smallmatrix}\big]\,(k - r_{\mu})!^{(r_{\mu})} (l - r_{\rho})!^{(r_{\rho})} e_{\nu,\hat{k}} \otimes e_{\rho,\hat{l}}\, u^{\overline{k}} 0^{\overline{l}} \bigg) \star  \eta^{\vee} \\
& = (r_{\mu} u)^{-1}\big(\mathsf{A}_{\mu,\rho}(u,v) - \mathsf{A}_{\mu,\rho}(0,v)\big) + \mathsf{B}_{\mu,\rho}(u,0) \star \eta^{\vee}\,,
\end{split}
\end{equation*}
where we have used $\eta(e_{\rho,l},e_{\nu,m}) = r_{\rho}^{-1}\delta_{\nu,\rho}\delta_{l + m,r_{\rho}}$, the fact that $(l - r_{\rho})!^{(r_{\rho})} = 1$ when $l < r_{\rho}$, which are the only terms contributing to the sum, and recognised $\mathsf{B}_{\mu,\rho}(u,0)$ from \cref{Buvdef}. We now make the substitution $u \mapsto u/r_{\mu}$, sum over $\mu,\rho \in \tilde{\mathfrak{a}}$ and use $\bar{\mathsf{A}}(0,v) = \eta^{\vee}$. This results in
\[
v\bar{\mathsf{A}}(u,v) = u^{-1}\big(\bar{\mathsf{A}}(u,v) - \eta^{\vee}\big) + \bar{\mathsf{B}}(u,0) \star \eta^{\vee}\,,
\]
whose solution is
\[
\bar{\mathsf{A}}(u,v) = \frac{\eta^{\vee} - u\bar{\mathsf{B}}(u,0) \star \eta^{\vee}}{1 - uv}\,.
\]
As the projection of $\eta^{\vee}$ onto $(V^{{\rm s}})^{\otimes 2}$ coincides with $\eta^{\bot} \star \eta^{\vee}$, we deduce
\[
{}^{{\rm ss}}\bar{\mathsf{A}}(u,v) = \frac{(\eta^{\bot} - u\bar{\mathsf{B}}(u,0)) \star \eta^{\vee}}{1 - uv}\,,
\]
where we recognise $\mathsf{R}(u)$ from \cref{Rudef}.
\end{proof}

\medskip
\subsection{Review of Witten \texorpdfstring{$r$}{r}-spin classes}
\label{revwit}
\medskip

For $r \geq 2$, we denote $w_{g,n}^{r{\rm spin}}(k_1,\ldots,k_n)$ the Witten $r$-spin class, first imagined by \cite{Wit93} and constructed in \cite{PoVa01,Chi06}, cf. also \cite{PPZ15}. For all integers $g,n \geq 0$ such that $2g - 2 + n > 0$ and $k_1,\ldots,k_n \in \mathbb{Z}$, this is a Chow class
\[
w^{r{\rm spin}}_{g,n}(k_1,\ldots,k_n) \in CH^*(\overline{\mathcal{M}}_{g,n})\,.
\]

It is defined via the moduli space of $r$-spin structures $\overline{\mathcal{M}}_{g;k_1,\ldots,k_n}^{r{\rm spin}}$ parametrizing pointed curves $(C,p_1,\ldots,p_n)$ with a line bundle $L \rightarrow C$ together with an isomorphism of $L^{\otimes r}$ to $K_C^{{\rm log}}\big(- \sum_{i = 1}^n k_i p_i\big)$. Denoting $\mathcal{C}$ the universal curve and $\mathcal{L}$ the universal line bundle, and considering the projection
\[
\mathcal{L} \rightarrow \mathcal{C} \mathop{\longrightarrow}^{\pi} \overline{\mathcal{M}}_{g;k_1,\ldots,k_n}^{r{\rm spin}} \mathop{\longrightarrow}^{p} \overline{\mathcal{M}}_{g,n}\,,
\]
the naive definition is
\[
w_{g,n}^{r{\rm spin}}(k_1,\ldots,k_n) = r^{-g}\,p_*c_{{\rm top}}\big((R^1\pi_*\mathcal{L})^{\vee}\big)\,.
\]
This works in genus $0$, but if $g > 0$, then $R^1\pi_*\mathcal{L}$ is not a vector bundle as $R^0\pi_*\mathcal{L}$ is non-zero, so the general construction is more involved. For positive $k_i$, Witten's class vanishes if one of the $k_i$ is divisible by $r$. It is concentrated in codimension\footnote{The notion of codimension for Chow classes refers to homology. Therefore, when the Chow class can be realised by a cohomology class, this codimension is twice the cohomological degree.}\begin{equation}
\label{dimrspin}
D = \frac{(r-2)(g-1) - n + \sum_{i = 1}^n k_i}{r}\,.
\end{equation}
In particular, the class vanishes unless the right-hand side of \eqref{dimrspin} is an integer, and its integration on $\overline{\mathcal{M}}_{g,n}$ vanishes unless $D = \dim \overline{\mc{M}}_{g,n} = 3g-3+n$.

We are primarily interested in indices ranging over $[r]$, or $[r - 1]$ since the class vanishes for index equal to $r$, but the following property, conjectured by Jarvis--Kimura--Vaintrob \cite[Descent axiom 1.9]{JKV01},  explains the appearance of the $r$-fold factorial in all our formulas, see also \cite[Lemma 4.2.8]{Chi08}.
\begin{lemma} \cite[Proposition 5.1]{PoVa01}
Let $g,n,k_1,\ldots,k_n \geq 0$ be integers such that $2g - 2 + n > 0$, and decompose $k_i = d_ir + a_i$ with $a_i \in [r]$ and $d_i \geq 0$. We have
\[
w_{g,n}^{r{\rm spin}}(k_1,\ldots,k_n) = w_{g,n}^{r{\rm spin}}(a_1,\ldots,a_n) \prod_{i = 1}^n r^{-d_i} (k - r)!^{(r)} \psi_i^{d_i}\,.
\]
\end{lemma}
This formula is consistent with the case $k \in [r]$ due to the initial condition in the definition \eqref{rfoldfac} of the $r$-fold factorial.

\begin{remark} Witten's class can be interpreted as a cohomological field theory in the following way: its vector space $V = V^{r\textup{spin}}$ has a basis $(e_i)_{i = 1}^{r - 1}$ and we define for $a_1,\ldots,a_n \in [r - 1]$
\[
W_{g,n}^{r{\rm spin}}(e_{a_1} \otimes \cdots \otimes e_{a_n}) \coloneq w_{g,n}^{r{\rm spin}}(a_1,\ldots,a_n) \in CH^*(\overline{\mathcal{M}}_{g,n})\,.
\]
$V^{r{\rm spin}}$ is equipped with product and multiplication
\[
\big(e_a\,\big|\,e_b\big)  = \frac{\delta_{a + b,r}}{r}\,,\qquad  \big(e_{a_1} \bullet e_{a_2}\,\big|\,e_{a_3}\big) = W_{0,3}(e_{a_1} \otimes e_{a_2} \otimes e_{a_3}) = \delta_{a_1 + a_2 + a_3,r + 1}\,.
\]
We included a factor $r^{-1}$ to align with \eqref{pairingwithr}. It could be removed by rescaling the product without changing the class $W_{g,n}^{r{\rm spin}}$ and without effect on the formulas we are going to write.
\end{remark}

We define the partition function of the $r$-spin theory by
\begin{equation*}
\begin{split}
& Z^{r{\rm spin}}\big[(x_k)_{k > 0}\big] \quad \\
& = \exp\left(\sum_{\substack{g,n \in \mathbb{Z}_{\geq 0} \\ 2g - 2 + n > 0}} \frac{\hbar^{g - 1}}{n!} \sum_{k_1,\ldots,k_n > 0} \bigg(\int_{\overline{\mathcal{M}}_{g,n}} w_{g,n}^{r{\rm spin}}(k_1,\ldots,k_n) \bigg)\prod_{i = 1}^n r^{\lfloor \frac{k_i}{r} \rfloor} k_i x_{k_i}\right) \\
& =  \exp\left(\sum_{\substack{g,n \in \mathbb{Z}_{\geq 0} \\ 2g - 2 + n > 0}} \frac{\hbar^{g - 1}}{n!} \sum_{\substack{a_1,\ldots,a_n \in [r - 1] \\ d_1,\ldots,d_n \geq 0}} \bigg(\int_{\overline{\mathcal{M}}_{g,n}} w_{g,n}^{r{\rm spin}}(a_1,\ldots,a_n) \prod_{i = 1}^n \psi_i^{d_i}\bigg)\prod_{i = 1}^n (d_ir + a_i)!^{(r)}\,x_{d_ir + a_i}\right)\,,
\end{split}
\end{equation*}
where we took into account the dimension constraint \eqref{dimrspin} to get the second line. Due to the aforementioned vanishing, it is independent of the times with indices divisible by $r$. $Z^{r{\rm spin}}$ was identified in \cite{FSZ10} with the tau function for the $r$-KdV hierarchy.  It is also known that $Z^{r{\rm spin}}$ satisfies $\mc{W}(\mathfrak{gl}_r)$-constraints --- see \cite{BBCCN18} for the history of this result. From there, Milanov proved in \cite{Mil16} that $Z^{r{\rm spin}}$ is the partition function of the Airy structure of \cref{thm:W_gl_Airy_Coxeter} with $(r,s) = (r,r + 1)$, as well as the topological recursion \`a la Bouchard--Eynard \cite{BoEy13} for the associated correlators:
\begin{equation}
\label{rspinw}
\omega_{g,n}^{r{\rm spin}}(\zeta_1,\ldots,\zeta_n) = \sum_{\substack{a_1,\ldots,a_n \in [r - 1] \\ d_1,\ldots,d_n \geq 0}} \bigg(\int_{\overline{\mathcal{M}}_{g,n}} w_{g,n}(a_1,\ldots,a_n) \prod_{i = 1}^n \psi_i^{d_i}\bigg) \prod_{i = 1}^n \frac{(d_ir + a_i)!^{(r)}\,\dd \zeta_i}{\zeta_i^{dr_i + a_i + 1}} \,.
\end{equation}
This result was apparently also obtained by Bouchard and Eynard in an unpublished draft, and appeared in \cite{DNOPS19} in a form closer to the one we state here.
\begin{theorem}
\label{rspinthmTR} The correlators $\boldsymbol{\omega}^{r{\rm spin}}$ are computed by the topological recursion for the spectral curve (without crosscap)
\begin{equation}
\label{zetarspin} x(\zeta) = \zeta^{r},\qquad y(\zeta) = -\frac{\zeta}{r}\,,\qquad \omega_{0,2}(\zeta_1,\zeta_2) = \frac{\dd \zeta_1 \dd \zeta_2}{(\zeta_1 - \zeta_2)^2}\,.
\end{equation}
\end{theorem}
\begin{proof} We start from \cite[Theorem 7.3]{DNOPS19} which shows that the topological recursion for the spectral curve
\begin{equation}
\label{zrspin} x(z) = z^r,\qquad y(z) = z,\qquad \omega_{0,2}(z_1,z_2) = \frac{\dd z_1\dd z_2}{(z_1 - z_2)^2}
\end{equation}
yields
\begin{equation}
\label{thefirstco}\sum_{\substack{a_1,\ldots,a_n \in [r - 1] \\ d_1,\ldots,d_n \geq 0}} (-r)^{2 - 2g - n} \bigg(\int_{\overline{\mathcal{M}}_{g,n}} w_{g,n}^{r{\rm spin}}(a_1,\ldots,a_n) \prod_{i = 1}^n \psi_i^{d_i}\bigg) \prod_{i = 1}^n \frac{(d_ir + a_i)!^{(r)}\,\dd z_i}{z_i^{dr_i + a_i + 1}}\,.
\end{equation}
The spectral curve \eqref{zetarspin} can be obtained from \eqref{zrspin} by multiplying $y$ by $-\frac{1}{r}$. This multiplies \eqref{thefirstco} by $(-r)^{2g -2+ n}$. So the factors of $r$ cancel and we indeed obtain the correlators $\omega_{g,n}^{r{\rm spin}}(\zeta_1,\ldots,\zeta_n)$.
\end{proof}

\medskip

\subsection{Deformations on Witten \texorpdfstring{$r$}{r}-spin classes}

\medskip

We now recall well-known actions on family of classes, originating from the work of Givental. As our focus is not on cohomological field theories, some of the actions we allow may not preserve this property and do not belong stricto sensu to the Givental group. See e.g. \cite{Shad08,Tel12,PPZ15} for more background.

\medskip
\subsubsection{Translations}

\medskip

\label{Sectrans}

Given a formal series $\mathsf{T}(u) \in uV^{r{\rm spin}}\llbracket u\rrbracket $, we can define a new family of Chow classes
\[
\big[\widehat{\mathsf{T}}\cdot W^{r{\rm spin}}\big]_{g,n}\colon (V^{r{\rm spin}})^{\otimes n} \rightarrow CH^*(\overline{\mathcal{M}}_{g,n})\,.
\]
We first decompose
\[
\mathsf{T}(u) = \sum_{\substack{d \geq 1 \\ a \in [r - 1]}} \mathsf{T}_{rd + a}\,e_{a}u^{d}
\]
and assume that $\mathsf{T}_{r + 1} \neq 1$. Then we introduce
\[
\tilde{\mathsf{T}}(u) \coloneqq \frac{\mathsf{T}(u) - \mathsf{T}_{r + 1}e_{1}u}{1 - \mathsf{T}_{r + 1}} =  \sum_{\substack{d \geq 1 \\ a \in [r - 1] \\ (d,a) \neq (1,1)}} \frac{\mathsf{T}_{rd + a}}{1 - \mathsf{T}_{r + 1}}\,e_{a}\,u^{d}\,,
\] 
and set
\[
\big[\hat{\mathsf{T}}\cdot W_{g,n}^{r{\rm spin}}\big](-) = \sum_{m \geq 0} \frac{1}{m!}\,(\pi_m)_*W_{g,n + m}^{r{\rm spin}}\big(- \otimes \tilde{\mathsf{T}}(\psi_{n + 1}) \otimes \dotsb \otimes \tilde{\mathsf{T}}(\psi_{n + m})\big)\,,
\]
where $\pi_m\,:\,\overline{\mathcal{M}}_{g,n + m} \rightarrow \overline{\mathcal{M}}_{g,n}$ is the forgetful morphism. This definition is well-posed, i.e. the sum has finitely many non-zero terms. Indeed, if we evaluate on $\bigotimes_{i=1}^n e_{a_i} $, the codimension (after pushforward) of the summand proportional to $\prod_{j = 1}^m \mathsf{T}_{d_jr + b_j}$ with $d_jr + b_j \geq r + 2$ is
\[
\frac{1}{r}\bigg((g-1)(r-2) -(n+m) + \sum_{i = 1}^n a_i + \sum_{j = 1}^m (d_jr + b_j)\bigg) -m \geq \frac{1}{r}\bigg((g-1)(r-2) -n + \sum_{i = 1}^n a_i  \bigg) + \frac{m}{r} \,,
\]
which for fixed $g,n,a_1,\ldots,a_n$ becomes larger than $ \dim \overline{\mc{M}}_{g,n} = 3g-3+n$ for  $m$ large enough, forcing this summand to vanish. The coefficient $\mathsf{T}_{r + 1}$ plays a special role, which reflects the dilaton equation
\[
(p_1)_*\big(w_{g,n + 1}^{r{\rm spin}}(a_1,\ldots,a_n,1)\cdot \psi_{n + 1}\big) = (2g - 2 + n)\,w_{g,n}^{r{\rm spin}}(a_1,\ldots,a_n)\,.
\]
The change from $ \mathsf{T}$ to $ \tilde{\mathsf{T}}$ reflects this.

\medskip

\subsubsection{Sums over stable graphs}

\medskip

\label{sumstab}
Let now $V$ be a finite-dimensional vector space and consider a family of classes
\[
\Omega_{g,n}\colon V^{\otimes n} \rightarrow CH^*(\overline{\mathcal{M}}_{g,n}),\qquad g,n \in \mathbb{Z}_{\geq 0},\quad 2g - 2 + n > 0\,.
\]
Given a symmetric formal power series $\mathsf{B}(u_1,u_2) \in V^{\otimes 2} \llbracket u_1,u_2 \rrbracket $, we can define a new such family $\big[\widehat{\mathsf{B}}\cdot \Omega_{g,n}\big]_{g,n}$ by sums over stable graphs.

Let $\mathsf{G}_{g,n}$ be the set of stable graphs of type $(g,n)$. For a vertex $\mathsf{v}$ in a stable graph, we denote $h(\mathsf{v})$ the genus and $k(\mathsf{v})$ the valency.
\[
\big[\widehat{\mathsf{B}}\cdot \Omega\big]_{g,n} = \sum_{\Gamma \in \mathsf{G}_{g,n}} \frac{1}{|{\rm Aut}\,\Gamma|}  \eta_{\Gamma}\,(\iota_{\Gamma})_* \bigg[\prod_{\mathsf{v} \in {\rm Vert}(\Gamma)} \Omega_{h(\mathsf{v}),k(\mathsf{v})} \prod_{\{\mathsf{e},\mathsf{e}'\} \in {\rm Edge}(\Gamma)} \mathsf{B}(\psi_{\mathsf{e}},\psi_{\mathsf{e}'})\bigg]\,,
\]
where $\iota_{\Gamma}\,:\,\prod_{v \in {\rm Vert}(\Gamma)} \overline{\mathcal{M}}_{h(v),k(v)} \rightarrow \overline{\mathcal{M}}_{g,n}$ is the natural inclusion of the boundary stratum associated to $\Gamma$. To read this formula, half-edges label the punctures on the curves whose moduli spaces sit at the vertices. So, we have $\psi$-classes $\psi_{\mathsf{e}},\psi_{\mathsf{e}'}$ associated to an edge $\{\mathsf{e},\mathsf{e}'\}$, and for each half-edge there is a copy of $V^*$ coming from the vertex it starts from, and a copy of $V$ coming from the contribution of the edge it belongs to. The symbol $\eta_{\Gamma}$ indicates that we pair them in the natural way. This definition is well-posed because the dimension of the moduli spaces at the vertices is smaller than the one of $\overline{\mathcal{M}}_{g,n}$, so that only finitely many powers of $\psi$-classes can contribute in the sum.

\begin{remark}
This differs slightly from the so-called $R$-action in Givental formalism, by the fact that we do not decorate leaves of the stable graph, and we do not assume that $\mathsf{B}$ has a factorisation property in terms of an $R$-matrix in the style of \eqref{thefirstfacto}.
\end{remark}

\medskip
 
\subsection{Intersection theory for regularly admissible spectral curves}
\label{theregucnu}
\medskip

Let $\mc{S} = (C,x,y,\omega_{0,2})$ be a regularly admissible spectral curve equipped with a fundamental bidifferential of the second kind. In particular, $C$ must be smooth and we can write $ C$ and $ \mf{a} $ instead of $ \tilde{C}$ and $ \tilde{\mf{a}}$, respectively. In this context, as in \cref{Rudef} we rather take
\[
V = \bigoplus_{\alpha \in \mathfrak{a}} V^{r_\alpha{\rm spin}}\,.
\]
This amounts to set $e_{\alpha,r_{\alpha}} = 0$ in all subsequent formulas.
 
Following \cref{Tudefn} we have a generating series
\begin{equation}
\label{eq:243} \begin{split}
\mathsf{T}(u) & \coloneqq \bigg(\sum_{\alpha \in \mathfrak{a}} e_{\alpha,1}u\bigg) + \mathfrak{L}^+_{{\rm tot}}(\omega_{0,1}) \\
& = \sum_{\alpha \in \mathfrak{a}} \bigg( e_{\alpha,1} u + \sum_{k > 0} (k - r_{\alpha})!^{(r_{\alpha})}\,F_{0,1}\big[\begin{smallmatrix} \alpha \\ -k \end{smallmatrix}\big]\,e_{\alpha,\hat{k}}\,u^{\overline{k}}\bigg) \\
& \eqqcolon \sum_{\alpha \in \mathfrak{a}} \sum_{k > 0} \mathsf{T}_{\alpha,k}\,e_{\alpha,\hat{k}}\,u^{\overline{k}}\,.
\end{split}
\end{equation} 
Equivalently, the definition of $\mathsf{T}$ means that we have the expansion
\[
\mathsf{Loc}_{\alpha}(y - y(\alpha)) = - \frac{\zeta}{r_\alpha}+  \sum_{k > 0} \frac{k!^{(r_\alpha)} \mathsf{T}_{\alpha,k + r_{\alpha}}}{r_\alpha}\,\zeta^{k}\dd \zeta
\]
when $z \rightarrow \alpha$ in the local coordinate such that $x = x(\alpha) + \zeta^{r_\alpha}$. Due to the regularly admissible condition, we indeed have $\mathsf{T}(u) \in \mc{O}(u)$ and
\[
\mathsf{T}_{\alpha,r_{\alpha} + 1} = 1 + F_{0,1}\big[\begin{smallmatrix} \alpha \\ -(r_{\alpha} + 1) \end{smallmatrix}\big] \neq 1\,,
\]
so we can use it to act on Witten $r_{\alpha}$-spin class. Recall that we have a second generating series $\mathsf{B}(u,v)$ from \cref{Buvdef}.

\begin{definition}
We let these generating series act on the direct sum of Witten classes to define
\[
\Omega^{\mathcal{S}}_{g,n} \coloneqq \bigg[\widehat{\mathsf{B}}\widehat{\mathsf{T}}\cdot \Big(\bigoplus_{\alpha \in \mathfrak{a}} W^{r_{\alpha}{\rm spin}}\Big)\bigg]_{g,n} \colon V^{\otimes n} \rightarrow CH^*(\overline{\mathcal{M}}_{g,n})\,.
\]
\end{definition}

\begin{theorem}
\label{interrepreg} If $\mathcal{S} = (C,x,y,\omega_{0,2})$ is a regularly admissible spectral curve equipped with a fundamental bidifferential of the second kind $\omega_{0,2}$ and zero crosscap form, then for $2g - 2 + n > 0$
\[
(\mathfrak{L}^{-}_{{\rm tot}})^{\otimes n}(\omega_{g,n}) = \int_{\overline{\mathcal{M}}_{g,n}} \frac{\Omega^{\mathcal{S}}_{g,n}}{\prod_{i = 1}^n (1 - \epsilon_i\psi_i)}\,,
\] 
where $\epsilon_i$ is the variable in the $i$th Laplace transform.
\end{theorem} 
\begin{proof}
The spectral curve
\[ 
x_0\big(\begin{smallmatrix} \alpha \\ z \end{smallmatrix}\big) = z^{r_{\alpha}}\,,\qquad y_0\big(\begin{smallmatrix} \alpha \\ z \end{smallmatrix}\big) = -\frac{z}{r},\qquad \omega_{0,2}\big(\begin{smallmatrix} \alpha_1 & \alpha_2 \\ z_1 & z_2 \end{smallmatrix}\big) = \frac{\delta_{\alpha_1,\alpha_2}\dd z_1\dd z_2}{(z_1 - z_2)^2}
\] 
has $\mathsf{T} = 0$ and $\mathsf{B} = 0$, hence $\Omega^{\mathcal{S}} = \bigoplus_{\alpha \in \mathfrak{a}} W_{g,n}^{r_{\alpha}{\rm spin}}$. It is obtained by taking independent copies of \eqref{zetarspin} indexed by $\alpha \in \mathfrak{a}$. So, its correlators equal to \eqref{rspinw}, that is
\begin{equation}
\label{apply0} \omega_{g,n}^0\big(\begin{smallmatrix} \alpha_1 & \cdots & \alpha_n \\ z_1 & \cdots & z_n \end{smallmatrix}\big) = \sum_{\substack{d_1,\ldots,d_n \geq 0 \\ l_i \in [r_{\alpha_i}]}}  \bigg(\int_{\overline{\mathcal{M}}_{g,n}} W_{g,n}^0\big(\otimes_{i = 1}^n e_{\alpha_i,l_i}\big)\prod_{i = 1}^n \psi_i^{d_i}\bigg) \prod_{i = 1}^n \frac{(d_i r_{\alpha_i} + l_i)!^{(r_{\alpha_i})}\,\dd z_i}{z_i^{d_i r_{\alpha_i} + l_i + 1}}\,,
\end{equation}
where $W_{g,n}^0 = \bigoplus_{\alpha \in \mathfrak{a}} W^{r_{\alpha}{\rm spin}}$. Applying $\bigotimes_{i = 1}^n \mathfrak{L}^{-}_{\alpha_i}$ to \cref{apply0} cancels the factorials and replaces $z^{-(d_ir_{\alpha_i} + l_i + 1)}\dd z$ with $e_{\alpha_i,l_i}^* \epsilon^{d_i}$. The sum over $d_i$ can then be performed and this entails the claim in this special case.

Applying \cref{mainth2} to the special case, let $Z_0$ be the partition function of the Airy structure corresponding this special case. Its coefficients are
\begin{equation}
\label{Fgn000} F_{g,n}^0\big[\begin{smallmatrix} \alpha_1 & \cdots & \alpha_n \\ k_1 & \cdots & k_n \end{smallmatrix}\big] \coloneqq \bigg(\int_{\overline{\mathcal{M}}_{g,n}} W_{g,n}\big(\otimes_{i = 1}^n e_{\alpha_i,\hat{k}_i}\big) \prod_{i = 1}^n \psi_i^{\overline{k}_i}\bigg) \prod_{i = 1}^n k_i!^{(r_{\alpha_i})}\,.
\end{equation}
It corresponds to $F_{0,1}^{0}\big[\begin{smallmatrix} \alpha \\ -(r_{\alpha} + 1) \end{smallmatrix}\big] = -1$.

Applying now \cref{mainth2} to a general regularly admissible spectral curve, the partition function of the corresponding Airy structure is
\[
Z = \exp\left( \frac{1}{2\hbar} \sum_{\substack{\alpha,\beta \in \mathfrak{a} \\ k,l > 0}} F_{0,2}\big[\begin{smallmatrix} \alpha & \beta \\ -k & -l \end{smallmatrix}\big]\,\frac{J_{k}^{\alpha}J_{l}^{\beta}}{k\,l}\right)\cdot Z_1,\qquad Z_1 = \exp\left(\frac{1}{\hbar} \sum_{\substack{\alpha \in \mathfrak{a} \\ k > 0}}  \Big(F_{0,1}\big[\begin{smallmatrix} \alpha \\ -k \end{smallmatrix}\big] + \delta_{k,r_{\alpha} + 1}\Big)\frac{J_{k}^{\alpha}}{k}\right)\cdot Z_0\,.
\]
The operation taking $Z_0$ to $Z_1$ is the shift of times:
\[
x_{\alpha,k} \rightarrow x_{\alpha,k} + \frac{F_{0,1}\big[\begin{smallmatrix} \alpha \\ -k \end{smallmatrix}\big] + \delta_{k,r_{\alpha} + 1}}{k} = x_{\alpha,k} + \frac{\mathsf{T}_{\alpha,k}}{k!^{(r_{\alpha})}}\,,
\]
taking into account \eqref{eq:243}. In terms of the coefficients $F_{g,n}^1$ of $Z_1$, we get
\begin{equation*}
\begin{split}
F_{g,n}^1\big[\begin{smallmatrix} \alpha_1 & \cdots & \alpha_n \\ k_1 & \cdots & k_n \end{smallmatrix}\big] & = \sum_{m \geq 0} \frac{1}{m!} \sum_{\substack{\beta_1,\ldots,\beta_n \in \mathfrak{a} \\ l_1,\ldots,l_m > 0}} F_{g,n + m}^{0}\big[\begin{smallmatrix} \alpha_1 & \cdots & \alpha_n & \beta_1 & \cdots & \beta_m \\ k_1 & \cdots & k_n & l_1 & \cdots & l_m \end{smallmatrix}\big]\,\prod_{i = 1}^m \frac{\mathsf{T}_{\beta_i,l_i}}{l_i!^{(r_{\beta_i})}} \\
& = \bigg(\int_{\overline{\mathcal{M}}_{g,n}} \big[\widehat{\mathsf{T}}\cdot  W^0\big]_{g,n}\big(\otimes_{i = 1}^n e_{\alpha_i,\hat{k}_i}\big) \prod_{i = 1}^n \psi_i^{\overline{k}_i}\bigg) \prod_{i = 1}^n k_i!^{(r_{\alpha_i})}
\end{split}
\end{equation*}
by using \eqref{Fgn000} and comparing with \cref{Sectrans}. Note the cancellation of the factorials that was the motivation for our definition of $\mathsf{T}$ in \eqref{eq:243}. Applying $\mathfrak{L}^-_{{\rm tot}}$ to the corresponding correlators kills the remaining factorials, we would prove the desired formula in the case $\mathsf{B} = 0$.

For general $\mathsf{B}$ we should still take $Z_1$ to $Z$, and this amounts at the level of coefficients to summing over stable graphs. Comparing with \cref{sumstab}, one can check in a similar way that the factorials completely disappear, so that the sum over $d_i$ becomes the geometric series in the Laplace variable $\epsilon_i$.
\end{proof} 
 
\medskip

\subsubsection{The conjectural $(r,s)$ classes}
\label{rsclconjsec}

\medskip

The basic case of irregularly admissible smooth spectral curves with one ramification point is
\[
x = z^r,\qquad y = -\frac{z^{s - r}}{r}\,,\qquad \omega_{0,2}(z_1,z_2) = \frac{\dd z_1\dd z_2}{(z_1 - z_2)^2}\,,
\]
with $r \geq 2$, $s \in [r - 1]$ and $r = \pm 1\,\,{\rm mod}\,\,s$. It correspond to the Airy structures of \cref{thm:W_gl_Airy_Coxeter}, already obtained in \cite{BBCCN18}. The coefficients of its partition function have the following basic properties from \Cref{lem:homogen} and \cref{prop:symm_cond_omega03}:
\begin{itemize}
\item[$(i)$] Homogeneity: $F_{g,n}[p_1,\ldots,p_n] = 0$ unless we have $\sum_{m = 1}^{n} p_m = s(2g - 2 + n)$.
\item[$(ii)$] Dilaton equation: $F_{g,n + 1}[s,p_1,\ldots,p_n] = s(2g - 2 + n)F_{g,n}[p_1,\ldots,p_n]$ for $2g - 2 + n > 0$.
\item[$(iii)$] Special values:
\begin{equation*} 
\begin{split}
F_{1,1}[p] & = \frac{r^2 - 1}{24}\,\delta_{p,s}\,, \\ 
F_{0,3}[p_1,p_2,p_3] & = c\,p_1p_2p_3\delta_{p_1 + p_2 + p_3,s}\,,
\end{split} 
\end{equation*} 
where $c$ is as in \cref{defcmu}.
\item[$(iv)$] If $s = 1$, then $F_{0,n} = 0$ for any $n \geq 3$. Indeed, as $c = 0$ in this case, we have $F_{0,3} = 0$ and as it is the only initial data needed for topological recursion \eqref{theFGNsum} in genus $0$, all the genus $0$ sector vanishes.
\item[$(v)$] $F_{g,n}[p_1,\ldots,p_n] = 0$ whenever there exists $i \in [n]$ such that $r|p_i$.
\end{itemize}
\begin{remark} There is no string equation, as the operator $H_{i = 2,k = -1}$ is not part of the Airy structure.
\end{remark}

Mimicking \cref{rspinthmTR} and taking into account these properties, we are led to propose the following conjecture --- in a slightly more precise form than \cite[Section 6.2]{BBCCN18}.

\begin{conjecture}
\label{conjrsclass} For each $r \geq 2$ and $s \in [r - 1]$ such that $r = \pm 1\,\,{\rm mod}\,\,s$, there exists cohomology classes $w_{g,n}^{(r,s)}(a_1,\ldots,a_n) \in CH^*(\overline{\mathcal{M}}_{g,n})$ indexed by $a_i \in [r]$ and $g,n \in \mathbb{Z}_{\geq 0}$ such that $2g - 2 + n > 0$, so that
\begin{itemize}
\item[$(o)$] for any $d_i \geq 0$
\begin{equation}
\label{rsclassfgn}F_{g,n}[d_1r + a_1,\ldots,d_nr + a_n] = \prod_{i = 1}^n (d_ir + a_i)!^{(r)} \int_{\overline{\mathcal{M}}_{g,n}} w_{g,n}^{(r,s)} (a_1, \dotsc, a_n) \prod_{i = 1}^n \psi_i^{d_i}\,.
\end{equation}
\item[$(i)$] $w_{g,n}^{(r,s)}(\mathbf{a})$ has pure Chow codimension
\[
2\bigg( \frac{ \sum_{i = 1}^n a_i - s(2g - 2 + n)}{r} + (3g - 3 + n)\bigg)\,.
\]
\item[$(ii)$] denoting $\pi\,:\,\overline{\mathcal{M}}_{g,n + 1} \rightarrow \overline{\mathcal{M}}_{g,n}$ the forgetful morphism, we have the dilaton equation
\begin{equation}
\label{rsdilaton}\psi_{n + 1} \pi^*\big(w_{g,n}^{(r,s)}(\mathbf{a})\big) = w_{g,n + 1}^{(r,s)}(s,\mathbf{a})\,.
\end{equation}
\item[$(iii)$] we have the special values
\begin{equation*}
\begin{split}
w_{0,3}^{(r,s)}(a_1,a_2,a_3) & = c \delta_{a_1 + a_2 + a_3,s} \mathbf{1} \in H^0(\overline{\mathcal{M}}_{0,3})\,, \\
w_{1,1}^{(r,s)}(a) & = \delta_{a,s}\,\frac{r^2 - 1}{s}\psi_1 \in H^2(\overline{\mathcal{M}}_{1,1})\,.
\end{split}
\end{equation*}
\item[$(iv)$] $w_{0,n}^{(r,s = 1)} = 0$ for all $n$.
\item[$(v)$] $w_{g,n}^{(r,s)}(a_1,\ldots,a_n) = 0$ whenever there exists $i \in [n]$ such that $a_i = r$.
\end{itemize}
\end{conjecture}

The conjecture is proved in the case $(r,s) = (2,1)$ by Norbury \cite{Nor17} and $s = r -1$ for general $r \geq 2$ by Chidambaram, Garcia-Failde and Giacchetto \cite{CGFG22}. They construct $w_{g,n}^{(r,r - 1)}$ as a cohomology class obtained by pushforward from the moduli space of $r$-spin curves; in their work it is denoted $\Theta_{g,n}^{r}$, or simply $\Theta_{g,n}$ if $r = 2$.

Assuming \cref{conjrsclass} holds, we will deform \eqref{rsclassfgn} to generalise \cref{interrepreg} for any smooth admissible spectral curve. Before this, we need to discuss action of translation and sums over stable graphs on the $(r,s)$ class.

\medskip

\subsubsection{Deformation of the $(r,s)$ classes}
\label{defrsdefrs}

\medskip

We introduce
\[
W^{(r,s)}_{g,n}\colon \begin{array}{lll} (V^{r{\rm spin}})^{\otimes n} & \longrightarrow & H^*(\overline{\mathcal{M}}_{g,n}) \\ \bigotimes_{i = 1}^n e_{a_i} & \longmapsto & w_{g,n}^{(r,s)}(a_1,\ldots,a_n) \end{array}\,.
\]
If we have a formal series $\mathsf{T}(u) \in V^{r{\rm spin}}[\![u]\!]$, we will see that under certain conditions we can define analogously to \cref{Sectrans} a new family of cohomology classes
\[
\big[\hat{\mathsf{T}} W^{(r,s)}\big]_{g,n}\colon (V^{r{\rm spin}})^{\otimes n} \rightarrow H^*(\overline{\mathcal{M}}_{g,n}) \,.
\]
Assuming $\mathsf{T}_{s} \neq 1$, this is done in terms of the modified generating series
\[
\tilde{\mathsf{T}}(u) \coloneqq \frac{\mathsf{T}(u) - \mathsf{T}_{s}e_{s}}{1 - \mathsf{T}_{s}} = \sum_{\substack{d \geq 0 \\ a \in [r - 1] \\ rd + a > s}}\frac{\mathsf{T}_{rd + a}e_{a}u^{d}}{1 - \mathsf{T}_{s}}\,,
\] 
via the formula
\begin{equation}
\label{TWRDSF}\big[\hat{\mathsf{T}} W^{(r,s)}\big]_{g,n}(-) = \sum_{m \geq 0} \frac{1}{m!} (\pi_m)_* W_{g,n + m}^{(r,s)}\big(- \otimes \tilde{\mathsf{T}}(\psi_{n + 1}) \otimes \cdots \otimes \tilde{\mathsf{T}}(\psi_{n + m})\big)\,,
\end{equation}
where $\pi_m\,:\,\overline{\mathcal{M}}_{g,n + m} \rightarrow \overline{\mathcal{M}}_{g,n}$ is the forgetful morphism. As in \eqref{Sectrans}, the handling of $\mathsf{T}_{s}$ is tailored to be compatible with the dilaton equation \eqref{rsdilaton}.

\begin{lemma}
\label{lemdefTrs} Assume that $\mathsf{T}_{s} \neq 1$ and $\mathsf{T}_{a} = 0$ for $a < s$. Then \cref{TWRDSF} is well-defined, i.e. for any evaluation on an element of $(V^{r{\rm spin}})^{\otimes n}$ the sum on the right-hand side is finite.
\end{lemma} 
\begin{proof}
The argument is similar to \cref{Sectrans}. If we evaluate on $\bigotimes_{i = 1}^n e_{a_i}$, due to $(i)$ the complex codimension of the summand proportional to $\prod_{j = 1}^m \mathsf{T}_{d_jr + b_j}$ with $d_jr + b_j \geq s + 1$ is
\[
\frac{1}{r}\bigg(s(2g - 2 + n + m) - \sum_{i = 1}^n a_i - \sum_{j = 1}^m (d_jr + b_j)\bigg) \leq \frac{1}{r}\bigg(s(2g - 2 + n) - \sum_{i = 1}^n a_i - m\bigg)\,,
\]
which for fixed $a_1,\ldots,a_n$ is negative for $m$ is large enough, forcing this summand to vanish.
\end{proof}

On the other hand, the action of any $\mathsf{B} \in V^{r{\rm spin}}[\![u,v]\!]$ via sums over stable graphs is well-defined since it always involves finite sums.

\medskip

\subsubsection{Intersection theory for admissible smooth local spectral curves}

\medskip

Let $\mathcal{S} = (C,x,y,\omega_{0,2})$ be an admissible smooth spectral curve equipped with a fundamental bidifferential of the second kind. Recall from \cref{Laplacesection} the definition of the vector space $V$ and the generating series $\mathsf{T}$ and $\mathsf{B}$ that can be associated to $\mathcal{S}$.
\begin{definition}
We let them act on the direct sum of $(r_{\alpha},s_{\alpha})$-classes to define
\[
\Omega_{g,n}^{\mathcal{S}} = \bigg[\hat{\mathsf{B}}\hat{\mathsf{T}}\bigg(\bigoplus_{\alpha \in \mathfrak{a}} W^{(r_{\alpha},s_{\alpha})}\bigg)\bigg]_{g,n}\,:\,V^{\otimes n} \rightarrow CH^*(\overline{\mathcal{M}}_{g,n})\,.
\]
Here we suppose the $(r,s)$ class is a Chow class to put them on the same footing as the Witten $r$-spin classes, and denote $W^{(r,r+1)} \coloneqq W^{r{\rm spin}}$ for uniformity. The admissibility condition for irregular ramification points matches the condition of \cref{lemdefTrs} so the definition is well-posed. 
\end{definition}

\begin{theorem}
\label{thmint2}Assume \cref{conjrsclass} holds, and let $\mathcal{S} = (C,x,y,\omega_{0,2})$ be an admissible smooth spectral curve equipped with a fundamental bidifferential of the second kind $\omega_{0,2}$ and zero crosscap form. Then, for $2g - 2 + n > 0$
\[
(\mathfrak{L}_{{\rm tot}}^{-})^{\otimes n}(\omega_{g,n}) = \int_{\overline{\mathcal{M}}_{g,n}} \frac{\Omega_{g,n}^{\mathcal{S}}}{\prod_{i = 1}^n (1 - \epsilon_i \psi_i)}\,,
\]
where $\epsilon_i$ is the variable of the $i$th Laplace transform.
\end{theorem}
\begin{proof}
The proof, which relies on the correspondence of \cref{mainth2} --- already proved in \cite{BBCCN18} --- is similar to that of \cref{interrepreg}, so we omit the details: the regular ramification points are treated by \cref{interrepreg} itself, and the irregular ramification points using \cref{defrsdefrs}.
\end{proof}

 \medskip
\subsubsection{Intersection numbers on the flat basis}
\label{Ssmonfg}
\label{ELSVwgnsec}

\medskip

\begin{definition}
\label{presequde}Let $\mathcal{S} = (C,x,y,\omega_{0,2})$ be a spectral curve equipped with a fundamental bidifferential of the second kind. We say that $\mathcal{S}$ is nearly compact if there exists a smooth compact connected curve $\bar{C}$ containing $\tilde{C}$ such that $\dd \tilde{x}$ admits an analytic continuation to a meromorphic $1$-form on $\bar{C}$ without zeroes on $\bar{C} \setminus \tilde{C}$, and $\omega_{0,2}$ admits an analytic continuation as a fundamental bidifferential of the second kind on $\bar{C}$.
\end{definition}
This definition is tailored to allow $\tilde{x}$ to be multivalued (for instance due to logarithmic singularities) on $\bar{C}$. Many examples of non compact but nearly compact spectral curves can be found in mirror symmetry and in Hurwitz theory. In the nearly compact case the results of Sections~\ref{higheret}-\ref{primsec} apply to $\bar{C}$: we have the factorisation property for $\bar{\mathsf{B}}$ via the $\mathsf{R}$-matrix (\cref{compactfactospin}) and we can consider the basis of primary and descendants differentials on $\bar{C}$ (\cref{ABuv}).

\begin{definition}
We introduce the covector of primary differentials:
\[
\dd \Xi(z) = \bigg(\sum_{\substack{\mu \in \tilde{\mathfrak{a}} \\ a \in [r_{\mu}]}} \dd \xi^{\,\mu}_{-a}(z)\,e^*_{\mu,a}\bigg) \in H^0\big(\tilde{C},K_{\tilde{C}}(*\tilde{\mathfrak{a}})\big) \otimes V^*
\]
\end{definition}

\begin{corollary}
\label{ELSVwgnCor} Consider the situation described in \cref{interrepreg} or \cref{thmint2}, and assume that the spectral curve $\mathcal{S}$ is nearly compact. We have the following decomposition on the basis of primary differentials and their descendants:
\[
\omega_{g,n}(z_1,\ldots,z_n)  = \int_{\overline{\mathcal{M}}_{g,n}} \Omega^{\mathcal{S}}_{g,n} \otimes  \bigg(\bigotimes_{i = 1}^n \dd\Xi(z_i)\bigg)  \cdot \bigg(\bigotimes_{i = 1}^n \frac{\mathsf{R}(\psi_i) \star \eta^{\vee}}{1 - \mathcal{D}_i \psi_i}\bigg)\,.
\]
Here, $\mathcal{D}_i$ is the operator $\mathcal{D} \colon \varphi \mapsto - \dd\big(\varphi/\dd \tilde{x}\big)$ acting from the right the $1$-forms in $\dd \Xi(z_i)$, and $\cdot$ means that the first tensor factor from the $i$-th factor on the right must be inserted by the $i$-th position in the multilinear map $\Omega_{g,n}^{\mathcal{S}}$, while the second tensor factor must be inserted in the linear form $\dd \Xi_i(z_i)$.
\end{corollary}
\begin{proof} 
Under the assumptions we already know that the coefficients of decomposition of the $\omega_{g,n}$ on the $\xi$-basis are
 \begin{equation}
\label{Fgnome} F_{g,n}\big[\begin{smallmatrix} \mu_1 & \cdots & \mu_n \\ k_1 & \cdots & k_n \end{smallmatrix}\big] = \int_{\overline{\mathcal{M}}_{g,n}} \Omega^{\mathcal{S}}_{g,n}\big((e_{\mu_i,\hat{k}_i})_{i = 1}^n\big)  \prod_{i = 1}^n k_i!^{(r_{\mu_i})} \psi_i^{\overline{k_i}} \,.
 \end{equation}
 We rather want to express them on the $\widehat{\xi}$-basis. With the change of basis \eqref{changeba} this yields:
 \begin{equation*}
 \begin{split}
 & \quad \omega_{g,n}(z_1,\ldots,z_n) \\
 & = \sum_{\substack{\mu_1,\ldots,\mu_n \in \tilde{\mathfrak{a}} \\ k_1,\ldots,k_n > 0}} \sum_{\substack{\rho_1,\ldots,\rho_n \in \tilde{\mathfrak{a}} \\ l_1,\ldots,l_n > 0}} \bigg(\int_{\overline{\mathcal{M}}_{g,n}} \Omega^{\mathcal{S}}_{g,n}\big(e_{\mu_1,\hat{k}_1} \otimes \cdots \otimes e_{\mu_n,\hat{k}_n}\big)\,\prod_{i = 1}^n k_i!^{(r_{\mu_i})} \psi_i^{\overline{k_i}} \bigg)\,A\big[\begin{smallmatrix} \mu_i & \rho_i \\ -k_i & -l_i \end{smallmatrix}\big]\, \dd \widehat{\xi}^{\,\rho_i}_{-l_i}(z_i) \\ 
 & = \int_{\overline{\mathcal{M}}_{g,n}} \Omega^{\mathcal{S}}_{g,n} \otimes \bigg(\bigotimes_{i = 1}^n \dd\Xi(z_i)\bigg)  \cdot \bigotimes_{i = 1}^n \bigg(\sum_{\substack{\mu_i,\rho_i \in \tilde{\mathfrak{a}} \\ k_i,l_i > 0}}  A\big[\begin{smallmatrix} \mu_i & \rho_i \\ -k_i & -l_i \end{smallmatrix}\big]\,k_i!^{(r_{\mu_i})} e_{\mu_i,\hat{k}_i} \otimes e_{\rho_i,l_i}\, \psi_i^{\overline{k_i}} \mathcal{D}^{\overline{l}_i}_i\bigg)  \\
 & = \int_{\overline{\mathcal{M}}_{g,n}} \Omega^{\mathcal{S}}_{g,n} \otimes  \bigg(\bigotimes_{i = 1}^n \dd\Xi(z_i)\bigg)  \cdot \bigotimes_{i = 1}^n \bar{\mathsf{A}}(\psi_i,\mathcal{D}_i)\,,
 \end{split} 
 \end{equation*} 
where $\cdot$ is the contraction of multilinear maps with tensors in the order prescribed by the statement of the claim, and $\mathcal{D}_i$ acts from the right on the $1$-forms in $\dd \Xi(z_i)$. The situations of \cref{interrepreg} or \cref{thmint2} correspond to smooth spectral curves, for which \eqref{Fgnome} vanishes whenever one of the $k_i$ is divisible by $r_{\mu_i}$ --- see $(v)$ in \cref{rsclconjsec}. Subsequently, we can replace $\bar{A}(\psi_i,\mathcal{D}_i)$ by its projection ${}^{{\rm ss}}\bar{\mathsf{A}}(\psi_i,\mathcal{D}_i)$. We conclude by using the second formula of \cref{ABuv}.
\end{proof} 
 
 \medskip
 
 \subsubsection{TR-ELSV formula and quasi-polynomiality}
\label{SecTRELSV}
\medskip

In the situation of \Cref{ELSVwgnCor},  let $P \subset \bar{C}$ be a set of points which are not zeroes of $\dd \tilde{x}$. Given $p \in P$, we introduce:
\[
d_p \coloneqq \big(- {\rm ord}_p \,\dd \tilde{x}\big) \in \mathbb{Z}_{\geq 0},\qquad c_p = \Res_{p} \dd \tilde{x}\,.
\]
and we choose a local coordinate $X$ near $p$ such that $X(p) = 0$ and 
\[
\dd \tilde{x} = \left\{\begin{array}{lll} c_p\,\frac{\dd X}{X} && {\rm if}\,\,d_p = 1\,, \\[3pt] X^{-d_p}\dd X && {\rm if}\,\,d_p \neq 1\,. \end{array}\right.
\] 
\begin{definition}
We introduce for $(p,\ell) \in P \times \mathbb{Z}_{> 0}$ and $\mu \in \tilde{\mathfrak{a}}$, $k \in [r_{\mu}]$ the quantities
\begin{equation*}
\begin{split}
S\big[\begin{smallmatrix} p & \mu \\ \ell & -k \end{smallmatrix}\big] & = \Res_{z = p} \frac{\dd\xi_{-k}^{\mu}(z)}{X(z)^{\ell}} = \Res_{z = p} \Res_{z' = \mu} \frac{k\,\omega_{0,2}(z,z')}{X(z)^{\ell}\zeta(z')^{k}} \,, \\
\boldsymbol{S}\big[\begin{smallmatrix} p \\ \ell \end{smallmatrix}\big] & = \sum_{\substack{\mu \in \tilde{\mathfrak{a}} \\ k \in [r_{\mu}]}} S\big[\begin{smallmatrix} p & \mu \\ \ell & -k \end{smallmatrix}\big]\, e_{\mu,k}^* \,.
\end{split} 
\end{equation*}
They give the all-order series expansion (indicated with the symbol $\approx$) of the primary differentials near $p$:
\[
\dd \Xi (z) \mathop{\approx}_{z \rightarrow p} \,\,\dd\bigg(\sum_{\ell > 0} \boldsymbol{S}\big[\begin{smallmatrix} p \\ \ell \end{smallmatrix}\big]\,\frac{X^{\ell}}{\ell}\bigg)\,.
\]  
and are closely related to the $S$-matrix in Givental formalism.
\end{definition}

\begin{lemma}
\label{lemDpi} With the convention that $\mathcal{D}$ acts from the right on primary differentials, we have for $p \in P$:
\[
\dd \Xi (z) \,\frac{1}{1 - \mathcal{D}\psi}\,\, \mathop{\approx}_{z \rightarrow p}\,\,\left\{\begin{array}{lll} \dd\bigg(\sum_{\ell > 0} \dfrac{\boldsymbol{S}\big[\begin{smallmatrix} p \\ \ell \end{smallmatrix}\big]}{1 + c_p^{-1} \ell \psi}\,\frac{X^{\ell}}{\ell}\bigg) && {\rm if}\,\,d_p = 1\,, \\ [12pt] \dd\bigg(\sum_{\ell > 0} \boldsymbol{S}\big[\begin{smallmatrix} p \\ \ell \end{smallmatrix}\big]\, G_{\frac{\ell}{d_p-1}}\big((d_p-1)X^{d_p - 1} \psi\big)\,\frac{X^{\ell}}{\ell}\bigg) && {\rm otherwise}\,. \end{array} \right.
\]
where in the second line we introduced
\begin{equation*}
G_a(x) \coloneqq \sum_{m\geq 0} (-1)^m \prod_{j=0}^{m-1}(j+a) ~ x^m = \int_{0}^{\infty} \frac{e^{-t/x}}{(1 + t)^{a}}\,\frac{\dd t}{x} = \frac{e^{1/x}\Gamma\big(1-a,\tfrac{1}{x}\big)}{x^{a}}\,,
\end{equation*}
and $\Gamma(1-a,x)$ is the incomplete Gamma function.
\end{lemma}
\begin{proof}
Let us write for uniformity $\dd \tilde{x} = \kappa_p X^{-d_p} \dd X$ with $\kappa_p = c_p \neq 0$ if $d_p = 1$ and $\kappa_p = 1$ otherwise. The operator $\mathcal{D}$ can be expressed in terms of $X$:
\[
\mathcal{D}\varphi = - \dd\bigg(\frac{X^{d_p}\varphi}{\kappa_p \dd X}\bigg)\,.
\]
In particular for $\ell > 0$
\[
\mathcal{D} (X^{\ell - 1}\dd X) = -\frac{(d_p - 1 + \ell)}{\kappa_p}\,X^{d_p - 1 + \ell - 1}\dd X\,,
\]
and by induction for $m \geq 0$
\[
\mathcal{D}^m (X^{\ell - 1}\dd X) = \left\{\begin{array}{lll} (-1)^m(\ell/c_p)^{m}\,X^{\ell - 1}\dd X && {\rm if}\,\,d_p = 1\,, \\ (-1)^m(d_p - 1)^m \prod_{j = 1}^{m} \Big(j + \frac{\ell}{d_p - 1}\Big)\,X^{m(d_p - 1) + \ell - 1}\dd X\,. && {\rm otherwise}\,. \end{array}\right.
\]
We then multiply by $\psi^m$ and sum over $m \geq 0$.
\end{proof}

Since $p$ is not a zero of $\dd \tilde{x}$, for $2g - 2 + n > 0$ the form $\omega_{g,n}(z_1,\ldots,z_n)$ is holomorphic at $z_i = p \in P$, and so we have an all-order series expansion of the form
\[
\omega_{g,n}(z_1,\ldots,z_n) \approx \dd_1 \cdots \dd_n\bigg(\sum_{\ell_1,\ldots,\ell_n > 0} H_{g,n}\big[\begin{smallmatrix} p_1 & \cdots & p_n \\ \ell_1 & \cdots & \ell_n \end{smallmatrix}\big]\,\prod_{i = 1}^n \frac{X_i^{\ell_i}}{\ell_i} \bigg)\,.
\]
Equivalently:
\[
H_{g,n}\big[\begin{smallmatrix} p_1 & \cdots & p_n \\ \ell_1 & \cdots & \ell_n \end{smallmatrix}\big] = \Res_{z_1 = p_1} \cdots \Res_{z_n = p_n} \frac{\omega_{g,n}(z_1,\ldots,z_n)}{\prod_{i = 1}^n X(z_i)^{\ell_i}}\,.
\]
When we only expand near simple poles of $\dd \tilde{x}$, the formula for $H_{g,n}$ becomes particularly simple, due to \cref{lemDpi}.
 
\begin{corollary}[TR-ELSV formula]
\label{TRESLVcor}Under the assumptions of \Cref{ELSVwgnCor} and assuming that $P \subset \bar{C}$ is a set of simple poles of $\dd \tilde{x}$, we have for $2g - 2 + n > 0$, $p_1,\ldots,p_n \in P$ and $\ell_1,\ldots,\ell_n > 0$
\[
H_{g,n}\big[\begin{smallmatrix} p_1 & \cdots & p_n \\ \ell_1 & \cdots & \ell_n \end{smallmatrix}\big] = \int_{\overline{\mathcal{M}}_{g,n}} \bigg(\Omega^{\mathcal{S}}_{g,n} \otimes \bigotimes_{i = 1}^n \frac{\boldsymbol{S}\big[\begin{smallmatrix} p_i \\ \ell_i \end{smallmatrix}\big]}{1 +  c_{p_i}^{-1}\ell_i \psi_i}\bigg) \cdot \bigg(\bigotimes_{i = 1}^n \mathsf{R}(\psi_i) \star \eta^{\vee}\bigg)\,. 
\] 
where $c_p = \Res_p \dd \tilde{x}$.
\end{corollary}
\begin{remark}
Apart from the non-trivial dependence in $\ell_i$ contained in $\boldsymbol{S}$, the dependence in the $\ell_i$ is polynomial, coming from the geometric series expansion of
\[
\frac{1}{1 + c_{p_i}^{-1}\ell_i\psi_i}
\]
from which only finitely many terms contribute. This type of behavior is called ``quasi-polynomiality''. The assumptions of \Cref{TRESLVcor} in particular apply for spectral curves of Hurwitz theory and Gromov--Witten theory, where the $H_{g,n}$ store the desired enumerative invariants.
\cref{lemDpi} shows that, in the situation where enumerative invariants are stored in the expansion of  $\omega_{g,n}$ near points which are not simple poles of $\dd \tilde{x}$, the quasi-polynomiality is a priori destroyed and the structure is more complicated: the coefficient of $\dd\big(X^{\ell_i}(z_i)/\ell_i\big)$ will receive contributions from $\boldsymbol{S}\big[\begin{smallmatrix} p_i \\ \ell_i' \end{smallmatrix}\big]$ for all $\ell_i' \leq \ell_i$ such that $(d_p - 1)$ divides $(\ell_i - \ell_i')$.  
\end{remark}

\medskip

\section{Open intersection numbers}
 \label{sec:open}
 
 \medskip
 
In this section, we shall propose precise conjectures about open $r$-spin intersection numbers using the partition function of the Airy structure with twist $ \sigma $ of cycle type $ (r-1,1)$ with no dilaton shift attached to the fixed point, as obtained in \cref{thm:W_gl_Airy_arbitrary_autom}. Before this, we review the relation between this partition function for $r = 2$ and the open intersection theory, and Safnuk~\cite{Saf16}'s topological recursion for it. The latter has peculiar features related to the reducibility of the spectral curve, our general approach shines a new light on this. These relations depend on some foundational conjectures in open intersection theory scattered in the literature and that we make explicit.

\medskip

\subsection{Review of open intersection theory}
\label{OpenIntKP}

\medskip

The enumerative geometry of open Riemann surfaces was developed by Pandharipande--Solomon--Tessler in genus $0$ in \cite{PST15}, and its extension to all genera was announced by Solomon and Tessler. The upshot is that for $g,n,b,m \geq 0$ such that
\begin{equation}
\label{chibarE}-\overline{\chi} \coloneqq 2\overline{g} - 2 + 2n + m > 0\,,\qquad \overline{g} \coloneqq 2g + b - 1\,,
\end{equation}
there is a moduli space $\mathcal{M}_{g,n;b,m}$ parametrizing Riemann surfaces of genus $g$ with $b$ (unlabelled) boundary components, $n$ labelled interior marked points and $m$ labelled boundary marked points, equipped with spin structure and a ``grading''. It is a real orbifold of dimension
\[
D \coloneqq 6g - 6 + 3b + 2n + m\,,
\]
and admits several connected components indexed by the distribution of the boundary marked points on the boundary components. Note that $\overline{g}$ and $\overline{\chi}$ are respectively the genus and the Euler characteristic of the surface doubled along its boundary. This moduli space admits a compactification $\overline{\mathcal{M}}_{g,n;b,m}$ on which one can seek to define and calculate intersection numbers. Denoting $\mathbb{L}_i$ the cotangent line bundle at the $i$th interior marked point, and according to the statement of \cite[Theorem 1.1]{ABT17} --- whose proof by Solomon and Tessler has not yet appeared --- it is possible to define
\begin{equation}
\label{withoutb} \big \langle \tau_{d_1}^{\circ} \cdots \tau_{d_n}^{\circ} (\tau_0^{\partial})^{m} \big\rangle_{g,n;b,m} = 2^{1 - g - \frac{b + m}{2}} \int_{\overline{\mathcal{M}}_{g,n;b,m}} e\Big(\bigoplus_{i = 1}^n \mathbb{L}_{i}^{\oplus d_i},\mathrm{s}\Big) \in \mathbb{Q}\,,
\end{equation}
whenever $d_1,\ldots,d_n \geq 0$ are such that
\[
D = \sum_{i = 1}^n 2d_i\,.
\]
Here, $e$ is the Euler class relative to some boundary condition $\mathrm{s}$, and \cite{PST15,ST} give suitable boundary conditions so that the number is unambiguously defined.

\begin{remark} For $b = k = 0$, the moduli space $\overline{\mathcal{M}}_{g,n;0,0}$ coincides with the moduli space of spin structures, which is a $\mathbb{Z}_2$-orbifold cover of the Deligne--Mumford moduli space of pointed Riemann surfaces $\overline{\mathcal{M}}_{g,n}$. In fact, $\overline{\mathcal{M}}_{g,n;0,0}$ has two connected components, distinguished by the parity of the corresponding spin structures. On the component containing the even (resp. odd) spin structures the degree of the cover is  $2^{g - 1}(2^g + 1)$ (resp $2^{g - 1}(2^g - 1)$), and the virtual fundamental class (= the Witten $2$-spin class) is the fundamental class of the even component minus the one of the odd components. After pushforward this yields a multiple of the fundamental class of $\overline{\mathcal{M}}_{g,n}$ by a factor of $\frac{1}{2}\big(2^{g - 1}(2^g + 1) - 2^{g - 1}(2^g - 1)\big)= 2^{g - 1}$, where the extra $\frac{1}{2}$ comes from the $\mathbb{Z}_2$-orbifold structure of $\overline{\mathcal{M}}_{g,n;0,0}$. So, the conventional factor of $2$ in \eqref{withoutb} is such that we retrieve for $b = k = 0$ the usual $\psi$-class intersection 
\[ 
\big\langle \tau_{d_1}^{\circ} \cdots \tau_{d_n}^{\circ} \big\rangle_{g,n;0,0} = \int_{\overline{\mathcal{M}}_{g,n}} \prod_{i = 1}^n\psi_i^{d_i},\qquad \psi_i = c_1(\mathbb{L}_i).
\]
We thank Ran Tessler for a remark on this point, see also \cite[Lemma 6.20 and Remark 6.21]{Tes15}.
\end{remark}

It is expected that one can also define geometrically boundary descendants, that we would like to denote
\[
\big \langle \tau_{d_1}^{\circ}\cdots \tau_{d_n}^{\circ} \tau_{k_1}^{\partial} \cdots \tau_{k_m}^{\partial} \big\rangle_{g,n;b,m}  \in \mathbb{Q}\,,
\]
where now $d_1,\ldots,d_n,k_1,\ldots,k_m \geq 0$ satisfy the dimension constraint 
\begin{equation}
\label{dimcon}
D = \sum_{i = 1}^n 2d_i + \sum_{j = 1}^{m} 2k_j\,.  
\end{equation}
or equivalently
\begin{equation}
\label{3deim} 3(2g - 2 + n + m+ b) = \sum_{i = 1}^n (2d_i +1) + \sum_{j = 1}^m (2k_j +2)\,.
\end{equation}
One could then consider the generating series
\[
Z^{{\rm open}}\big[Q;t^\circ;t^{\partial}\big] = \exp\left(\sum_{\substack{g,b,m,n \geq 0 \\ \overline{\chi} < 0}} \!\!\!\! \frac{\hbar^{g - 1 + \frac{b}{2}}\,Q^{b}}{m!n!} \big \langle \tau_{d_1}^{\circ} \cdots \tau_{d_n}^{\circ} \tau_{k_1}^{\partial} \cdots \tau_{k_m}^{\partial} \big\rangle_{g,n;b,m} \prod_{i = 1}^n (2d_i + 1)!!\,t^{\circ}_{d_i} \prod_{j = 1}^m (2k_j + 2)!!\,t^{\partial}_{k_j}\right)\,,
\]
where $b$ is the number determined from $g,n,m,d_i,k_i$ via \eqref{dimcon}.

A combinatorial model for the intersection numbers \eqref{withoutb} --- i.e. without boundary descendants --- has been proposed in \cite{ABT17}, refining \cite{Tes15} where only their sum over all possible $b$s was obtained. As a consequence, their generating series has a matrix integral description. It turns out this matrix integral allows naturally for the insertion of extra parameters, thus defining a generating series of the form
\[  
Z^{{\rm ABT}}\big[Q;t^{\circ};t^{\partial}\big] = \exp\left(\sum_{\substack{g,b,m,n \geq 0 \\ \overline{\chi} < 0}} \!\!\!\! \frac{\hbar^{g - 1 + \frac{b}{2}} Q^{b}}{m!n!} \big\langle \tau_{d_1}^{\circ} \cdots \tau_{d_n}^{\circ} \tau_{k_1}^{\partial} \cdots \tau_{k_m}^{\partial} \big\rangle_{g,n;b,m}^{{\rm ABT}} \prod_{i = 1}^n (2d_i + 1)!!\, t^{\circ}_{d_i} \prod_{j = 1}^m (2k_j + 2)!!\,t^{\partial}_{k_j}\right)\,.  
\]
We will not need its precise definition, which can be found in \cite[Equation 3.14 and Lemma 3.2]{ABT17}.

The Kontsevich-Penner matrix model is another important character of the story. It is defined by
\[
\mathcal{Z}^{{\rm KP},N}(\Lambda) = \int_{\mathcal{H}_{N}} \frac{\dd H}{c_{\Lambda,N,\hbar}}\,\exp\bigg\{\hbar^{-\frac{1}{2}} {\rm Tr}\bigg(\frac{H^3}{6} - \frac{H^2\Lambda}{2}\bigg)\bigg\} \bigg(\frac{\det\big(\Lambda)}{\det(\Lambda - H)}\bigg)^{Q}\,, \\
\]
where $\mathcal{H}_{N}$ is the space of hermitian matrices of size $N$ and $c_{\Lambda,N,\hbar}$ is some normalizing constant. It determines a unique generating series of the form
\[
Z^{{\rm KP}}\big[Q;t^{\circ};t^{\partial}\big] = \exp\left(\sum_{\substack{g,b,m,n \geq 0 \\ \overline{\chi} < 0}} \!\!\!\! \frac{\hbar^{g - 1 + \frac{b}{2}}\,Q^{b}}{m!n!} \big\langle \tau_{d_1}^{\circ} \cdots \tau_{d_n}^{\circ} \tau_{k_1}^{\partial} \cdots \tau_{k_m}^{\partial} \big\rangle_{g,n;b,m}^{{\rm KP}} \prod_{i = 1}^{n} (2d_i + 1)!!\,t_{d_i}^{\circ}  \prod_{j = 1}^m (2k_j + 2)!!\,t_{k_j}^{\partial}\right)
\]
such that for any $N \geq 1$
\[
\mathcal{Z}^{{\rm KP}}_{N}(\Lambda) = Z^{{\rm KP}}\bigg[Q;\Big(t_d^{\circ} = \frac{\hbar^{\frac{1}{2}}\,{\rm Tr}\,\Lambda^{-(2d + 1)}}{2d + 1}\Big)_{d \geq 0}\,;\,\Big(t_k^{\partial} = \frac{\hbar^{\frac{1}{2}}\,{\rm Tr}\,\Lambda^{-(2k + 2)}}{2k + 2}\Big)_{k \geq 0}\bigg]\,.
\]

\begin{remark} In contrast with the aforementioned works, we have included in our definition of the generating series a variable $\hbar$ which is redundant because of the dimension constraint \eqref{dimcon}, and we have not written separately the contribution of the closed Riemann surfaces ($b = 0$). Our definition can be obtained from \cite{ABT17} by the following substitutions:
\begin{equation*}
\begin{split}
&  N \rightarrow Q,\qquad H \rightarrow \hbar^{-\frac{1}{6}}H,\qquad \Lambda \rightarrow \hbar^{-\frac{1}{6}}\Lambda,\qquad t_{d} \rightarrow \hbar^{\frac{d - 1}{3}}\,(2d + 1)!!\,t_{d}^{\circ},\qquad s_{k} \rightarrow \hbar^{\frac{k}{3} - \frac{1}{6}}\,(2k + 2)!!\,t_{k}^{\partial} \,,
\end{split}
\end{equation*}
under which $\tau^{o,ext}_N \rightarrow Z^{{\rm ABT}}$ and $\tau_N \rightarrow Z^{{\rm KP}}$. Our definition can be obtained from \cite{Ale15} by substituting there
\[
N \rightarrow Q,\qquad T_{2d + 1} \rightarrow \hbar^{\frac{d - 1}{3}}\,t_{2d + 1}^{\circ},\qquad T_{2k + 2} \rightarrow \hbar^{\frac{k}{3} - \frac{1}{6}}\,t_{k}^{\partial}\,.
\]
Finally, note that the definition of the Kontsevich-Penner matrix model of \cite{ABT17} can be obtained from the one in \cite{Ale15} by the substitution $\Phi \rightarrow \Lambda - H$.
\end{remark}

It is expected that these three collections of numbers coincide.
\begin{conjecture}
\label{conj1}There exists a geometric definition of the open intersection numbers with boundary descendants, and it is such that $Z^{{\rm open}} = Z^{{\rm ABT}}$, that is 
\[
\big\langle \tau_{d_1}^{\circ} \cdots \tau_{d_n}^{\circ} \tau_{k_1}^{\partial} \cdots \tau_{k_m}^{\partial} \big\rangle_{g,n;b,m} =  \big\langle \tau_{d_1}^{\circ} \cdots \tau_{d_n}^{\circ} \tau_{k_1}^{\partial} \cdots \tau_{k_m}^{\partial} \rangle_{g,n;b,m}^{{\rm ABT}}\,.
\]
\end{conjecture}
\begin{conjecture}
\label{conj2}There exists a geometric definition of the open intersection numbers with boundary descendants, and it is such that $Z^{{\rm open}} = Z^{{\rm KP}}$, that is
\[
 \big\langle \tau_{d_1}^{\circ} \cdots \tau_{d_n}^{\circ} \tau_{k_1}^{\partial} \cdots \tau_{k_m}^{\partial} \big\rangle_{g,n;b,m} = \big\langle \tau_{d_1}^{\circ} \cdots \tau_{d_n}^{\circ} \tau_{k_1}^{\partial} \cdots \tau_{k_m}^{\partial} \rangle_{g,n;b,m}^{{\rm KP}}\,.
\]
\end{conjecture}
\begin{conjecture}
\label{conj3} We have $Z^{{\rm KP}} = Z^{{\rm ABT}}$.
\end{conjecture}
Obviously, any two of the conjectures imply the third one. They are supported by partial results:
\begin{itemize}
\item The specialisation of \cref{conj1} to $t_{k}^{\partial} = 0$ for $k > 0$ is proved in \cite{ABT17} conditionally to \cite[Theorem 1.1]{ABT17} whose proof was announced by Solomon and Tessler but has not appeared yet. The same specialisation in \cref{conj2} was anticipated in \cite{Saf16} and proved there for $g = 0,\tfrac{1}{2},1$.
\item The specialisation of \cref{conj3} to $Q = 1$ is proved in \cite{Ale15} via integrability techniques and for $Q = \pm 1$ in \cite[Section 4.2]{ABT17} by matrix integral techniques.
\item The string and dilaton equations satisfied by $Z^{{\rm ABT}}$ and $Z^{{\rm KP}}$ are the same \cite[Section 4.3]{ABT17}.
\end{itemize}

Alexandrov has proposed various collections of differential operators relevant to the study of the Kontsevich--Penner model --- and therefore to open intersection theory in light of \cref{conj2}. To summarise what is relevant for our exposition:
\begin{enumerate}
\item[(a)]  in \cite{Ale15}, Alexandrov uses a representation of the $\mathfrak{gl}_1$-Heisenberg algebra to construct operators $(\widehat{L}_k^o)_{k \geq 0}$ and $(\widehat{M}_{k}^{o})_{k \geq -2}$ --- see Equations (7.4) and (7.14) therein --- annihilating the $Q = 1$ specialization of $Z^{{\rm KP}}$.
\item[(b)] in \cite{Saf16}, Safnuk introduced a modification of these operators, denoted $(\widehat{L}_k)_{k \geq -1}$ and $(\widehat{M}_{k})_{k \geq -2}$ --- see Equations (2.9) and (2.10) therein --- which still annihilate the $Q = 1$ specialization of $Z^{{\rm KP}}$.
\item[(c)] \label{Ale17constraints} in \cite{Ale17}, Alexandrov uses a twisted representation of the Heisenberg algebra of $\mathfrak{gl}_3$ to construct a free field representation of $\mc{W}(\mathfrak{sl}_3)$ and operators $(\widehat{\mathcal{L}}_{k}^{Q})_{k \geq -1}$ and $(\widehat{\mathcal{M}}_{k}^{Q})_{k \geq -2}$ --- see Equation (72) therein --- annihilating $Z^{{\rm KP}}$.
\end{enumerate}

All those operators are related by taking (possibly infinite) linear combinations, but it turns out choosing one or the other set of operators affects the structure of the recursion one deduces for $\langle \cdots \rangle^{{\rm KP}}$. For sake of comparison and completeness, we review in \cref{Safnukreviewopen} the definition of the operators in  (a) and (b) and the topological recursion with strange features that Safnuk derived from the operators in (b). In \cref{TRopenrel}, we explain that the operators in (c) directly compare to Airy structures for $\sigma = (12)(3)$ and thus provide a CEO-like topological recursion, whose structure is more transparent and more general than \cite{Saf16}.

\medskip

\subsection{Review of Safnuk's recursion}

\label{Safnukreviewopen}

\medskip

Consider the following representation of the Heisenberg VOA for $\mathfrak{gl}_1$
\[
J(z) = \sum_{k \in \mathbb{Z}} \frac{J_k}{z^{k + 1}},\qquad J_k \coloneqq \left\{\begin{array}{lll} \partial_{t_{k}} & & {\rm if}\,\,k > 0 \\ 0 & & {\rm if}\,\,k = 0 \\ -kt_{-k} & & {\rm if}\,\,k < 0 \end{array}\right.\,,
\] 
and introduce the normal ordered products
\begin{equation*}
\begin{split}
L(z) & = \frac{1}{2}\,\norder{J(z)^2} = \sum_{k \in \mathbb{Z}} \frac{L_k}{z^{k + 2}}\,, \\
M(z) & = \frac{1}{3}\,\norder{J(z)^3} = \sum_{k \in \mathbb{Z}} \frac{M_{k}}{z^{k + 3}}\,.
\end{split}
\end{equation*}
These operators form a representation of the $\mc{W}(\mathfrak{sl}_3)$-algebra. The collection of operators mentioned in (a) are:
\begin{equation*}
\begin{split}
\widehat{L}_k^o &\coloneqq  L_{2k} + (k+2) J_{2k} - J_{2k+3} + \delta_{k,0} \Big(\frac{1}{8} + \frac{3}{2} \Big)\,, \\
\widehat{M}_k^o &\coloneqq  M_{2k} + 2(k+3)L_{2k} -2L_{2k+3} -2(k+3) J_{2k+3} + J_{2k+6} + \Big( \frac{95}{12} + 6k + \frac{4}{3}k^2\Big) J_{2k} + \frac{23}{4} \delta_{k,0} \,.
\end{split}
\end{equation*}
Then Alexandrov proved in~\cite{Ale15} that $(\widehat{L}_k^o)_{k \geq 0}$ and $(\widehat{M}_{k}^o)_{k \geq -2}$ annihilate the specialization of $Z^{{\rm KP}}$
\begin{equation*}
\left\{\begin{array}{lll} \hat{L}_k^o \tau_1= 0 && k \geq 0\,,\\
\hat{M}_k^o \tau_1= 0 &&k \geq -2\,. \end{array}\right.
\end{equation*}

The collection of modified operators mentioned in (b) read:
\[
\widehat{L}_k \coloneqq  \hat{L}_k^o\,, \qquad \widehat{M}_k \coloneqq  - \widehat{M}_k^o + 2(k+2) \widehat{L}_k^o\,.
\]
Let us sketch the strategy of Safnuk in \cite{Saf16} --- for a better comparison, we set $\hbar = 1$ till the end of this paragraph. We first rewrite these operators. To this end, we change the currents to include as zero-mode $Q = 1$ and a dilaton shift (Safnuk does this later in the computation):
\[
\tilde{J}(z) \coloneqq  \sum_{k\in \Z} \frac{\tilde{J}_k}{z^{k+1}} \,, \qquad \tilde{J}_k \coloneqq  J_k - \delta_{k,-3} + \delta_{k,0}\,.
\]
Let us also define, following Safnuk, the differential operators
\[
\mc{D}_1 \coloneqq  \dd z \Big( -\frac{\dd}{\dd z} + \frac{1}{z} \Big) \,; \qquad \mc{D}_2 \coloneqq  \frac{(\dd z)^2}{2} \Big( \frac{\dd^2}{\dd z^2} -\frac{3}{z} \frac{\dd}{\dd z} + \frac{3}{z^2} \Big)\,
\]
the one-form $ \eta = -z^2\dd z$, and the projection operators
\[
\mc{P}^{(i)} \colon \C \llbracket z^{\pm 1} \rrbracket \dd z \to  \C \llbracket z^{-2} \rrbracket z^{-i} \dd z\, \qquad i \in \{2,3\}\,.
\]
Then, taking only the parts of the generating series that annihilate $Z^{{\rm KP}}|_{Q = 1}$, we get
\begin{equation*}
\begin{split}
\mc{J}(z) &\coloneqq  \tilde{J}(z)\dd z\,,\\
\mc{L}(z) &\coloneqq  \sum_{k \geq -1} \frac{\dd z}{z^{2k+4}} \hat{L}_k = \mc{P}^{(2)} \bigg( \frac{1}{2\eta } \Big( \norder{\mc{J}^2} + \mc{D}_1\mc{J} + \frac{(\dd z)^2}{4z^2} \Big) \bigg)\,,\\
\mc{M}(z) &\coloneqq  \sum_{k \geq -2} \frac{\dd z}{z^{2k+7}} \hat{M}_k  = \mc{P}^{(3)} \bigg( \frac{1}{3\eta^2} \Big( \norder{\mc{J}^3} - \mc{D}_2\mc{J} + \frac{3(\dd z)^2}{4z^2} \mc{J} \Big) \bigg)\,.
\end{split} 
\end{equation*}
The last term in the last two lines can be absorbed by \emph{defining}
\[
\label{thesquareZ} \mc{J}^2(z) = \norder{\mc{J}^2(z)} + \frac{(\dd z)^2}{4z^2}\,,\qquad \mc{J}^{3}(z) = \mc{J}(z)\cdot \mc{J}^2(z)\,.
\]
The square of $\mathcal{J}(z)$ itself will not be defined due to infinite sums, and somehow the definition \eqref{thesquareZ} implements the right operator product expansion from the $\mc{W}(\mathfrak{sl}_3)$-algebra, cf. \cref{Wikz}. We then get the following operators annihilating $Z^{{\rm KP}}|_{Q = 1}$:
\[
\mc{L} = \mc{P}^{(2)} \bigg( \frac{1}{2\eta } \Big( \mc{J}^2+ \mc{D}_1\mc{J} \Big) \bigg)\,, \qquad \mc{M} = \mc{P}^{(3)} \bigg( \frac{1}{3\eta^2} \Big( \mc{J}^3 - \mc{D}_2\mc{J} \Big) \bigg)\,.
\]
Now, if we write $Z^{{\rm KP}}|_{Q = 1} = e^F$, and commute this through these operators, this gives the equations
\begin{equation*}
\begin{split}
0&= \mc{P}^{(2)} \bigg( \frac{1}{2\eta } e^{-F} \Big( \mc{J}^2+ \mc{D}_1\mc{J} \Big) e^F\bigg) = \mc{P}^{(2)} \bigg( \frac{1}{2\eta } \Big( U^2+ \mc{D}_1 U \Big) \cdot 1 \bigg)\,,\\
0&= \mc{P}^{(3)} \bigg( \frac{1}{3\eta^2}e^{-F} \Big( \mc{J}^3 - \mc{D}_2\mc{J} \Big) e^F\bigg)= \mc{P}^{(3)} \bigg( \frac{1}{3\eta^2} \Big( U^3 - \mc{D}_2U \Big)\cdot 1 \bigg)\,,
\end{split}
\end{equation*}
where 
\[
U(z) \coloneqq  e^{-F} \mc{J}(z) e^F = \mc{J}(z) + [\mc{J}(z),F]
\]
is the operator appearing in \cite{Saf16}. In order to recover a spectral curve topological recursion from this, one should define
\[
 \delta_{z} \coloneqq  \mc{J}_-(z) = \sum_{k\geq 1} \frac{\dd z}{z^{k+1}}\,\partial_{t_{k}},\qquad  \omega_{g,n}(z_1,\ldots,z_n) \coloneqq  \delta_{z_1} \cdots \delta_{z_n} F_{g,n}\,.
 \]
If we also introduce the unstable terms
\[
\omega_{0,1}(z) \coloneqq  \eta(z)\,, \qquad \omega_{0,2} = [\delta_1, \mc{J}_2 ] = \frac{\dd z_1 \dd z_2}{(z_1-z_2)^2} \,, \qquad \omega_{\frac{1}{2},1} \coloneqq \frac{\dd z}{z} \,,
\]
coming respectively from the dilaton shift, the positive part of $\mc{J}(z)$ and the zero mode $J_0 = 1$, we see that
\[
U(z) = \delta_{z} + \omega_{0,1}(z) + \omega_{\frac{1}{2},1}(z) + \delta^{-1} \omega_{0,2} + \sum_{2g - 2 + n > 0} \delta F_{g,n}\,.
\]
Reinterpreting the projection operators as residues with the recursion kernel results in
\begin{theorem}[{\cite[Theorem~5.3]{Saf16}}]
The $\omega_{g,n}$ obey a modified topological recursion on the spectral curve with crosscap form
\[
\C,\qquad  x(z) = \frac{z^2}{2},\qquad  y(z) = -z\,,\qquad \omega_{0,2}(z_1,z_2) = \frac{\dd z_1 \dd z_2}{(z_1-z_2)^2}\,,\qquad \omega_{\frac{1}{2},1} = \frac{\dd z}{z}\,,
\]
given by
\begin{equation*}
\begin{split}
\omega_{g,n+1}(z_0,z_{[n]}) = \Res_{w = 0}& \bigg( K^{(2)}(z_0,w) \Big( \mc{W}'_{g,2,n}(w,w;z_{[n]}) + \mc{D}_1 \omega_{g-\frac{1}{2},n+1}(w,z_{[n]}) \Big) \\
&+  K^{(3)}(z_0,w) \Big( \mc{W}'_{g,3,n}(w,w,w;z_{[n]}) + \mc{D}_2 \omega_{g-1,n+1}(w,z_{[n]}) \Big) \bigg)\,,
\end{split}
\end{equation*}
where $\mc{W}'$ is as in \Cref{DisconnConnCor} and
\[
K^{(j)}(z,w) = \Big( (-1)^j \int_{\zeta = 0}^{-w} \omega_{0,2}(z,\zeta ) - \int_{\zeta = 0}^w \omega_{0,2}(z,\zeta) \Big) \frac{1}{2j (-w^2\dd w)^{j-1}}\,.
\]
\end{theorem}

This form of the topological recursion has some odd features. For one, it is a recursion of order $3$, even though the degree of $x$ is only $2$. This is also noted by Safnuk, who computes the quantum curve in \cite[Theorem~7.1]{Saf16}, and obtains the semi-classical limit $ y(y^2-2x) =0$. Moreover, the formula does not include summation over fibers of $x$ --- rather, the variable $ w$ is inserted several times. Finally, the operators $\mc{D}_j$ introduce uncommon derivatives. Safnuk posits that these quirks come from the reducibility of the curve, but in the next section, we will see this relation is not straightforward.

\medskip

\subsection{Relation to topological recursion}
\label{TRopenrel}

\medskip

In a previous work, it was shown that the operators in (c), i.e. in \cite{Ale17}, coincide, up to a change of variables, with the reduction to $\mc{W}(\mathfrak{sl}_3)$ of an $\mc{W}(\mathfrak{gl}_3)$-Airy structure, leading to the following result.

\begin{theorem}[{\cite[Proposition~6.3]{BBCCN18}}]
\label{ZKPAiry} $Z^{{\rm KP}}$ is the partition function of the Airy structure of \cref{thm:W_gl_Airy_arbitrary_autom} with parameters $(r_1,s_1,r_2,s_2) = (2,3,1,\infty )$, $t_1 = \tfrac{1}{2}$, and $Q_1 = -Q_2 = Q$, specialised to
\begin{equation}
 \label{specin} x^1_{2d+1} = t_{d}^{\circ},\qquad x^1_{2d + 2} = t_{d}^{\partial},\qquad x^2_{d + 1} = 0,\qquad (d \geq 0) \,.
\end{equation}
\end{theorem}
Thanks to Propositions~\ref{HASequivALE} and \ref{mainth2}, we can now convert this into a CEO-like topological recursion on the reducible spectral curve consisting of the union of two components intersecting at $z = 0$
\begin{equation}
\label{spreducible} C = C_1 \cup C_2\,,\qquad C_1\,: \left\{\begin{array}{l} x(z) = z^2 \\ y(z) = -\frac{z}{2} \end{array}\right.,\qquad C_2\,:\,\left\{\begin{array}{l} x(z) = z \\ y(z) = 0 \end{array}\right.\,.
\end{equation}
Due to the constraint $H_{i = 1,k=d + 1} \cdot Z = 0$ in the Airy structure, the partition function only depends on the variable $x^1_{2d + 2} - x^2_{d + 1}$. Accordingly, no information is lost after the specialisation \eqref{specin}. More precisely, each occurrence of $x^1_{2d + 1}$ in a monomial is associated with the insertion of $(2d + 1)!! \tau^{\circ}_{d}$, each occurrence of $x^2_{d + 1}$ with the insertion of $(2d + 2)!!\tau^{\partial}_{d}$, and each occurrence of $x^1_{2d + 2}$ with $- (d+ 1)!!\tau_d^{\partial}$. Taking this into account and using \eqref{omgnxp}, the suitable definition for the correlators in terms of open intersection numbers is
\begin{equation*}
\begin{split}
& \quad \omega_{h,n + m}^{{\rm KP}}\big(\begin{smallmatrix} 1 & \cdots & 1 & 2 & \cdots & 2 \\ z_1 & \cdots & z_n  & z_{n + 1} & \cdots & z_{n + m}\end{smallmatrix}\big) \\
& = \sum_{\substack{g,b \in \mathbb{Z}_{\geq 0} \\ g + \frac{b}{2} = h}} \sum_{N_{\circ} \sqcup N_{\partial} = [n]} \sum_{\substack{d_1,\ldots,d_{n} \geq 0 \\ k_1,\ldots,k_{m} \geq 0}} (-1)^{m}\,Q^{b}\,\bigg\langle \prod_{i \in N_{\circ}} \tau^{\circ}_{d_i} \prod_{i \in N_{\partial}} \tau^{\partial}_{d_i} \prod_{j = 1}^m \tau_{k_j}^{\partial}\bigg\rangle_{g,|N_{\circ}|;b,m + |N_{\partial}|}^{{\rm KP}} \\
& \qquad \cdot \prod_{i \in N_{\circ}} \frac{(2d_i + 1)!!\,\dd z_i}{z_i^{2d_i + 2}}\prod_{i \in N_{\partial}} \frac{(d_i + 1)!\,\dd z_i}{z_i^{2d_i + 3}}\prod_{j = 1}^m \frac{(k_j + 1)!\,\dd z_{n + j}}{z_{n + j}^{k_j + 2}}\,,
\end{split}
\end{equation*} 
where the number of boundaries $b$ is determined in terms of $g,n,d,k,m$ via \eqref{dimcon}.

\begin{corollary}
\label{coroKP} For $g \in \tfrac{1}{2}\mathbb{Z}_{\geq 0}$ and $n \geq 1$ such that $2g - 2 + n > 0$, $\omega_{g,n}^{{\rm KP}}$ is computed by the topological recursion on the spectral curve \eqref{spreducible} equipped with bidifferential and crosscap form
\[
\omega_{0,2}\big(\begin{smallmatrix} \mu_1 & \mu_2 \\ z_1 & z_2 \end{smallmatrix}\big) = \delta_{\mu_1,\mu_2}\,\frac{\dd z_1 \dd z_2}{(z_1 - z_2)^2}\,,\qquad \omega_{\frac{1}{2},1}\big(\begin{smallmatrix} \mu \\ z \end{smallmatrix}\big) = (-1)^{\mu + 1}\,Q\,\frac{\dd z}{z}\,.
\]
for $\mu,\mu_1,\mu_2 \in \{1,2\}$.
\end{corollary}
 
Several sanity checks of this corollary can be proposed. At $Q = 0$, the variables $t^{\partial}$ become irrelevant and $Z^{{\rm KP}}$ specialises to the Witten-Kontsevich partition function \cite{Wit91,Kont92}
\begin{equation}
\label{ZKWQ0} Z^{{\rm KP}}[Q = 0,t^{\circ},t^{\partial}] = \sum_{\substack{g \geq 0,\,\,n \geq 1 \\ 2g - 2 + n > 0}} \frac{\hbar^{g - 1}}{n!} \bigg(\int_{\overline{\mathcal{M}}_{g,n}} \prod_{i = 1}^n \psi_i^{d_i}\bigg) \prod_{i = 1}^n (2d_i + 1)!!\,t_{d_i}^{\circ}\,.
\end{equation}
This can also be checked on the topological recursion side. Indeed, \cref{rspinthmTR} for $r=2$ states that the topological recursion for the spectral curve
\begin{equation}
\label{xyzeta} x(\zeta) = \zeta^2\,,\qquad y(\zeta) = -\frac{\zeta}{2}\,,\qquad \omega_{0,2}(\zeta_1,\zeta_2) = \frac{\dd \zeta_1\dd \zeta_2}{(\zeta_1 -\zeta_2)^2}
\end{equation}
produces the $(g,n)$-correlator for $2g - 2 + n > 0$:
\[
\omega_{g,n}(\zeta_1, \dotsc, \zeta_n) = \sum_{d_1,\ldots,d_n \geq 0} \bigg(\int_{\overline{\mathcal{M}}_{g,n}} \prod_{i = 1}^n \psi_i^{d_i}\bigg) \prod_{i = 1}^n \frac{(2d_i + 1)!!\dd \zeta_i}{\zeta_i^{2d_i + 2}}\,.
\]
Note that this case was proved before with a slightly different normalisation by Eynard--Orantin \cite{EyOr07}, see also \cite{Ebook}.\par
At $Q = 0$ we are in position to apply \cref{propdecouple}, showing that the second component decouples, i.e. $\omega_{g,n}^{{\rm KP}}\big(\begin{smallmatrix} 1 & \cdots & 1 \\ z_1 & \cdots & z_n \end{smallmatrix}\big)\big|_{Q = 0}$ coincides with the correlators of \eqref{xyzeta}. So, the correlators of the spectral curve \eqref{spreducible} at $ Q = 0$ agree with the correlators associated to \eqref{ZKWQ0}, as predicted by the corollary.

\medskip
 
\subsection{Review of open \texorpdfstring{$r$}{r}-spin theory}

\medskip

It is expected that there exists an open analog of Witten $r$-spin theory, related to a space $\overline{\mathcal{M}}_{g,n;b,m}^{r{\rm spin}}$, which specialises to the open theory of \cref{OpenIntKP} for $r = 2$. It would give for each $(g,n,b,m)$ such that $-\overline{\chi} > 0$ a collection of numbers
\[ 
\big\langle\tau_{d_1}^{\circ}(a_1) \cdots \tau_{d_n}^{\circ}(a_n) \tau_{k_1}^{\partial} \cdots \tau_{k_m}^{\partial}\big\rangle_{g,n;b,m}^{r{\rm spin}} \in \mathbb{Q}
\] 
indexed by $d_i,k_i \geq 0$ and $a_i \in [r]$, which we can collect in a generating series:
\begin{equation}
\label{openrspinZ}
\begin{split}
& \quad Z^{{\rm open}\,r{\rm spin}}\big[Q;t^{\circ};t^{\partial}\big] \\
& = \exp\left(\sum_{\substack{g,b,n,m \geq 0 \\ -\overline{\chi} > 0}} \!\!\frac{\hbar^{g - 1 + \frac{b}{2}}Q^{b}}{m!n!} \!\! \sum_{\substack{d_1,\ldots,d_n \geq 0 \\ a_1,\ldots,a_n \in [r] \\ k_1,\ldots,k_m \geq 0}} \! \Big\langle \prod_{i = 1}^n \tau_{d_i}^{\circ}(a_i) \prod_{j = 1}^m \tau_{k_j}^{\partial}\Big\rangle_{g,n;b,m}^{r{\rm spin}} \prod_{i = 1}^n (d_ir + a_i)!^{(r)}t_{a_i,d_i}^\circ  \prod_{j = 1}^m (rk_j + r)!^{(r)} t_{k_j}^\partial \right)\,.
\end{split}
\end{equation}

These numbers should vanish unless
\begin{equation}
\label{thedimopenr}(r + 1)(2g - 2 + b + n + m) = \sum_{i = 1}^n (rd_i + a_i) + \sum_{j = 1}^m (rk_j + r)\,.
\end{equation}
Besides, for $m \geq 1$, each insertion\footnote{For $d = 0$ this matches the normalisation \cite[Theorem 1.5 and Section 6.1]{BCT18}, namely $\tau_0^{\circ}(r) \leftrightarrow (-1/r)\tau_0^{\partial}$. The $\tau_d^{\partial}$ correspond to boundary descendent insertion: as they have not been defined geometrically yet, we do not have a natural normalisation to compare to. The factor $r^{d + 1}$ is natural from the numerical perspective in combination with the identity $(dr + r)!^{(r)} = r^{d + 1} (d + 1)!$.}  of $\tau_{d}^\circ(r)$ should amount to an insertion of $(-1/r^{d + 1})\tau^\partial_d$. 
 
There are several possible choices of conventions (in particular, for orientations) that could affect these numbers by a prefactor depending only on the topology. We fix them by the normalisation of the consistency relations with the intersection numbers that are already defined. For $r = 2$ and $m \geq 1$, we want to retrieve the open intersection numbers of \cref{OpenIntKP}
\begin{equation}
\label{compare2pin} \bigg\langle \prod_{i = 1}^{n} \tau_{d_i}^\circ(1) \prod_{j = 1}^m \tau_{k_j}^\partial \bigg\rangle_{g,n;b,m}^{2{\rm spin}} = \bigg\langle \prod_{i = 1}^n \tau^\circ_{d_i} \prod_{j = 1}^m \tau_{k_j}^\partial \bigg\rangle_{g,n;b,m}\,.
\end{equation}
Notice that the dimension constraint \eqref{3deim} forces $b + m$ to be even, so this is indeed an identity in $\mathbb{Q}$. In absence of boundaries $b = m = 0$, we want to retrieve the Witten $r$-spin class intersections of \cref{revwit}
\[
\big\langle \tau_{d_1}^{\circ}(a_1) \cdots \tau_{d_n}^{\circ}(a_n) \big\rangle_{g,n;0,0}^{r{\rm spin}} = \int_{\overline{\mathcal{M}}_{g,n}} w_{g,n}^{r{\rm spin}}(a_1,\ldots,a_n) \prod_{i = 1}^n \psi_i^{d_i}\,.
\]

For disks without boundary descendants --- that is $(g,b) = (0,1)$ and $k_j = 0$ --- the open $r$-spin intersection numbers have been defined in \cite{BCT20} in the form
\begin{equation}
\label{bibibibibib}\big \langle \tau_{d_1}^{\circ}(a_1)\cdots \tau_{d_n}^{\circ}(a_n) (\tau_0^{\partial})^m \big\rangle_{0,n;1,m} = \int_{\mathcal{P}\overline{\mathcal{M}}_{0,n;1,m}^{r{\rm spin}}} e\bigg(\mathbb{W} \oplus \bigoplus_{i = 1}^n \mathbb{L}_{i}^{\oplus d_i},{\rm s}\bigg)\,,
\end{equation}
where $\mathcal{P}\overline{\mathcal{M}}$ is a partial compactification of the moduli space of $r$-spin disks and $\mathbb{W}$ is a bundle which is the open analogue of $R^1\pi_*\mathcal{L}$. We stress that $a_i$ in \cite{BCT18,BCT20} corresponds to our $a_i - 1$. These numbers are computed explicitly for $d_i = 0$ in \cite[Theorem 1.2]{BCT18}, and in particular
\[
\big\langle \tau_0^\circ(a) \tau_0^\partial \big\rangle_{0,1;1,1} = \delta_{a,1}\,.
\]
The dimension constraint \eqref{thedimopenr} is the natural generalisation of \cite[Section 6.2.1]{BCT18} allowing boundary descendants, and coincides with \eqref{3deim} for $r = 2$.

Bertola and Yang have constructed in \cite{BY15} a particular solution of the extended $r$-KdV hierarchy, generalising the $r = 2$ construction of \cite{Bur15}. Up to a change of normalisation, this solution is mentioned in  \cite{BCT18} under the name $\Phi$ and depends on a redundant parameter $\varepsilon$ and times $(T_k)_{k>0}$. We shall use the latter normalisation, and for uniformity denote it  $Z^{r{\rm BY}}[\varepsilon ; (T_k)_{k > 0}]$. It gives, for each $(\overline{g},n,m)$ such that $\overline{\chi} > 0$ (see \eqref{chibarE}), a collection of numbers
\[
\big\langle \tau_{d_1}^\circ(a_1)\cdots\tau_{d_n}^\circ(a_n)\tau_{k_1}^\partial\cdots \tau_{k_m}^\partial\big\rangle^{r{\rm BY}}_{\overline{g};n;m}
\]
by writing down the following expansion:
\begin{equation*}
\begin{split}
& \quad Z^{r{\rm BY}}\Big[\varepsilon = (-r\hbar)^{\frac{1}{2}}\,;\, \big(T_{dr + a} = (-r)^{d+\frac{1}{2} -\frac{3(dr + a)}{2(r + 1)}} (t_{a,d}^\circ - \delta_{a,r}r^{d + 1}t_{d}^\partial)\big)_{\substack{a \in [r] \\ d \geq 0}}\Big] \\
& = \exp\left(\sum_{\substack{\overline{g},n,m \geq 0 \\ -\overline{\chi} > 0}} \frac{\hbar^{\frac{\overline{g} - 1}{2}}}{m!n!}  \sum_{\substack{d_1,\ldots,d_n \geq 0 \\ a_1,\ldots,a_n \in [r] \\ k_1,\ldots,k_m \geq 0}} \bigg\langle \prod_{i = 1}^n \tau_{d_i}^{\circ}(a_i)\prod_{j = 1}^m \tau_{k_j}^{\partial} \bigg\rangle_{\overline{g};n;m}^{r{\rm BY}} \prod_{i = 1}^n (rd_i + a_i)!^{(r)}t_{a_i,d_i}^\circ  \prod_{j = 1}^m (rk_j + r)!^{(r)} t_{k_j}^\partial \right)\,.
\end{split}
\end{equation*}
In absence of an extra variable in $Z^{r{\rm BY}}$ playing the role that $Q$ has in \eqref{openrspinZ}, one cannot define numbers depending individually on $(g,b)$, but only on the doubled genus $\overline{g} = 2g + b - 1$.

\begin{conjecture}
There exists a geometric definition of the open $r$-spin intersection numbers and it satisfies
\begin{equation}
\label{theconjBYro}\big\langle \tau_{d_1}^\circ(a_1)\cdots\tau_{d_n}^\circ(a_n)\tau_{k_1}^\partial\cdots \tau_{k_m}^\partial\big\rangle^{r{\rm BY}}_{\overline{g};n;m} = \sum_{\substack{g,b \in \mathbb{Z}_{\geq 0} \\ 2g + b - 1 =\overline{g}}} \big\langle \tau_{d_1}^\circ(a_1)\cdots\tau_{d_n}^\circ(a_n)\tau_{k_1}^\partial\cdots \tau_{k_m}^\partial\big\rangle^{r{\rm spin}}_{g,n;b,m} \,.
\end{equation}
\end{conjecture}
This conjecture is formulated in the restricted case $t^{\partial}_d = 0$ for $d > 0$ in \cite{BCT18}, perhaps because no geometric construction of boundary descendants in the open $r$-spin theory is available yet. Under this restriction, it is supported by the following results:
\begin{itemize}
\item[$\bullet$] the $r = 2$ case was proved in \cite{Bur15} and in agreement with \eqref{compare2pin} we have
$$
Z^{{\rm ABT}}\big[Q = 1;(t_d^{\circ})_{d \geq 0};(t_k^\partial)_{k \geq 0}\big] = Z^{{\rm open}\,2{\rm spin}}\big[Q = 1,(t^\circ_{1,d} = t^{\circ}_d\,,\,t^\circ_{2,d} = 0)_{d \geq 0};(t^\partial_k)_{k \geq 0}\big].
$$
\item[$\bullet$] the conjecture is proved in \cite{BCT18} for $\overline{g} = 0$ for general $r$.  In that case there is a single term $(g,b) = (0,1)$ in the right-hand side of \eqref{theconjBYro}.
\end{itemize}

 \medskip
 \subsection{Conjectural relation to topological recursion}
 \label{courserspin}
 \medskip
We now propose a direct generalisation of \cref{TRopenrel}. We consider the Airy structure of \cref{thm:W_gl_Airy_arbitrary_autom} with
\[
d = 2\,,\qquad (r_1,s_1;r_2,s_2) = (r,r+1;1,\infty)\,,\qquad Q_1 = -Q_2 = Q\,,\qquad t_1 = \frac{1}{r}\,,
\]
which is also given in \cite[Theorem 4.16]{BBCCN18}. We denote $Z^{r \star}$ its partition function and we decompose its coefficients as
\begin{equation}
\label{Fhgnng}\begin{split}
&\quad F_{h,n + m}^{r \star}\big[\begin{smallmatrix} 1 & \cdots & 1 &  2 & \cdots & 2 \\ d_1r + a_1 & \cdots & d_nr + a_n & k_1 + 1 & \cdots & k_m + 1 \end{smallmatrix}\big] \\
& = \sum_{\substack{g,b \in \mathbb{Z}_{\geq 0} \\ g + \frac{b}{2} = h}}  (-1)^{m} \, Q^b  \prod_{i = 1}^n (d_ir + a_i)!^{(r)} \prod_{j = 1}^m (k_j + 1)!  \big\langle \tau_{d_1}^\circ(a_1) \cdots \tau_{d_n}^\circ(a_n) \tau_{k_1}^\partial \cdots \tau_{k_m}^\partial \big\rangle_{g,n;b,m}^{r\star}\,.
\end{split}
\end{equation}
According to \cref{SecTR}, the corresponding correlators should be defined as
\begin{equation}
\label{cornonsun}
\begin{split}
& \quad \omega_{h,n + m}^{r\star}\big(\begin{smallmatrix} 1 & \cdots & 1 & 2 & \cdots & 2 \\ z_1 & \cdots & z_n  & z_{n + 1} & \cdots & z_{n + m}\end{smallmatrix}\big) \\
& \coloneqq \sum_{\substack{g,b \in \mathbb{Z}_{\geq 0} \\ g + \frac{b}{2} = h}} \sum_{N_{\circ} \sqcup N_{\partial} = [n]} \sum_{\substack{a\,:\, N_\circ \rightarrow [r-1] \\ d_1,\ldots,d_{n} \geq 0 \\ k_1,\ldots,k_{m} \geq 0}} (-1)^{m} \,Q^{b}\,\bigg\langle \prod_{i \in N_{\circ}} \tau^{\circ}_{d_i}(a_i) \prod_{i \in N_{\partial}} \tau^{\partial}_{d_i} \prod_{j = 1}^m \tau_{k_j}^{\partial}\bigg\rangle_{g,|N_{\circ}|;b,m + |N_{\partial}|}^{r\star} \\
& \quad \cdot \prod_{i \in N_{\circ}} \frac{(d_ir + a_i)!^{(r)}\,\dd z_i}{z_i^{d_ir + a_i + 1}}\prod_{i \in N_{\partial}} \frac{(d_i + 1)!\,\dd z_i}{z_i^{d_ir + r + 1}}\prod_{j = 1}^m \frac{(k_j + 1)!\,\dd z_{n + j}}{z_{n + j}^{k_j + 2}}\,,
\end{split}
\end{equation} 
Note that we have converted all $\tau^{\circ}_d(r)$ into $-r^{-(d + 1)}\tau^{\partial}_d$, which turned the $r$-fold factorial $(rd + r)!^{(r)} = r^{d + 1}(d + 1)!$ into a usual factorial in the corresponding factors. By \cref{mainth2}, they satisfy the topological recursion on the reducible spectral curve with two components intersecting at $z = 0$:
\begin{equation}
\label{spreducibler} C = C_1 \cup C_2\,,\qquad C_1\,: \left\{\begin{array}{l} x(z) = z^r \\ y(z) = -\frac{z}{r} \end{array}\right.,\qquad C_2\,:\,\left\{\begin{array}{l} x(z) = z \\ y(z) = 0 \end{array}\right.\,,
\end{equation}
equipped with
\[
\omega_{0,2}\big(\begin{smallmatrix} \mu_1 & \mu_2 \\ z_1 & z_2 \end{smallmatrix}\big) = \delta_{\mu_1,\mu_2}\,\frac{\dd z_1 \dd z_2}{(z_1 - z_2)^2}\,,\qquad \omega_{\frac{1}{2},1}\big(\begin{smallmatrix} \mu \\ z \end{smallmatrix}\big) = (-1)^{\mu + 1}\,Q\,\frac{\dd z}{z}\,.
\]

For comparison, let us examine the basic properties of $\langle \cdots \rangle^{r\star}$. Firstly, due to the constraint  $(J^1_{kr} + J^2_k)Z^{r\star} = 0$ for $k > 0$ and the definition \eqref{Fhgnng}, each insertion of $\tau_d^\circ(r)$ amounts to the insertion of $-r^{-(d + 1)}\tau^\partial_d$ while incrementing the number $m$ by $1$. Secondly, the constraint $H_{i = 2,k = 0} Z^{r\star} = 0$ gives, by computations similar to those of \cref{SecHom}, the string equation
\begin{equation}
\label{stringopen}\begin{split}
& \quad \bigg\langle \tau_{0}^\circ(1) \prod_{i = 1}^n \tau_{d_i}^\circ(a_i) \prod_{j = 1}^m \tau_{k_j}^\partial\bigg\rangle_{g,1+n;b,m}^{r\star} \\
& = \sum_{l = 1}^n \bigg\langle \tau_{d_l - 1}^\circ(a_l) \prod_{i \neq l} \tau_{d_i}^\circ(a_i) \prod_{j = 1}^m \tau_{k_j}^\partial \bigg\rangle_{g,n;b,m}^{r\star} + \sum_{l = 1}^m \bigg\langle \tau_{k_l - 1}^\partial  \prod_{i = 1}^n \tau_{d_i}^\circ (a_i) \prod_{j \neq l} \tau_{k_j}^\partial\bigg\rangle_{g,n:b,m}^{r\star} \\
& \quad + \delta_{g,b,m,0}\delta_{n,2}\delta_{d_1,d_2,0}\delta_{a_1 + a_2,r} + \delta_{g,m,0}\delta_{b,n,1}\delta_{d_1,0}\delta_{a_1,r} + \delta_{g,n,0}\delta_{b,m,1}\delta_{k_1,0}\,,
\end{split}
\end{equation}
with obvious vanishing conventions for insertion of negative indices. This last term gives the special value
\begin{equation}
\label{iungfisgun}
\big\langle \tau_0^\circ(1) \tau_0^\partial\big\rangle_{0,1;1,1}^{r\star} = 1\,.
\end{equation}
Note this is compatible with the computation of \eqref{eq:omega1half2_exc_cross} with the specialisation $t_1 = 1/r$. Indeed, the latter yields
\[ 
\omega_{\frac{1}{2},2}\big(\begin{smallmatrix} 1 & 2 \\ z_1 & z_2 \end{smallmatrix}\big) = - rQ\frac{\dd z_1}{z_1^2}\,\frac{\dd z_2}{z_2^2} = - Q \frac{\dd z_1}{z_1^2}\,\frac{r!^{(r)}\dd z_2}{z_2^2}\,,
\]
and thus after taking \eqref{cornonsun} into account, \eqref{iungfisgun} describes the only non-vanishing intersection number for $(g,n;b,m) = (0,1;1,1)$.  Thirdly, we have from \Cref{lem:homogen2} the dilaton equation
\begin{equation}
\label{opendilaton}\begin{split}
\bigg\langle \tau_{1}^\circ(1) \prod_{i = 1}^n \tau_{d_i}^\circ(a_i)\prod_{j = 1}^m \tau_{k_j}^\partial\bigg\rangle_{g,1+n;b,m}^{r\star} & = (2g - 2 + n + m)\bigg\langle \prod_{i = 1}^n \tau_{d_i}^\circ(a_i) \prod_{j = 1}^m \tau_{k_j}^\partial\bigg\rangle_{g,n;b,m}^{r\star} \\
& \quad + \frac{r - 1}{24}\delta_{g,1}\delta_{n,b,m,0} + \frac{1}{2}\delta_{g,n,m,0}\delta_{b,2}\,,
\end{split}
\end{equation}
and the homogeneity property which says that $\big\langle \prod_{i = 1}^n \tau_{d_i}^\circ(a_i) \prod_{j = 1}^m \tau_{k_j}^\partial \big\rangle_{g,n;b,m}^{r\star}$ vanishes unless the dimension constraint \eqref{thedimopenr} holds.

We predict that the partition function $Z^{r\star}$ describes the full open $r$-spin intersection theory in any genera and with arbitrary descendants.
\begin{conjecture}
There is a geometric definition of the open $r$-spin intersection numbers, and it satisfies
\[
\big\langle \tau_{d_1}^\circ(a_1) \cdots \tau_{d_n}^\circ(a_n) \tau_{k_1}^\partial \cdots \tau_{k_m}^\partial \big\rangle_{g,n;b,m}^{r {\rm spin}} =  \big\langle \tau_{d_1}^\circ(a_1) \cdots \tau_{d_n}^\circ(a_n) \tau_{k_1}^\partial \cdots \tau_{k_m}^\partial \big\rangle_{g,n;b,m}^{r\star} \,.
\]
\end{conjecture}

A weaker prediction involving only quantities whose definition is available at the time of writing, is that the Bertola-Yang $Z^{r{\rm BY}}$ partition function satisfies  $\mc{W}(\mathfrak{gl}_r)$-constraints with zero mode values $Q_1 = -Q_2 = 1$. Including the expected normalisations, this would translate into the following.
\begin{conjecture}
We have for $b,m > 0$
\[
\langle \tau_{d_1}^\circ(a_1) \cdots \tau_{d_n}^\circ(a_n) \tau_{k_1}^\partial \cdots \tau_{k_m}^\partial \big\rangle_{\overline{g};n;m}^{r {\rm BY}} = \sum_{\substack{g,b \in \mathbb{Z}_{\geq 0} \\ g + \frac{b}{2} = \overline{g}}} \big\langle \tau_{d_1}^\circ(a_1) \cdots \tau_{d_n}^\circ(a_n) \tau_{k_1}^\partial \cdots \tau_{k_m}^\partial \big\rangle_{g,n;b,m}^{r \star}\,.
\]
\end{conjecture}

In support of the conjectures, we see that the basic properties listed for $\langle \cdots \rangle^{r\star}$ match the ones listed for $\langle \cdots \rangle$ in the range of parameters in which the comparison is possible. The dilaton equation of \cite[Proposition 5.3]{BCT18}  matches the restriction of \eqref{opendilaton} to $(g,b) = (0,1)$ and $k_j = 0$. The string equation of \cite[Proposition 5.2]{BCT18} matches the same restriction of \eqref{stringopen}, and observe that in absence of boundary descendants the second sum in the right-hand side of \eqref{stringopen} is absent. The identification of $\tau_d^{\circ}(r)$ with $-r^{-(d +1)} \tau_d^{\partial}$ for $d = 0$  is manifest in \cite[Theorem 1.1]{BCT18}. Their extension to $g > 0$ expected in \cite[Section 6.2]{BCT18} also matches our proposal.

TRR relations involving in a linear way the open $r$-spin intersection numbers mentioned in \eqref{bibibibibib} and the closed $r$-spin intersection numbers are given in \cite[Theorem 4.1]{BCT18}. The information of \eqref{bibibibibib} should be encoded for us in $F_{\frac{1}{2},n}^{r\star}$, as it is easy to see that the $W$-constraints indeed imply some quadratic relation with a similar structure involving $F_{0,n}^{r\star}$ in a non-linear way and $F_{\frac{1}{2},n}^{r\star}$ in a linear way. As $F_{0,n}^{r\star}$  contains only information from the closed sector and satisfies $W$-constraints on its own, the structure somehow resembles the TRR relation, but establishing an exact match is left to future work.

\newpage

\printbibliography[heading=bibintoc]

\end{document}